\documentclass[11pt,a4paper,reqno]{amsart}%
\usepackage{amsthm,amsmath,amsfonts,amssymb,xcolor,amsxtra,appendix,bookmark,dsfont,bm,mathrsfs}
\usepackage[margin=35mm]{geometry}
\usepackage[latin1]{inputenc}
\usepackage{color} 
\usepackage{graphicx}
\usepackage{aliascnt}
\usepackage[shortlabels]{enumitem}
\usepackage{mathtools}
\usepackage{xparse}
\usepackage{tabularx,booktabs}
\usepackage[capitalize]{cleveref}
	\crefrangeformat{enumi}{(items~#3#1#4--#5#2#6)}
\numberwithin{equation}{section}
\usepackage{tikz}
\usepackage{etoolbox}
\usepackage{amstext}
\usetikzlibrary{calc,math}
\usepackage{hyperref}

\theoremstyle{plain}
\newtheorem{theorem}{Theorem}[section]
\newaliascnt{lemma}{theorem}
	\newtheorem{lemma}[lemma]{Lemma}
	\aliascntresetthe{lemma}
\newaliascnt{corollary}{theorem}
	\newtheorem{corollary}[corollary]{Corollary}
	\aliascntresetthe{corollary}
\newaliascnt{proposition}{theorem}
	\newtheorem{proposition}[proposition]{Proposition}
	\aliascntresetthe{proposition}

\theoremstyle{definition}
\newtheorem*{definition}{Definition}
\newaliascnt{example}{theorem}
	\newtheorem{example}[example]{Example}
	\aliascntresetthe{example}
\newtheorem{innercustomdef}{Assumption}
	\newenvironment{customdef}[1]
	{\renewcommand\theinnercustomdef{#1}\innercustomdef}
	{\endinnercustomdef}
	\Crefname{innercustomdef}{Assumption}{Assumptions}

\theoremstyle{remark}
\newaliascnt{remark}{theorem}
	\newtheorem{remark}[remark]{Remark}
	\aliascntresetthe{remark}

\newcommand{\R}{\mathbb{R}}
\newcommand{\Z}{\mathbb{Z}}
\newcommand{\C}{\mathbb{C}}
\newcommand{\N}{\mathbb{N}}
\newcommand{\T}{\mathbb{T}}

\newcommand\1{{\ensuremath {\mathds 1} }}
\newcommand{\tr}{{\rm Tr}\,}
\newcommand{\eps}{\varepsilon}
\newcommand{\nn}{\nonumber}

\DeclareMathOperator{\sgn}{sgn}

\DeclarePairedDelimiter\myp()
\DeclarePairedDelimiter\myt\{\}
\DeclarePairedDelimiter\myb[]
\DeclarePairedDelimiter\abs\lvert\rvert
\DeclarePairedDelimiter\norm\lVert\rVert
\DeclarePairedDelimiter\expec\langle\rangle
\DeclarePairedDelimiter\floor\lfloor\rfloor

\newcommand{\id}{\, \mathrm{d}}
\newcommand{\ids}{\mathrm{d}}

\newcommand\nonarg{{}\cdot{}}

\NewDocumentCommand\For{ s m }{
	\IfBooleanF{#1}{\quad}
	\quad
	\text{#2}
}

\DeclarePairedDelimiterX\Set[1]\{\}{%
	
	#1
}

\NewDocumentCommand\normt{ s o m m }{
	\IfBooleanTF{#1}{
		\norm*{#3}
	}{
		\IfNoValueTF{#2}{
			\norm{#3}
		}{
			\norm[#2]{#3}
		}
	}
	_{#4}
}

\DeclareRobustCommand{\SkipTocEntry}[5]{}

\allowdisplaybreaks

\title[On the asymptotic behavior of a weakly interacting Fermi gas]{On the asymptotic behavior at the kinetic time of a weakly interacting Fermi gas}

\author[P. S. Madsen]{Peter S. Madsen}
\address{Department of Mathematics, LMU Munich, Theresienstrasse 39, 80333 Munich, Germany}
\email{madsen@math.lmu.de}

\author[P. T. Nam]{Phan Thành Nam}
\address{Department of Mathematics, LMU Munich, Theresienstrasse 39, 80333 Munich, Germany}
\email{nam@math.lmu.de}

\author[H. Spohn]{Herbert Spohn}
\address{Departments of Mathematics and Physics, Technical University of Munich, 85747 Garching, Germany}
\email{spohn@ma.tum.de}

\author[M.-B. Tran]{Minh-Binh Tran}
\address{Department of Mathematics, Texas A\&M University, College Station, TX 77843, USA}
\email{minhbinh@tamu.edu}

\begin{document}
\date{\today}

\begin{abstract}
	This paper is devoted to the dynamics of a weakly interacting Fermi gas at the kinetic time regime $t\sim \lambda^{-2}$ where $\lambda \ll 1$ is the strength of the interaction potential.
	We prove that if the initial state is close to equilibrium, then the two-point time correlation function of the many-body quantum dynamics can be computed effectively.
	In fact, we show that its leading order behavior is determined completely by the collisional frequency of the Boltzmann-Nordheim collision operator at equilibrium.
	This settles a prediction by Lukkarinen-Spohn \cite{lukkarinen2009not}, and thus gives a justification of the quantum Boltzmann equation from many-body quantum mechanics.
\end{abstract}

\maketitle

 \tableofcontents
 
\section{Introduction}
Quantum kinetic equations are widely used in condensed matter physics.
For the first time they were written down by Nordheim \cite{Nordheim} in 1928 shortly after the discovery of quantum mechanics.
Nordheim replaced the classical two-body scattering by its quantum version, either bosons or fermions, and realized that thereby the collision operator acquires cubic terms.
He then established the H-Theorem confirming that the only stationary solutions of the transport equation are the Fermi-Dirac, resp. Bose-Einstein, equilibrium momentum distributions.
The well-known contribution by Uehling and Uhlenbeck \cite{UehlingUhlenbeck:TPI:1933} explored the Chapman-Enskog expansion for this novel transport equation describing the physics of quantum gases.

A more systematic derivation of Nordheim's equation was initiated by van Hove \cite{van1954quantum} in the early 50's, being followed up by Prigogine \cite{prigogine1962prigogine}, Balescu \cite{balescu1975equilibrium}, and Hugenholtz \cite{hugenholtz1983derivation}, to mention only a few.
Van Hove proposed to study a weakly interacting quantum fluid at nonzero density using time-dependent perturbation theory with respect to the strength $\lambda$ of the nonlinearity.
Then the mean free time is of order $\lambda^{-2}$, thereby fixing the kinetic space-time scale.
In such a set-up, the scattering cross section appears only in the Born approximation.

To guess the collision operator of the kinetic equation, it suffices to compute the second order expansion $\lambda^2 t$.
To explain why, it is of advantage to consider a spatially homogeneous initial state.
Under the noninteracting dynamics all quasifree states are stationary.
Let us thus start from one of such states.
Adding a small interaction the time evolution will generate a new state, in general no longer quasifree.
Under the free dynamics the system relaxes quickly to some quasifree state.
Nordheim's collision operator tells us how this quasifree state differs from the initial one.
Iterating the scheme yields the kinetic transport equation, as confirmed by Erd\H{o}s, Salmhofer, and Yau \cite{erdHos2008quantum}.
They prove that if in the course of time the true state stays sufficiently close to being quasifree, then Nordheim's equation follows.

To derive the Boltzmann transport equation from an underlying microscopic many-particle model is a long standing challenge in kinetic theory.
The first mathematical advance was achieved in 1975 by O.E. Lanford \cite{lanford1975time}.
He considered a classical system of hard spheres in the low density Boltzmann-Grad limit and proved the validity of the Boltzmann equation for a short kinetic time interval.
In a very recent break-through advance the extension to arbitrary long kinetic times was accomplished by Deng, Hani, and Ma \cite{DenHanMa-24pre,DenHanMa-25pre,Deng-26-ICM}.

Advancing to a mathematical analysis of the kinetic regime for quantum many-body systems, 
we consider weakly interacting spinless fermions on the lattice $\Z^d$.
This has the advantage that the field operators have norm $1$ and the momentum space is a $d$-dimensional torus.
There have been prior attempts. We refer to a series articles by 
M. Pulvirenti et al \cite{Benedetto-2004,Benedetto-2005,Benedetto-2006,Benedetto-2007}, who study quantum particles with classical statistics.
T. Chen et al \cite{CarChe-26} investigate a Fermi gas at low temperatures and prove the transport equation in a scaling limit different from the kinetic one.

A structurally similar model is the weakly nonlinear discrete Schr\"{o}dinger equation governing the complex wave field $\psi \myp{x,t}$, $x\in\Z^d$,
\begin{equation}
	\label{1}
	i \frac{d}{dt}\psi(x,t) 
	= \sum_{y \in \Z^d} \alpha \myp{x-y} \psi \myp{y,t} +\lambda \abs{\psi \myp{x,t}}^2 \psi \myp{x,t}.
\end{equation}
For this model, the kinetic limit has been studied in \cite{LukkarinenSpohn:WNS:2011}, although with yet many restrictions.
The initial data are distributed according to thermal equilibrium and the limiting behavior of the thermal average $\expec{ \hat{\psi} \myp{k,t} \hat{\psi} \myp{k',0}^{\ast} }$ is investigated, where $\hat{\psi} \myp{k,t}$ is the Fourier transform of the solution to \eqref{1}.
For small $\lambda$ this average turns out to be of the form
\begin{equation}
\label{2}
	\expec[\big]{ \hat{\psi} \myp{k,t} \hat{\psi} \myp{k',0}^{\ast} } 
	\simeq{} \delta\myp{k-k'} \exp\myb[\big]{ \myp[\big]{i \omega \myp{k}+ i\lambda R \myp{k}- \lambda^2 \gamma \myp{k} }t }.
\end{equation}
Here $\omega$ is the dispersion relation of the linear part of the equation, which is renormalized by $\lambda R$.
Both $\omega, R$ are real functions.
For onsite interaction, $R$ turns out to be independent of $k$.
The coefficient $\gamma$ is complex with a strictly positive real part, which reflects the dissipation resulting from the collision operator. 

Quantizing the classical wave field $\psi$, and imposing Fermi statistics, yields the model of our study.
However, the onsite nonlinearity as in Eq. \eqref{1} vanishes and leads to free fermion dynamics.
Therefore the quartic interaction term of the hamiltonian has to be suitably spread to neighbouring lattice sites.
As a result, the renormalization $R$ becomes a smooth $k$-dependent function, which poses a serious challenge.
Being of order $\lambda$ the renormalization is small compared to $\omega$.
Still, all oscillatory estimates have to be reconsidered.  
In our contribution we assume $d \geq 4$ and a kinetic time $\tau$, to say, microscopic times of order $\lambda^{-2} \tau$ with $\tau$ sufficiently small. 
We prove that to leading order in the weak-coupling limit, the two-point correlation function is governed by the collisional frequency associated with the Boltzmann-Nordheim collision operator at equilibrium.
At kinetic times, the many-body quantum dynamics exhibits the damping as predicted by the effective quantum kinetic equation.
This settles a prediction in \cite{lukkarinen2009not}.
From the perspective of quantum transport theory, we cover only a tiny piece of a much larger picture.
Still our contribution is the first mathematical result proving validity of the quantum kinetic limit.

From a technical point of view, we adopt the strategy developed in \cite{LukkarinenSpohn:WNS:2011} in the context of the discrete nonlinear Schr\"{o}dinger equation \eqref{1}.
This means, we expand the dynamical evolution as a perturbation series in $\lambda$.
Since the formal power series has zero radius convergence, it is cut at order $\abs{\log \lambda}$ with a suitable estimate on the remainder.
The two-point time correlation function is written as a series of high-dimensional oscillatory integrals.
However there is one crucial difference: while of order $\lambda$, the renormalization term is $k$-dependent and requires new estimates throughout our analysis.

The kinetic limit of the continuum nonlinear Schr\"{o}dinger equation with hamiltonian
\begin{equation}
	\label{3}
	H = \int_{\myb{0,L}^d} \ids x \myp[\Big]{ \abs{\nabla \psi \myp{x}}^2 + \frac{\lambda}{2} \abs{\psi \myp{x}}^4 }
\end{equation}	
has been studied by Deng and Hani in \cite{DenHan-23,DenHan-23pre,DenHan-23pre2}.
They consider a $d$-torus of linear size $L \simeq \lambda^{-A}$ with a fixed $A \geq 1$ and assume initial conditions to be a Gaussian field with variations on the kinetic scale.
The kinetic limit is established for arbitrarily large $\tau$.

In the next section we describe more precisely the model, the limiting Boltzmann-Nordheim equation, the required assumptions, and present our main theorem.

\addtocontents{toc}{\SkipTocEntry}
\subsection*{Acknowledgements}
The authors wish to thank Jani Lukkarinen for helpful discussions.
P.S. Madsen wishes to thank the Department of Mathematics at Texas A\&M University for hosting him on several occasions.
This work was partially funded by the Deutsche Forschungsgemeinschaft (DFG, German Research Foundation) via TRR 352 (Project-ID 470903074). 
P.S. Madsen was partially funded by NSF CAREER DMS-2303146 and NSF Grants DMS-2305523 and DMS-2306379.
P.T. Nam was partially funded by the European Research Council via the ERC Consolidator Grant RAMBAS (Project Nr. 10104424).
M.-B. Tran is funded in part by a Humboldt Fellowship, NSF CAREER DMS-2303146 and NSF Grants DMS-2305523 and DMS-2306379.

\section{Setting and main result} 
We consider a system of fermions on the lattice $\Omega= \Set{0, \dotsc, L-1}^d$ ($d\geq 1$) interacting via a two-body potential $\lambda V \myp{x-y}$.
The function $V$ is fixed and the small coupling constant $\lambda \to 0^+$ models the weakly interacting regime.
Here we impose periodic boundary conditions, which is natural to work in momentum space later. 

\addtocontents{toc}{\SkipTocEntry}
\subsection*{Hamiltonian in Fock space} 
In the grand canonical ensemble, the system is described by the Hamiltonian 
\begin{equation*}
	H_{\lambda} = \bigoplus_{n=0}^{\infty} \myp[\bigg]{ \sum_{k=1}^n \mathfrak{h}_{x_k} +\frac{\lambda}{2} \sum_{1\leq k<\ell \leq n} V \myp{x_k-x_{\ell}} }
\end{equation*}
on the fermionic Fock space 
\begin{equation*}
	\mathcal{F} 
	={} \bigoplus_{n=0}^{\infty} \myp[\bigg]{ \bigwedge^n \ell^2 \myp{\Omega} },
\end{equation*}
where $\mathfrak{h}$ is a suitable one-body operator on $\ell^2 \myp{\Omega}$.
It is convenient to introduce the Fourier transform 
\begin{equation*}
	\hat{f} \myp{k} ={} \sum_{x \in \Omega} f \myp{x} e^{-2\pi i k\cdot x}, 
	\quad k\in \Omega^{\ast} ={} \Set[\bigg]{0, \frac{1}{L}, \dotsc,\frac{L-1}{L} }^d
\end{equation*}
and work instead exclusively in momentum space.
The set $\Omega^{\ast}$ is a subset of the $d$-dimensional torus $ \T^d = \myb{0,1}^d$. 
In this convention, the physical momentum is $p=2\pi k$, and the inverse Fourier transform is 
\begin{equation*}
	f \myp{x} 
	={} \frac{1}{\abs{\Omega}} \sum_{k\in\Omega^{\ast}} \hat{f} \myp{k} e^{2\pi i k\cdot x}, 
\end{equation*}
where $\abs{\Omega}=L^d$. 
In the sequel, we also employ the following short hand notations
\begin{equation}
\label{Shorthand1}
	\int_{\Omega^{\ast}} \ids k 
	={} \frac{1}{\abs{\Omega}} \sum_{k\in\Omega^{\ast}}
\end{equation}
and
\begin{equation}
\label{Shorthand2}
	\expec{f, g} 
	={} \sum_{x\in\Omega} f \myp{x}^{\ast} g \myp{x},
\end{equation}
as well as the standard Japanese bracket
\begin{equation}
\label{JapaneseBracket}
	\expec{x}
	= \sqrt{1+x^2}, \quad \forall x\in\R.
\end{equation}

In the momentum space, we will use the usual creation and annihilation operators $\hat{a} \myp{k}^{\ast}, \hat{a} \myp{k}$, which  satisfy the anti-commutation relations
\begin{equation}
\label{eq:anticom}
	\myt{\hat{a} \myp{k}, \hat{a} \myp{k'}^{\ast}} ={} \delta_L \myp{k-k'}, 
	\quad \myt{ \hat{a} \myp{k}, \hat{a} \myp{k'} } ={} 0, 
	\quad \myt{\hat{a} \myp{k}^{\ast}, \hat{a} \myp{k'}^{\ast}} ={} 0,
\end{equation}
where $\myt{A,B}=AB+BA$, and $\delta_L: \myp{\Z / L}^d \to \R$ is the discrete $\delta$-function
\begin{equation}
\label{eq:discrete_delta}
	\delta_L \myp{k} \coloneq L^d \1 \myp{k \mod 1 = 0}.
\end{equation}
The Hamiltonian $H_{\lambda}$ can then be rewritten as 
\begin{align}
\label{eq:Hamiltonian}
	H_{\lambda} = \int_{\Omega^{\ast}} \ids k \omega \myp{k} \hat{a} \myp{k}^{\ast} \hat{a} \myp{k} + \frac{\lambda}{2} &\iiiint_{\myp{\Omega^{\ast}}^4}  \ids k_1 \ids k_2 \ids k_3 \ids k_4 \delta_L \myp{k_1+k_2-k_3-k_4} \nn \\
	&\times \widehat{V} \myp{k_2-k_3} \hat{a} \myp{k_1}^{\ast} \hat{a} \myp{k_2}^{\ast} \hat{a} \myp{k_3} \hat{a} \myp{k_4},
\end{align}
where $\omega \myp{k}$ is the dispersion relation (its properties will be specified later). 

\addtocontents{toc}{\SkipTocEntry}
\subsection*{Two-point time correlation function}
For any operator $A$ on Fock space, the Heisenberg time evolution is defined by
\begin{equation*}
	A \myp{t} \coloneq{} e^{itH_{\lambda}} A e^{-itH_{\lambda}}.
\end{equation*}
This satisfies the equation
\begin{equation*}
	\frac{\id}{\id t} A \myp{t} 
	={} i \myb{H_{\lambda},A \myp{t}},
	\quad A \myp{0}={} A. 
\end{equation*}
In particular, for the annihilation operators we have inserting the expression \eqref{eq:Hamiltonian} of $H_{\lambda}$ and using \eqref{eq:anticom},
\begin{align}
	\frac{\id}{\id t} \hat{a} \myp{k_1,t} 
	={} - i \omega \myp{k_1} \hat{a} \myp{k_1,t} - i\lambda &\iiint_{\myp{\Omega^{\ast}}^3} \ids k_2 \ids k_3 \ids k_4 \delta_L \myp{k_1+k_2-k_3-k_4} \nn \\
	&\times \widehat{V} \myp{k_2-k_3} \hat{a} \myp{k_2,t}^{\ast} \hat{a} \myp{k_3,t} \hat{a} \myp{k_4,t},
\label{eq:at}
\end{align}
and similarly for the creation operators.

We are interested in the two-point time correlation function $\expec{\hat{a} \myp{k}^{\ast} \hat{a} \myp{k',t}}$.
We will restrict the attention to the initial Gibbs state $Z^{-1}e^{-\frac{1}{T} H_{\lambda}}$, namely
\begin{equation*}
	\expec{A}_{\lambda} \coloneq{} \frac{1}{Z} \tr \myp{e^{-\frac{1}{T} H_{\lambda}} A},
	\quad Z ={} \tr e^{-\frac{1}{T} H_{\lambda}},
\end{equation*}
where $T > 0$ is the temperature.
Here we ignore the chemical potential for simplicity. 
The result remains valid for a general chemical potential $\mu \neq 0$, which only changes the equilibrium background. 
Moreover, we work in finite volume throughout and take the limit $L \to \infty$ only in the final step (the volume dependence of the estimates is important, of course).

\begin{lemma}
\label{lem:wignerfunction}
	\begin{enumerate}[(a)]
		\item If $\expec{\nonarg}$ is a translation- and gauge invariant state on $\Omega$, then there exists a function $W_L^{\lambda} : \Omega^{\ast} \to \R$ such that for all $k,k' \in \Omega^{\ast}$ and all $t \in \R$,
		\begin{equation}
		\label{eq:wignerfunction}
			\expec[\big]{\hat{a} \myp{k,t}^{\ast} \hat{a} \myp{k',t}} = \delta_L \myp{k-k'} W_L^{\lambda} \myp{k,t},
		\end{equation}
		with $0 \leq W_L^{\lambda} \myp{k,t} \leq 1$.
		
		\item In the case of the Gibbs state, denoted by $\expec{\nonarg}_{\lambda}$, the two-point correlation function is time independent.
		Furthermore, at $\lambda=0$ (the non-interacting gas) and at infinite volume $\Omega \nearrow \Z^d$, it is simply
		\begin{equation}
		\label{eq:wignerfunction0}
			\expec{\hat{a} \myp{k}^{\ast} \hat{a} \myp{k'} }_0 
			={} \delta \myp{k-k'}  W^0 \myp{k},
			\quad W^0 \myp{k} ={} \frac{1}{e^{\frac{1}{T}\omega\myp{k}}+1},
		\end{equation}
		for $k, k' \in \T^d$, where $\delta$ denotes the usual $\delta$-distribution at the origin.
	\end{enumerate}
\end{lemma}
\begin{proof}
	By translation invariance, we immediately have
	\begin{align}
		\expec[\big]{\hat{a} \myp{k,t}^{\ast} \hat{a} \myp{k',t}}
		={}& \sum_{x,y\in \Omega} e^{-2\pi i k \cdot \myp{x-y}} e^{-2 \pi i y \cdot \myp{k-k'}} \expec[\big]{a \myp{0,t}^{\ast} a \myp{x-y,t}} \nn \\
		={}& \sum_{y \in \Omega} e^{-2 \pi i y \cdot \myp{k-k'}} \sum_{x \in \Omega} e^{-2\pi i k \cdot x} \expec[\big]{a \myp{0,t}^{\ast} a \myp{x,t}} \nn \\
		={}& \delta_L \myp{k-k'} \sum_{x \in \Omega} e^{-2\pi i k \cdot x} \expec[\big]{a \myp{0,t}^{\ast} a \myp{x,t}},
	\label{eq:wignerfunctioncalc}
	\end{align}
	establishing \eqref{eq:wignerfunction} with
	\begin{equation*}
		W_L^{\lambda} \myp{k,t} = \sum_{x \in \Omega} e^{-2\pi i k \cdot x} \expec[\big]{a \myp{0,t}^{\ast} a \myp{x,t}}.
	\end{equation*}
	The fact that $0\leq W_L^{\lambda} $ follows directly from \eqref{eq:wignerfunction} since $\hat{a} \myp{k,t}^{\ast} \hat{a} \myp{k,t}$ is a non-negative operator, and $W_L^{\lambda} \leq 1$ follows in the same way, using the anti-commutation relation \eqref{eq:anticom}.
	
	For the Gibbs state, it is clear that $\expec[\big]{\hat{a} \myp{k,t}^{\ast} \hat{a} \myp{k',t}}_{\lambda} = \expec{\hat{a} \myp{k}^{\ast} \hat{a} \myp{k'} }_{\lambda} $, since $e^{-itH_{\lambda}}$ and $e^{-\beta H_{\lambda}}$ commute.
	The computation of $W^0 \myp{k}$ for the ideal gas is well-known.
\end{proof}

Our goal is to prove an effective description for the two-point time correlation function $\expec{\hat{a} \myp{k}^{\ast} \hat{a} \myp{k',t}}_{\lambda}$ in the weakly interacting regime $\lambda\to 0^+$ up to the kinetic time $t \sim \lambda^{-2}$. 

\addtocontents{toc}{\SkipTocEntry}
\subsection*{Prediction by the Boltzmann-Nordheim kinetic theory}
In the following heuristic discussion, we will take the formal limit $L\to \infty$ and set $\T^d = \myb{0,1}^d$ for simplicity.
If we think of $\expec{\nonarg}$ as a quasi-free state, then we can solve the Cauchy problem \eqref{eq:at} by Duhamel expansion and find that in the (rescaled) kinetic time $t\lambda^{-2} = O \myp{\lambda^{-2}}$, the two-point correlation function is approximated by
\begin{equation}\label{eq:kinetic-conjecture}
	\expec{\hat{a} \myp{k,t \lambda^{-2}}^{\ast} \hat{a} \myp{k',t \lambda^{-2}} }
	\simeq{} \delta \myp{k-k'} W \myp{k,t},
\end{equation}
where $W \myp{k,t}$ solves the Boltzmann-Nordheim equation
\begin{equation}
\label{eq:BN}
	\frac{\partial}{\partial t} W \myp{k,t}
	={} \mathcal{C} \myp[\big]{W \myp{t}} \myp{k}
\end{equation}
with the collision operator 
\begin{align}
	\mathcal{C} \myp{W} \myp{k_1}
	={} \pi &\int_{\T^{3d}} \ids k_2 \ids k_3 \ids k_4 \delta \myp{k_1+k_2-k_3-k_4} \delta \myp{ \omega \myp{k_1} + \omega \myp{k_2} - \omega \myp{k_3} - \omega \myp{k_4} } \nn \\
	&\times \abs{\widehat{V} \myp{k_2-k_3} - \widehat{V} \myp{k_2-k_4}}^2 \myp[\big]{ \widetilde{W}_1 \widetilde{W}_2 W_3 W_4 - W_1 W_2 \widetilde{W}_3 \widetilde{W}_4 }.
\label{eq:BN-C}
\end{align}
Here we have introduced the shorthand notation $W_j \coloneq W \myp{k_j}, j=1,2,3,4$, as well as $\widetilde{W} \coloneq{} 1-W$. 
Note that $W^0$ from \eqref{eq:wignerfunction0} is the equilibrium to \eqref{eq:BN}, as
\begin{equation*}
	\mathcal{C} \myp{W^0} ={} 0.
\end{equation*}
Therefore, we may expect that $\lim_{t\to \infty} W \myp{t} = W^0$.
A straightforward computation shows that the remainder 
\begin{equation*}
	\overline{W} \coloneq{} W-W^0
\end{equation*}
satisfies the equation 
\begin{equation}
\label{eq:W-W-infty}
	\frac{\partial}{\partial t} \overline{W} \myp{k_1}
	={} K \myb{\overline{W} \myp{t}} - \mu \myp{k_1} \overline{W} \myp{t} + {\rm error \ terms},
\end{equation}
where $K$ is a compact operator of convolution type, and the collisional frequency is defined to be
\begin{align}
	\mu \myp{k_1}
	={} \pi &\int_{\T^{3d}} \ids k_2 \ids k_3 \ids k_4 \delta \myp{k_1+k_2-k_3-k_4} \delta \myp{\omega \myp{k_1} + \omega \myp{k_2} - \omega \myp{k_3} - \omega \myp{k_4} } \nn \\
	&\times \widehat{V} \myp{k_2-k_3} \myp[\big]{ \widehat{V} \myp{k_2-k_3} - \widehat{V} \myp{k_2-k_4} } \myp[\big]{ W^0_3 W^0_4 -  W^0_2 W^0_4 + W^0_2 \widetilde{W}^0_3}.
\label{eq:mu-infty}
\end{align}

The well-posedness theory for the Boltzmann-Nordheim equation \eqref{eq:BN} can be established using the techniques from \cite{LukMeiSpo-15}; in particular, one has the uniform bound $0 \leq W \leq 1$ provided that the initial datum satisfies the same bound (which corresponds to Pauli's exclusion principle for fermions).
We also refer to the recent work \cite{LukPirVuo-26pre} for related results. 

Establishing the leading-order behavior \eqref{eq:kinetic-conjecture} is a long-standing problem in mathematical physics. Heuristically, it follows from the idea that the many-body quantum dynamics still admits a suitable quasi-free approximation; see, for example, \cite{ErdSalYau-04}. However, propagating such a strong bound rigorously is notoriously difficult.

In the classical setting, the analogue of \eqref{eq:kinetic-conjecture} for the kinetic-time behavior of the nonlinear Schr\"odinger equation is known as the {\em kinetic conjecture} \cite{LukkarinenSpohn:WNS:2011}, which was resolved recently in \cite{DenHan-23} for short kinetic times and in \cite{DenHan-23pre} for long kinetic times.

\medskip

In the following, we will justify this heuristic argument for the many-body quantum problem.
For simplicity, we will consider the simplified two-point time correlation function $\expec{\hat{a} \myp{k}^{\ast} \hat{a} \myp{k',t} }_{\lambda}$ at thermal equilibrium instead of the usual two-point function $\expec{\hat{a} \myp{k,t}^{\ast} \hat{a} \myp{k',t} }$ for a non-stationary initial state.
Since $W_L^{\lambda}$ is in this case time independent, we will from this point on hide the $t$-dependence from the notation, $W \myp{k} = W \myp{k,t} $, and instead introduce
\begin{equation}
	\label{eq:Wnotations}
	W \myp{k,\sigma} \coloneq{} \begin{cases}
		W \myp{k} & \text{if } \sigma = -1, \\
		1-W \myp{k} & \text{if } \sigma = 1.
	\end{cases}
\end{equation}

\subsection{Assumptions}\label{sec:assumptions}
Before stating our main result, we list the assumptions that will be used.
To treat the fermionic quantum case, we essentially need the same assumptions as in \cite{LukkarinenSpohn:WNS:2011} for the weakly non-linear Schr\"odinger equation.
However, since we will allow for a more general class of interaction potentials, we need to consider the following \emph{modified dispersion relation},
\begin{equation}
	\label{ModifiedDispersion}
	\omega^{\lambda} \coloneq{} \omega + \lambda R_{\lambda},
\end{equation}
with
\begin{equation}
\label{QuantityR}
	R_{\lambda} \myp{k_1} 
	\coloneq{} \int_{\Omega^{\ast}} \ids k_2 W_L^{\lambda} \myp{k_2} \myb{\widehat{V} \myp{0} - \widehat{V} \myp{k_1-k_2}},
\end{equation}
where $W_L^{\lambda}$ is given by \cref{lem:wignerfunction}.
The reason for introducing this modification is to absorb first order effects in $\lambda$.
This will become clear later when we expand $\hat{a} \myp{k,t}$ in a pertubation series using the Duhamel formula.
The assumptions listed below are the same as in \cite{LukkarinenSpohn:WNS:2011}, except that we require the estimates in \cref{ass:DR2,ass:DR3,ass:DR4} to hold for $\omega^{\lambda}$ uniformly in $\lambda$ in the weak coupling regime.
\begin{customdef}{(V)}
\label{ass:V}
	The potential $V: \Z^d \to \R$ is even and non-negative.
	Its Fourier transform $\widehat{V}$ is real-analytic on $\T^d$.
\end{customdef}
\begin{customdef}{(L1)}[$\ell_1$-clustering]
\label{ass:L1}
	There exists $\lambda_0>0$ and $c_0>0$ independent of $n$ such that for $0<\lambda\leq \lambda_0$ and all $n\geq 4$, the following estimate holds true
	\begin{equation}
		\sup_{\Omega, \sigma \in \Set{\pm 1}^n} \sum_{x \in \Omega} \delta \myp{{x_1=0}} \abs[\Bigg]{ \expec[\bigg]{ \prod^n_{j=1} a \myp{x_j,\sigma_j} }_{\lambda}^T }
		\leq{} \lambda c_0^n n!,
	\end{equation}
	where $a \myp{x,1} = a \myp{x}$, $a \myp{x,-1} = a \myp{x}^{\ast}$, and the truncation operator is defined in \eqref{Def:Truncation:2} below.
	In addition,
	\begin{equation}
		\limsup_{\Omega \nearrow \Z^d} \sum_{x \in \Omega} \abs[\big]{ \expec{a \myp{0}^{\ast} a \myp{x}}_{\lambda} - \expec{a \myp{0}^{\ast} a\myp{x}}_0 } 
		\leq{} \lambda 2 c_0^2.
	\end{equation}
\end{customdef}
\begin{customdef}{(DR1)}
\label{ass:DR1}
	The periodic extension of $\omega$ is real-analytic and $\omega \myp{-k} = \omega \myp{k}$.
\end{customdef}
\begin{customdef}{(DR2)}[$\ell_3$-dispersivity]
\label{ass:DR2}
	Let us consider the free propagator 
	\begin{equation}
	\label{eq:propagator}
		p_t^{\lambda} \myp{x} 
		\coloneq{} \int_{\T^d} \ids k \, e^{i2\pi x\cdot k} e^{-it\omega^{\lambda} \myp{k}}.
	\end{equation} 
	There exist $\lambda_0,\delta > 0$ such the following inequality holds true for all $0 \leq \lambda < \lambda_0$, $t \in \R$, 
	\begin{equation}
		\normt{p_t^{\lambda}}{3}^3 
		={} \sum_{x \in \Z^d} \abs{p_t^{\lambda} \myp{x}}^3
		\lesssim \expec{t}^{-1-\delta}.
	\end{equation}
\end{customdef}
\begin{customdef}{(DR3)}[Constructive interference]
\label{ass:DR3}
	There exists a set $M^{\mathrm{sing}} \subset \T^d$ consisting of a union of a finite number of closed, one-dimensional, smooth submanifolds, and constants $C,\lambda_0 > 0$, such that for all $0 \leq \lambda < \lambda_0$ $t \in \R$, $k_0\in \T^d$, and $\sigma\in \Set{\pm 1}$,
	\begin{equation}
	\label{5.20a}
		\abs[\bigg]{ \int_{\T^d} \ids k \, e^{-it \myp{\omega^{\lambda} \myp{k}+\sigma \omega^{\lambda} \myp{k-k_0}}} } 
		\leq{} \frac{C \expec{t}^{-1}}{d \myp{k_0,M^{\mathrm{sing}} }},
	\end{equation}
	where $d \myp{k_0,M^{\mathrm{sing}}}$ is the distance from $k_0$ to $M^{\mathrm{sing}}$. 
\end{customdef}
\begin{customdef}{(DR4)}[Crossing bounds]
\label{ass:DR4}
	Define for $t_0,t_1,t_2\in \R$, $u_1,u_2\in \T^d$, and $x\in \Z^d$,
	\begin{equation}
	\label{eq:defp2tx}
	    \mathcal{K} \myp{x;t_0,t_1,t_2,u_1,u_2} 
	    \coloneq{} \int_{\T^d} \ids k \, e^{i 2\pi x\cdot k} e^{-i \myp{ t_0 \omega^{\lambda} \myp{k}+t_1 \omega^{\lambda} \myp{k+u_1}+ t_2 \omega^{\lambda} \myp{k+u_2} }}.
	\end{equation}
	We assume that there is a measurable function $F^{\mathrm{cr}}: \T^d \times \R_+ \to \myb{0,\infty}$ and constants $0<\gamma \leq 1$, $c_1,c_2,\lambda_0 > 0$ such that the following hold for all $0 \leq \lambda < \lambda_0$.
\begin{enumerate}
	\item For any $u_j\in \T^d$, $\sigma_j \in \Set{\pm 1}$, $j=1,2,3$, and $0<\zeta\leq 1$, the following bounds are satisfied:
	\begin{equation}
		\int_{\R^2} \ids s \ids t \, e^{-\zeta \abs{s}} \normt{p_t^{\lambda}}{3}^2 \normt{\mathcal{K} \myp{\nonarg; t,\sigma_1 s,\sigma_2 s,u_1,u_2} }{3} 
		\leq{} \zeta^{\gamma-1} F^{\mathrm{cr}} \myp{u_2-u_1,\zeta}, 
		\label{eq:DR4_1}
	\end{equation}
	and for any $n\in\{1,2,3\}$,
	\begin{equation}
		\int_{\R^2} \ids s \ids t \, e^{-\zeta \abs{s}} \prod_{j=1}^3 \norm{\mathcal{K} \myp{\nonarg; t,\sigma_j s,0,u_j,0} }_3
		\leq{} \zeta^{\gamma-1} F^{\mathrm{cr}} \myp{u_n,\zeta}.
		\label{eq:DR4_2}
	\end{equation}

	\item Let $0<\zeta \leq 1$, $u,k_0\in \T^d$, $\alpha\in \R$, $\sigma\in \Set{\pm 1}$, and $n\in \Set{1,2,3}$.
	Further, denoting $k= \myp{k_1,k_2,k_0-k_1-k_2}$, then
	\begin{equation}
	\label{eq:crossingest2}
		\iint_{\myp{\T^d}^2} \ids k_1 \ids k_2 F^{\mathrm{cr}} \myp{k_n+u;\zeta} \frac{1}{ \abs{\alpha-\Theta \myp{k,\sigma}+i\zeta} } 
		\leq{} c_1 \expec{\ln \zeta}^{1+c_2},
	\end{equation}
	where $\Theta: \myp{\T^d}^3\times \Set{\pm 1} \to \R$ is defined by
	\begin{equation}
	\label{eq:defTheta} 
		\Theta \myp{k,\sigma} 
		={} \omega^{\lambda} \myp{k_3} -\omega^{\lambda} \myp{k_1}+\sigma \myp{ \omega^{\lambda} \myp{k_2}-\omega^{\lambda} \myp{k_1+k_2+k_3} }.
	\end{equation}
\end{enumerate}
\end{customdef}
Note that $\Theta$ of course depends on $\lambda$, but we will suppress this dependence from the notation throughout.
We note also that for fermions at thermal equilibrium, the $\ell^1$-clustering estimate of \cref{ass:L1} can be justified using the work of Salmhofer \cite{Salmhofer-09}.

\subsection{Main result}
We recall the \emph{cell step function} $\myb{k}: \R^d \to \Omega^{\ast}$ defined by $\myb{k}_i \coloneq \floor{L \myp{k_i \mod 1}}/L$.
Since $\myb{k}$ is periodic, it is also identifiable with a map $\T^d \to \Omega^{\ast}$.
Then, for any $F: \Omega^{\ast} \to \C$, we have the following formula relating the discrete sum over $\Omega^{\ast}$ to a Lebesgue integral,
\begin{equation}
\label{eq:discrete-lebesque}
	\int_{\Omega^{\ast}} F \myp{k} \id k
	= \int_{\T^d} F \myp{\myb{k}} \id k,
\end{equation}
where $F \myp{\myb{k}}$ is now a piecewise constant step function on $\T^d$.
The step function $\myb{\nonarg}$ of course depends on $L$, but we suppress this from the notation.
\begin{theorem}
\label{thm:main}
	Suppose that $d \geq 4$ and that the assumptions of \cref{sec:assumptions} are satisfied.
	Consider for $k,k' \in \Omega^{\ast}$ the two-point time correlation function 
	\begin{equation}
	\label{eq:Qlambda_def}
		Q^{\lambda} \myp{k',k,t} \coloneq \expec[\big]{ \hat{a} \myp{-k',0}^{\ast} e^{i\omega^{\lambda} \myp{k} t /\eps} \hat{a} \myp{k,t/\eps} }_{\lambda}, 
	\end{equation}
	where $\eps =\lambda^2$, $\expec{\nonarg}_{\lambda}$ is the Gibbs state $Z^{-1} e^{-\frac{1}{T} H_{\lambda}}$, and $\omega^{\lambda} \myp{k}$ is the modified dispersion relation \eqref{ModifiedDispersion}.
	There exists $t_0>0$ such that for any $\abs{t} <t_0$, we have
	\begin{equation}
		\label{eq:main}
		\lim_{\lambda\to 0} \sup_{k,k'\in \T^d} \limsup_{L\to\infty} \abs[\bigg]{\frac{1}{L^d} Q^{\lambda} \myp[\big]{\myb{k'},\myb{k},t} - \1 \myp{k+k'=0} W^0 \myp{k} e^{-\nu_1 \myp{k} \abs{t} -it \nu_2 \myp{k}} } 
		={} 0,
	\end{equation}
	where $\nu_1,\nu_2$ are real functions defined by a quantity which is proportional to the collisional frequency $\mu$ of the Boltzmann-Nordheim operator at equilibrium.
	To be precise, $\nu \myp{k} =\nu_1 \myp{k} + i\nu_2 \myp{k}$, with
	\begin{align}
	\label{eq:nu-infty}
		\nu \myp{k_1}
		={}&-\int_0^{\infty} \ids t \int_{\T^{3d}} \ids k_2 \ids k_3 \ids k_4 \delta \myp{k_1+k_2-k_3-k_4} \nn \\
		&\times \widehat{V} \myp{k_2-k_3} \myp[\big]{\widehat{V} \myp{k_2-k_3} - \widehat{V} \myp{k_2-k_4}} e^{it \myp{\omega \myp{k_1} + \omega \myp{k_2} - \omega \myp{k_3} - \omega \myp{k_4}}} \nn \\
		&\times \myp[\big]{W^0 \myp{k_3} W^0 \myp{k_4} - W^0 \myp{k_2} W^0 \myp{k_4} + W^0 \myp{k_2} \myp{1- W^0 \myp{k_3}} }.
	\end{align}
\end{theorem}
As pointed out in \cite{lukkarinen2009not}, the real part of the decay coefficient $\nu$ satisfies $\mathrm{Re} \, \nu > 0$, since a direct calculation gives
\begin{align*}
	\mathrm{Re} \, \nu \myp{k_1}
	={} \frac{\pi}{2} \int_{\T^{3d}} \ids k_2 \ids k_3 \ids k_4 &\delta \myp{k_1+k_2-k_3-k_4} \delta \myp{\omega \myp{k_1} + \omega \myp{k_2} - \omega \myp{k_3} - \omega \myp{k_4} } \nn \\
	&\times \myp[\big]{ \widehat{V} \myp{k_2-k_3} - \widehat{V} \myp{k_2-k_4} }^2 \myp{W^0_1}^{-1} \widetilde{W}^0_2 W^0_3 W^0_4.
\end{align*}
\begin{remark}[$Q^{\lambda}$ vanishes away from the diagonal]
\label{rem:Qdiagonal}
	For any $k,k' \in \Omega^{\ast}$ with $k+k'\neq 0$, it holds already at finite volume that $Q^{\lambda} \myp{k',k,t} = 0$.
	Indeed, denoting
	\begin{equation}
	\label{eq:F_diag}
		F \myp{x,t} \coloneq{} \expec[\big]{a \myp{0}^{\ast} a \myp{x,t}}_{\lambda},
	\end{equation}
	we have by translation invariance and the fact that the time evolution commutes with spatial translation,
	\begin{equation*}
		\expec[\big]{a \myp{y}^{\ast} a \myp{x,t}}_{\lambda}
		={} F \myp{x-y,t}.
	\end{equation*}
	Hence, repeating the calculation in \eqref{eq:wignerfunctioncalc}, we obtain
	\begin{equation*}
		\expec[\big]{\hat{a} \myp{-k'}^{\ast} \hat{a} \myp{k,t}}_{\lambda}
		={} \delta_L \myp{k+k'} \widehat{F} \myp{k,t},
	\end{equation*}
	implying that we can write
	\begin{equation}
	\label{eq:space_time_corr2}
		Q^{\lambda} \myp{k',k,t}
		={} \delta_L \myp{k+k'} e^{i \omega^{\lambda} \myp{k} t/\eps} \widehat{F} \myp{k,t/\eps}.
	\end{equation}
\end{remark}
\cref{thm:main} focuses on the short kinetic time.
The case of long kinetic time is more challenging and requires new ingredients going beyond the existing techniques in \cite{LukkarinenSpohn:WNS:2011}.
As a corollary to \cref{thm:main}, we obtain the following, which is essentially \cite[Claim 6.4]{lukkarinen2009not}.
\begin{corollary}
\label{cor:main}
	Suppose that $d \geq 4$ and that the assumptions of \cref{sec:assumptions} are satisfied.
	For any $f,g\in \ell^1 \myp{\Z^d}$, we  denote by $\hat{f}$, $\hat{g}$ their Fourier transforms and consider the following testing two-point correlation operator
	\begin{equation}
		\label{eq:space_time_corr}
		Q^{\lambda} \myb{g,f} \myp{t} 
		\coloneq{} \iint_{\myp{\Omega^{\ast}}^2} \mathrm{d}k\mathrm{d}k' \hat{g} \myp{k}^{\ast} \hat{f} \myp{k'} \expec[\big]{ \hat{a} \myp{k',0}^{\ast} e^{i\omega^{\lambda} \myp{k} t /\eps} \hat{a} \myp{k,t/\eps} }_{\lambda}.
	\end{equation}
	There exists $t_0>0$ such that for all $\abs{t} <t_0$, we have
	\begin{equation}
	\label{eq:cor_main}
		\lim_{\lambda\to 0} \limsup_{L\to\infty} \abs[\bigg]{ Q^{\lambda} \myb{g,f} \myp{t}- \int_{\T^d} \ids k \hat{g} \myp{k}^{\ast} \hat{f} \myp{k} W^0 \myp{k} e^{-\nu_1 \myp{k} \abs{t} -it \nu_2 \myp{k}} } 
		={} 0.
	\end{equation}
\end{corollary}
\begin{proof}
	The corollary follows easily from \cref{thm:main} by using that $\hat{g}$ and $\hat{f}$ are continuous and bounded on $\T^d$.
	First, note that we may rewrite
	\begin{align*}
		Q^{\lambda} \myb{g,f} \myp{t}
		={}& \iint_{\myp{\Omega^{\ast}}^2} \ids k \ids k' \hat{g} \myp{k}^{\ast} \hat{f} \myp{-k'} Q^{\lambda} \myp{k',k,t} \nn \\
		={}& \iint_{\myp{\T^d}^2} \ids k \ids k' \hat{g} \myp{\myb{k}}^{\ast} \hat{f} \myp{-\myb{k'}} \delta_L \myp{\myb{k}+\myb{k'}} \frac{Q^{\lambda} \myp{\myb{k'},\myb{k},t}}{L^d},
	\end{align*}
	where we in the second equality used \eqref{eq:space_time_corr2} to insert the delta-function, as well as \eqref{eq:discrete-lebesque} to write the expression using Lebesgue integrals.
	Then, assuming for simplicity that $t \geq 0$, we resolve the delta-function and estimate
	\begin{align*}
		\MoveEqLeft[6] \abs[\bigg]{ Q^{\lambda} \myb{g,f} \myp{t} - \int_{\T^d} \ids k \hat{g} \myp{k}^{\ast} \hat{f} \myp{k} W^0 \myp{k} e^{-\nu \myp{k} t} } \nn \\
		\leq{}& \int_{\T^d} \ids k \abs[\big]{\hat{g} \myp{\myb{k}}^{\ast} \hat{f} \myp{\myb{k}}} \abs[\Big]{ \frac{1}{L^d}Q^{\lambda} \myp{-\myb{k},\myb{k},t} - W^0 \myp{k} e^{-\nu \myp{k} t} } \nn \\
		&+ \int_{\T^d} \ids k \abs[\big]{W^0 \myp{k} e^{-\nu \myp{k} t}} \abs[\big]{\hat{g} \myp{\myb{k}}^{\ast} \hat{f} \myp{\myb{k}} - \hat{g} \myp{k} \hat{f} \myp{k}}.
	\end{align*}
	Taking the $\limsup$ in $L$, the second term vanishes by dominated convergence, since $\hat{g} \myp{\myb{k}} \to \hat{g} \myp{k}$ pointwise for $k \in \T^d$, and the same holds for $\hat{f}$.
	By Fatou's lemma, the first term is bounded by
	\begin{equation*}
		\normt{\hat{g}}{\infty} \normt{\hat{f}}{\infty} \int_{\T^d} \ids k \limsup_{L\to \infty} \abs[\Big]{ \frac{1}{L^d}Q^{\lambda} \myp{-\myb{k},\myb{k},t} - W^0 \myp{k} e^{-\nu \myp{k} t} },
	\end{equation*}
	where the integrand is uniformly bounded in $k$ and vanishes in the $\lambda \to 0$ limit, by \cref{thm:main}.
\end{proof}

\addtocontents{toc}{\SkipTocEntry}
\subsection*{Restriction to times \texorpdfstring{$t > 0$}{t>0}}
\begin{lemma}
\label{lem:wigner_estimate}
	If the assumptions \ref{ass:L1} and \ref{ass:DR1} are satisfied, then for all $0 < \lambda\leq \lambda_0$,
	\begin{equation}
		\limsup_{L \to \infty} \sup_{k \in \T^d} \abs[\big]{W_L^{\lambda} \myp{\myb{k}} - W^0 \myp{k}}
		\leq 2 c_0^2 \lambda.
	\end{equation}
\end{lemma}
\begin{proof}
	We have from \cref{lem:wignerfunction} that
	\begin{equation*}
		\frac{1}{e^{\frac{1}{T} \omega \myp{k}} + 1} 
		= W^0 \myp{k}
		= \sum_{x \in \Z^d} e^{-2 \pi i k \cdot x} \expec[\big]{a \myp{0}^{\ast} a \myp{x}}_0,
	\end{equation*}
	so since the left hand side is smooth by our assumptions \ref{ass:DR1} on $\omega$, the inverse Fourier transform $x \mapsto \expec{a \myp{0}^{\ast} a \myp{x}}_0$ is in $\ell^1 \myp{\Z^d}$.
	We can now estimate
	\begin{equation*}
		\abs[\big]{W_L^{\lambda} \myp{\myb{k}} - W^0 \myp{k}}
		\leq{} \abs[\big]{W^0 \myp{\myb{k}} - W^0 \myp{k}} + \abs[\big]{W_L^{\lambda} \myp{\myb{k}} - W^0 \myp{\myb{k}}},
	\end{equation*}
	where the first term converges to zero uniformly on $\T^d$ as $\Omega$ grows, and the second term is bounded by
	\begin{equation*}
		\abs[\big]{W_L^{\lambda} \myp{\myb{k}} - W^0 \myp{\myb{k}}}
		\leq{} \sum_{x \in \Z^d \setminus \Omega} \abs[\big]{\expec{a \myp{0}^{\ast} a \myp{x}}_0} + \sum\limits_{x \in \Omega} \abs[\big]{\expec{a \myp{0}^{\ast} a \myp{x}}_{\lambda} - \expec{a \myp{0}^{\ast} a \myp{x}}_0}.
	\end{equation*}
	Here, the first term tends to zero as $L \to \infty$, and the second term is bounded by $\lambda 2 c_0^2$ by \cref{ass:L1}.
\end{proof}
\begin{lemma}
	\cref{thm:main} is true for $t = 0$.
	Furthermore, assuming that the theorem holds for $0<t<t_0$, then it also holds for $-t_0 < t < 0$.
\end{lemma}
\begin{proof}
	We fix $k,k' \in \Omega^{\ast}$.
	For $t=0$, we have by \cref{lem:wignerfunction},
	\begin{equation*}
		Q^{\lambda} \myp{k',k,0} 
		={} \expec[\big]{ \hat{a} \myp{-k'}^{\ast} \hat{a} \myp{k} }_{\lambda}
		={} \delta_L \myp{k+k'} W_L^{\lambda} \myp{k},
	\end{equation*}
	so using \cref{lem:wigner_estimate}, we can replace $W_L^{\lambda}$ by $W^0$ and immediately obtain \eqref{eq:main} at the initial time.
	
	Now, we suppose that the theorem is true for $0<t<t_0$, and re-use the notation \eqref{eq:F_diag} from \cref{rem:Qdiagonal}.
	Noting that
	\begin{equation*}
		F \myp{-x,-t}^{\ast}
		={} \expec[\big]{a \myp{-x,-t}^{\ast} a \myp{0}}_{\lambda}
		={} \expec[\big]{a \myp{0,-t}^{\ast} a \myp{x}}_{\lambda}
		={} \expec[\big]{a \myp{0}^{\ast} a \myp{x,t}}_{\lambda}
		={} F \myp{x,t},
	\end{equation*}
	which implies
	\begin{equation*}
		\widehat{F} \myp{k,-t}
		={} \sum\limits_{x \in \Omega} e^{2\pi i k \cdot x} F \myp{-x,-t}
		={} \widehat{F} \myp{k,t}^{\ast},
	\end{equation*}
	we conclude from \eqref{eq:space_time_corr2} that for arbitrary $k,k' \in \Omega^{\ast}$,
	\begin{equation*}
		Q^{\lambda} \myp{k',k,-t}
		={} \myp[\big]{Q^{\lambda} \myp{k',k,t}}^{\ast}.
	\end{equation*}
	It then follows easily for $0<t<t_0$ that we can replace $t$ by $-t$ in \eqref{eq:main}.
\end{proof}

The rest of the paper is organized as follows.
In \cref{sec:feynman_diagrams} we develop the iterated Duhamel expansion and the corresponding diagrammatic representation. 
In \cref{Sec:MomentumGraph} we review the momentum graph formalism and the classification of graphs needed in the proof.
In \cref{sec:main_lemmata} we introduce the key parameters and auxiliary estimates.
The graph bounds are established in \cref{sec:graph_estimates}. 
Finally, in \cref{sec:mainproof} we combine these ingredients to prove the main theorem.

\section{Feynman diagrams}
\label{sec:feynman_diagrams}
\subsection{Feynman diagrams for Duhamel expansions}
\subsubsection{Duhamel expansions}
We rewrite \eqref{eq:at} as
\begin{align}
\label{eq:duhamel1bis}
	\frac{\partial}{\partial t} \hat{a} \myp{k_1,t} 
	={}& - i\omega \myp{k_1} \hat{a} \myp{k_1,t} - i\lambda \iiint_{\myp{\Omega^{\ast}}^3} \ids k_2 \ids k_3 \ids k_4 \delta_L \myp{k_1-k_2-k_3-k_4} \nn \\
	&\times \widehat{V} \myp{k_2+k_3} \hat{a} \myp{-k_2,t}^{\ast} \hat{a} \myp{k_3,t} \hat{a} \myp{k_4,t}. 
\end{align}
We will continue to expand \eqref{eq:duhamel1bis} in several layers of Duhamel expansions.
This strategy leads to a technical difficulty of treating the wave number near the singular manifold $M^{\mathrm{sing}}$.
As a result, we introduce the following proposition about the existence of cut-off functions, whose construction is detailed in \cite[Section $7.1$]{LukkarinenSpohn:WNS:2011}.

\begin{proposition}\label{Propo:Cutoff}
	Let $b > 0$.
	There are constants $C>0,\lambda_0'>0$ such that for any $ \myp{k_1,k_2,k_3} \in \myp{\T^d}^3$, $0<\lambda<\lambda_0'$ and for any pair of indices $i\neq j$, $i,j\in\Set{1,2,3}$, there exist smooth cut-off functions $\Phi_0: \myp{\T^d}^3 \to \myb{0,1}$ and $\Phi_1 \myp{k_1,k_2,k_3} = 1- \Phi_0 \myp{k_1,k_2,k_3},$ such that
	\begin{enumerate}[(i)]
		\item If $k_i+k_j=0$ then $\Phi_1=0$ and $\Phi_0=1$.
		\item $0 \leq \Phi_1 \myp{k} \leq C \lambda^{-b} d \myp{k_i+k_j,M^{\mathrm{sing}}}$.
		\item $\Phi_1 \myp{k_3,k_2,k_1} = \Phi_1 \myp{k_1,k_2,k_3}$, $\Phi_0 \myp{k_3,k_2,k_1} = \Phi_0 \myp{k_1,k_2,k_3}$.
		\item  
		\begin{equation}
		\label{PsiProperty4}
			0 \leq \Phi_0 \myp{k_1,k_2,k_3} 
			\leq \sum_{1\leq i<j \leq 3} \1 \myp[\big]{d \myp{k_i+k_j,M^{\mathrm{sing}}} 
			< \lambda^b}.
		\end{equation}
		\item There exist smooth functions $F_1,F_0: \T^d \to \R$ such that $F_0 = 1 - F_1$ and $\Phi_1 \myp{k_1,k_2,k_3} = F_1 \myp{k_1+k_2} F_1 \myp{k_2+k_3} F_1 \myp{k_3+k_1}$.
		Moreover, there exist constants $C_F,C_F'>0$ such that for any $0<\lambda<\lambda_0'$,
		\begin{enumerate}[label=(\roman{enumi}.\arabic*)]
			\item $0 \leq  F_1 \myp{k}, F_0 \myp{k} \leq 1$.
			\item $ \abs{\nabla F_1 \myp{k}}$, $ \abs{\nabla F_0 \myp{k}} \leq C_F \lambda^{-b}$, for all $k \in \T^d$.
 			\item If $k\in M^{\mathrm{sing}}$, then $F_1 \myp{k}=0$ and $F_0 \myp{k}=1$. 
  			\item If $d \myp{k,M^{\mathrm{sing}}} \geq \lambda^b$, then $F_1 \myp{k}=1$ and $F_0 \myp{k}=0$. 
    		\item $0 \leq F_1 \myp{k} \leq C_F \lambda^{-b} d \myp{k,M^{\mathrm{sing}}}$ for all $k \in \T^d$. 
			\item 
			\begin{equation}
			\label{PsiProperty5}
				\int_{\T^d} \mathrm{d}k F_0 \myp{k} 
				\leq \int_{\T^d} \mathrm{d}k \1 \myp[\big]{ d \myp{k,M^{\mathrm{sing}}} < \lambda^b }
				\leq C'_F \lambda^{b \myp{d-1}}.
			\end{equation}
		\end{enumerate}
	\end{enumerate}
\end{proposition}
Of course, the functions $\Phi_0$, $\Phi_1$, $F_0$, $F_1$ all depend on $\lambda$, but we will suppress this from the notation throughout.
Plugging the expression $1=\Phi_0+\Phi_1$ into \eqref{eq:duhamel1bis}, we find
\begin{align}
\label{eq:duhamel1}
	\frac{\partial}{\partial t} \hat{a} \myp{k_1,t} 
	={}& - i \omega(k_1)  \hat{a} \myp{k_1,t} - i \lambda \iiint_{\myp{\Omega^{\ast}}^3} \mathrm{d}k_2 \mathrm{d}k_3 \mathrm{d}k_4 \delta \myp{k_1-k_2-k_3-k_4} \nn \\
	&\qquad \times \widehat{V} \myp{k_2+k_3} \Phi_0 \myp{k_2,k_3,k_4} \hat{a} \myp{-k_2,t}^{\ast} \hat{a} \myp{k_3,t} \hat{a} \myp{k_4,t} \nn \\
	&- i \lambda \iiint_{\myp{\Omega^{\ast}}^3} \mathrm{d}k_2 \mathrm{d}k_3 \mathrm{d}k_4 \delta \myp{k_1-k_2-k_3-k_4} \nn \\
	&\qquad \times \widehat{V} \myp{k_2+k_3} \Phi_1 \myp{k_2,k_3,k_4} \hat{a} \myp{-k_2,t}^{\ast} \hat{a} \myp{k_3,t} \hat{a} \myp{k_4,t}.
\end{align}
We now define a truncated product $\myb{\hat{a}^{\ast} \hat{a} \hat{a}}^T$ by writing
\begin{align}
\label{Def:Truncation}
	\hat{a} \myp{k_2,t}^{\ast} \hat{a} \myp{k_3,t} \hat{a} \myp{k_4,t}
	={}& \hat{a} \myp{k_4,t} \expec{\hat{a} \myp{k_2,t}^{\ast} \hat{a} \myp{k_3,t}}_{\lambda} - \hat{a} \myp{k_3,t} \expec{\hat{a} \myp{k_2,t}^{\ast} \hat{a} \myp{k_4,t}}_{\lambda} \nn \\
	&+ \myb{\hat{a} \myp{k_2,t}^{\ast} \hat{a} \myp{k_3,t} \hat{a} \myp{k_4,t} }^T.
\end{align}
Inserting \eqref{Def:Truncation} into the term of \eqref{eq:duhamel1} containing $\Psi_0$ yields
\begin{align*}
	\frac{\partial}{\partial t} \hat{a} \myp{k_1,t} 
	={}& - i \omega \myp{k_1}  \hat{a} \myp{k_1,t} - i \lambda \iiint_{\myp{\Omega^{\ast}}^3} \mathrm{d}k_2 \mathrm{d}k_3 \mathrm{d}k_4 \delta_L \myp{k_1-k_2-k_3-k_4} \nn \\
	&\qquad\qquad\begin{aligned}[t]
			\times \widehat{V} \myp{k_2+k_3} \Phi_0 \myp{k_2,k_3,k_4} \myb{&\hat{a} \myp{k_4,t} \delta_L \myp{k_2+k_3} W_L^{\lambda} \myp{k_3} \\
			&-\hat{a} \myp{k_3,t} \delta_L \myp{k_2+k_4} W_L^{\lambda} \myp{k_4} }
		\end{aligned} \nn \\
	&- i \lambda \iiint_{\myp{\Omega^{\ast}}^3} \mathrm{d}k_2 \mathrm{d}k_3 \mathrm{d}k_4 \delta_L \myp{k_1-k_2-k_3-k_4} \widehat{V} \myp{k_2+k_3} \nn \\
	&\qquad\qquad\begin{aligned}[t]
		\times \myb[\big]{&\Phi_1 \myp{k_2,k_3,k_4} \hat{a} \myp{-k_2,t}^{\ast} \hat{a} \myp{k_3,t} \hat{a} \myp{k_4,t} \\
		&+ \Phi_0 \myp{k_2,k_3,k_4} \myb{ \hat{a} \myp{-k_2,t}^{\ast} \hat{a} \myp{k_3,t} \hat{a} \myp{k_4,t}}^T },
	\end{aligned} 
\end{align*}
which, after integrating out the delta functions and taking into account the properties of the cut-off function $\Phi_0$, leads to
\begin{align}
\label{eq:duhamel1b}
	\frac{\partial}{\partial t} \hat{a} \myp{k_1,t}
	={}& - i \omega \myp{k_1} \hat{a} \myp{k_1,t} - i \lambda \hat{a} \myp{k_1,t} \int_{\Omega^{\ast}} \mathrm{d}k_2 W_L^{\lambda} \myp{k_2} \myb{\widehat{V} \myp{0} - \widehat{V} \myp{k_1-k_2}} \nn \\
	&- i \lambda \iiint_{\myp{\Omega^{\ast}}^3} \mathrm{d}k_2 \mathrm{d}k_3 \mathrm{d}k_4 \delta_L \myp{k_1-k_2-k_3-k_4} \widehat{V} \myp{k_2+k_3} \nn \\
	&\qquad\qquad\begin{aligned}[b]
	\times \myb[\big]{&\Phi_1 \myp{k_2,k_3,k_4} \hat{a} \myp{-k_2,t}^{\ast} \hat{a} \myp{k_3,t} \hat{a} \myp{k_4,t} \\
		&+ \Phi_0 \myp{k_2,k_3,k_4} \myb{ \hat{a} \myp{-k_2,t}^{\ast} \hat{a} \myp{k_3,t} \hat{a} \myp{k_4,t}}^T },
	\end{aligned} 
\end{align}
At this point, we recall the definition of the modified dispersion relation \eqref{ModifiedDispersion}, \eqref{QuantityR}, whose role is to absorb the terms containing $\hat{a} \myp{k_1,t}$ in \eqref{eq:duhamel1b}.
In order to do that, we introduce the following new presentation of the field operators $\hat{a}$, 
\begin{equation}
\label{eq:hatanew}
	b \myp{k,1,t} \coloneq e^{i\omega^{\lambda} \myp{k}t} \hat{a} \myp{k,t}, 
	\qquad
	b \myp{k,-1,t} \coloneq e^{-i\omega^{\lambda} \myp{k}t} \hat{a} \myp{-k,t}^{\ast}.
\end{equation}
A new equation for $b$ can then be derived from \eqref{eq:duhamel1b},
\begin{align}
\label{eq:duhamelforb1}
	\frac{\partial}{\partial t}  b \myp{k_1,1,t} 
	={}& - i \lambda \iiint_{\myp{\Omega^{\ast}}^3} \mathrm{d}k_2 \mathrm{d}k_3 \mathrm{d}k_4 \delta_L \myp{k_1-k_2-k_3-k_4} \widehat{V} \myp{k_2+k_3} \nn \\
	&\times \exp{ \myb[\big]{-it \myp{-\omega^{\lambda} \myp{k_1} - \omega^{\lambda} \myp{k_2} + \omega^{\lambda} \myp{k_3} + \omega^{\lambda} \myp{k_4}} } } \nn \\
	&\times \myb[\big]{\Phi_1 \myp{k_2,k_3,k_4} b \myp{k_2,-1,t} b \myp{k_3,1,t} b \myp{k_4,1,t} \nn \\
	&\qquad + \Phi_0 \myp{k_2,k_3,k_4} \myb{ b \myp{k_2,-1,t} b \myp{k_3,1,t} b \myp{k_4,1,t}}^T}. 
\end{align}
Note that this step relies on the stationarity of the Gibbs state $\expec{\nonarg}_{\lambda}$.
For general initial states, $\omega^{\lambda}$ will depend also on $t$ through $W_L^{\lambda}$, which would make absorbing the first two terms in \eqref{eq:duhamel1b} more complicated.

A similar computation can also be done for $b \myp{k_1,-1,t}$, starting from
\begin{align*}
	\frac{\partial}{\partial t} \hat{a} \myp{k_1,t}^{\ast} 
	={} i \omega(k_1) \hat{a} \myp{k_1,t}^{\ast} + i \lambda &\iiint_{\myp{\Omega^{\ast}}^3} \mathrm{d}k_2 \mathrm{d}k_3 \mathrm{d}k_4  \delta_L \myp{k_1+k_2+k_3+k_4} \\
	&\times \widehat{V} \myp{k_3+k_4} \hat{a} \myp{-k_2,t}^{\ast} \hat{a} \myp{-k_3,t}^{\ast} \hat{a} \myp{k_4,t},
\end{align*}
and using, similarly to \eqref{Def:Truncation},
\begin{align}
	\hat{a} \myp{k_2,t}^{\ast} \hat{a} \myp{k_3,t}^{\ast} \hat{a} \myp{k_4,t} 
	={}& \hat{a} \myp{k_2,t}^{\ast} \expec{\hat{a} \myp{k_3,t}^{\ast} \hat{a} \myp{k_4,t}}_{\lambda} - \hat{a} \myp{k_3,t}^{\ast} \expec{\hat{a} \myp{k_2,t}^{\ast} \hat{a} \myp{k_4,t}}_{\lambda} \nn \\
	&+  \myb{ \hat{a} \myp{k_2,t}^{\ast} \hat{a} \myp{k_3,t}^{\ast} \hat{a} \myp{k_4,t}}^T.
\label{Def:Truncation2}
\end{align}
The following general equation for $b \myp{k_1,\sigma,t}$ can then be derived,
\begin{align}
	\frac{\partial}{\partial t} b \myp{k_1,\sigma,t} 
	={}& - i\sigma \lambda \iiint_{\myp{\Omega^{\ast}}^3} \mathrm{d}k_2 \mathrm{d}k_3 \mathrm{d}k_4 \delta_L \myp{k_1-k_2-k_3-k_4} \nn \\
	&\times\frac{1}{2} \myb[\big]{ \myp{1+\sigma} \widehat{V} \myp{k_2+k_3} + \myp{1-\sigma} \widehat{V} \myp{k_3+k_4}} \nn \\
	&\times \exp{ \myb[\big]{-it \myp{-\sigma\omega^{\lambda} \myp{k_1} - \omega^{\lambda} \myp{k_2} + \sigma\omega^{\lambda} \myp{k_3} + \omega^{\lambda} \myp{k_4}} }} \nn \\
	&\times  \myb[\big]{ \Phi_1 \myp{k_2,k_3,k_4} b \myp{k_2,-1,t} b \myp{k_3,\sigma,t} b \myp{k_4,1,t} \nn \\
	&\qquad + \Phi_0 \myp{k_2, k_3,k_4} \myb{ b \myp{k_2,-1,t} b \myp{k_3,\sigma,t} b \myp{k_4,1,t}}^T}. 
\label{eq:duhamelforb2}
\end{align}
To simplify the expression, we denote
\begin{align}
	\Theta \myp{k_2,k_3,k_4,\sigma} \coloneq{}& \omega^{\lambda} \myp{k_4} -\omega^{\lambda} \myp{k_2} + \sigma \myp{\omega^{\lambda} \myp{k_3} -\omega^{\lambda}\myp{k_2+k_3+k_4}}, \nn \\
	U \myp{k_2,k_3,k_4,\sigma} \coloneq{}& \frac{1}{2} \myb[\big]{ \myp{1+\sigma} \widehat{V} \myp{k_2+k_3} + \myp{1-\sigma} \widehat{V} \myp{k_3+k_4} }, \label{Def:Uinteraction}\\
	\Psi_j \myp{k_2,k_3,k_4,\sigma} \coloneq{}& \Phi_j \myp{k_2, k_3,k_4} U \myp{k_2,k_3,k_4,\sigma}, \qquad j = 1,2, \nn
\end{align}
where $\Theta$ is the same as in \eqref{eq:defTheta}.
Plugging these expressions into \eqref{eq:duhamelforb2}, we arrive at
\begin{align}
	\MoveEqLeft[2] \frac{\partial}{\partial t} b \myp{k_1,\sigma,t} 
	={} - i\sigma \lambda \iiint_{\myp{\Omega^{\ast}}^3} \mathrm{d}k_2 \mathrm{d}k_3 \mathrm{d}k_4 \delta_L \myp{k_1-k_2-k_3-k_4} \nn \\
	&\times \exp{\myb[\big]{-it \Theta \myp{k_2,k_3,k_4,\sigma}}} \myb[\big]{ \Psi_1 \myp{k_2,k_3,k_4,\sigma} b \myp{k_2,-1,t} b \myp{k_3,\sigma,t} b \myp{k_4,1,t} \nn \\
	&\qquad +\Psi_0 \myp{k_2,k_3,k_4,\sigma} \myb{b \myp{k_2,-1,t} b \myp{k_3,\sigma,t} b \myp{k_4,1,t}}^T}. 
\label{eq:duhamelforbUSimplified}
\end{align}
We want to use this to obtain an expression for general products of $b$ operators.
For this, we recall the product rule
\begin{equation*}
	\frac{\partial}{\partial t} \prod_{j=1}^n b \myp{k_j,\sigma_j,t} 
	={} \sum_{m=1}^n \myp[\bigg]{ \prod_{j=1}^{m-1} b \myp{k_j,\sigma_j,t} } \frac{\partial}{\partial t} b \myp{k_m,\sigma_m,t} \myp[\bigg]{ \prod_{j=m+1}^n b \myp{k_j,\sigma_j,t} },
\end{equation*}
where we use the general convention that the operators are ordered from left to right according to the order of the index in the $\prod$-symbol, i.e., the product on the left hand side above is by definition $b \myp{k_1,\sigma_1,t} \cdots b \myp{k_n,\sigma_n,t}$.
Introducing a new parameter $\kappa>0$ which will be specified later, plugging \eqref{eq:duhamelforbUSimplified} into the expression above leads to
\begin{align}
	\MoveEqLeft[2] \frac{\partial}{\partial t} \myb[\Big]{ e^{\kappa t} \prod_{j=1}^n b \myp{k_j,\sigma_j,t} } \nn \\
	={}& \kappa e^{\kappa t} \prod_{j=1}^n b \myp{k_j,\sigma_j,t} - i\lambda \sum_{m=1}^n \sigma_m \iiint_{\myp{\Omega^{\ast}}^3} \mathrm{d}k'_2 \mathrm{d}k'_3 \mathrm{d}k'_4 \delta_L \myp{k_m-k_2'-k_3'-k_4'} \nn \\
	&\quad\quad \times \exp{ \myb[\big]{\kappa t-it \Theta \myp{k_2',k_3',k_4',\sigma_m} }} \nn \\
	&\quad\quad \times \prod_{j=1}^{m-1} b(k_j,\sigma_j,t) \myb[\big]{ \Psi_1(k_2,k_3,k_4,\sigma)b(k_2',-1,t) b(k_3',\sigma_m,t) b(k_4',1,t) \nn \\
	&\quad\quad \quad  +\Psi_0(k_2',k_3',k_4',\sigma_m)[b(k_2',-1,t) b(k_3',\sigma_m,t) b(k_4',1,t)]^T } \prod_{j=m+1}^{n} b(k_j,\sigma_j,t).
\label{eq:duhamelproductlambda}
\end{align}
Integrating both sides of the above equation yields
\begin{align}
	\MoveEqLeft[2] \prod_{j=1}^n b \myp{k_j,\sigma_j,t}
	={} e^{-\kappa t} \prod_{j=1}^n b \myp{k_j,\sigma_j,0} + \kappa \int_0^t \mathrm{d}s e^{-\myp{t-s} \kappa} \prod_{j=1}^n b \myp{k_j,\sigma_j,s} \nn \\
	&- i\lambda \int_0^t \mathrm{d}s \sum_{m=1}^n \sigma_m \iiint_{\myp{\Omega^{\ast}}^3} \mathrm{d}k'_2 \mathrm{d}k'_3 \mathrm{d}k'_4  \delta_L \myp{k_m-k_2'-k_3'-k_4'} \nn \\
	&\quad \times \exp{\myb[\big]{-\myp{t-s}\kappa - is \Theta \myp{k_2',k_3',k_4',\sigma_m}}} \nn \\
	&\quad \times \prod_{j=1}^{m-1} b \myp{k_j,\sigma_j,s} \myb[\big]{ \Psi_1 \myp{k_2',k_3',k_4',\sigma_m} b \myp{k_2',-1,s} b \myp{k_3',\sigma_m,s} b \myp{k_4',1,s} \nn \\
	&\quad \quad +\Psi_0 \myp{k_2',k_3',k_4',\sigma_m} \myb{b \myp{k_2',-1,s} b \myp{k_3',\sigma_m,s} b \myp{k_4',1,s}}^T } \prod_{j=m+1}^{n} b \myp{k_j,\sigma_j,s},
\label{eq:duhamelproduct}
\end{align}
which provides the basis for the full pertubation expansion of $b \myp{k,\sigma,t}$ that we will use to treat the correlation function \eqref{eq:Qlambda_def}.

Observe that the term in \eqref{eq:duhamelproduct} associated to $\Psi_1$ contains a product of $n+2$ field operators $b$.
The idea is then to iteratively plug \eqref{eq:duhamelproduct} back into itself in the term containing $\Psi_1$, using different values of $\kappa$ at each step.
Starting with $n=1$, we label the terms in \eqref{eq:duhamelproduct} as
\begin{align*}
	b \myp{k,\sigma,t}
	={}& \mathcal{E}_0^0 \myp{t,k,\sigma,\kappa_0} \myb{b \myp{0}} + \kappa_0 \int_0^t \id s \mathcal{E}_0^1 \myp{s,t,k,\sigma,\kappa_0} \myb{b \myp{s}} \\
	&+ \int_0^t \id s \mathcal{E}_1^2 \myp{s,t,k,\sigma,\kappa_0} \myb{b \myp{s}} + \int_{0}^t \id s \mathcal{E}_1^3 \myp{s,t,k,\sigma,\kappa_0} \myb{b \myp{s}},
\end{align*}
where $\mathcal{E}_0^0$ is the first term in \eqref{eq:duhamelproduct}, $\mathcal{E}_0^{1}$ comes from the second term, $\mathcal{E}_1^2$ is the term containing $\Psi_0$, and $\mathcal{E}_1^3$ is the term with $\Psi_1$.
The subscripts appearing on the $\mathcal{E}$'s will indicate that any term in $\mathcal{E}_n$ contains $2n+1$ factors of $b$.
Because of the integral over $s$, it is convenient to imagine that we split the time interval $\myb{0,t}$ into two slices of length $s_0 = s$ and $s_1 = t-s$.
Denoting $\sigma_1 =-1$, $\sigma_2 = \sigma$, and $\sigma_3 = 1$, we can write for instance
\begin{align*}
	\MoveEqLeft[2] \mathcal{E}_1^3 \myp{s,t,k,\sigma,\kappa_0} \myb{b\myp{s_0}}
	={} -i\lambda \int_{\myp{\Omega^{\ast}}^3} \ids k_1 \ids k_2 \ids k_3 \delta_L \myp{k-k_1-k_2-k_3} \\
	&\times \sigma \Psi_1 \myp{k_1,k_2,k_3,\sigma} \prod\limits_{j=1}^3 b\myp{k_j,\sigma_j,s_0} \int_{\R_+} \ids s_1 \delta \myp{t-s_0-s_1} e^{-s_1 \kappa_0 - is_0 \Theta \myp{k_1,k_2,k_3,\sigma}},
\end{align*}
and similarly for $\mathcal{E}_1^2$.
Then, plugging \eqref{eq:duhamelproduct} into $\mathcal{E}_1^3$ replaces $\mathcal{E}_1^3$ with a sum of four new terms $\mathcal{E}_1^0$, $\mathcal{E}_1^1$, $\mathcal{E}_2^2$, $\mathcal{E}_2^3$.

Repeating this process $N$ times, we obtain an $N$-time Duhamel generalization of \eqref{eq:duhamelproduct}.
In this expansion, the time interval $ \myb{0,t} $ is divided into $N+1$ time slices $ \myb{0,s_0} $, $ \myb{s_0,s_0+s_1} $, $\dots$, $ \myb{s_0+\cdots+s_{N-1},t} $ and $t=s_0+\cdots+s_{N}$.
In order to better represent these Duhamel expansions, we will construct their presentations in terms of Feynman diagrams in the next subsection. 

\subsubsection{Feynman diagrams}\label{Duhamel}
As discussed in the previous subsection, the time interval $ \myb{0,t}$ is divided into $N+1$ time slices, each of length $s_i$.
We define the sets of indices
\begin{equation}
\label{IndexSet1}
	I_N=\Set{1,\dotsc,N} \mbox{ and } I_{N',N}=\Set{N',N'+1,\dotsc,N}.
\end{equation}
We now how to construct the Feynman diagrams corresponding to the expansion.
The time slices are represented from the bottom to the top of the diagram, with the lengths, $s_0$, $s_1$, $\dots$, $s_N$, as shown in \cref{Fig1}. 

\begin{figure}
	\centering
	\includegraphics{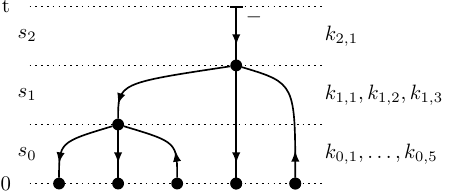}
	\caption{An example of a Feynman diagram. At time slice $s_0$, the edges are $k_{0,1},k_{0,2},k_{0,3},k_{0,4},k_{0,5}$, with the parities $-,-,+,-,+$. At time slice $s_1$, the edges are $k_{1,1},k_{1,2},k_{1,3}$, with the parities $-,-,+$.  At time slice $s_2$, the edge is $k_{2,1}$, with parity $-$. }
\label{Fig1}
\end{figure}

\subsubsection*{Indexing the segments}
At time slice $s_j$, since we expand the term containing $\Psi_1$, the delta function associated to $\Psi_1$ in \eqref{eq:duhamelproduct} means that the three momenta $k_2,k_3,k_4$ are combined into the momentum $k_1$.
This is represented on the diagram by the fact that at time slice $s_j$, there is exactly one triplet of the segments of time slice $s_{j-1}$ fuses into one segment of time slice $s_j$.
Those segments of time slice $s_{j-1}$ are denoted by $k_{j-1,\ell_j},k_{j-1,\ell_j+1},k_{j-1,\ell_j+2}$ and the one at time slice $s_j$ is denoted by $k_{j,\ell_j}$.
By the delta function associated to $\Psi_1$ in \eqref{eq:duhamelforbUSimplified}, we have $k_{j,\ell_j} = k_{j-1,\ell_j} + k_{j-1,\ell_j+1} + k_{j-1,\ell_j+2}$.
In these notations, the $j$ is the index of time slice $s_j$ and $\ell_j$ is the index of the segment where the fusion happens i.e. where we plug \eqref{eq:duhamelproduct} into itself at time slice $s_j$.
As we notice from \eqref{eq:duhamelproduct} there is a sign parameter $\sigma$ associated to each momentum $k$.
We call those parameters \emph{parities}.
We then denote these parameters for $k_{j-1,\ell_j},k_{j-1,\ell_j+1},k_{j-1,\ell_j+2}$, and $k_{j,\ell_j}$ by $\sigma_{j-1,\ell_j}, \sigma_{j-1,\ell_j+1}, \sigma_{j-1,\ell_j+2}$, and $\sigma_{j,\ell_j}$.
Now, at the very first time slice $N$, whose length is $s_N$, there is only one segment $k_{N,1}$ being decomposed into $k_{N-1,1}, k_{N-1,2},  k_{N-1,3}$.
In this case, $\ell_N=1$.
We denote the number of segments at this time slice by $m_0=1$.
However, at the second time slice $s_{N-1}$, $\ell_{N-1}$ can be either $1,2$, or $3$.
Each choice of $\ell_{N-1}$ leads to a different Feynman diagram.
The number of segments at this time slice is $m_1=m_0+2=3$.
By induction, the number of segments at time slice $s_{j}$ is $m_{N-j}=1+2(N-j)$.
In this time slice $s_{j}$, the fusion only happens at the segment $k_{j,\ell_j}$, where the three segments $\sigma_{j-1,\ell_j},\sigma_{j-1,\ell_j+1},\sigma_{j-1,\ell_j+2}$ are combined.
This leads us to the following way of indexing the segments in two consecutive time slices $s_{j-1}$ and $s_j$: $k_{j,l}=k_{j-1,l}, \sigma_{j,l}=\sigma_{j-1,l}, $ for  $l \in \Set{1,\cdots,\ell_j-1}$ and $k_{j,l}=k_{j-1,l+2}, \sigma_{j,l}=\sigma_{j-1,l+2}, $ for $l \in \Set{\ell_j+1,\cdots,m_{N-i}}$, with the notice that $m_{N-i}=m_{N-i+1}-2$.
Looking closer into equation \eqref{eq:duhamelproduct}, we can see that the parities associated to $k_1,k_2,k_3,k_4$ are $\sigma,-1,\sigma,1$ in the term containing $\Psi_1$.
This essentially means that $ \sigma_{j,\ell_j}=\sigma_{j-1,\ell_j+1}$, $ \sigma_{j-1,\ell_j}=-1$ and $ \sigma_{j-1,\ell_j+2}=1$.
We now denote  $\ell = \myp{\ell_1,\cdots, \ell_N}$, which is a vector in the space 
\begin{equation}
	G_N \coloneq{} I_{m_{N-1}} \times I_{m_{N-2}} \times\cdots \times I_{m_0}.
\end{equation}
We also need the notation 
\begin{equation}
\label{IndexSet2}
	\mathcal{I}_{n;m_0} \coloneq{} \Set{(j,l) \mid 0\le j\le n-1, 1\le l \le m_0+2(n-j) }, 
	\mbox{ and } 
	\mathcal{I}_n \coloneq{} \mathcal{I}_{n;1}, 
\end{equation}
which collects all index pairs associated to the above segments, excluding the final time slice.

Graphically, we place an \emph{interaction vertex} in the diagram at the places where three segments fuse into one.
We also place an \emph{initial time vertex} at the bottom of each of the $2N+1$ edges at the bottom of the diagram (at time $t=0$).
The parities $\sigma_{j,l}$ are represented on the diagram by adding an arrow on the corresponding edge.
An upward pointing arrow represents the parity $\sigma = 1$, while a downward pointing arrow represents $\sigma = -1$.
See \cref{Fig1} for a simple example of a Feynman diagram.

With all of these definitions in hand, we are able to write down the full Duhamel expansion, up to $N$ time slices. 

\subsubsection*{The full Duhamel expansion}
Plugging  \eqref{eq:duhamelproduct} back into itself at the terms containing $\Psi_1$, we find by induction that
\begin{align}
	b \myp{k,\sigma,t} 
	={}& \sum_{n=0}^{N-1} \mathcal{E}_{n}^0 \myp{t,k,\sigma,\kappa} \myb{b \myp{0}} + \sum_{n=0}^{N-1} \kappa_n \int_0^t \mathrm{d}s_0 \mathcal{E}_{n}^1 \myp{s_0,t,k,\sigma,\kappa} \myb{b \myp{s_0}} \nn \\
	&+ \sum_{n=1}^{N} \int_0^t \mathrm{d}s_0 \mathcal{E}_{n}^2 \myp{s_0,t,k,\sigma,\kappa} \myb{b \myp{s_0}} + \int_0^t \mathrm{d}s_0 \mathcal{E}_{N}^3 \myp{s_0,t,k,\sigma,\kappa} \myb{b \myp{s_0}}.
\label{eq:fullDuhamel}
\end{align}
The terms  in the formulation \eqref{eq:fullDuhamel} are explicitly written as follows.

(i) $\kappa$ denotes the vector $ \kappa = \myp{\kappa_0,\cdots,\kappa_{N-1}}$ in $\R_+^N$.

(ii) The first quantity $\mathcal{E}_{n}^0 \myp{t,k,\sigma,\kappa}$ in \eqref{eq:fullDuhamel} takes the form (for $n\geq 1$)
\begin{align}
	\mathcal{E}_{n}^0 \myp{t,k_{n,1},\sigma_{n,1},\kappa} \myb{b}
	={}& \myp{-i\lambda}^n \sum_{\ell \in G_n, \bar{\sigma} \in \Set{\pm 1}^{\mathcal{I}_n}} \int_{\myp{\Omega^{\ast}}^{\mathcal{I}_n}} \mathrm{d} \bar{k} \Delta_{n,\ell} \myp{\bar{k},\bar{\sigma}} \prod_{j=1}^{1+2n} b \myp{k_{0,j},\sigma_{0,j}} \nn \\
	&\times \prod_{j=1}^n \myb[\big]{\sigma_{j,\ell_j} \Psi_1 \myp{k_{j-1:\ell_{j}},\sigma_{j,\ell_j}} } \nn \\
	&\times \int_{\myp{\R_+}^{I_{0,n}}} \mathrm{d}s \delta \myp[\bigg]{t-\sum_{j=0}^n s_j} \prod_{j=0}^{n} e^{-s_j \kappa_{n-j}} \prod_{j=1}^{n} e^{-it_j \myp{s} \Theta_{j-1:\ell_{j}} \myp{\bar{k},\bar{\sigma}}},
\label{eq:DefE0}
\end{align}
in which $t_j \myp{s} = \sum_{l=0}^{j-1} s_l$ is the total time before time slice $s_j$; $k_{j:l}$ denotes the triplet $(k_{j,l},k_{j,l+1},k_{j,l+2})$; $\bar{k}$ and $\bar{\sigma}$ represent all of the quantities $k_{j,l},\sigma_{j,l}$ appearing in the integration; and $\Theta_{j:l} \myp{\bar{k},\bar{\sigma}}$ is the function $\Theta \myp{k_{j:l},\sigma_{j+1,l}}$.
Moreover, $\mathcal{E}_{0}^0 \myp{t,k,\sigma,\kappa} \myb{b} = e^{-\kappa_0 t} b \myp{k,\sigma}$.

(iii) The function $\Delta_{n,\ell}$ contains all of the $\delta$-functions appearing up to the $n$th time slice.
To be precise,
\begin{align}
	\Delta_{n,\ell} \myp{\bar{k},\bar{\sigma}}
	={} \prod_{j=1}^{n} \myt[\bigg]{ &\prod_{l=1}^{\ell_j-1} \myb[\Big]{\delta_L \myp{k_{j,l}-k_{j-1,l}} \1 \myp{\sigma_{j,l}=\sigma_{j-1,l}} }\nn \\
	&\times \delta_L \myp[\big]{k_{j,\ell_j}-k_{j-1,\ell_j}-k_{j-1,\ell_j+1}-k_{j-1,\ell_j+2} } \nn \\
	&\times \prod_{l=\ell_j+1}^{m_{n-j}} \myb[\Big]{ \delta_L \myp{k_{j,l}-k_{j-1,l+2}} \1 \myp{\sigma_{j,l} = \sigma_{j-1,l+2}}} \nn \\
	&\times \1 \myp{\sigma_{j-1,\ell_j}=-1} \1 \myp{\sigma_{j-1,\ell_j+1} = \sigma_{j,\ell_j}} \1 \myp{\sigma_{j-1,\ell_{j+2}}=1} }.
\label{eq:DefDelta}
\end{align}
This function makes sure that $k_{j,l}=k_{j-1,l}, \sigma_{j,l}=\sigma_{j-1,l} $, for  $l \in \Set{1,\cdots,\ell_j-1}$; $k_{j,l}=k_{j-1,l+2}, \sigma_{j,l}=\sigma_{j-1,l+2} $, for $ l \in \Set{\ell_j+1,\cdots,m_{N-j}} $; $ \sigma_{j,\ell_j}=\sigma_{j-1,\ell_j+1}$; $ \sigma_{j-1,\ell_j}=-1$; $ \sigma_{j-1,\ell_j+2}=1$; and $k_{j,\ell_j}-k_{j-1,\ell_j}-k_{j-1,\ell_j+1}-k_{j-1,\ell_j+2}=0$, as discussed in the previous subsection. 

(iv) The last quantity $\mathcal{E}_n^3$ can now be defined as:
\begin{align}
	\MoveEqLeft[3] \mathcal{E}^3_{n} \myp{s_0,t,k_{n,1},\sigma_{n,1},\kappa} \myb{b \myp{s_0}} \nn \\
	={}& \myp{-i\lambda}^n \sum_{\ell \in G_n, \bar{\sigma} \in \Set{\pm 1}^{\mathcal{I}_n}} \int_{\myp{\Omega^{\ast}}^{\mathcal{I}_n}} \mathrm{d} \bar{k} \Delta_{n,\ell} \myp{\bar{k},\bar{\sigma}}
	\prod_{j=1}^n \myb[\big]{ \sigma_{j,\ell_j} \Psi_1 \myp{k_{j-1:\ell_j},\sigma_{j,\ell_j}} } \nn \\
	&\times \prod_{j=1}^{1+2n} b \myp{k_{0,j},\sigma_{0,j}, s_0} \int_{\myp{\R_+}^{I_{n}}} \mathrm{d}s \delta \myp[\bigg]{t-\sum_{j=0}^n s_j} \prod_{j=1}^{n} e^{-s_j \kappa_{n-j}-i t_j \myp{s_0,s} \Theta_{j-1:\ell_{j}} \myp{\bar{k},\bar{\sigma}} }.
\label{eq:DefE3}
\end{align}
The only difference between the two formulas \eqref{eq:DefE0} and \eqref{eq:DefE3} is the integration with respect to $\mathrm{d}s$.
In \eqref{eq:DefE0} the integration is taken over $ \myp{\R_+}^{I_{0,n}}$ and in \eqref{eq:DefE3} it is over $\myp{\R_+}^{I_{n}}$, and the factor $e^{-s_0 \kappa_n}$ does not appear in \eqref{eq:DefE3}. 

(v) The second quantity $\mathcal{E}_n^1$ is defined as follows.
We set $\mathcal{E}^1_{0} \myp{s_0,t,k,\sigma,\kappa} \myb{b \myp{s_0}} = e^{- \myp{t-s_0} \kappa_0} b \myp{k,\sigma,s_0}$, and for $n>0$, put
\begin{align}
\label{eq:DefE1}
	\MoveEqLeft[6] \mathcal{E}^1_n \myp{s_0,t,k_{n,1},\sigma_{n,1},\kappa} \myb{b \myp{s_0}}
	={}  \int_0^{t-s_0} \mathrm{d}r e^{-r \kappa_n} \mathcal{E}^3_{n} \myp{s_0+r,t,k_{n,1},\sigma_{n,1},\kappa} \myb{b\myp{s_0}} \\
	={}& \myp{-i\lambda}^n \int_0^{t-s_0} \id r \sum_{\ell \in G_n, \bar{\sigma} \in \Set{\pm 1}^{\mathcal{I}_n}} \int_{\myp{\Omega^{\ast}}^{\mathcal{I}_n}} \mathrm{d}\bar{k} \Delta_{n,\ell} \myp{\bar{k}, \bar{\sigma}} \nn \\
	&\times e^{-r \kappa_n} \prod_{j=1}^{1+2n} b \myp{k_{0,j},\sigma_{0,j}, s_0} \prod_{j=1}^n \myb[\big]{\sigma_{j,\ell_j} \Psi_1 \myp{k_{j-1:\ell_j},\sigma_{j,\ell_j}}} \nn \\
	&\times \int_{\myp{\R_+}^{I_{n}}} \mathrm{d}s \delta \myp[\bigg]{t-r-\sum_{j=0}^n s_j} \prod_{j=1}^{n} e^{-s_j \kappa_{n-j} - i t_j \myp{s_0+r,s} \Theta_{j-1:\ell_{j}} \myp{\bar{k},\bar{\sigma}}}. \nn 
\end{align}

(vi) And, finally,
\begin{align}
	\MoveEqLeft[6] \mathcal{E}^2_{n} \myp{s_0,t,k_{n,1},\sigma_{n,1},\kappa} \myb{b \myp{s_0}} \nn \\
	={}& \myp{-i\lambda}^n \sum_{\ell \in G_n, \bar{\sigma} \in \Set{\pm 1}^{\mathcal{I}_n}} \int_{\myp{\Omega^{\ast}}^{\mathcal{I}_n}} \mathrm{d} \bar{k} \Delta_{n,\ell} \myp{\bar k,\bar{\sigma}} \sigma_{1,\ell_1} \Psi_0 \myp{k_{0:\ell_1},\sigma_{1,\ell_1}} \nn \\
	&\times \prod_{j=2}^n \myb[\big]{\sigma_{j,\ell_j} \Psi_1 \myp{k_{j-1:\ell_j},\sigma_{j,\ell_j}}} \prod_{j=1}^{\ell_1-1} b \myp{k_{0,j},\sigma_{0,j}, s_0} \nn \\
	&\times \myb[\Bigg]{ \prod_{l=\ell_1}^{\ell_1+2} b \myp{k_{0,l},\sigma_{0,l},s_0} }^T \prod_{j=\ell_1+3}^{1+2n} b \myp{k_{0,j},\sigma_{0,j},s_0} \nn \\
	&\times \int_{\myp{\R_+}^{I_{n}}} \mathrm{d}s \delta \myp[\bigg]{t-\sum_{j=0}^n s_j} \prod_{j=1}^{n} e^{-s_j \kappa_{n-j}-i t_j \myp{s_0,s} \Theta_{j-1:\ell_{j}} \myp{\bar{k},\bar{\sigma}}},
\label{eq:DefE2}
\end{align}
where $\bar{k}$ and $\bar\sigma$ represent all of the quantities $k_{j,l},\sigma_{j,l}$ appearing in the integration.

\subsubsection{Expansions of $Q^{\lambda}$.}
\label{Sec:Q_expansions}
For the remainder of the section, we fix $k,k' \in \Omega^{\ast}$ and consider $Q^{\lambda} \myp{k',k,t}$.
Plugging \eqref{eq:fullDuhamel} into the definition \eqref{eq:Qlambda_def} of $Q^{\lambda}$, we obtain an expansion
\begin{equation}
\label{Def:OperatorQLambda2}
	Q^{\lambda} \myp{k',k,t}
	={} \expec[\big]{b \myp{k',-1,0} b \myp{k,1,t}}_{\lambda}
	={} Q^{main} + Q^{err}_1 + Q^{err}_2 + Q^{err}_3,
\end{equation}
where 
\begin{equation}
\label{Def:OperatorQMain}
	Q^{main} \myp{k',k,t} = \sum_{n=0}^{N-1} \expec[\big]{b \myp{k',-1,0} \mathcal{E}_{n}^0 \myp{t/\eps,k,1,\kappa} \myb{b \myp{0}} }_{\lambda},
\end{equation}
and
\begin{equation}\label{Def:OperatorQErrors}
	\begin{aligned}
		Q^{err}_1 \myp{k',k,t} ={}& \sum_{n=0}^{N-1} \kappa_n \int_0^{t/\eps} \ids s \expec[\big]{ b \myp{k',-1,0} \mathcal{E}^1_n \myp{s,t/\eps,k,1,\kappa} \myb{b \myp{s}} }_{\lambda}, \\
		Q^{err}_2 \myp{k',k,t} ={}& \sum_{n=1}^{N} \int_0^{t/\eps} \ids s \expec[\big]{ b \myp{k',-1,0} \mathcal{E}^2_n \myp{s,t/\eps,k,1,\kappa} \myb{b \myp{s}} }_{\lambda},\\
		Q^{err}_3 \myp{k',k,t} ={}& \int_0^{t/\eps} \ids s \expec[\big]{ b \myp{k',-1,0} \mathcal{E}^3_N \myp{s,t/\eps,k,1,\kappa} \myb{b \myp{s}} }_{\lambda}.
	\end{aligned}
\end{equation}
To study in details the operators $Q^{main}, Q^{err}_1, Q^{err}_2, Q^{err}_3$, we will need the following definition of truncated moments, which is a generalization of \eqref{Def:Truncation} and \eqref{Def:Truncation2}. 
\begin{definition}
\label{Def:Truncation:0}
	For any finite, non-empty, ordered set $I$, we denote by $\pi \myp{I}$ the set of its ordered partitions: $S\in \pi \myp{I}$ if and only if $S\subset \mathcal{P} \myp{I}$ such that each $A\in S$ is ordered (carrying the order from $I$) and non-empty, $\cup_{A\in S} A=I$ and if $A, A'\in S$, $A\ne A$ then $A'\cap A=\emptyset$. We also assume that if $A\in S$, then $A$ has an even number of elements. Moreover, $\pi \myp{\emptyset}=\emptyset$. 
	
	For any state $\expec{\nonarg}$, consider the moment
	\begin{equation*}
		\expec[\bigg]{ \prod_{j=1}^n  a \myp{x_j,\sigma_j} }.
	\end{equation*}
	This quantity is $0$ if $n$ is odd.
	We define its truncation in the recursive manner 
	\begin{equation}
	\label{Def:Truncation:2}
		\expec[\bigg]{ \prod_{j=1}^n a \myp{x_j,\sigma_j} }
		= \sum_{S \in \pi \myp{I_n}} \epsilon(S) \prod_{A \in S} \expec[\bigg]{ \prod_{j \in A} a \myp{x_j,\sigma_j} }^T,
	\end{equation}
	where $\epsilon(S)$ is the sign of the permutation corresponding to $S$, which is defined as follows:
	Labelling the elements of $S = \Set{A_1, \dotsc, A_m}$, $\epsilon \myp{S}$ is the sign of the permutation that maps the ordered set $I_n$ to the ordered set $\Set{A_1,\dotsc,A_m}$.
	Since each $A_j$ is assumed to have an even number of elements, $\epsilon \myp{S}$ is independent of the labeling of the elements of $S$ (see \cref{lem:cluster_sign_order} below).
	
	We also define the truncated correlation function
	\begin{equation}
	\label{Def:Truncation:3}
		{C}_n \myp{k,\sigma,\lambda,\Omega}
		\coloneq{} \sum_{x \in \Omega^n} \1_{\Set{x_1=0}} e^{-2\pi i \sum_{j=1}^n x_j\cdot k_j} \expec[\bigg]{ \prod_{j=1}^n a \myp{x_j,\sigma_j} }^T.
	\end{equation}
\end{definition}
\begin{lemma}
\label{Lemma:CummulantEstimate}
Assume that the state $\expec{\nonarg}$ is translation invariant.
Then, for any ordered index set $I$, and any $k\in \myp{\Omega^{\ast}}^I$, $\sigma \in \Set{\pm 1}^I$, we have
\begin{equation}
\label{Lemma:CummulantEstimate:1}
	\expec[\bigg]{ \prod_{j\in I} \hat{a} \myp{k_j,\sigma_j} }
	= \sum_{S \in \pi \myp{I}} \epsilon(S) \prod_{A \in S} \myb[\bigg]{ \delta_L \myp[\bigg]{ \sum_{j \in A} k_j } C_{\abs{A}} \myp{k_A,\sigma_A,\lambda,\Omega} },
\end{equation}
where $k_A \coloneq \myp{k_j}_{j \in A}$ contains the components of $k$ with index belonging to $A$.
\end{lemma}
\begin{proof}
	We need to study
	\begin{equation*}
		\expec[\bigg]{ \prod_{j\in I} \hat{a} \myp{k_j,\sigma_j} }
		= \sum_{x \in \Omega^I} e^{-2\pi i \sum_j x_j\cdot k_j} \expec[\bigg]{\prod_{j\in I} a \myp{x_j,\sigma_j} }.
	\end{equation*}
	By inserting \eqref{Def:Truncation:2} into the above equation and rearranging the terms, we find
	\begin{align*}
		\expec[\bigg]{ \prod_{j\in I} \hat{a} \myp{k_j,\sigma_j} } 
		={}& \sum_{x\in\Omega^I} e^{-2\pi i \sum_j x_j\cdot k_j} \sum_{S \in \pi \myp{I}} \epsilon \myp{S} \prod_{A\in S} \expec[\bigg]{ \prod_{j \in A} a \myp{x_j,\sigma_j} }^T \nn \\
		={}& \sum_{S \in \pi \myp{I}} \epsilon \myp{S} \prod_{A\in S} \myb[\Bigg]{ \sum_{x \in \Omega^{\abs{A}}} e^{-2\pi i \sum_{j \in A} x_j\cdot k_j} \expec[\bigg]{ \prod_{j \in A} a \myp{x_j,\sigma_j} }^T }.
	\end{align*}
	\eqref{Lemma:CummulantEstimate:1} then follows by the translation invariance and \eqref{Def:Truncation:3}.
\end{proof}
\begin{remark}[Truncated pair correlations]
\label{rem:pair_correl_trunc}
	Suppose that $n=2$ and consider the Gibbs state $\expec{\nonarg}_{\lambda}$.
	Since truncating a moment with just two factors is a trivial operation, we have by definition \eqref{Def:Truncation:3} and \cref{lem:wignerfunction} that
	\begin{align}
		C_2 \myp{\myp{k_1,k_2}, \myp{\sigma_1,\sigma_2},\lambda,\Omega}
		={}& \sum\limits_{x_2 \in \Omega} e^{-2 \pi i x_2 \cdot k_2} \expec[\big]{a \myp{0,\sigma_1} a \myp{x_2,\sigma_2}}_{\lambda} \nn \\
		={}& W_L^{\lambda} \myp{k_2,\sigma_1} \1 \myp{\sigma_1 = - \sigma_2},
	\label{eq:pair_correl_trunc}
	\end{align}
	where we used the notation from \eqref{eq:Wnotations}.
\end{remark}
\begin{remark}[Permutations for pairing clusters]
	Consider a cluster decomposition $S = \Set{A_1, \dotsc, A_m}$ where all the clusters form pairs, $\abs{A_j} = 2$.
	Denoting $A_j = \Set{a_{j,1},a_{j,2}}$ and ordering the $A_j$'s such that $a_{j,1} < a_{k,1}$ when $j < k$, it is apparent that the permutation corresponding to $S$ belongs to the set of pairings
	\begin{align*}
		P_{2m} \coloneq{} \myt[\big]{\tau \in S_{2m} \mid\, &\tau \myp{2j-1} < \tau \myp{2j} \text{ for } j=1,\dotsc, m, \\
		&\tau \myp{2j-1} < \tau \myp{2j+1} \text{ for } j=1,\dotsc, m-1 },
	\end{align*}
	which is reminiscent of Wick's theorem, see e.g. \cite[Theorem 10.2]{Solovej-2014notes}.
\end{remark}

A straightforward application of \cref{Lemma:CummulantEstimate} above leads to.
\begin{proposition}[Main term]
\label{Proposition:Ampl}
For any $N\geq 1$, we have
\begin{equation}
	\label{Proposition:Ampl:1}
	Q^{main} \myp{k',k,t}
	={} \sum_{n=0}^{N-1} \sum_{\ell \in G_n} \sum_{S \in \pi \myp{I_{0,2n+1}}} \mathcal{F}^{main}_{n} \myp{S,\ell,k',k,t/\eps,\kappa},
\end{equation}
where, defining $\mathcal{I}_n'' = \mathcal{I}_n \cup \Set{\myp{n,1}} \cup \Set{\myp{0,0}}$,
\begin{align}
	\MoveEqLeft[2] \mathcal{F}^{main}_{n} \myp{S,\ell,k',k,t/\eps,\kappa}
	= \myp{-i\lambda}^n \epsilon \myp{S} \sum_{\bar{\sigma} \in \Set{\pm 1}^{\mathcal{I}_n''}} \int_{\myp{\Omega^{\ast}}^{\mathcal{I}_n''}} \ids \bar{k} \Delta_{n,\ell} \myp{\bar{k},\bar{\sigma}} \nn \\
	&\times \delta_L \myp{k_{n,1}-k} \delta_L \myp{k_{0,0}-k'} \prod_{A\in S} \myb[\bigg]{ \delta_L \myp[\bigg]{ \sum_{j\in A} k_{0,j} } C_{\abs{A}} \myp{k_{0,A},\sigma_{0,A},\lambda,\Omega} } \nn \\
	&\times \1 \myp{\sigma_{n,1}=1} \1 \myp{\sigma_{0,0}=-1} \prod_{j=1}^n \myb[\big]{ \sigma_{j,\ell_j} \Psi_1 \myp{k_{j-1:\ell_j},\sigma_{j,\ell_j}} } \nn \\
	&\times \int_{\myp{\R_+}^{I_{0,n}}} \mathrm{d}s \delta \myp[\bigg]{ \frac{t}{\eps}-\sum_{j=0}^n s_j} \prod_{j=0}^{n} e^{-s_j \kappa_{n-j}} \prod_{j=0}^{n-1} e^{- i s_j\sum_{l=j+1}^n \Theta_{l-1:\ell_{l}} \myp{\bar{k},\bar{\sigma}}}.
	\label{eq:DefFmain}
\end{align}
\end{proposition}
In the above expression, we have manually inserted the two additional momentum integrals
\begin{equation}
\label{eq:top_momenta}
	1 = \int_{\Omega^{\ast}} \ids k_{n,1} \delta_L \myp{k_{n,1}-k} \int_{\Omega^{\ast}} \ids k_{0,0} \delta_L \myp{k_{0,0}-k'},
\end{equation}
which will ensure that there is a momentum integral associated to every edge connected to an interaction vertex in the Feynmann diagram corresponding to \eqref{eq:DefFmain}, see \cref{Sec:Diagrams} below.
Let us consider the error terms $Q^{err}_1$, $Q^{err}_2$ and $Q^{err}_3$, that would lead to sums over terms of the type:
\begin{equation}
	\label{Def:OperatorQErrors:1}
	\int_0^{t/\eps} \ids s \expec[\big]{ b \myp{k',-1,0} \mathcal{E}^j_n (s,t/\eps,k,1,\kappa) \myb{b \myp{s}} }_{\lambda},
\end{equation}
which can be bounded, using the Jensen and Cauchy-Schwarz inequalities, as
\begin{align}
	\MoveEqLeft[4] \abs[\bigg]{ \int_0^{t/\eps} \ids s \expec[\big]{ b \myp{k',-1,0} \mathcal{E}^j_n (s,t/\eps,k,1,\kappa) \myb{b \myp{s}} }_{\lambda} }^2 \nn \\
	\leq{}& \frac{t}{\eps} \expec[\Big]{ \abs[\big]{ b \myp{k',-1,0}^{\ast} }^2 }_{\lambda} \int_0^{t/\eps} \ids s \expec[\Big]{ \abs[\big]{ \mathcal{E}^j_n \myp{s,t/\eps,k,1,\kappa} \myb{b \myp{s}} }^2 }_{\lambda},
	\label{Def:OperatorQErrors:2}
\end{align}
where the absolute value should be understood in the sense of operators, $\abs{A}^2 = A^{\ast} A$.
Notice that the first factor can be expressed by \cref{lem:wignerfunction} as 
\begin{equation*}
	\expec[\Big]{ \abs[\big]{ b \myp{k',-1,0}^{\ast} }^2 }_{\lambda}
	={} \expec[\big]{ \hat{a} \myp{-k'}^{\ast} \hat{a} \myp{-k'} }_{\lambda}
	={} L^d W_L^{\lambda} \myp{-k'}
	\leq{} L^d.
\end{equation*}
As a result, recalling the prefactor $L^{-d}$ in \eqref{eq:main}, instead of $Q^{err}_1$, $Q^{err}_2$ and $Q^{err}_3$, we need to estimate
\begin{equation}
\label{Def:OperatorQErrors:Bis}
\begin{aligned}
	\bar{Q}^{err}_1 \myp{k,t} \coloneq{}& \sum_{n=0}^{N-1} \kappa_n \myp[\bigg]{\frac{1}{L^d} \frac{t}{\eps} \int_0^{t/\eps} \ids s \expec[\Big]{ \abs[\big]{ \mathcal{E}^1_n \myp{s,t/\eps,k,1,\kappa} \myb{b \myp{s}} }^2 }_{\lambda}}^{\frac{1}{2}}, \\
	\bar{Q}_2^{err} \myp{k,t} \coloneq{}& \sum_{n=1}^{N} \myp[\bigg]{\frac{1}{L^d} \frac{t}{\eps} \int_0^{t/\eps} \ids s \expec[\Big]{ \abs[\big]{ \mathcal{E}^2_n \myp{s,t/\eps,k,1,\kappa} \myb{b \myp{s}} }^2 }_{\lambda} }^{\frac{1}{2}}, \\
	\bar{Q}_3^{err} \myp{k,t} \coloneq{}& \myp[\bigg]{ \frac{1}{L^d} \frac{t}{\eps} \int_0^{t/\eps} \ids s \expec[\Big]{ \abs[\big]{ \mathcal{E}^3_N \myp{s,t/\eps,k,1,\kappa} \myb{b \myp{s}} }^2 }_{\lambda} }^{\frac{1}{2}}.
\end{aligned}
\end{equation}
To proceed further, we will need the following concept of interlacing two time sequences.
\begin{definition}
\label{Def:J}
	Let $n,n' \in \N_0$.
	Let $J$ be a map from $I_{n+n'}$ to $\Set{\pm 1}$.
	Consider the inverse images $J^{-1} \myp{1}$ and $J^{-1} \myp{-1}$.
	If  $J^{-1} \myp{1}$ and $J^{-1} \myp{-1}$ has $n$ and $n'$ elements, respectively, then $J$ is said to  \emph{interlace} $ \myp{n,n'}$.
	We then define further two maps
	\begin{equation}
	\label{eq:J_plusminus}
		J_+ \myp{j,J} \coloneq \sum_{l=1}^j \1 \myp{J \myp{l}=1},
		\qquad
		J_- \myp{j,J} \coloneq \sum_{l=1}^j \1 \myp{J(l)=-1}.
	\end{equation}
	We also define $\gamma \myp{j;J}$ as follows
	\begin{align}
		\gamma \myp{j;J} 
		={}& \sum_{l=1}^{2 \myp{n-m}+1} \sigma_{m,l} \omega^{\lambda} \myp{k_{m,l}} + \sum_{l=1}^{2 \myp{n-m'}+1} \sigma_{m',l}' \omega^{\lambda} \myp{k_{m',l}'} \nn \\
		&- i \myp{\kappa_{n-m}+\kappa_{n-m'}},
	\label{eq:DefVarthetajJ}
	\end{align}
	with $m = m \myp{j} = J_+ \myp{j-2,J}+1$ and $m' = m' \myp{j} = J_- \myp{j-2,J}+1$.
\end{definition}
We then have the following Lemma, that allows us to interlace the Duhamel expansions.
The proof of this lemma can be found in \cite[Lemma 4.6]{LukkarinenSpohn:WNS:2011}.
\begin{lemma}
\label{Lemma:Interlace}
	For $t>0$, $n,n'\ge 0$, and suppose that $\gamma^+_l, \gamma_{l'}^- \in \C$ are given for $l \in I_{0,n}$, $l' \in I_{0,n'}$.
	It holds true that
	\begin{align}
		\MoveEqLeft[4] \int_{\myp{\R_+}^{I_{0,n}}} \ids s \delta \myp[\bigg]{t-\sum_{l=0}^n s_l} \prod_{l=0}^n e^{-i s_l \gamma^+_l} \times \int_{\myp{\R_+}^{I_{0,n'}}} \ids s' \delta \myp[\bigg]{ t-\sum_{l=0}^{n'} s_l'} \prod_{l=0}^{n'} e^{-i s_l' \gamma^-_l} \nn \\
		={}& \sum_{J \mbox{ interlaces } \myp{n,n'}} \int_{\myp{\R_+}^{I_{0,n+n'}}} \ids r \delta \myp[\bigg]{t-\sum_{j=0}^{n+n'} r_j } \prod_{j=0}^{n+n'} e^{-ir_j \myp{\gamma^+_{J_+ \myp{j;J}} + \gamma^-_{J_- \myp{j;J}} }}.
	\label{Lemma:Interlace:2}
	\end{align}
\end{lemma}
A straightforward application of Lemma \ref{Lemma:Interlace} leads to the following representations. 
\begin{proposition}
\label{Propo:QERR3}
	The components of $\bar{Q}^{err}_3$ can be expressed as
	\begin{align}
		\MoveEqLeft[4] \expec[\Big]{ \abs[\big]{ \mathcal{E}^3_n \myp{s,t/\eps,k,1,\kappa} \myb{b \myp{s}} }^2 }_{\lambda} \nn \\
		={}& \sum_{J \mbox{ interlaces } \myp{n-1,n-1}} \sum_{\ell,\ell' \in G_n} \sum_{S\in \pi \myp{I_{4n+2}}} \mathcal{F}^3_n \myp{S,J,\ell,\ell',k,t/\eps-s,\kappa},
	\label{Def:OperatorQErrors:Bis3}
	\end{align}
	where, denoting $ \mathcal{I}_n' = \mathcal{I}_n \cup \Set{\myp{n,1}} $,
	\begin{align}
		\mathcal{F}^3_n \myp{S,J,\ell,\ell',k,\tau,\kappa }
		={}& (-\lambda^2)^{n} \epsilon \myp{S} \sum_{\bar{\sigma},\bar{\sigma}' \in \Set{\pm 1}^{\mathcal{I}_n'}} \int_{\myp{\Omega^{\ast}}^{\mathcal{I}_n'}} \ids \bar{k} \int_{\myp{\Omega^{\ast}}^{\mathcal{I}_n'}} \ids \bar{k}' \Delta_{n,\ell} \myp{\bar{k},\bar{\sigma}} \nn \\
		&\times \Delta_{n,\ell'} \myp{\bar{k}',\bar{\sigma}'} \delta_L \myp{k_{n,1}-k} \delta_L \myp{k_{n,1}'+k} \1 \myp{\sigma_{n,1}=1} \nn \\
		&\times \1 \myp{\sigma_{n,1}'=-1} \prod_{A\in S} \myb[\bigg]{ \delta_L \myp[\bigg]{ \sum_{j\in A} K_j } C_{\abs{A}} \myp{K_A, o_A, \lambda, \Omega} } \nn \\
		&\times \prod_{j=1}^n \myb[\big]{ \sigma_{j,\ell_j} \Psi_1 \myp{k_{j-1:\ell_{j}},\sigma_{j,\ell_j}} \sigma_{j,\ell'_j}' \Psi_1 \myp{-k'_{j-1:\ell'_{j}},-\sigma'_{j,\ell'_j}} } \nn \\
		&\times \int_{\myp{\R_+}^{I_{2,2n}}} \ids s \delta \myp[\bigg]{ \tau-\sum_{j=2}^{2n} s_j } \prod_{j=2}^{2n} e^{-i s_j \gamma \myp{j;J}},
		\label{eq:DefF3}
	\end{align}
	where $K = \myp{k_{0,\nonarg}',k_{0,\nonarg}}$ and $o = \myp{\sigma_{0,\nonarg}', \sigma_{0,\nonarg}}$ denote the combined momentum and parity vectors at the initial time slice.
\end{proposition}
\begin{proof}
	We denote by $k_{n,1} = k$ the momentum associated to the top edge of the diagrams in the expansion of $\mathcal{E}_n^3$.
	Using first by \cite[Lemma 4.4]{LukkarinenSpohn:WNS:2011} that 
	\begin{equation*}
		\sum\limits_{j=1}^{2n+1} \sigma_{0,j} \omega^{\lambda} \myp{k_{0,j}}
		={} \sigma_{n,1} \omega^{\lambda} \myp{k_{n,1}} + \sum\limits_{j=1}^n \Theta_{j-1:\ell_j} \myp{\bar{k},\bar{\sigma}},
	\end{equation*}
	we can rewrite
	\begin{align}
		\prod_{j=1}^{2n+1} b \myp{k_{0,j}, \sigma_{0,j},s_0}
		={}& \prod_{j=1}^{2n+1} e^{i \sigma_{0,j} \omega^{\lambda} \myp{k_{0,j}} s_0} \hat{a} \myp{k_{0,j}, \sigma_{0,j},s_0} \nn \\
		={}& e^{i \sigma_{n,1} \omega^{\lambda} \myp{k_{n,1}}s_0} \prod_{j=1}^n e^{is_0 \Theta_{j-1;\ell_j} \myp{\bar{k},\bar{\sigma}}} \prod_{j=1}^{2n+1} \hat{a} \myp{k_{0,j}, \sigma_{0,j},s_0},
	\label{eq:b_prod_phase}
	\end{align}
	and thus (when ignoring the first factor on the right hand side above) the oscillating exponential in \eqref{eq:DefE3} becomes
	\begin{align*}
		\MoveEqLeft[6] \prod_{j=1}^{n} e^{is_0 \Theta_{j-1;\ell_j} \myp{\bar{k},\bar{\sigma}} -i t_j \myp{s_0,s} \Theta_{j-1:\ell_{j}} \myp{\bar{k},\bar{\sigma}} } \\
		={}& \prod_{j=2}^n \prod_{m=1}^{j-1} e^{-i s_m \Theta_{j-1:\ell_{j}} \myp{\bar{k},\bar{\sigma}} }
		={} \prod_{m=1}^{n-1} \prod_{j=m+1}^n e^{-i s_m \Theta_{j-1:\ell_{j}} \myp{\bar{k},\bar{\sigma}} }.
	\end{align*}
	It then follows that we can write
	\begin{align}
		\mathcal{E}^3_n \myp{s_0,t,k,1,\kappa} \myb{b \myp{s_0}} 
		={}& \myp{-i\lambda}^n \sum_{\ell \in G_n, \bar{\sigma} \in \Set{\pm 1}^{\mathcal{I}_n'}} \int_{\myp{\Omega^{\ast}}^{\mathcal{I}_n}} \ids \bar{k} \Delta_{n,\ell} \myp{\bar{k},\bar{\sigma}} \1 \myp{\sigma_{n,1}=1} \nn \\
		&\times e^{i \omega^{\lambda} \myp{k_{n,1}} s_0} \prod_{j=1}^{1+2n} \hat{a} \myp{k_{0,j},\sigma_{0,j}, s_0} \prod_{j=1}^n \myb[\big]{ \sigma_{j,\ell_j} \Psi_1 \myp{k_{j-1:\ell_j},\sigma_{j,\ell_j}} } \nn \\
		&\times \int_{\myp{\R_+}^{I_{n}}} \ids s \delta \myp[\bigg]{t-\sum_{j=0}^n s_j} \prod_{j=1}^{n} e^{-i s_j \gamma_j^+},
	\label{eq:E3_expansion}
	\end{align}
	where the $\gamma_j^+$ are defined by
	\begin{align}
		\gamma^+_n ={}& -i\kappa_0, \nn \\
		\gamma^+_l ={}& \sum_{j=l+1}^n \Theta_{j-1:\ell_j} \myp{\bar{k},\bar{\sigma}} - i\kappa_{n-l}, \mbox{ for } 1\leq l\leq n-1.
	\label{eq:gammaplus}
	\end{align}
	Next, we use this to expand the left hand side of \eqref{Def:OperatorQErrors:Bis3}, however, in the expansion of the conjugated factor we make a change of variables $\sigma_{i,j}' = -\sigma_{i,2\myp{n-i+1}-j}$, $k_{i,j}' = -k_{i,2\myp{n-i+1}-j}$, and $\ell_i' = 2 \myp{n-i+1}-\ell_i$.
	This corresponds to inverting the order and swapping the signs of all variables in each time slice, or, graphically speaking, inverting the corresponding Feynmann diagram and changing all the parities.
	Now, using that $\Theta \myp{- \myp{k_3,k_2,k_1},-\sigma} = -\Theta \myp{\myp{k_1,k_2,k_3},\sigma} $, the oscillating factor in the time integral from the conjugate tree can, as before, be written $ \prod_{j=1}^{n} e^{-i s_j \gamma_j^-} $, with $\gamma_j^-$ given by
	\begin{align}
		\gamma^-_{n'} ={}& -i\kappa_0, \nn \\
		\gamma^-_{l'} ={}& \sum_{j=l'+1}^n \Theta_{j-1:\ell_j'} \myp{\bar{k}',\bar{\sigma}'} - i\kappa_{n'-l}, \mbox{ for } 1\leq l' \leq n'-1.
	\label{eq:gammaminus}
	\end{align}
	Further, by stationarity of the initial state, the resulting expansion contains products of the form
	\begin{equation*}
		\expec[\bigg]{\prod_{j=1}^{1+2n} \hat{a} \myp{k_{0,j}', \sigma_{0,j}'} \prod_{j=1}^{1+2n} \hat{a} \myp{k_{0,j}, \sigma_{0,j}} }_{\lambda}.
	\end{equation*}
	After inserting the truncated correlation functions using \cref{Lemma:CummulantEstimate}, and applying \cref{Lemma:Interlace} to interlace the time sequences, we thus arrive at
	\begin{align*}
		\MoveEqLeft[4] \expec[\Big]{ \abs[\big]{ \mathcal{E}^3_n \myp{s,t/\eps,k,1,\kappa} \myb{b \myp{s}} }^2 }_{\lambda} \\
		={}& \sum_{J \textit{ interlaces } \myp{n-1,n-1}} \sum_{\ell,\ell' \in G_n} \sum_{S\in \pi \myp{I_{4n+2}}} \mathcal{F}^{3}_{n} \myp{S,J,\ell,\ell',k,t/\eps-s,\kappa},
	\end{align*}
	with
	\begin{align*}
		\mathcal{F}^{3}_{n}
		={}& \myp{-\lambda^2}^n \epsilon \myp{S} \sum\limits_{\sigma,\sigma' \in \Set{\pm 1}^{\mathcal{I}_n'}} \int_{\myp{\Omega^{\ast}}^{\mathcal{I}_n'}} \ids \bar{k} \int_{\myp{\Omega^{\ast}}^{\mathcal{I}_n'}} \ids \bar{k}' \Delta_{n,\ell} \myp{\bar{k},\bar{\sigma}} \Delta_{n,\ell'} \myp{\bar{k}',\bar{\sigma}'} \\
		&\times \delta_L \myp{k_{n,1}-k} \delta_L \myp{k_{n,1}'+k} \1 \myp{\sigma_{n,1}=1} \1 \myp{\sigma_{n,1}'=-1} \nn \\
		&\times e^{is \myp{\omega^{\lambda} \myp{k_{n,1}} - \omega^{\lambda} \myp{-k_{n,1}'} }} \prod_{A\in S} \myb[\bigg]{ \delta_L \myp[\bigg]{ \sum_{j\in A} K_j } C_{\abs{A}} \myp{K_A, o_A, \lambda, \Omega} } \\
		&\times \prod_{j=1}^n \myb[\big]{ \sigma_{j,\ell_j} \Psi_1 \myp{k_{j-1:\ell_{j}},\sigma_{j,\ell_j}} \sigma_{j,\ell'_j}' \Psi_1 \myp{-k'_{j-1:\ell'_{j}},-\sigma'_{j,\ell'_j}} } \nn \\
		&\times \int_{\myp{\R_+}^{I_{0,2n-2}}} \ids r \delta \myp[\bigg]{\frac{t}{\eps} -s -\sum_{j=0}^{2n-2} r_j } \prod_{j=0}^{2n-2} e^{-i r_j \myp{\gamma_{J_+ \myp{j;J}+1}^+ + \gamma_{J_- \myp{j;J}+1}^-} },
	\end{align*}
	where we again manually inserted the momentum integrals for the top edges.
	Recall here by definition of the top momenta that $k_{n,1} = k = -k_{n,1}' $ in $\T^d$.
	This means that $\omega^{\lambda} \myp{-k_{n,1}'} = \omega^{\lambda} \myp{k_{n,1}}$, so the $s$-dependent factor in the third line above vanishes.
	Finally, shifting the indices in the time integral upwards by two and setting $m = m\myp{j} =J_+ \myp{j-2,J}+1$ and $m' = m\myp{j}' =J_- \myp{j-2,J}+1$, we arrive at \eqref{eq:DefF3} after noting that
	\begin{align*}
		\MoveEqLeft[6] \gamma_{J_+ \myp{j-2;J}+1}^+ + \gamma_{J_- \myp{j-2;J}+1}^-
		={} \gamma_{m}^+ + \gamma_{m'}^- \\
		={}& \sum_{i=m+1}^n \Theta_{i-1:\ell_i} \myp{\bar{k},\bar{\sigma}} + \sum_{i=m'+1}^n \Theta_{i-1:\ell_i'} \myp{\bar{k}',\bar{\sigma}'} - i \myp{\kappa_{n-m} + \kappa_{n-m'}} \\
		={}& \sum_{i=1}^{2 \myp{n-m}+1} \sigma_{m,i} \omega^{\lambda} \myp{k_{m,i}} + \sum_{i=1}^{2 \myp{n-m'}+1} \sigma_{m',i}' \omega^{\lambda} \myp{k_{m,i}'} - i \myp{\kappa_{n-m} + \kappa_{n-m'}} \\
		={}& \gamma \myp{j;J},
	\end{align*}
	by \cite[Lemma 4.4]{LukkarinenSpohn:WNS:2011} and the definition \eqref{eq:DefVarthetajJ} of $\gamma \myp{j;J}$.
\end{proof}
\begin{proposition}\label{Propo:QERR1}
	The components of $\bar{Q}^{err}_1$ can be expressed as
	\begin{align}
		\MoveEqLeft[4] \expec[\Big]{ \abs[\big]{ \mathcal{E}^1_n \myp{\tilde{s},t/\eps,k,1,\kappa} \myb{b \myp{\tilde{s}}} }^2 }_{\lambda} \nn \\
		={}& \sum_{J \textit{ interlaces } \myp{n,n}} \sum_{\ell,\ell' \in G_n}\sum_{S\in \pi \myp{I_{4n+2}}} \mathcal{F}^{1}_{n} \myp{S,J,\ell,\ell',k,t/\eps-\tilde{s},\kappa},
		\label{Def:OperatorQErrors:Bis1}
	\end{align}
	where, denoting $ \mathcal{I}_n' =\mathcal{I}_n \cup \Set{\myp{n,1}} $,
	\begin{align}
		\mathcal{F}^1_n \myp{S,J,\ell,\ell',k,\tau,\kappa} 
		={}& (-\lambda^2)^{n} \epsilon \myp{S} \sum_{\sigma,\sigma' \in \Set{\pm 1}^{\mathcal{I}_n'}} \int_{ \myp{\Omega^{\ast}}^{\mathcal{I}_n'}} \ids \bar{k} \int_{ \myp{\Omega^{\ast}}^{\mathcal{I}_n'}} \ids \bar{k}' \Delta_{n,\ell} \myp{\bar{k},\bar{\sigma}} \nn \\
		&\times \Delta_{n,\ell'} \myp{\bar{k}',\bar{\sigma}'} \delta_L \myp{k_{n,1}-k} \delta_L \myp{k_{n,1}'+k} \1 \myp{\sigma_{n,1}=1} \nn \\
		&\times \1 \myp{\sigma_{n,1}'=-1} \prod_{A\in S} \myb[\bigg]{ \delta_L \myp[\bigg]{ \sum_{j\in A} K_j} C_{\abs{A}} \myp{K_A, o_A, \lambda, \Omega} } \nn \\
		&\times \prod_{j=1}^n \myb[\big]{\sigma_{j,\ell_j} \Psi_1 \myp{k_{j-1:\ell_{j}},\sigma_{j,\ell_j}} \sigma_{j,\ell_j'}' \Psi_1 \myp{-k'_{j-1:\ell'_{j}},-\sigma'_{j,\ell_j'}} } \nn \\
		&\times \int_{\myp{\R_+}^{I_{0,2n}}} \ids s \delta \myp[\bigg]{ \tau-\sum_{j=0}^{2n} s_j} \prod_{j=0}^{2n} e^{-i s_j \gamma \myp{j;J}}.
		\label{eq:DefF1}
	\end{align}
\end{proposition}
\begin{proof}
	The proof follows the same line of argumentation as \cref{Propo:QERR3}, the main difference being the extra time integration in \eqref{eq:DefE1}.
	For notational convenience, we rename the integration variable $r$ in \eqref{eq:DefE1} to $s_0$.
	Using \eqref{eq:b_prod_phase} (with $s_0$ replaced by $\tilde{s}$) to extract the phase factor from the product of $b$ operators, we obtain as before the oscillating factor
	\begin{align*}
		\MoveEqLeft[4] \prod_{j=1}^{n} e^{i\tilde{s} \Theta_{j-1;\ell_j} \myp{\bar{k},\bar{\sigma}} -i t_j \myp{\tilde{s}+s_0,s} \Theta_{j-1:\ell_{j}} \myp{\bar{k},\bar{\sigma}} }
		={} \prod_{m=0}^{n-1} \prod_{j=m+1}^n e^{-i s_m \Theta_{j-1:\ell_{j}} \myp{\bar{k},\bar{\sigma}} },
	\end{align*}
	which in turn leads to the total exponential
	\begin{equation*}
		\prod_{j=1}^n e^{-s_j \kappa_{n-j}} \prod_{m=0}^{n-1} e^{-i s_m \sum_{j=m+1}^n \Theta_{j-1:\ell_j} \myp{\bar{k}, \bar{\sigma}}}
		={} \prod_{m=0}^n e^{-i s_m \gamma_m^+}.
	\end{equation*}
	This way, we obtain an expression analogous to \eqref{eq:E3_expansion}, except the time integral is over $\myp{\R_+}^{I_{0,n}}$ and contains an extra factor $ e^{-i s_0 \gamma_0^+} $, and the delta function in the time integral has $t$ replaced by $t-\tilde{s}$.
	The rest of the proof is exactly the same as in \cref{Propo:QERR3}, and we arrive at \eqref{Def:OperatorQErrors:Bis1} and \eqref{eq:DefF1}.
\end{proof}
\begin{proposition}
	\label{Propo:QERR2}
	The components of $\bar{Q}^{err}_2$ can be expressed as
	\begin{align}
		\MoveEqLeft[4] \expec[\Big]{ \abs[\big]{ \mathcal{E}^2_n \myp{s,t/\eps,k,1,\kappa} \myb{b \myp{s}} }^2 }_{\lambda} \nn \\
		={}& \sum_{J \textit{ interlaces } \myp{n-1,n-1}} \sum_{\ell,\ell' \in G_n} \sum_{S\in \pi \myp{I_{4n+2}}} \mathcal{F}^{2}_{n} \myp{S,J,\ell,\ell',k,t/\eps-s,\kappa},
	\label{Def:OperatorQErrors:Bis2}
	\end{align}
	where $\mathcal{F}^{2}_{n} = 0$ if the cluster decomposition $S$ contains a pairing between two initial time vertices that are connected to the same of the two bottom interaction vertices (amputated vertices) in the corresponding diagram (see \cref{Sec:AmputatedDiagram} below).
	Otherwise, denoting $ \mathcal{I}_n' = \mathcal{I}_n \cup \Set{\myp{n,1}} $, we have
	\begin{align}
		\mathcal{F}^2_n \myp{S,J,\ell,\ell',k,\tau,\kappa}
		={}& (-\lambda^2)^{n} \epsilon \myp{S} \sum_{\sigma,\sigma'\in \Set{\pm 1}^{\mathcal{I}_n'}} \int_{\myp{\Omega^{\ast}}^{\mathcal{I}_n'}} \ids \bar{k} \int_{\myp{\Omega^{\ast}}^{\mathcal{I}_n'}} \ids \bar{k}' \Delta_{n,\ell} \myp{\bar{k},\bar{\sigma}} \nn \\
		&\times \Delta_{n,\ell'} \myp{\bar{k}',\bar{\sigma}'} \delta_L \myp{k_{n,1}-k} \delta_L \myp{k_{n,1}'+k} \1 \myp{\sigma_{n,1}=1} \nn \\
		&\times \1 \myp{\sigma_{n,1}'=-1} \prod_{A\in S} \myb[\bigg]{ \delta_L \myp[\bigg]{ \sum_{j\in A} K_j } C_{\abs{A}} \myp{K_A, o_A, \lambda, \Omega} } \nn \\
		&\times \myb[\big]{\sigma_{1,\ell_1} \Psi_0 \myp{k_{0:\ell_{1}},\sigma_{1,\ell_1}} \sigma_{1,\ell'_1}' \Psi_0 \myp{-k'_{0:\ell'_{1}},-\sigma'_{1,\ell'_1}} } \nn \\
		&\times \prod_{j=2}^n \myb[\big]{ \sigma_{j,\ell_j} \Psi_1 \myp{k_{j-1:\ell_{j}},\sigma_{j,\ell_j}} \sigma_{j,\ell'_j}' \Psi_1 \myp{-k'_{j-1:\ell'_{j}},-\sigma'_{j,\ell'_j}} } \nn \\
		&\times \int_{\myp{\R_+}^{I_{2,2n}}} \ids s \delta \myp[\bigg]{ \tau-\sum_{j=2}^{2n} s_j } \prod_{j=2}^{2n} e^{-i s_j \gamma \myp{j;J}}.
		\label{eq:DefF2}
	\end{align}
\end{proposition}
\begin{proof}
	Noting that \eqref{eq:b_prod_phase} also holds for a product containing the truncated factor in the definition \eqref{eq:DefE2} of $\mathcal{E}_n^2$, the proof follows exactly the lines of \cref{Propo:QERR3} up until the point of inserting the cumulant expansion \eqref{Lemma:CummulantEstimate:1} into the expectation
	\begin{align*}
		\MoveEqLeft[6] \expec[\Bigg]{ \prod_{j=1}^{\ell_1'-1} \hat{a} \myp{k_{0,j}',\sigma_{0,j}'} \myb[\Bigg]{ \prod_{j=\ell_1'}^{\ell_1'+2} \hat{a} \myp{k_{0,j}',\sigma_{0,j}'} }^T \prod_{j=\ell_1'+3}^{1+2n} \hat{a} \myp{k_{0,j}',\sigma_{0,j}'} \\  
		&\times \prod_{j=1}^{\ell_1-1} \hat{a} \myp{k_{0,j},\sigma_{0,j}} \myb[\Bigg]{ \prod_{j=\ell_1}^{\ell_1+2} \hat{a} \myp{k_{0,j},\sigma_{0,j}} }^T \prod_{j=\ell_1+3}^{1+2n} \hat{a} \myp{k_{0,j},\sigma_{0,j}} }_{\lambda}.
	\end{align*}
	For the sake of readability, we will do this step just in the case $n=1$, where only the truncated factors remain.
	In this simple case, where the vector of parities is forced to be $ \myp{\sigma',\sigma} = \myp{-1,-1,1,-1,1,1}$, we use the definitions \eqref{Def:Truncation} and \eqref{Def:Truncation2} of the truncated factors to expand the expectation into a sum of nine terms,
	\begin{align*}
		\MoveEqLeft[4] \expec[\Big]{ \myb[\big]{ \hat{a} \myp{k_1',-1} \hat{a} \myp{k_2',-1} \hat{a} \myp{k_3',1} }^T \myb[\big]{ \hat{a} \myp{k_1,-1} \hat{a} \myp{k_2,1} \hat{a} \myp{k_3,1} }^T }_{\lambda} \\
		={}& \expec[\big]{ \hat{a} \myp{k_1',-1} \hat{a} \myp{k_2',-1} \hat{a} \myp{k_3',1} \hat{a} \myp{k_1,-1} \hat{a} \myp{k_2,1} \hat{a} \myp{k_3,1} }_{\lambda} \\
		&\quad- \expec[\big]{ \hat{a} \myp{k_1',-1} \hat{a} \myp{k_1,-1} \hat{a} \myp{k_2,1} \hat{a} \myp{k_3,1} }_{\lambda} \expec[\big]{ \hat{a} \myp{k_2',-1} \hat{a} \myp{k_3',1} }_{\lambda} \\
		&\quad+ \expec[\big]{ \hat{a} \myp{k_2',-1} \hat{a} \myp{k_1,-1} \hat{a} \myp{k_2,1} \hat{a} \myp{k_3,1} }_{\lambda} \expec[\big]{ \hat{a} \myp{k_1',-1} \hat{a} \myp{k_3',1} }_{\lambda} \\
		&- \expec[\big]{ \hat{a} \myp{k_1',-1} \hat{a} \myp{k_2',-1} \hat{a} \myp{k_3',1} \hat{a} \myp{k_3,1} }_{\lambda} \expec[\big]{ \hat{a} \myp{k_1,-1} \hat{a} \myp{k_2,1} }_{\lambda} \\
		&\quad+ \expec[\big]{ \hat{a} \myp{k_1',-1} \hat{a} \myp{k_3,1} }_{\lambda} \expec[\big]{ \hat{a} \myp{k_2',-1} \hat{a} \myp{k_3',1} }_{\lambda} \expec[\big]{ \hat{a} \myp{k_1,-1} \hat{a} \myp{k_2,1} }_{\lambda} \\
		&\quad- \expec[\big]{ \hat{a} \myp{k_2',-1} \hat{a} \myp{k_3,1} }_{\lambda} \expec[\big]{ \hat{a} \myp{k_1',-1} \hat{a} \myp{k_3',1} }_{\lambda} \expec[\big]{ \hat{a} \myp{k_1,-1} \hat{a} \myp{k_2,1} }_{\lambda} \\
		&+\expec[\big]{ \hat{a} \myp{k_1',-1} \hat{a} \myp{k_2',-1} \hat{a} \myp{k_3',1} \hat{a} \myp{k_2,1} }_{\lambda} \expec[\big]{ \hat{a} \myp{k_1,-1} \hat{a} \myp{k_3,1} }_{\lambda} \\
		&\quad -\expec[\big]{ \hat{a} \myp{k_1',-1} \hat{a} \myp{k_2,1} }_{\lambda} \expec[\big]{ \hat{a} \myp{k_2',-1} \hat{a} \myp{k_3',1} }_{\lambda} \expec[\big]{ \hat{a} \myp{k_1,-1} \hat{a} \myp{k_3,1} }_{\lambda} \\
		&\quad +\expec[\big]{ \hat{a} \myp{k_2',-1} \hat{a} \myp{k_2,1} }_{\lambda} \expec[\big]{ \hat{a} \myp{k_1',-1} \hat{a} \myp{k_3',1} }_{\lambda} \expec[\big]{ \hat{a} \myp{k_1,-1} \hat{a} \myp{k_3,1} }_{\lambda}.
	\end{align*}
	Now, inserting the cumulant expansion \eqref{Lemma:CummulantEstimate:1} into each of these, it follows that any given contribution from the first term above is cancelled exactly when the corresponding partition $S$ contains one of the pairs from the truncated products \eqref{Def:Truncation}, \eqref{Def:Truncation2}, that is,
	\begin{align*}
		\MoveEqLeft[8] \expec[\Big]{ \myb[\big]{ \hat{a} \myp{k_1',-1} \hat{a} \myp{k_2',-1} \hat{a} \myp{k_3',1} }^T \myb[\big]{ \hat{a} \myp{k_1,-1} \hat{a} \myp{k_2,1} \hat{a} \myp{k_3,1} }^T }_{\lambda} \\
		={}& \sum_{S \in \tilde{\pi} \myp{I_6}} \epsilon(S) \prod_{A \in S} \myb[\bigg]{ \delta_L \myp[\bigg]{ \sum_{j\in A} K_j } C_{\abs{A}} \myp{K_A, o_A,\lambda,\Omega} },
	\end{align*}
	where $\tilde{\pi} \myp{I_6}$ is the set of all ordered partitions of $\Set{k_1',k_2',k_3',k_1,k_2,k_3}$ \emph{not} containing any of the pairings $\Set{k_1',k_3'}$, $\Set{k_2',k_3'}$, $\Set{k_1,k_2}$, $\Set{k_1,k_3}$.
	Note also that that any partition containing one of the pairs $\Set{k_1',k_2'}$, $\Set{k_2,k_3}$ will not contribute in the expression above, because the corresponding $C_2$-factors vanish by \eqref{eq:pair_correl_trunc}.
	Thus, for the simple case $n=1$, the amplitude $\mathcal{F}_1^{2}$ vanishes for any cluster decomposition $S$ containing a pairing between the momenta $\Set{k_1',k_2',k_3'}$, or between the momenta $\Set{k_1,k_2,k_3}$, which are exactly the sets of momenta connected to the bottom two vertices in the corresponding diagram.
	The same conclusion holds generally for any $n \geq 1$, that is, $\mathcal{F}_n^2$ vanishes if $S$ contains a pairing in $\Set{k_{0,\ell_1'}',k_{0,\ell_1'+1}',k_{0,\ell_1'+2}'}$, or in $\Set{k_{0,\ell_1},k_{0,\ell_1+1},k_{0,\ell_1+2}}$.
	The expression \eqref{eq:DefF2} for other cluster decompositions is then obtained by following the same steps as in the proof of \cref{Propo:QERR3}.
\end{proof}

\subsection{Diagrams for main and amputated terms}
\label{Sec:Diagrams}
This subsection is devoted to the diagrams of the two key quantities that we need to estimate: the main and the error terms.
\subsubsection{Main diagrams}
\label{Sec:MainDiagram}
The key difference between \eqref{eq:DefFmain} and \eqref{eq:DefE0} is the appearance of the delta function $\delta_L \myp[\big]{\sum_{j\in A} k_{0,j}}$ acting on the set $A$ and the two new parities $\1 \myp{\sigma_{n,1}=1}$, $\1 \myp{\sigma_{0,0}=-1}$.
The construction of the diagram for \eqref{eq:DefFmain} can be done by modifying the graphs introduced in Subsection \ref{Duhamel} as follows (see Figure \ref{Fig2}):
\begin{figure}
	\centering
	\includegraphics{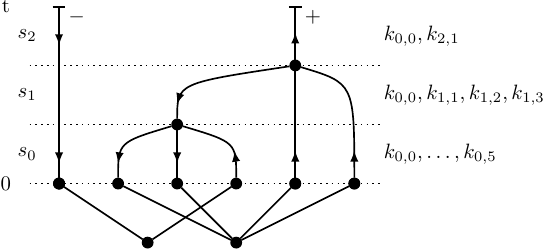}
	\caption{An example of a Feynman diagram with clusters, also featuring the additional initial vertex with momentum $k_{0,0}$.
		Two new cluster vertices are added at the bottom of the diagram, to connect $k_{0,0},k_{0,3}$ and $k_{0,1},k_{0,2},k_{0,4},k_{0,5}$.}
\label{Fig2}
\end{figure}
\begin{enumerate}[(i)]
	\item For each $A$, we add a \emph{cluster vertex} at the bottom of the graph and connect this new vertex to all of the vertices $ k_{0,j}$, $j\in A$, appearing in the delta function $\delta_L \myp[\big]{\sum_{j\in A} k_{0,j} }$.  
	
	\item To represent the parity  $\1 \myp{\sigma_{n,1}=1}$, we just assign the value $1$ to $\sigma_{n,1}$, which corresponds to the parity of the top of the graph.
	
	\item To represent the parity $\1 \myp{\sigma_{0,0}=-1}$, we will also add a vertex $\myp{0,0}$ and a segment $k_{0,0}$ with $\sigma_{0,0}=-1$.
	This segment goes from the vertex $\myp{0,0}$ to the top line.
	
	\item Thanks to the parities $\sigma_{n,1}=1$ and  $\sigma_{0,0}=-1$, the remaining parities can be computed from top to bottom using the same rule discussed in \cref{Duhamel}.
	At each iteration, the parity remains unchanged in the middle line and we set $-1$ for the segment on the left and $1$ for the segment on the right.
	The cluster vertices do not affect the parities.
	
	\item In this diagram, when all cluster vertices and their edges are removed, the diagram splits into two components.
	The left one contains only one edge $k_{0,0}$ and is called the \emph{minus tree}.
	The right one is called the \emph{plus tree}.
	
	\item The time-dependent factor in \eqref{eq:DefFmain} can also be expressed under the form
	\begin{equation}
	\label{eq:gamma_j}
		\prod_{j=0}^n e^{-i s_j \gamma_j},
		\mbox{ where } 
		\gamma_j = \sum_{l=j+1}^n \Theta_{l-1:\ell_l} \myp{\bar{k},\bar{\sigma}} - i \kappa_{n-j}.
	\end{equation}
	Since $-\omega^{\lambda} \myp{k_{n,1}} = \sigma_{0,0} \omega^{\lambda} \myp{k_{0,0}}$, we have that
	\begin{equation*}
		\mathrm{Re} \gamma_j = \sum_{l=1}^{2(n-j)+1} \sigma_{j,l} \omega^{\lambda} \myp{k_{j,l}} - \omega^{\lambda} \myp{k_{n,1}}.
	\end{equation*}
	
	\item Each fusion vertex carries a factor $-i\lambda \Psi_1$ and each cluster vertex carries a factor $C_{\abs{A}}$. 
\end{enumerate}

\subsubsection{Diagrams for error terms}
\label{Sec:AmputatedDiagram}
In this subsection, we will construct Feynman diagrams for the terms $\mathcal{F}^{1}_{n},\mathcal{F}^{2}_{n},\mathcal{F}^{3}_{n}$.
The diagram can be constructed in a similar manner with that of \eqref{eq:DefFmain} (see Figure \ref{Fig3}).
\begin{figure}
\centering
	\includegraphics{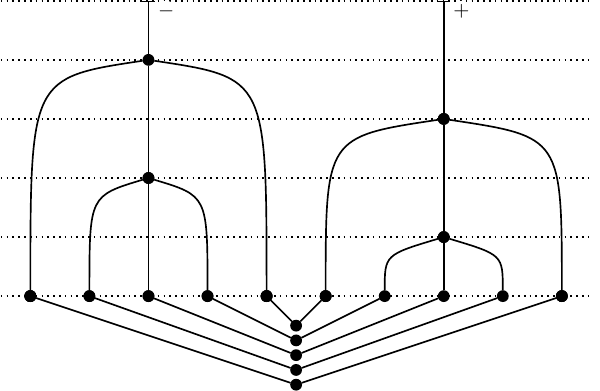}
	\caption{In this diagram, the plus and the minus trees are denoted by the plus $+$ and the minus $-$ signs. Each vertex at the bottom of the plus tree is paired to a vertex at the bottom of the minus tree.}
\label{Fig3}
\end{figure}
\begin{enumerate}[(i)]
	\item The appearance of $k,k'$ and $\sigma$, $\sigma'$ indicates that we need two graphs.
	Moreover, the quantity $\1 \myp{\sigma_{n,1}=1} \1 \myp{\sigma_{n,1}'=-1}$ shows that the two tops of the two  graphs have the parities $-1$ and $1$.
	We set the left one to be $-1$.
	
	\item Thanks to (i), when all cluster vertices and their edges are removed, the diagram also splits into two components.
	The left one is called the \emph{minus tree} and the right one is called the \emph{plus tree}.
	
	\item The delta function $\delta \myp[\big]{ t-\sum_{j=2}^{2n}s_j }$ in the expressions for $\mathcal{F}_n^{2}$ and $\mathcal{F}_n^{3}$ indicates that the first two time slices have zero length.
	We say that these time slices are \emph{amputated}, and we also call the corresponding diagrams \emph{amputated diagrams}.
	This can be denoted on the diagram for instance by shading the first two time slices.
	
	\item To make the identification of the amputated terms with their corresponding Feynman diagrams unique, we can without loss of generality assume that the first fusion from the bottom always takes place in the minus tree.
	
	\item For each $A$, we add a  cluster vertex at the bottom of the graph and connect the new vertex to all of the vertices $ k_{0,j}$, and $ k_{0,j}'$, $j\in A$.
	The plus and minus trees are indeed connected only by the cluster vertices.
	
	\item For $\mathcal{F}_n^1$ and $\mathcal{F}_n^3$, each interaction vertex in the plus tree carries a factor $i\lambda \sigma_{j,\ell_j} \Psi_1 \myp{k_{j-1:\ell_{j-1}},\sigma_{j,\ell_j}}$, and each interaction vertex in the minus tree carries a factor $-i\lambda \sigma_{j,\ell'_j}' \Psi_1 \myp{-k'_{j-1:\ell'_{j-1}},\sigma'_{j,\ell'_j}}$.
	The same is true for $\mathcal{F}_n^2$, except $\Psi_1$ is replaced by $\Psi_0$ at the amputated vertices at the bottom of the diagram.
	Each cluster vertex carries a factor $C_{\abs{A}}$ with a corresponding $\delta_L$-function. 
\end{enumerate}

\subsection{Signs of cluster decompositions, and leading motives}
\subsubsection{Properties of the sign of a cluster decomposition $S$}
The presence of the sign $\epsilon\myp{S}$ in the definition \eqref{Def:Truncation:2} of the truncated moments is a distinctly \emph{fermionic} feature, and it can be seen as an expression of the Pauli exclusion principle.
By commutativity, it does not appear in the classical setting \cite{LukkarinenSpohn:WNS:2011}, or in the case of bosons.
Ultimately, $\epsilon \myp{S}$ will be responsible for some of the signs appearing in \eqref{eq:nu-infty} in the main theorem.

Here we investigate more closely the sign $\epsilon \myp{S}$ for a given cluster decomposition $S$ of the set $I_n$, in particular the behaviour of the sign when attaching two diagrams together, as explained below.
First, we have:
\begin{lemma}
\label{lem:cluster_sign_order}
	Let $S = \Set{A_1, \dotsc, A_k}$ be a clustering of $I_n$, and assume that $\abs{A_j}$ is even for all $j$.
	Then the sign $\epsilon \myp{S}$ is independent of the ordering of the elements $A_j$ of $S$.
\end{lemma}
\begin{proof}
	Denoting the cluster elements by $A_j = \Set{a_{j,1}, \dotsc, a_{j,n_j}}$ with $n = \sum_j n_j$, the permutation $\tau_S$ corresponding to $S$ is by definition
	\begin{equation}
	\label{eq:S_permutation}
		\tau_S =	\begin{pmatrix}
						1		&	2		&	\cdots	&	n_1			&	n_1+1	&	\cdots	&	n-1			&	n \\
						a_{1,1}	&	a_{1,2}	&	\cdots	&	a_{1,n_1}	&	a_{2,1}	&	\cdots	&	a_{k,n_k-1}	&	a_{k,n_k}
					\end{pmatrix},
	\end{equation}
	and so $\epsilon \myp{S} \coloneq \sgn \myp{\tau_S} = \myp{-1}^{N\myp{\tau_S}}$, where $N\myp{\tau_S}$ is the number of inversions of $\tau_S$, that is, the number of pairs $i<j$ with $\tau_S \myp{i} > \tau_S \myp{j}$.
	Note that by construction, there is no contribution to the number of inversions coming from pairs within any given cluster $A_j$, since we assume $a_{j,1} < a_{j,2} < \cdots < a_{j,n_j}$, so all inversions of $\tau_S$ must come from pairs between the clusters.
	
	Now, fix $j \in \Set{1,\dotsc,n-1}$ and suppose that we switch the order of $A_j$ and $A_{j+1}$ in $S$.
	It is clear that this action does not change the number of inversions between $A_j$ and any of the other clusters, apart from $A_{j+1}$.
	It follows that we can assume that $S$ has only two clusters $S= \Set{A_1,A_2}$.
	Now, an inversion $i<j$ of $\tau_S$ must satisfy $i \in \Set{1,\dotsc,n_1}$, $j \in \Set{n_1+1,\dotsc,n_1+n_2}$, and $\tau_S \myp{i} = a_{1,i} > a_{2,j-n_1} = \tau_S \myp{j}$.
	After switching the order of $A_1$ and $A_2$, the permutation $\tau_{\widetilde{S}}$ of $\widetilde{S} = \Set{A_2,A_1}$ attains these values at $n_2 + i$ and $j-n_1$, respectively.
	That is, the corresponding pair in $\widetilde{S}$ will be $j-n_1 < n_2+i$ with $\tau_{\widetilde{S}} \myp{j-n_1} = a_{2,j-n_1} < a_{1,i} = \tau_{\widetilde{S}} \myp{n_2+i}$.
	It is now clear that and inversion of $S$ corresponds to a non-inversion of $\widetilde{S}$ (and vice versa), so
	\begin{equation*}
		N \myp{\tau_S} = n_1 + n_2 - N \myp{\tau_{\widetilde{S}}},
	\end{equation*}
	from which it follows that $\epsilon \myp{S} = \epsilon \myp{\widetilde{S}}$, since $n_1$, $n_2$ are assumed to be even.
\end{proof}
\begin{remark}[Reading the sign $\epsilon \myp{S}$ from a diagram]
\label{rem:cluster_sign_graph}
	Consider any clustering $S = \Set{A_1, \dotsc, A_k}$ of $I_n$.
	We make the following observations that allows one to graphically read the sign $\epsilon \myp{S}$ from the drawing of a diagram.
	\begin{enumerate}[(a)]
		\item Assume the sets $A_j = \Set{a_{j,1}, \dotsc, a_{j,n_j}}$ are ordered such that $a_{i,1} < a_{j,1}$ whenever $i<j$.
		Assume further that on the diagram, all cluster vertices are placed on the same horizontal line below the initial time vertices in such a way that the vertex of $A_i$ is to the left of $A_j$ whenever $i<j$.
		Then it is apparent that if a cluster edge of $A_i$ intersects a cluster edge of $A_j$, it exactly corresponds to an inversion of the permutation of $S$ (just keep the representation \eqref{eq:S_permutation} in mind).
		Consequently, the number of inversions of the permutation determined by $S$ is exactly equal to the number of intersections $g$ between the cluster edges of the graph, see e.g. \cref{fig:C1}.
		In particular, $\epsilon \myp{S} = \myp{-1}^g$.
		
		\item If $\abs{A_j}$ is even for all $j$, then the placement of the cluster vertices is irrelevant when determining $\epsilon \myp{S}$ (as long as the cluster vertices are placed below the initial time vertices).
		The sign $\epsilon \myp{S}$ can still be calculated as $\epsilon \myp{S} = \myp{-1}^{g}$.
		
		To see this, take two clusters $A_1$ and $A_2$ and consider what happens if we move one of the cluster vertices.
		Since $\abs{A_1}$ and $\abs{A_2}$ are both even, the number of intersections between their cluster edges can only change by an even number, as long as the cluster vertex remains below the horizontal $t = 0$ line, see e.g. \cref{fig:C1} for an illustration.
		Hence the parity of the total number of intersections remains the same.
	\end{enumerate}
	\begin{figure}[h]
		\centering
		\includegraphics{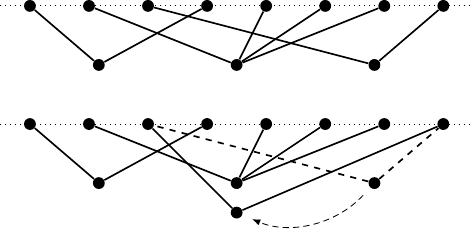}
		\caption{Two illustrations of the cluster decomposition $S = \Set{ \Set{1,4}, \Set{2,5,6,7}, \Set{3,8}}$.
			At the top, the clusters are drawn according to the instructions of \cref{rem:cluster_sign_graph} (a), so the sign of the corresponding permutation can be read as $\epsilon \myp{S} = \myp{-1}^5 = -1$.
			At the bottom, the rightmost cluster vertex have been moved, changing the number of intersections by an even number, thus leaving the parity invariant.}
		\label{fig:C1}
	\end{figure}
\end{remark}

Under certain conditions, it is possible to attach one Feynman diagram to the bottom of another one, and later it will be useful to know how the signs of the corresponding clusters are related.
Take two Feynman diagrams $\mathscr{G}_1,\mathscr{G}_2$ with clusters $S_1,S_2$.
Consider a cluster $A \in S_1$, and assume that $\mathscr{G}_2$ has a number of (plus or minus) trees that is equal to the number of elements in $A$.
(So far, we have only considered Feynman diagrams consisting of either one or two trees, but nothing prevents us from generalising to diagrams with an arbitrary number of trees).
If, furthermore, the parities of the cluster edges of $A$ coincide with the parities of the edges at the top of $\mathscr{G}_2$, then $\mathscr{G}_2$ can be attached to the bottom of $\mathscr{G}_1$ at the cluster $A$.
The resulting graph $\mathscr{G}$ is also a Feynman diagram with interaction vertices equal to the union of the interaction vertices of $\mathscr{G}_1$ and $\mathscr{G}_2$, and with clusters $S = \myp{S_1 \cup S_2} \setminus \Set{A}$, up to re-ordering.
At the bottom of $\mathscr{G}$, each initial time vertex of $\mathscr{G}_1$ connected to $A$ is replaced by the set of initial time vertices of $\mathscr{G}_2$ at the bottom of the corresponding tree of $\mathscr{G}_2$.

\begin{lemma}[Multiplicativity of the sign $\epsilon \myp{S}$]
	Let $\mathscr{G}_1,\mathscr{G}_2$ be two Feynman diagrams with corresponding cluster decompositions $S_1,S_2$.
	Assume that $S_1,S_2$ contain only clusters with an even number of elements, and that $S_1$ contains a cluster where $\mathscr{G}_2$ can be attached, forming a composite diagram $\mathscr{G}$.
	Then the sign of the cluster decomposition $S$ of $\mathscr{G} $satisfies
	\begin{equation*}
		\epsilon \myp{S} = \epsilon \myp{S_1} \epsilon \myp{S_2}.
	\end{equation*}
\end{lemma}
\begin{proof}[Sketch of proof]
	For simplicity, we consider only the simple case where $\mathscr{G}_2$ consists of two trees, and where $\mathscr{G}_2$ attaches to $\mathscr{G}_1$ at a pairing.
	Note that since $\mathscr{G}_2$ is a Feynman diagram, there is an odd number of initial vertices at the bottom of each of its trees.
	Furthermore, since each cluster of $S_2$ has an even number of elements, the number of cluster edges from the plus tree connecting to the minus tree must be odd.
	
	In the diagram $\mathscr{G}$, we need to determine the parity of the total number of intersections between cluster edges.
	Note that for any given cluster edge coming from $S_2$, the number of intersections between this edge and all other cluster edges originating from $S_2$ is the same as in the original graph $\mathscr{G}_2$.
	The same is true when considering intersections between cluster edges originating from $S_1 \setminus \Set{A}$, where $A$ is the cluster to which $\mathscr{G}_2$ is attached.
	Hence, we only need to argue that the number of intersections between $S_2$ and $S_1 \setminus \Set{A}$ has the same parity as the number of intersections between $A$ and $S_1 \setminus \Set{A}$ in the original graph $\mathscr{G}_1$.
	
	Consider thus a cluster edge $e$ from a cluster in $S_1 \setminus \Set{A}$ that intersects with a cluster edge $e'$ from $A$ in the original graph $\mathscr{G}_1$.
	The initial vertex $v'$ connected to $e'$ in $\mathscr{G}_1$ gets replaced by an odd number of initial vertices after attaching $\mathscr{G}_2$, corresponding to the initial vertices from one of the trees of $\mathscr{G}_2$.
	Furthermore, an odd number of cluster edges connect these new vertices in $\mathscr{G}$ to the other tree coming from $\mathscr{G}_2$.
	The remaining (even) number of new initial time vertices in the tree that replaces $v'$ will not contribute to the number of intersections with $e$, simply by drawing them so that they do not intersect with the cluster edges from $S_1$.
	This means that the intersection between $e$ and $e'$ gets replaced by an odd number of intersections between $e$ the cluster edges connecting the new vertices replacing $v'$ to the other tree in $\mathscr{G}_2$.
	In other words, the number of intersections $N\myp{S}$ between the cluster edges in $\mathscr{G}$ must satisfy
	\begin{equation*}
		N \myp{S} = N \myp{S_1} + N \myp{S_2} + 2m,
	\end{equation*}
	for some $m \in \N_0$, and the assertion of the lemma follows.
\end{proof}

\subsubsection{Leading motives and main terms}
\label{sec:leadingmotives}
It turns out (and will be proved later) that with suitable choices of $N$ and $\kappa$, the only terms giving a non-zero contribution in the weakly interacting limit come from the main term \eqref{Proposition:Ampl:1}, and only a certain type of diagrams will contribute to the limit.
These leading graphs are precisely the graphs where the clusters $S$ consist only of pairings, and which are formed by iteration of \emph{leading motives}, also called \emph{immediate recollisions}, which can be attached to pairings at the bottom of any existing graph.
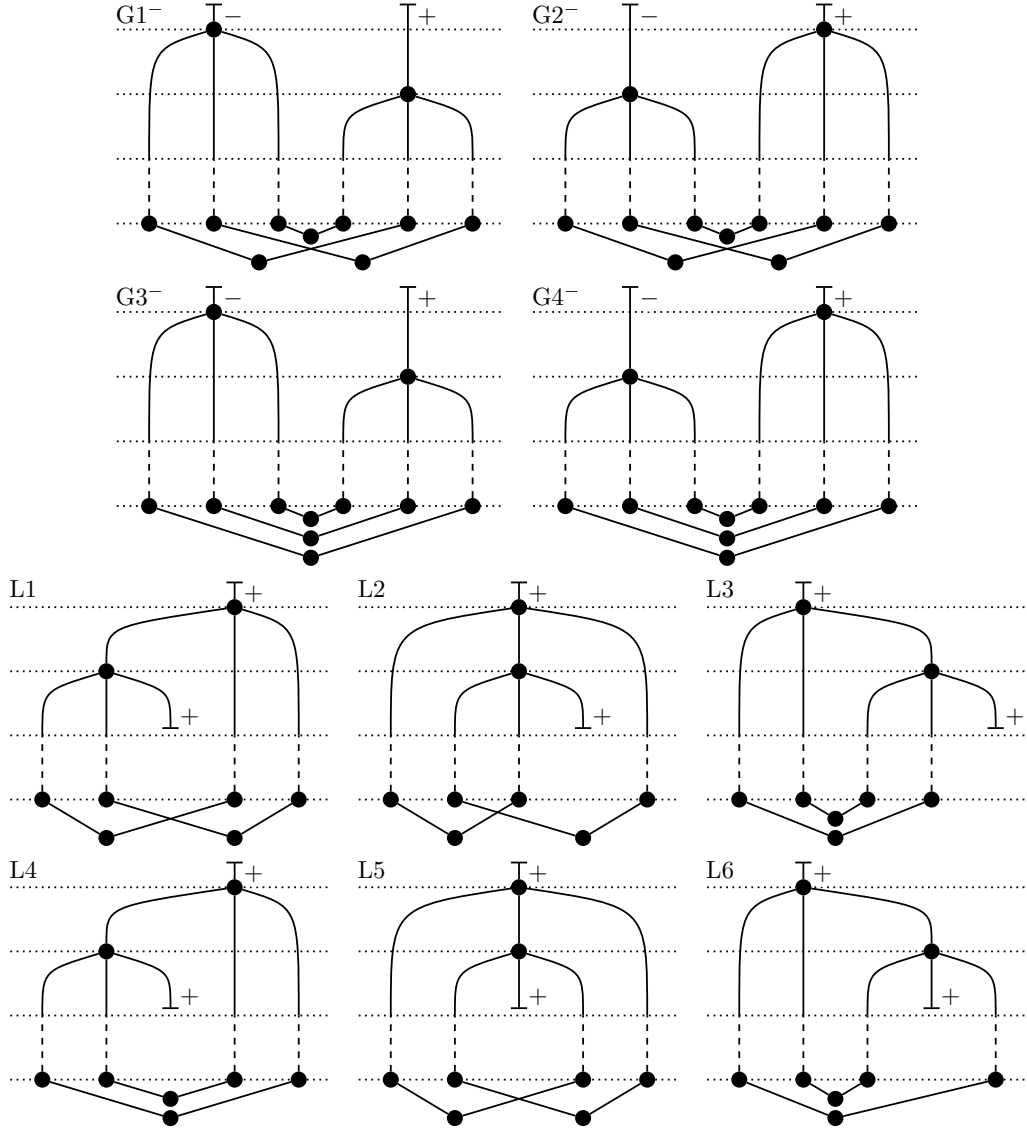
\begin{figure}
	\centering
	\resizebox{0.384\textwidth}{!}{
		\begin{tikzpicture}[every node/.style={fill=black,circle,inner sep=0pt,minimum size=0.25cm},every path/.style={outer sep=0pt,line width=0.8pt}]
			\foreach \x in {0,...,3} \draw[dotted]	(-0.5,\x) -- (5.5,\x); 
			
			\foreach \name in {1,...,6} \node (0\name) at (\name-1,0) {};
			\foreach \name in {01,02,03,04,05,06} \path[dashed] (\name) edge (\name |- ,1); 
			
			\node (2) at (4,2) {};
			\node (3) at (1,3) {};
			\draw	(02 |- ,1)	edge (3)
			(05 |- ,1)	edge (2)
			(01 |- ,1) .. controls ($(3)+(-1,-0.3)$) .. (3) 
			(03 |- ,1) .. controls ($(3)+(1,-0.3)$) .. (3) 
			(04 |- ,1) .. controls ($(2)+(-1,-0.3)$) .. (2) 
			(06 |- ,1) .. controls ($(2)+(1,-0.3)$) .. (2); 

			\draw[-|]	(2) to node [at end,below,xshift=0.3cm,fill=white] {$+$} ++(0,1.4);
			\draw[-|]	(3) to node [at end,below,xshift=0.3cm,fill=white] {$-$} ++(0,0.4);
			
			\node (c1) at (2.5,-0.2) {};
			\node (c2) at (1.7,-0.6) {};
			\node (c3) at (3.3,-0.6) {};
			\path	(c1)	edge (03)
			edge (04)
			(c2)	edge (01)
			edge (05)
			(c3)	edge (02)
			edge (06);
			
			\node[rectangle,fill=white,anchor=north west] at (current bounding box.north west) {G1$^-$}; 
		\end{tikzpicture}
	}
	\resizebox{0.384\textwidth}{!}{
		\begin{tikzpicture}[every node/.style={fill=black,circle,inner sep=0pt,minimum size=0.25cm},every path/.style={outer sep=0pt,line width=0.8pt}]
			\foreach \x in {0,...,3} \draw[dotted]	(-0.5,\x) -- (5.5,\x); 
			
			\foreach \name in {1,...,6} \node (0\name) at (\name-1,0) {};
			\foreach \name in {01,02,03,04,05,06} \path[dashed] (\name) edge (\name |- ,1); 
			
			\node (2) at (1,2) {};
			\node (3) at (4,3) {};
			\draw	(02 |- ,1)	edge (2)
			(05 |- ,1)	edge (3)
			(01 |- ,1) .. controls ($(2)+(-1,-0.3)$) .. (2) 
			(03 |- ,1) .. controls ($(2)+(1,-0.3)$) .. (2) 
			(04 |- ,1) .. controls ($(3)+(-1,-0.3)$) .. (3) 
			(06 |- ,1) .. controls ($(3)+(1,-0.3)$) .. (3); 
			
			\draw[-|]	(2) to node [at end,below,xshift=0.3cm,fill=white] {$-$} ++(0,1.4);
			\draw[-|]	(3) to node [at end,below,xshift=0.3cm,fill=white] {$+$} ++(0,0.4);
			
			\node (c1) at (2.5,-0.2) {};
			\node (c2) at (1.7,-0.6) {};
			\node (c3) at (3.3,-0.6) {};
			\path	(c1)	edge (03)
			edge (04)
			(c2)	edge (01)
			edge (05)
			(c3)	edge (02)
			edge (06);
			
			\node[rectangle,fill=white,anchor=north west] at (current bounding box.north west) {G2$^-$}; 
		\end{tikzpicture}
	}
	\\[\medskipamount]
	\resizebox{0.384\textwidth}{!}{
		\begin{tikzpicture}[every node/.style={fill=black,circle,inner sep=0pt,minimum size=0.25cm},every path/.style={outer sep=0pt,line width=0.8pt}]
			\foreach \x in {0,...,3} \draw[dotted]	(-0.5,\x) -- (5.5,\x); 
			
			\foreach \name in {1,...,6} \node (0\name) at (\name-1,0) {};
			\foreach \name in {01,02,03,04,05,06} \path[dashed] (\name) edge (\name |- ,1); 
			
			\node (2) at (4,2) {};
			\node (3) at (1,3) {};
			\draw	(02 |- ,1)	edge (3)
			(05 |- ,1)	edge (2)
			(01 |- ,1) .. controls ($(3)+(-1,-0.3)$) .. (3) 
			(03 |- ,1) .. controls ($(3)+(1,-0.3)$) .. (3) 
			(04 |- ,1) .. controls ($(2)+(-1,-0.3)$) .. (2) 
			(06 |- ,1) .. controls ($(2)+(1,-0.3)$) .. (2); 
			
			\draw[-|]	(2) to node [at end,below,xshift=0.3cm,fill=white] {$+$} ++(0,1.4);
			\draw[-|]	(3) to node [at end,below,xshift=0.3cm,fill=white] {$-$} ++(0,0.4);
			
			\node (c1) at (2.5,-0.2) {};
			\node (c2) at (2.5,-0.5) {};
			\node (c3) at (2.5,-0.8) {};
			\path	(c1)	edge (03)
			edge (04)
			(c2)	edge (02)
			edge (05)
			(c3)	edge (01)
			edge (06);
			
			\node[rectangle,fill=white,anchor=north west] at (current bounding box.north west) {G3$^-$}; 
		\end{tikzpicture}
	}
	\resizebox{0.384\textwidth}{!}{
		\begin{tikzpicture}[every node/.style={fill=black,circle,inner sep=0pt,minimum size=0.25cm},every path/.style={outer sep=0pt,line width=0.8pt}]
			\foreach \x in {0,...,3} \draw[dotted]	(-0.5,\x) -- (5.5,\x); 
			
			\foreach \name in {1,...,6} \node (0\name) at (\name-1,0) {};
			\foreach \name in {01,02,03,04,05,06} \path[dashed] (\name) edge +(0,1); 
			
			\node (2) at (1,2) {};
			\node (3) at (4,3) {};
			\draw	(02 |- ,1)	edge (2)
			(05 |- ,1)	edge (3)
			(01 |- ,1) .. controls ($(2)+(-1,-0.3)$) .. (2) 
			(03 |- ,1) .. controls ($(2)+(1,-0.3)$) .. (2) 
			(04 |- ,1) .. controls ($(3)+(-1,-0.3)$) .. (3) 
			(06 |- ,1) .. controls ($(3)+(1,-0.3)$) .. (3); 
			
			\draw[-|]	(2) to node [at end,below,xshift=0.3cm,fill=white] {$-$} ++(0,1.4);
			\draw[-|]	(3) to node [at end,below,xshift=0.3cm,fill=white] {$+$} ++(0,0.4);
			
			\node (c1) at (2.5,-0.2) {};
			\node (c2) at (2.5,-0.5) {};
			\node (c3) at (2.5,-0.8) {};
			\path	(c1)	edge (03)
			edge (04)
			(c2)	edge (02)
			edge (05)
			(c3)	edge (01)
			edge (06);
			
			\node[rectangle,fill=white,anchor=north west] at (current bounding box.north west) {G4$^-$}; 
		\end{tikzpicture}
	}
	\\[\medskipamount]
	\resizebox{0.32\textwidth}{!}{
		\begin{tikzpicture}[every node/.style={fill=black,circle,inner sep=0pt,minimum size=0.25cm},every path/.style={outer sep=0pt,line width=0.8pt}]
			\foreach \x in {0,...,3} \draw[dotted]	(-0.5,\x) -- (4.5,\x); 
			
			\node (01) at (0,0) {};
			\node (02) at (1,0) {};
			\node (03) at (3,0) {};
			\node (04) at (4,0) {};
			\foreach \name in {01,02,03,04} \path[dashed] (\name) edge (\name |- ,1); 
			
			\node (2) at (1,2) {};
			\node (3) at (3,3) {};
			\draw	(02 |- ,1)	edge (2)
			(03 |- ,1)	edge (3)
			(01 |- ,1) .. controls ($(2)+(-1,-0.3)$) .. (2) 
			(04 |- ,1) .. controls ($(3)+(1,-0.3)$) .. (3) 
			(2) .. controls ($(2)+(0,0.7)$) .. (3); 
			
			\draw[-|]	(2) .. controls ($(2)+(1,-0.3)$) .. ($(2)+(1,-0.9)$) node [above,xshift=0.3cm,fill=white] {$+$};
			\draw[-|]	(3) to node [at end,below,xshift=0.3cm,fill=white] {$+$} ++(0,0.4);
			
			\node (c1) at (1,-0.6) {};
			\node (c2) at (3,-0.6) {};
			\path	(c1)	edge (01)
			edge (03)
			(c2)	edge (02)
			edge (04);
			
			\node[rectangle,fill=white,anchor=north west] at (current bounding box.north west) {L1}; 
		\end{tikzpicture}
	}
	\resizebox{0.32\textwidth}{!}{
		\begin{tikzpicture}[every node/.style={fill=black,circle,inner sep=0pt,minimum size=0.25cm},every path/.style={outer sep=0pt,line width=0.8pt}]
			\foreach \x in {0,...,3} \draw[dotted]	(-0.5,\x) -- (4.5,\x); 
			
			\node (01) at (0,0) {};
			\node (02) at (1,0) {};
			\node (03) at (2,0) {};
			\node (04) at (4,0) {};
			\foreach \name in {01,02,03,04} \path[dashed] (\name) edge (\name |- ,1); 
			
			\node (2) at (2,2) {};
			\node (3) at (2,3) {};
			\draw	(2) edge (3)
			(03 |- ,1)	edge (2)
			(01 |- ,1) .. controls ($(3)+(-2,-0.3)$) .. (3) 
			(02 |- ,1) .. controls ($(2)+(-1,-0.3)$) .. (2) 
			(04 |- ,1) .. controls ($(3)+(2,-0.3)$) .. (3); 

			\draw[-|]	(2) .. controls ($(2)+(1,-0.3)$) .. ($(2)+(1,-0.9)$) node [above,xshift=0.3cm,fill=white] {$+$};
			\draw[-|]	(3) to node [at end,below,xshift=0.3cm,fill=white] {$+$} ++(0,0.4);
			
			\node (c1) at (1,-0.6) {};
			\node (c2) at (3,-0.6) {};
			\path	(c1)	edge (01)
			edge (03)
			(c2)	edge (02)
			edge (04);
			
			\node[rectangle,fill=white,anchor=north west] at (current bounding box.north west) {L2}; 
		\end{tikzpicture}
	}
	\resizebox{0.32\textwidth}{!}{
		\begin{tikzpicture}[every node/.style={fill=black,circle,inner sep=0pt,minimum size=0.25cm},every path/.style={outer sep=0pt,line width=0.8pt}]
			\foreach \x in {0,...,3} \draw[dotted]	(-0.5,\x) -- (4.5,\x); 
			
			\node (01) at (0,0) {};
			\node (02) at (1,0) {};
			\node (03) at (2,0) {};
			\node (04) at (3,0) {};
			\foreach \name in {01,02,03,04} \path[dashed] (\name) edge (\name |- ,1); 
			
			\node (2) at (3,2) {};
			\node (3) at (1,3) {};
			\draw	(02 |- ,1)	edge (3)
			(04 |- ,1)	edge (2)
			(01 |- ,1) .. controls ($(3)+(-1,-0.3)$) .. (3) 
			(03 |- ,1) .. controls ($(2)+(-1,-0.3)$) .. (2) 
			(2) .. controls ($(2)+(0,0.7)$) .. (3); 

			\draw[-|]	(2) .. controls ($(2)+(1,-0.3)$) .. ($(2)+(1,-0.9)$) node [above,xshift=0.3cm,fill=white] {$+$};
			\draw[-|]	(3) to node [at end,below,xshift=0.3cm,fill=white] {$+$} ++(0,0.4);
			
			\node (c1) at (1.5,-0.3) {};
			\node (c2) at (1.5,-0.6) {};
			\path	(c1)	edge (02)
			edge (03)
			(c2)	edge (01)
			edge (04);
			
			\node[rectangle,fill=white,anchor=north west] at (current bounding box.north west) {L3}; 
		\end{tikzpicture}
	}
	\\[\medskipamount]
	\resizebox{0.32\textwidth}{!}{
		\begin{tikzpicture}[every node/.style={fill=black,circle,inner sep=0pt,minimum size=0.25cm},every path/.style={outer sep=0pt,line width=0.8pt}]
			\foreach \x in {0,...,3} \draw[dotted]	(-0.5,\x) -- (4.5,\x); 
			
			\node (01) at (0,0) {};
			\node (02) at (1,0) {};
			\node (03) at (3,0) {};
			\node (04) at (4,0) {};
			\foreach \name in {01,02,03,04} \path[dashed] (\name) edge (\name |- ,1); 
			
			\node (2) at (1,2) {};
			\node (3) at (3,3) {};
			\draw	(02 |- ,1)	edge (2)
			(03 |- ,1)	edge (3)
			(01 |- ,1) .. controls ($(2)+(-1,-0.3)$) .. (2) 
			(04 |- ,1) .. controls ($(3)+(1,-0.3)$) .. (3) 
			(2) .. controls ($(2)+(0,0.7)$) .. (3); 
			
			\draw[-|]	(2) .. controls ($(2)+(1,-0.3)$) .. ($(2)+(1,-0.9)$) node [above,xshift=0.3cm,fill=white] {$+$};
			\draw[-|]	(3) to node [at end,below,xshift=0.3cm,fill=white] {$+$} ++(0,0.4);
			
			\node (c1) at (2,-0.3) {};
			\node (c2) at (2,-0.6) {};
			\path	(c1)	edge (02)
			edge (03)
			(c2)	edge (01)
			edge (04);
			
			\node[rectangle,fill=white,anchor=north west] at (current bounding box.north west) {L4}; 
		\end{tikzpicture}	
	}
	\resizebox{0.32\textwidth}{!}{
		\begin{tikzpicture}[every node/.style={fill=black,circle,inner sep=0pt,minimum size=0.25cm},every path/.style={outer sep=0pt,line width=0.8pt}]
			\foreach \x in {0,...,3} \draw[dotted]	(-0.5,\x) -- (4.5,\x); 
			
			\node (01) at (0,0) {};
			\node (02) at (1,0) {};
			\node (03) at (3,0) {};
			\node (04) at (4,0) {};
			\foreach \name in {01,02,03,04} \path[dashed] (\name) edge (\name |- ,1); 
			
			\node (2) at (2,2) {};
			\node (3) at (2,3) {};
			\draw	(2) edge (3);
			\draw	(01 |- ,1) .. controls ($(3)+(-2,-0.3)$) .. (3) 
			(02 |- ,1) .. controls ($(2)+(-1,-0.3)$) .. (2) 
			(03 |- ,1) .. controls ($(2)+(1,-0.3)$) .. (2) 
			(04 |- ,1) .. controls ($(3)+(2,-0.3)$) .. (3); 
			
			\draw[-|]	(2) to node [at end,above,xshift=0.3cm,fill=white] {$+$} ++(0,-0.9);
			\draw[-|]	(3) to node [at end,below,xshift=0.3cm,fill=white] {$+$} ++(0,0.4);
			
			\node (c1) at (1,-0.6) {};
			\node (c2) at (3,-0.6) {};
			\path	(c1)	edge (01)
			edge (03)
			(c2)	edge (02)
			edge (04);
			
			\node[rectangle,fill=white,anchor=north west] at (current bounding box.north west) {L5}; 
		\end{tikzpicture}
	}
	\resizebox{0.32\textwidth}{!}{
		\begin{tikzpicture}[every node/.style={fill=black,circle,inner sep=0pt,minimum size=0.25cm},every path/.style={outer sep=0pt,line width=0.8pt}]
			\foreach \x in {0,...,3} \draw[dotted]	(-0.5,\x) -- (4.5,\x); 
			
			\node (01) at (0,0) {};
			\node (02) at (1,0) {};
			\node (03) at (2,0) {};
			\node (04) at (4,0) {};
			\foreach \name in {01,02,03,04} \path[dashed] (\name) edge (\name |- ,1); 
			
			\node (2) at (3,2) {};
			\node (3) at (1,3) {};
			\draw	(02 |- ,1)	edge (3)
			(01 |- ,1) .. controls ($(3)+(-1,-0.3)$) .. (3) 
			(03 |- ,1) .. controls ($(2)+(-1,-0.3)$) .. (2) 
			(04 |- ,1) .. controls ($(2)+(1,-0.3)$) .. (2) 
			(2) .. controls ($(2)+(0,0.7)$) .. (3); 
			
			\draw[-|]	(2) to node [at end,above,xshift=0.3cm,fill=white] {$+$} ++(0,-0.9);
			\draw[-|]	(3) to node [at end,below,xshift=0.3cm,fill=white] {$+$} ++(0,0.4);
			
			\node (c1) at (1.5,-0.3) {};
			\node (c2) at (1.5,-0.6) {};
			\path	(c1)	edge (02)
			edge (03)
			(c2)	edge (01)
			edge (04);
			
			\node[rectangle,fill=white,anchor=north west] at (current bounding box.north west) {L6}; 
		\end{tikzpicture}
	}
	\caption{Half of the leading motives.
	The top four are \emph{gain motives}, while the bottom six are called \emph{loss motives}.
	The remaining leading motives (denoted by G1$,\dotsc,$G4,L1$^-,\dotsc,$L6$^-$) can be obtained from the ones depicted by first inverting the parities of all edges, and then inverting the order of the edges below each interaction vertex (for the loss motives, this corresponds to "inverting" the entire motive).}
	\label{fig:LeadingMotives}
\end{figure}
There are twenty such leading motives, half of them are shown in \cref{fig:LeadingMotives}.
The gain motives (e.g. the four top motives in \cref{fig:LeadingMotives}) can be used to replace any pairing that has the correct parity.
For instance, the depicted gain motives can only replace a pairing whose left leg has the parity $-1$, while the remaining "inverted" motives can replace a pairing whose left leg has the parity $+1$.
The loss motives may be attached to any edge in the graph that has the correct parity, so any of the twelve loss motives may be attached to any pairing by simply attaching it to the leg that has the suitable parity.
\begin{remark}[Leading motives]
\label{rem:leading_motives}
	The leading motives have two defining features.
	The first feature is that they \emph{preserve the incoming momentum}, which for the loss motives simply means that the momenta corresponding to the top and bottom edges (where the motives are attached to a diagram) coincide.
	This is easily checked by noting that the momentum at the top edge of a loss motive is equal to the sum of the five momenta at the bottom, four of which cancel pairwise, due to the pairings.
	For gain motives, the momenta corresponding to the two top edges are equal, but with opposite signs.
	
	The second feature of the leading motives is that they \emph{preserve the phase}, which simply means that the phase factors $\Theta$ corresponding to the two vertices are the same, but with \emph{opposite} signs.
	Suppose that we attach any leading motive to an existing Feynman diagram, and that the vertices of the motive in the resulting diagram are $v_{i-1}$ and $v_i$.
	Then, since $\Theta_{i-1:\ell_i} \myp{\bar{k},\bar{\sigma}} = - \Theta_{i:\ell_{i+1}} \myp{\bar{k},\bar{\sigma}}$, the presence of the vertices $v_{i-1}$ and $v_i$ will only give a contribution to the phase factor $\gamma_i$.
	All other phase factors $\gamma_j$ will be independent of $v_{i-1}$ and $v_i$ by \eqref{eq:gamma_j}, because the contributions from these two vertices cancel out.
\end{remark}
\begin{example}[Amplitude for the loss motive L1]
\label{ex:L1}
	\begin{figure}[h]
		\centering
		\includegraphics{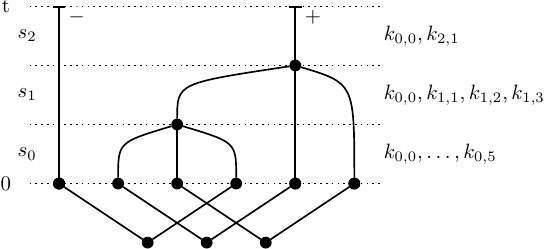}
		\caption{The leading motive L1 attached to a single pairing with parities $-,+$.}
		\label{fig:L1}
	\end{figure}
	For later reference, let us calculate the amplitude $\mathcal{F}_2^{main}$ for the diagram in \cref{fig:L1}, obtained by attaching the motive L1 to a single pairing with parities $-,+$.
	In this case, $\ell = \myp{1,1} $ and $S = \Set{\Set{0,3},\Set{1,4},\Set{2,5}}$.
	Since there are three intersections between cluster edges in the diagram, we have $\epsilon \myp{S} = -1$.
	It as also apparent from the diagram that the momenta are related by
	\begin{equation*}
		k_{0,0} ={} -k_{0,3},
		\qquad k_{0,1} ={} - k_{0,4},
		\qquad k_{0,2} ={} - k_{0,5},
	\end{equation*}
	\begin{equation*}
		k_{1,1} ={} k_{0,1} + k_{0,2} + k_{0,3},
	\end{equation*}
	and
	\begin{equation*}
		k_{2,1} ={} k_{1,1} + k_{1,2} + k_{1,3}
		={} \sum_{j=1}^5 k_{0,j}
		={} k_{0,3}.
	\end{equation*}
	From the two interaction vertices we get the factors
	\begin{align*}
		-\Psi_1 \myp{k_{0,1},k_{0,2},k_{0,3},-1}
		={}& -\Phi_1 \myp{k_{0,1},k_{0,2},k_{0,3}} U \myp{k_{0,1},k_{0,2},k_{0,3},-1} \\
		={}& -\Phi_1 \myp{k_{0,1},k_{0,2},k_{0,3}} \widehat{V} \myp{k_{0,2}+k_{0,3}}, \\
		\Psi_1 \myp{k_{1,1},k_{1,2},k_{1,3},1}
		={}& \Phi_1 \myp{k_{1,1},k_{0,4},k_{0,5}} \widehat{V} \myp{k_{1,1}+k_{0,4}}.
	\end{align*}
	Since all the clusters in the diagram are pairs, by \cref{rem:pair_correl_trunc}, we get the three factors $C_2 \myp{\myp{k_{0,0},k_{0,3}},\myp{-1,1},\lambda,\Omega} = W_L^{\lambda} \myp{k_{0,3},-1} = W_L^{\lambda} \myp{k_{0,3}}$, $W_L^{\lambda} \myp{k_{0,4}}$, and $W_L^{\lambda} \myp{k_{0,5}}$ from the cluster vertices.
	What remains is to calculate the phase factors in \eqref{eq:DefFmain},
	\begin{equation*}
		\Theta_{1:1} \myp{\bar{k},\bar{\sigma}}
		={} \Theta \myp{k_{1,1},k_{0,4},k_{0,5},1},
	\end{equation*}
	and note that
	\begin{align*}
		\Theta_{0:1} \myp{\bar{k},\bar{\sigma}}
		={}& \Theta \myp{k_{0,1},k_{0,2},k_{0,3},-1}
		={} \omega^{\lambda} \myp{k_{0,3}} - \omega^{\lambda} \myp{k_{0,1}} - \myp{\omega^{\lambda} \myp{k_{0,2}} - \omega^{\lambda} \myp{k_{1,1}}} \\
		={}& -\myp{ \omega^{\lambda} \myp{-k_{0,5}} - \omega^{\lambda} \myp{k_{1,1}} + \omega^{\lambda} \myp{k_{0,4}} - \omega^{\lambda} \myp{k_{0,3}} } \\
		={}& - \Theta \myp{k_{1,1},k_{0,4},k_{0,5},1},
	\end{align*}
	so the phase in the initial time slice vanishes.
	Finally, since all other momenta of the graph can be expressed in terms of the four momenta $k_{2,1}, k_{1,1}, k_{0,4}, k_{0,5}$, we can resolve all the $\delta_L$-functions, except for the one associated to the top interaction vertex, and express $\mathcal{F}_2^{main}$ in terms of the variables $k_0 = k_{2,1}$, $k_1\coloneq k_{1,1}$, $k_2 \coloneq k_{0,4}$, and $k_3 \coloneq k_{0,5}$.
	After noting from the relation $k_{2,1}+k_{0,0}$ that the delta functions of the top two edges give an additional contribution
	\begin{equation*}
		\int_{\Omega^{\ast}} \ids k_{0,0} \delta_L \myp{k_{2,1}-k_0} \delta_L \myp{k_{0,0}-k_0'}
		={} \delta_L \myp{k_0+k_0'},
	\end{equation*}
	we find that \eqref{eq:DefFmain} reduces to
	\begin{align}
		\MoveEqLeft[4] \mathcal{F}^{main}_{2} \myp{S,\ell,k_0',k_0,t/\eps,\kappa} \nn \\
		={}& -\lambda^2 \delta_L \myp{k_0+k_0'} \int_{\myp{\Omega^{\ast}}^3} \ids k \delta_L \myp{k_0-k_1-k_2-k_3} \widehat{V} \myp{k_1+k_2}^2 \nn \\
		&\times \Phi_1 \myp{-k_2,-k_3,k_0} \Phi_1 \myp{k_1,k_2,k_3} W_L^{\lambda} \myp{k_0} W_L^{\lambda} \myp{k_2} W_L^{\lambda} \myp{k_3} \nn \\
		&\times \int_{\myp{\R_+}^3} \ids s \delta \myp[\bigg]{ \frac{t}{\eps}-\sum_{j=0}^2 s_j} \prod_{j=0}^{2} e^{-s_j \kappa_{2-j}} e^{- i s_1 \Theta \myp{k_1,k_2,k_3,1}}.
	\label{eq:L1}
	\end{align}
\end{example}

\section{A review of momentum graphs}\label{Sec:MomentumGraph}
In Subsections \ref{Sec:MainDiagram} and \ref{Sec:AmputatedDiagram}, we introduced Feynman diagrams for the main and error terms of the Duhamel expansion of $Q^{\lambda} \myp{k',k,t}$.
We will combine these two diagrams into a single class of diagrams called momentum graphs, and review the general properties of these graphs.
Since the general structure of the graphs appearing in our problem coincides completely with those appearing in the setting of the weakly nonlinear Schr\"odinger equation \cite{LukkarinenSpohn:WNS:2011}, we can reuse the toolbox that has already been developed for the momentum graphs in \cite{LukkarinenSpohn:WNS:2011}.
For the sake of completeness, we recall in this section the relevant notions and results on momentum graphs, but mostly refer to \cite{LukkarinenSpohn:WNS:2011} for the proofs.

We add one more delta-function to the amplitude terms (e.g. \eqref{eq:DefFmain}) by multiplying with a factor
\begin{align}
\label{eq:factormomentumgraphs}
	1 ={} \int_{\Omega^{\ast}} \ids k_{e_0} \delta_L \myp{k_{e_0}-k-k'},
\end{align}
where $\myp{k,k'}=\myp{k_{n,1},k_{0,0}}$ for a main diagram and  $\myp{k,k'} = \myp{k_{n,1},k_{n,1}'}$ for the diagram of an error term.
In the new graph, the additional delta function is represented by introducing two new vertices and one new edge to the graph.
One of the new vertices is a \emph{fusion vertex} connecting two edges $k$ and $k'$ to the new edge $k_{e_0}$, and the other vertex is a \emph{root vertex} placed at the top of the graph, at the other end of $k_{e_0}$.
The new graph is called the \emph{momentum graph}.
We denote a momentum graph by $\mathscr{G}= \myp{\mathscr{V},\mathscr{E}}$, where $\mathscr{V}$ and $\mathscr{E}$ are the sets of vertices and edges.

This graph depends on the cluster partitioning $S$, the number of interaction vertices $n'$ and $n$ in the minus and plus trees, the interaction histories $\ell'$ and $\ell$ in the minus and plus trees, and finally the interlacing function $J$.
In the case of a main diagram, $n'=0$, and then $\ell',J$ are not relevant. 
As a consequence, given $S,n',n,\ell',\ell,J$, we can reconstruct the momentum graph.

One important distinction between the present work and \cite{LukkarinenSpohn:WNS:2011} is that our main result is uniform in the momenta $k,k'$, while the corresponding result in \cite{LukkarinenSpohn:WNS:2011}, analogous to \cref{cor:main}, relies on averaging over the top momenta of the graphs.
However, we can keep the same structure of the momentum graphs as in \cite{LukkarinenSpohn:WNS:2011} by manually inserting momentum integrals and delta functions for the top edges, as in \eqref{eq:top_momenta}.
The only difference from \cite{LukkarinenSpohn:WNS:2011} is then that the two top edges are each associated with an additional delta function.
To handle the presence of these delta functions, the only additional result on momentum graphs that we need is contained in \cref{lemma:odd_dependence} below.

\subsubsection*{Construction of momentum graphs}
Given the parameters $S,n',n,\ell',\ell,J$, the momentum graph can be reconstructed the following way (see Figure \ref{Fig4}). 
\begin{figure}
\centering
	\includegraphics{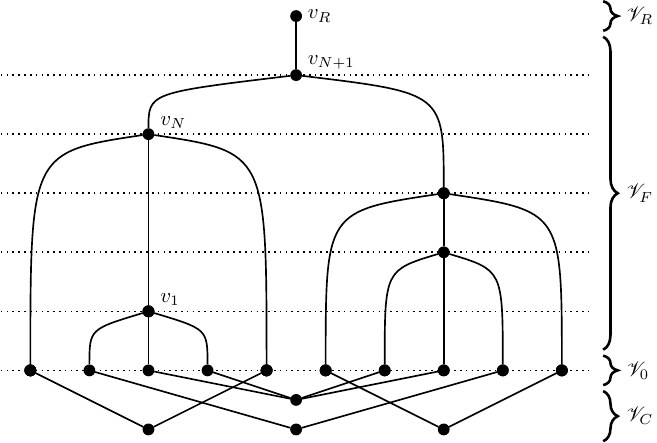}
	\caption{An example of a momentum graph with $n=n'=2$ and cluster decomposition $S = \Set{ \Set{1,5}, \Set{2,9}, \Set{3,4,7,8}, \Set{6,10}}$.
		The sets $\mathscr{V}_R$, $\mathscr{V}_F$, $\mathscr{V}_0$, $\mathscr{V}_C$ and the vertices $v_R$, $v_N$ are marked on the graph.}
\label{Fig4}
\end{figure}

\begin{itemize}
	\item Let $N=n+n'$ denote the total number of interactions.
	We start with $\mathscr{G}^0= \myp{\mathscr{V}^0,\mathscr{E}^0}$, where $\mathscr{V}^0$ contains two vertices $v_R$ and $v_{N+1}$, and $\mathscr{E}^0$ contains only one edge $e_0$ that connects those two vertices.

	\item At the first iteration, we attach one new edge $e_1$ to $v_{N+1}$.
	This edge belongs to the minus tree.
	Using $J$, we can determine if the final interaction should be in the minus tree.
	If this is the case, we label the edge $e_1 = \Set{v_{N+1},v_N}$; otherwise, we do not label it.
	
	\item In the second iteration, we attach a new edge $e_2$ to $v_{N+1}$.
	This edge then belongs to the plus tree.
	If the final interaction does not occur in the minus tree, then it is in the plus tree, and we label $e_2 = \Set{v_{N+1},v_N}$.
	
	\item The next three steps are to attach another three edges to $v_N$, from left to right (with parities $-1,\sigma,1$, as explained in \cref{Sec:MainDiagram,Sec:AmputatedDiagram}).
	For the moment, we leave the vertices that these edges connect to unlabelled.
	
	\item Using $\ell,\ell'$ and $J$, we can determine which edge should connect to the next interaction vertex.
	This edge has an unlabelled vertex which we now denote by $v_{N-1}$.
	
	\item This procedure of attaching three new edges to the most newly created interaction vertex and then labelling the next interaction vertex is iterated $N$ times in total.
	This produces a tree starting from $v_R$ with a vertex set composed of $N+2$ labelled and $2N+2$ unlabelled vertices.

	\item The labelled vertices belong to $\mathscr{V}_R= \Set{v_R}$ or $\mathscr{V}_F = \Set{v_1,\cdots,v_{N+1}}$.
	The set $\mathscr{V}_F$ is called the \emph{fusion vertex set}.
	When excluding the topmost fusion vertex $v_{N+1}$, the remaining fusion vertices $\mathscr{V}_I =\mathscr{V}_F \setminus \Set{v_{N+1}}$ are called \emph{interaction vertices}.
	They correspond to the interaction vertices in the original diagrams from \cref{sec:feynman_diagrams}.
	The set of the remaining vertices is denoted by $\mathscr{V}_0$ and is called the \emph{initial time vertex} set.
	
	\item For each cluster $A$ in $S$, we add a vertex labelled $u_A$.
	The set $\mathscr{V}_C = \Set{u_A}_{A\in S}$ is called the \emph{cluster vertex set}.
	Then the vertex set of the final graph $\mathscr{G}$ is $\mathscr{V} = \mathscr{V}_R \cup \mathscr{V}_F \cup \mathscr{V}_0 \cup \mathscr{V}_C$.
	The remaining edges connecting the initial time vertices to the cluster vertices can be added in a natural way (going through $\mathscr{V}_0$ from left to right). 
\end{itemize} 
The above process gives us an unoriented graph $\mathscr{G} = \mathscr{G} \myp{S,J,n,\ell,n',\ell'}= \myp{\mathscr{V},\mathscr{E}}$.
For two edges $e$ and $e'$, we say that $e<e'$ if $e$ is created before $e'$.
This defines then a complete order $e\leq e'$ on $\mathscr{E}$.
 
\subsection{Properties of momentum graphs}
\label{Sec:PropMometum}
We collect here various results on the properties of momentum graphs.
We simply recall the relevant definitions and constructions, but otherwise state the results without proof, since all the proofs are contained in \cite{LukkarinenSpohn:WNS:2011}.
To study the properties of the momentum graph $\mathscr{G}$ constructed above, we need the following concepts.
\begin{itemize}
	\item Define the function $\mathcal{T}: \mathscr{V} \to \myb{0,N+2}$, in which $\mathcal{T} \myp{v} = N+2$ if $v\in \mathscr{V}_R$, $\mathcal{T} \myp{v}=j$ if $v=v_j$ for $j \in \Set{1,\cdots,N+1}$, and $\mathcal{T} \myp{v}=0$ in the other cases (meaning $v \in \mathscr{V}_0 \cup \mathscr{V_C}$).
	This function gives a natural time-order on the vertices of $\mathscr{G}$.
	
	\item We also define 
	\begin{equation*}
		\widehat{\mathcal{T}} \myp{e} 
		={} \max \Set{\mathcal{T} \myp{v} \mid v\in e} \mbox{ for }  e\in \mathscr{E}.
	\end{equation*}
	It is clear that $e\leq e'$ leads to $\widehat{\mathcal{T}} \myp{e} \geq \widehat{\mathcal{T}} \myp{e'}$. 
	
	\item For any $v \in \mathscr{V}$, we denote the set of edges attached to $v$ by $\mathscr{E} \myp{v} = \Set{e \in \mathscr{E} \mid v \in e}$. 
	
	\item For $v \in \mathscr{V} \setminus \mathscr{V}_R$, there is a delta function associated to it
	\begin{equation}
	\label{eq:deltavertex}
		\delta_L \myp[\bigg]{ \sum_{e\in \mathscr{E}_+ \myp{v}} k_e - \sum_{e \in \mathscr{E}_- \myp{v}} k_e },
	\end{equation}  
	where $\mathscr{E} \myp{v} = \mathscr{E}_+ \myp{v} \cup \mathscr{E}_- \myp{v}$.
	If $v \in \mathscr{V}_C$, then $\mathscr{E}_- \myp{v} = \emptyset$ and  $\mathscr{E}_+ \myp{v} = \mathscr{E} \myp{v}$.
	Otherwise, $\mathscr{E}_+ \myp{v} = \Set{e}$, where $e$ is the first edge attached to $v$ in the graph construction described above (meaning that $e$ maximizes $\widehat{\mathcal{T}} \myp{e'}$ for $e' \in \mathscr{E} \myp{v}$), and $\mathscr{E}_- \myp{v} = \mathscr{E} \myp{v} \setminus \Set{e}$.
	
	\item The main aim of this section is to develop a standard procedure for integrating out all the delta functions.
	This can be done (see below) by associating to each vertex a particular edge which will be used in the integration.
	This edge is called \emph{integrated} and the remaining edges are called \emph{free}.
	We denote by $\mathscr{E}'$ the set of integrated edges and by $\mathscr{F}$ the set of free edges.
	We call this process \emph{integrating the momentum constraints}.
	
	\item For any fusion vertex $v \in \mathscr{V}_F$, we then call the number of free edges in $\mathscr{E}_- \myp{v}$ the \emph{degree} of the fusion vertex and denote it by $\deg \myp{v}$.
\end{itemize}  

\subsubsection*{Resolution of the momentum constraints}
Below, we introduce the algorithm that determines the sets of free and integrated edges in the graph $\mathscr{G}$.
\begin{theorem}
\label{thm:momentum_resolution}
	Let $\mathscr{G}$ be a momentum graph.
	There exists a spanning tree $\mathscr{T}$ of $\mathscr{G}$ that determines an integration of the momentum constraints.
	It satisfies that for any free edge $f$, all the edges $e<f$ satisfy that $k_{e}$ is independent of $k_{f}$.
	Moreover, all free edges end at a fusion vertex: if $f$ is free, there is a fusion vertex $v \in \mathscr{V}_F$ and $v' \in \mathscr{V}$ such that $\mathcal{T} \myp{v} > \mathcal{T} \myp{v'}$ and $f$ is the edge connecting $v$ and $v'$.
\end{theorem}
\begin{proof}[Part of proof]
	The construction consists of two steps.
	First, we will construct an unoriented tree $\mathscr{T} = \myp{\mathscr{V}_{\mathcal{T}}, \mathscr{E}_{\mathcal{T}}}$ from the graph $\mathscr{G}$.
	This can be done by an iteration process.
	Let us first start with $\mathscr{T}^{\myp{0}} = \myp{ \mathscr{V}_{\mathcal{T}}^{\myp{0}}, \mathscr{E}_{\mathcal{T}}^{\myp{0}} }$ in which both $\mathscr{V}_{\mathcal{T}}^{\myp{0}},\mathscr{E}_{\mathcal{T}}^{\myp{0}}$ are empty sets.
	The key idea is to go through all edges in $\mathscr{E}$ in the opposite direction that they were created above.
	To be more precise, suppose that at step $l$, the graph $\mathscr{T}^{\myp{l-1}}$ is created.
	Now, at step $l$, the corresponding edge is denoted by $e$.
	If adding $e$ to $\mathscr{T}^{\myp{l-1}}$ creates a loop i.e. a subset of the graph in which we can start with any vertex, follow the edges and return to the same vertex (see Figure \ref{Fig5}), then we do not add $e$ to $\mathscr{T}^{\myp{l-1}}$ and set $\mathscr{T}^{\myp{l}} = \mathscr{T}^{\myp{l-1}}$. 
	
	\begin{figure}
	\centering
		\includegraphics{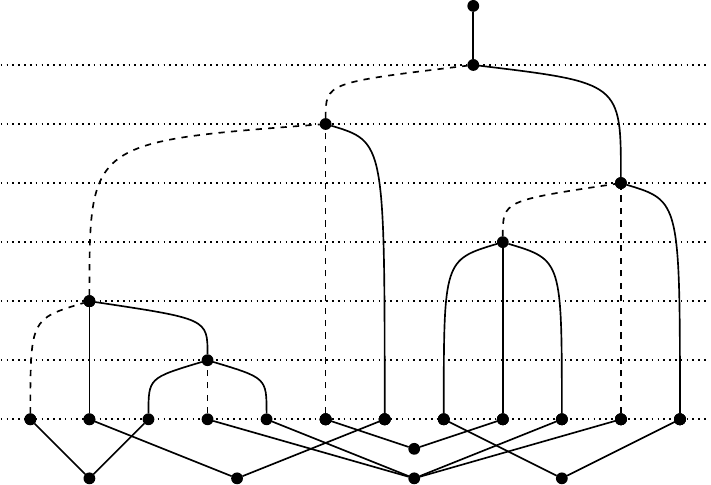}
		\caption{A momentum graph with $n'=3$ and $n=2$, where the free edges are denoted by dashed lines.
		The remaining (integrated) edges constitute the spanning tree constructed in the proof of \cref{thm:momentum_resolution}.}
	\label{Fig5}
	\end{figure}
	If adding $e$ to $\mathscr{T}^{\myp{l-1}}$ does not create a loop, then we define $\mathscr{T}^{\myp{l}}$ to be the graph created from this addition
	That is, $\mathscr{V}_{\mathcal{T}}^{\myp{l}}$ is created by adding the vertices of $e$ to $\mathscr{V}_{\mathcal{T}}^{\myp{l-1}}$ and $\mathscr{E}_{\mathcal{T}}^{\myp{l}}$ is created by adding  $e$ to $\mathscr{E}_{\mathcal{T}}^{\myp{l-1}}$.
	In the first case, since adding $e$ to $\mathscr{T}^{\myp{l-1}}$ would create a loop, the vertices of $e$ would already belong to $\mathscr{V}_{\mathcal{T}}^{\myp{l-1}} = \mathscr{V}_{\mathcal{T}}^{\myp{l}}$.
	In the second case, it is obvious that the vertices of $e$ also belong to $\mathscr{V}_{\mathcal{T}}^{\myp{l}}$.
	In both cases, no vertex in $e$ is lost in the iteration step.
	At the final step, we obtain the graph $\mathscr{T} = \myp{ \mathscr{V}_{\mathcal{T}}, \mathscr{E}_{\mathcal{T}} }$.
	It is then clear that $\mathscr{T}$ is a tree, and as every vertex of $\mathscr{V}$ belongs to some edge, we have $\mathscr{V}_{\mathcal{T}} = \mathscr{V}$.
	Moreover, no loop is contained in $\mathscr{E}_{\mathcal{T}} \subset \mathscr{E}$; and for every $e \in \mathscr{E} \setminus \mathscr{E}_{\mathcal{T}}$, adding $e$ to $\mathscr{E}_{\mathcal{T}}$ would create a loop  composed of edges $e'$ satisfying $e'\geq e$ and $\widehat{\mathcal{T}} \myp{e'} \leq \widehat{\mathcal{T}} \myp{e}$.
	
	Second, we create a new oriented tree $\mathscr{T}_0$ by adding an orientation to $\mathscr{T}$.
	We denote the root of the tree to be $v_R$ and iterate the following algorithm.
	Now, let us notice that for any vertex $v$, there is a unique path connecting $v_R$ and $v$.
	The path can be oriented such that it begins with $v$ and ends with $v_R$.
	This process is iterated for all the edges of the tree $\mathscr{T}$.
	The iterations may have two different vertices that share edges along the path, but these edges have the same orientation at all of the steps in the algorithm.
	Moreover, if two such paths share a vertex, the paths must coincide after this vertex; otherwise, a loop exists in the graph, which is a contradiction with the construction described above.
	The resulting graph has the following important feature: for every $v \in \mathscr{V} \setminus \mathscr{V}_R$, there is unique edge $ E \myp{v} \in \mathscr{E} \myp{v}$ pointing out of the vertex.
	Moreover, the map $E: \mathscr{V} \setminus \mathscr{V}_R \to \mathscr{E}_{\mathcal{T}}$ is one-to-one. 
	
	As a result, we can integrate all the delta functions associated to the vertices $v \in \mathscr{V} \setminus \mathscr{V}_R$, using the variable $k_{E \myp{v}}$ associated to the edge $E \myp{v}$.
	After integrating all the edges, the remaining variables consist of $k_e$, with $e \in \mathscr{E} \setminus \mathscr{E}_{\mathcal{T}}$.
	These are free integration variables, and we set $ \mathscr{F} = \mathscr{E} \setminus \mathscr{E}_{\mathcal{T}}$ to be the set of free edges and $\mathscr{E}'=\mathscr{E}_{\mathcal{T}}$ to be the set of integrated edges.
\end{proof}
\begin{remark}[Integrating the additional delta functions]
	\label{rem:additional_deltas}
	The procedure in \cref{thm:momentum_resolution} handles only the delta functions associated to the \emph{vertices} of the graph $\mathscr{G}$.
	However, the amplitude expressions in \cref{Sec:Q_expansions} carry additional delta functions for the top edges of the two trees of the graph (corresponding for instance to $k_{n,1},k_{n,1}'$ for an error term).
	Since $k_{n,1}$ is always an integrated edge by construction, we use the free edge $k_{n,1}'$ to integrate these delta functions.
	We obtain from the relation $k_{n,1}' + k_{n,1} = 0$ that
	\begin{equation*}
		\int_{\Omega^{\ast}} \ids k_{n,1}' \delta_L \myp{k_{n,1}-k} \delta_L \myp{k_{n,1}'+k}
		={} \int_{\Omega^{\ast}} \ids k_{n,1}' \delta_L \myp{k_{n,1}'+k}^2
		={} L^d,
	\end{equation*}
	and similarly for a main term,
	\begin{equation*}
		\int_{\Omega^{\ast}} \ids k_{0,0} \delta_L \myp{k_{n,1}-k} \delta_L \myp{k_{0,0}-k'}
		={} \delta_L \myp{k+k'} \int_{\Omega^{\ast}} \ids k_{0,0} \delta_L \myp{k_{0,0}-k'}
		={} \delta_L \myp{k+k'}.
	\end{equation*}
\end{remark}

The rest of this subsection mainly describes how the integrated edges of a graph depend on the free ones.
First, we define a few more concepts related to the oriented graph. 
\begin{itemize}
	\item For any $v \in \mathscr{V}$, let $\mathscr{F} \myp{v}$ denote the set of free edges attached to $v$, i.e. $ \mathscr{F} \myp{v} \coloneq \mathscr{E} \myp{v} \cap \mathscr{F}$.
	
	\item Any $v \in \mathscr{V} \setminus \mathscr{V}_R$ is associated with an \emph{edge parity} mapping $\sigma_v : \mathscr{E} \myp{v} \to \Set{\pm 1}$, defined by 
	\begin{equation}
	\label{eq:edgeparity}
		\sigma_v \myp{e} ={}
		\begin{cases}
			1, &\mbox{ if } e \in \mathscr{E}_+ \myp{v},\\
			-1, &\mbox{ if } e \in \mathscr{E}_- \myp{v}.
		\end{cases}
	\end{equation}
\end{itemize}
The following lemmas detail how the integrated edges belonging to an interaction vertex may depend on the free edges of the graph.
\begin{lemma}
\label{lemma:degreeofafusionvertex}
	Let $v$ be any fusion vertex.
	Then $\mathrm{deg} \myp{v} \in \Set{0,1,2}$. 
	\begin{enumerate}[(a)]
		\item If $v \in \mathscr{V}_I$ and $\mathrm{deg} \myp{v} = 1$, then $\mathscr{E}_- \myp{v} = \Set{f,e,e'}$, where $f$ is a free edge and $k_e = -k_f + q$, $k_{e'}=q'$, where both $q$ and $q'$ are independent of $k_f$. 
		
		\item If $v \in \mathscr{V}_I$ and $\mathrm{deg} \myp{v} = 2$, then $\mathscr{E}_- \myp{v} = \Set{f,f',e}$, where $f,f'$ are free edges and $k_e = -k_f -k_{f'} + q$, where $q$ is independent of both $k_f$ and $k_{f'}$. 
	\end{enumerate}
\end{lemma}
It is worth noting that the vertex degrees are \emph{not} a graph invariant; they depend on the construction of the spanning tree in \cref{thm:momentum_resolution}.
It is easy to construct a different spanning tree such that the vertices can have $\deg \myp{v} = 3$.
\begin{lemma}
\label{lemma:orientedpathedcomposition}
	For any integrated edge $e = \myp{v_1,v_2} \in \mathscr{E}'$, the following identities hold true
	\begin{align}
		k_e 
		={}& \sum_{v \in \mathscr{P} \myp{v_1}} \sum_{f = \myp{v,v_f} \in \mathscr{F} \myp{v}} \1 \myp{v_f \notin \mathscr{P} \myp{v_1}} \myp[\big]{-\sigma_{v_1} \myp{e} \sigma_v \myp{f} } k_f \nn \\
		={}& -\sigma_{v_1} \myp{e} \sum_{f \in \mathscr{F}} \1 \myp{\exists v \in f \cap \mathscr{P} \myp{v_1} \mbox{ and } f \cap \mathscr{P} \myp{v_1}^c \neq \emptyset} \sigma_v \myp{f} k_f.
	\label{eq:orientedpathedcomposition}
	\end{align}
	Moreover, if $f = \myp{v,v'}$, such that $f\neq e$, $v \in \mathscr{P} \myp{v_1}$ and $v' \notin \mathscr{P} \myp{v_1}$, then $f$ is free. 
\end{lemma}
In the above formula, the edge $e$ is said to \emph{depend} on the edges $f$, and $k_e$ is said to \emph{depend} on $k_f$.
To put it in words, the integrated edge $e$ depends on any given free edge $f$ if and only if there is there is a path in the oriented tree $\mathscr{T}_0$ going to $e$ from exactly \emph{one} of the vertices connected to $f$ (and not from the other), which is the case exactly when $e$ belongs to the loop that is defined by adding $f$ to the spanning tree.
As a consequence of the lemma, given a free edge $f$, any integrated edge can depend on it only through $\pm k_f$.
It is not possible to have any other prefactors in the linear relations of \eqref{eq:orientedpathedcomposition}.
\begin{corollary}
\label{corollary:decompositionfreeedgetofreeedge}
	For any edge $e \in \mathscr{E}$, there exists a unique set of free edges $\mathscr{F}_e$ and parities $\sigma_{e,f} \in \Set{\pm 1}$ such that
	\begin{equation}
	\label{eq:decompositionfreeedgetofreeedge}
		k_e ={} \sum_{f\in \mathscr{F}_e} \sigma_{e,f} k_f. 
	\end{equation}
	Moreover, $k_e$ is independent of all free momenta if and only if $k_e = 0$.
	This is equivalent to $\mathscr{F}_e = \emptyset$, which occurs only when the number of connected components of $ \mathscr{G} $ increases by $ 1 $ when $e$ is removed from the graph. 
\end{corollary}
This corollary implies for instance that the edge $e_0$ attached to the root vertex $v_R$ always satisfies $k_{e_0}=0$, since removing $e_0$ isolates $v_R$ in the graph.
Hence, by \eqref{eq:factormomentumgraphs}, the top momenta of the plus and minus trees always sum to zero.
We already used this fact in the expansions of the main term $Q^{main}$ and the error terms $\bar{Q}^{err}_j$ in \cref{sec:feynman_diagrams}.
\begin{lemma}
	\label{lemma:freeedgesorder2}
	Let $v_0 \in \mathscr{V}_I$ and suppose that $\mathrm{deg} \myp{v_0} = 2$.
	Let $f,f'$ be the two free edges ending at $v_0$. 
	\begin{enumerate}[(a)]
		\item Consider an integrated edge $e= \myp{v_1,v_2} \in \mathscr{E}'$.
		Then $k_e = F_e \myp{k_f,k_{f'}} + q_e$, where $ q_e $ is independent of $k_f$ and  $k_{f'}$, and $ F_e \myp{k_f,k_{f'}} $ is one of the seven functions: $0$, $\pm k_f$, $\pm k_{f'}$, $\pm \myp{k_f+k_{f'}}$.
		
		\item If $v \in \mathscr{V}_I$, and suppose $e,e' \in \mathscr{E}_- \myp{v}$, $e\neq e'$, then $k_e + k_{e'} = F \myp{k_f,k_{f'}} + q$, where $ q $ is independent of $k_f$ and  $k_{f'}$, and $ F \myp{k_f,k_{f'}} $ is also one of the seven functions: $0$, $\pm k_f$, $\pm k_{f'}$, $\pm \myp{k_f+k_{f'}}$.
		
		\item If $v = v_0$, and suppose $e,e' \in \mathscr{E}_- \myp{v}$, $e\neq e'$, then $k_e + k_{e'} = F \myp{k_f,k_{f'}} + q$, where $ q $ is independent of $k_f$ and  $k_{f'}$, and $ F \myp{k_f,k_{f'}} $ is one of the three functions $- k_f$, $- k_{f'}$, or $k_f+k_{f'}$.
	\end{enumerate}
\end{lemma}
In particular, the lemma rules out that an integrated momentum $k_e$ can depend on $k_f$ and $k_{f'}$ through $\pm \myp{k_f - k_{f'}}$.
\begin{lemma}
\label{lemma:removalofedges}
	Suppose that $e,e' \in \mathscr{E}$, $e \neq e'$, and $\mathscr{F}_e = \mathscr{F}_{e'}$.
	Then either $k_e = 0 = k_{e'}$ or if one removes both $e$ and $e'$, the graph $\mathscr{G}$ splits into two disconnected graphs.
\end{lemma}
\begin{lemma}
\label{lemma:dependenceoftwoedges0}
	Suppose that $e,e' \in \mathscr{E}$, $e \neq e'$.
	Then $\mathscr{F}_e = \mathscr{F}_{e'}$ if and only if there is a parity $\sigma \in \Set{\pm 1}$ such that $k_e = \sigma k_{e'}$ independently of the free momenta. 
\end{lemma}
\begin{definition}
\label{def:connection}
	Two sets $V_1, V_2 \subset \mathscr{V}_0$ are said to have a connection between them if there is a cluster that connects them in the sense that there is $A\in S$ such that $A \cap V_1 \neq \emptyset$ and $A \cap V_2\neq \emptyset$.
	If there is no connection between $V_1$ and $V_2$, they are said to be isolated from each other.
\end{definition}
The two following results are helpful to identify graphs where the corresponding amplitude is identically zero.
Conversely, they also assert useful properties for graphs with non-zero amplitudes.
\begin{lemma}
\label{lemma:dependenceoftwoedges}
	Let $v \in \mathscr{V}_I$ and suppose $e,e' \in \mathscr{E}_- \myp{v}$, $e \neq e'$.
	Then $k_e+k_{e'}$ is independent of all free momenta if and only if the initial time vertices at the bottom of the interaction trees beginning with $e$ and $e'$ are isolated from the rest of the initial time vertices.
	It also follows that $k_{e}+k_{e'}=0$. 
\end{lemma}
\begin{lemma}
\label{lemma:oddcluster}
	If the momentum graph has an edge $e$, which is different from the topmost edge, such that $k_e=0$, then $S$ contains a cluster with an odd number of vertices.
\end{lemma}
\begin{lemma}
\label{lemma:freemomenta}
	A momentum graph has exactly $2 \myp{n+n'} + 2 - \abs{S}$ free momenta.
\end{lemma}
As a side note, the proof of this in \cite{LukkarinenSpohn:WNS:2011} uses a construction of a second spanning tree along with the fact that the size of spanning trees is a graph invariant.
However, since by the construction of the spanning tree in \cref{thm:momentum_resolution} there is exactly one integrated edge for every vertex in the graph (excluding $v_R$), the result can be obtained by simply counting vertices and edges.

Finally, in order to prove the main theorem uniformly in the momenta, we will need the following refinement of \cref{lemma:oddcluster}
Since this is not contained in \cite{LukkarinenSpohn:WNS:2011}, we include the proof here.
\begin{lemma}
\label{lemma:odd_dependence}
	Let $\mathscr{G}$ be any momentum graph such that the cluster decomposition $S$ contains only clusters with an even number of elements.
	Then the following hold:
	\begin{enumerate}[(i)]
		\item For any edge $e \in \mathscr{E}$ which is not the unique edge attached to the root vertex, the momentum $k_e$ depends on an odd number of free momenta.

		\item Let $v \in \mathscr{V}_I$ be any interaction vertex, and $e_1,e_2 \in \mathscr{E}_- \myp{v}$.
		Then $k_{e_1} + k_{e_2}$ depends on an even number of free momenta.
		In particular, the sum cannot depend only on the free momentum of the top fusion vertex of the graph.
	\end{enumerate}
\end{lemma}
\begin{proof}
	The lemma is a simple consequence of the fact that all the delta functions associated to the vertices of the graph (except for the top fusion vertex) involve an even number of momenta.
	Since any free edge depends on precisely one free momentum by definition, we may consider only the integrated edges along the oriented spanning tree $\mathscr{T}_0$ of the graph.
	For any integrated edge $e \in \mathscr{E}'$, there is a unique oriented path from $e$ to the root vertex, as well as a number of paths from a non-empty subset of the leaves of $\mathscr{T}_0$ to the edge $e$.
	Along each such path going from a leaf of $\mathscr{T}_0$ to $e$, we may count the number of times the path passes through a vertex, and we define an edge index $n \myp{e}$ to be the largest number of vertices that are passed along any path in the spanning tree going from a leaf to $e$.
	
	Recalling that the spanning tree $\mathscr{T}_0$ is constructed such that each vertex in $\mathscr{G}$ has exactly one integrated edge pointing out of it, we may proceed to prove \emph{(i)} by induction in $n \myp{e}$ along the tree.
	First, if an edge $e$ has $n \myp{e}=0$ then the initial vertex of $e$ must be a leaf of the spanning tree, and all other edges in $\mathscr{G}$ that are connected to this vertex must be free.
	Since by construction, $e$ cannot be attached to the top fusion vertex or the root vertex of $\mathscr{G}$, and all other vertices have an even number of edges attached to them, it follows that $k_e$ is a linear combination of the odd number of free momenta that are attached to the initial vertex of $e$.
	
	Thus, for induction, we may assume that $\emph{(i)}$ holds for all edges $e$ with $n \myp{e} < k$ for some $k \geq 1$.
	Fix an edge $e \in \mathscr{E}'$, assume that $n \myp{e} = k$, and let $v$ be the vertex that $e$ points out of.
	Then all other integrated edges $e'$ attached to $v$ must point towards $v$, so they must satisfy $n \myp{e'} < n \myp{e}$, and hence also \emph{(i)}.
	By the delta function associated to $v$, we conclude that $k_e$ is a linear combination of the odd number of momenta $k_{e'}$ of the other edges $e' \in \mathscr{E} \myp{v}$, each of which depends on an odd number of free momenta, by the induction hypothesis.
	There may be some cancellations among these free momenta, but since they can only cancel in pairs, the only possibility is that $k_e$ itself also depends on an odd number of free momenta.
	This finishes the proof of \emph{(i)}.
	
	The second statement follows easily from the first, in combination with \cref{lemma:orientedpathedcomposition}.
	It is clear that $k_{e_1}+k_{e_2}$ is a sum of an even number of terms, so for the assertion in \emph{(ii)}, we only have to make sure that if both $k_{e_1}$ and $k_{e_2}$ depend on the same free momentum $k_f$, the dependence must be with opposite sign, so that the corresponding terms cancel in $k_{e_1}+k_{e_2}$.
	This is, however, a clear consequence of \cref{lemma:orientedpathedcomposition}, since both $e_1$ and $e_2$ must belong to the loop of $f$ in order to depend on it.
	Thus, denoting by $e_3 \in \mathscr{E}_- \myp{v}$ and $e_0 \in \mathscr{E}_+ \myp{v}$ the last two edges attached to $v$, both $e_3$ and $e_0$ are independent of $f$, and hence $k_{e_1} + k_{e_2} = k_{e_0} - k_{e_3}$ must be as well.
\end{proof}

\subsection{Classification of momentum graphs}
We now introduce the following new concepts for a momentum graph.
\begin{enumerate}[(i)]
	\item Let $j$ be a time slice in $I_{0,n+n'}$ of a momentum graph.
	If the length of this time slice is $0$, this time slice is then called \emph{amputated}.
	If it ends in an interaction vertex, whose degree is either $1$ or $2$, the time slice is called \emph{short}.
	Otherwise, it is a \emph{long} time slice.
	
	\item The graph is called \emph{irrelevant}, if the corresponding amplitude is equal to zero.
	Otherwise, it is \emph{relevant}.
	
	\item The graph is called \emph{pairing} if for every $A\in S$, we have $ \abs{A}=2$.
	Otherwise, it is \emph{non-pairing}.
	
	\item The graph is called \emph{higher order} if it is relevant and non-pairing. 
	
	\item The graph is called \emph{fully paired} if it is pairing and it has no interaction vertices of degree one. 
	
	\item A pairing graph is called \emph{partially paired} if it is not fully paired. 
\end{enumerate}
\begin{remark}
\label{rem:irrelevant_graphs}
	We have the following examples of irrelevant graphs:
	\begin{itemize}
		\item If $S$ contains a cluster $A$ with an odd number of vertices, then the graph is irrelevant because the corresponding $C_{\abs{A}}$ is identically zero.
		
		\item If $S$ does not respect the parities of the edges in the initial time slice, then the graph is irrelevant, also because some of the truncated correlation functions $C_{\abs{A}}$ vanish.
		
		\item If there is a vertex $v \in \mathscr{V}_I$ with two edges $e,e' \in \mathscr{E}_- \myp{v}$ satisfying $k_e+k_{e'} = 0$, then the graph is irrelevant because of the presence of the cut-off function $\Phi_1$ at the vertex.
		(The same conclusion also holds for the amputated vertices of the terms $\mathcal{F}_n^2$ because of \cref{Propo:QERR2}).
	\end{itemize}
\end{remark}

\subsubsection{The iterative cluster scheme}
An important tool for estimating the amplitudes associated with a momentum graph is knowing bounds on the number of interaction vertices of various degrees.
Denote by $n_j$ the number of interaction vertices of degree $j$.
Then, since we know from \cref{lemma:degreeofafusionvertex} that the degree of a vertex can only to be either $0,1$, or $2$, it is clear that the total number of interaction vertices can be expressed as $n+n'=n_0+n_1+n_2$.

Denote also by $n_l \myp{j}$ with $l=0,1,2$ the number of interaction vertices of degree $l$ below and including $v_j$.

To estimate these numbers, one can apply the \emph{iterative cluster scheme} (see Figure \ref{Fig6}),
which follows the evolution of the cluster structure when integrating out the delta functions from the bottom to the top of the graph, as described in \cref{Sec:PropMometum}.
We first define $S^{\myp{0}}$ to be $S$. 
\begin{figure}
	\centering
		\includegraphics{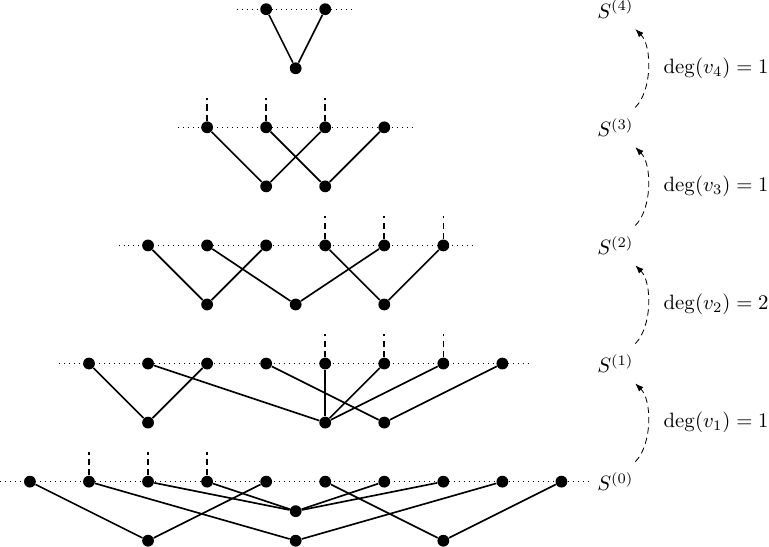}
		\caption{An illustration of the \emph{iterative cluster scheme} applied to the momentum graph in \cref{Fig4}, with the cluster decomposition $S = \Set{ \Set{1,5}, \Set{2,9}, \Set{3,4,7,8}, \Set{6,10}}$.
		At each level, we have denoted which three of the edges will be fused in the next iteration.
			}
	\label{Fig6}
\end{figure}

The cluster $S^{\myp{j}}$ is defined to be a clustering of the edges intersecting with the time slice $j$.
The construction $S^{\myp{j}}$ is done using an iterative procedure in which the interaction vertices are added to the graph, from the bottom to the top.
The additional vertex $v_j$ fuses the three edges in $\mathscr{E}_- \myp{v_j}$ into the new one in $\mathscr{E}_+ \myp{v_j}$.
The three edges of $\mathscr{E}_- \myp{v_j}$ belong to some clusters in the previous iteration $S^{\myp{j-1}}$.
The cluster $S^{\myp{j}}$ is constructed by joining all clusters of $S^{\myp{j-1}}$ containing an edge from $\mathscr{E}_- \myp{v_j}$, and replacing the three edges by the one in $\mathscr{E}_+ \myp{v_j}$.
The rest remain the same.  
	
If two of the three old edges belong to the same cluster in $S^{\myp{j-1}}$, the addition of the second one would create a loop in the construction of the tree. 
	
If all of the three old edges belong to the same cluster in $S^{\myp{j-1}}$, the addition of the second and third ones would create two loops in the construction of the tree. 
	
Therefore, this process determines the degree of the added interaction vertex.
If the new vertex fuses three separate clusters, it has degree 0.
If it fuses two separate clusters, it has degree 1.
Otherwise, it has degree 2. 
	
The iterative cluster scheme shows that $\abs{S^{\myp{j}}} = \abs{S^{\myp{j-1}}} - 2 + \deg \myp{v_j}$.
In addition, the structure of the clustering is conserved, meaning each $S^{\myp{j}}$ contains only even clusters. 

We then have the following two important lemmas concerning $n_0,n_1$ and $n_2$. 
\begin{lemma}
\label{lemma:numberofnonpairclusters}
	Consider a relevant graph, and define $r = n+n'+1-\abs{S}$.
	Then $n_{np} \leq r \leq n+n'$, where $n_{np}$ is the number of clusters that are not pairs, that is, $n_{np} \coloneq \abs{\Set{A\in S \mid \abs{A} > 2} }$.
	Moreover $r=0$ if and only if $S$ is a pairing.
	In addition, $n_2-n_0 = r$, $n_0 = \frac{1}{2} \myp{n+n'-r-n_1}$, $0 \leq n_1 \leq n+n'$, and $0 \leq n_0 \leq \floor{\frac{n+n'-r}{2}}$. 
\end{lemma}
\begin{lemma}
\label{lemma:n2jn0j}
	Consider a relevant graph, and an interaction vertex $v_j$, $1 \leq j \leq n+n'$, then
	\begin{equation*}
		n_2 \myp{j} \leq r + n_0 \myp{j}
		\mbox{  and  }
		n_0 \myp{j} \geq \frac{j- \myp{n_1+r}}{2},
	\end{equation*}
	where $r$ is defined in \cref{lemma:numberofnonpairclusters}.
\end{lemma}

\subsection{Fully paired graphs}
\label{sec:FullyPairedGraphs}
The goal of this section is to give a detailed study on fully paired graphs.
We have the following observations.
\begin{enumerate}[(i)]
	\item Since we have both $n_1=0$ and $r = 0$ for fully paired graphs, it follows immediately from the statements of \cref{lemma:numberofnonpairclusters} that $n_2 = n_0 = \frac{n+n'}{2}$.
	In particular, the number of interaction vertices in a fully paired graph must be even.
	
	\item From \cref{lemma:n2jn0j}, we have $n_0 \myp{j} \geq \frac{j}{2}$ and $n_2 \myp{j} \leq n_0 \myp{j}$, implying that necessarily $\deg \myp{v_1} = 0$ and $\deg \myp{v_{n+n'}} = 2$.
	Therefore, the degrees of the interaction vertices form a sequence $ \myp{0,2,0,2,\dotsc,0,2}$, or this behavior ends in two or more consecutive zeros in the middle.
	
	\item Let us now consider the first phase, for any relevant graph
	\begin{equation} 
	\label{eq:Varthetaj0}
		\mathrm{Re} \gamma \myp{0;J} 
		={} \sum_{l=1}^{2n+1} \sigma_{0,l} \omega^{\lambda} \myp{k_{0,l}} + \sum_{l=1}^{2n'+1} \sigma_{0,l}' \omega^{\lambda} \myp{k_{0,l}'}
		={} 0,
	\end{equation}
	due to the fact that the terms cancel each other pairwise because of the pairings of momenta and parities on the initial time slice.
	
	\item For pairing graphs, we will drop the dependence on $J$, for the sake of simplicity.
	In other words, we denote $\gamma \myp{m}$ instead of $\gamma \myp{m;J}$.
	
	\item Recall that the time slice $m<n+n'$ is called \emph{long} if $\mathrm{deg} \myp{v_{m+1}} = 0$.
	The long time slice $m$ is \emph{trivial} if $\mathrm{Re} \gamma \myp{m} = 0$.
	Note that the initial time slice in any fully paired graph is always long and trivial.
	
	\item We define the \emph{index of the last trivial long time slice} to be the last trivial long time slice in the initial sequence of trivial long time slices
	\begin{equation}
	\label{eq:indexofthelasttriviallong}
		m'_0 \coloneq \max \Set{0 \leq m \leq n+n' \mid \mathrm{Re} \gamma \myp{j} = 0 \mbox{ if } 0 \leq j \leq m, \mbox{ and the slice $j$ is long}}. 
	\end{equation}
	
	\item Let us consider a degree two interaction vertex $v$.
	The two free edges associated to $v$ can be denoted by $e_1$ and $e_2$, with $e_1<e_2$.
	We now denote the third (integrated) edge in $\mathscr{E}_- \myp{v}$ by $e_3$ and the unique edge in $\mathscr{E}_+ \myp{v}$ by $e_0$.
	The momenta associated to $e_1,e_2,e_3$ and $e_0$ are denoted buy $ k_1,k_2,k_3$ and $k_0$.
	It is then clear that $k_3=k_0-k_1-k_2$.
	The two edges $e_1$ and $e_2$ generate two loops.
	The directions of these loops start from $e_1$ or $e_2$ and end at $e_3$.
	Following the two directions, there is a unique vertex $v'$ where the two loops meet.
	We call it the \emph{X-vertex}, or the \emph{crossing-vertex} associated to the double loop of $v$.
	By construction, $v'$ must be attached to at least three distinct integrated edges, so it must be a degree zero vertex (since there are no degree one vertices in fully paired graphs).
	Any other vertices along the two loops are called $T$-vertices (or \emph{through-vertices}) of the double loop of $v$, and they are classified into three types:
	\begin{itemize}
		\item A $T_1$-vertex belongs to the path from  $v\to v'$ containing the edge $e_1$.
		
		\item A $T_2$-vertex belongs to the path from  $v\to v'$ containing the edge $e_2$.
		
		\item A $T_3$-vertex belongs to the path from  $v'\to v$.
	\end{itemize}
	The reason for this naming convention of the $T$-vertices is the following:
	If $w$ is a $T_j$-vertex, $j=1,2,3$, then there are exactly two momenta $k_e$, $e \in \mathscr{E} \myp{w}$ such that $k_e$ depends on $k_1$ and $k_2$, and the dependence is of the type $\pm k_j$.
	This is a consequence of \cref{lemma:orientedpathedcomposition} combined with \cref{lemma:freeedgesorder2}.
	
	\item For an interaction vertex $v_j$, $ 1 \leq j \leq n+n'$, we define
	\begin{equation}
	\label{eq:Thetaj}
		\Theta_j \coloneq{} \Theta \myp{k_{\mathscr{E}_- \myp{v_j}}, \sigma_+ \myp{v_j}}.
	\end{equation}
	Now, let us recall that $\mathrm{Re} \gamma \myp{0} = 0$, which means, for every time slice $j$ of a fully paired graph $0\le  j\le n+n'$, the following holds true
	\begin{equation}
	\label{eq:varthetaj}
		\mathrm{Re} \gamma \myp{j}
		={} \sum_{l=j+1}^{n+n'} \Theta_l 
		={} - \sum_{l=1}^{j} \Theta_l.
	\end{equation}
	
	\item  Let $m$ be an  arbitrary long time slice, $0 \leq m \leq n+n'-1$.
	This long time slice ends in the interaction vertex $v_{m+1}$ with $\deg \myp{v_{m+1}} = 0$. 
	Consider also an arbitrary vertex $v_{j_2}$ with $\mathrm{deg} \myp{v_{j_2}} = 2$, and denote the two free edges of $v_{j_2}$ by $k_1$ and $k_2$. 
	\begin{enumerate}[label=(\roman{enumi}.\arabic*)]
		\item If $\mathrm{Re} \gamma \myp{m}$ does not depend on $k_1$ and $k_2$, the time slice $m$ is said to be \emph{independent of the double loop of $v_{j_2}$}.
		
		\item If $\mathrm{Re} \gamma \myp{m}$ depends on $k_1$ and $k_2$ but $ \mathrm{Re} \gamma \myp{m} - \Theta_{j_2}$ does not, the time slice $m$ is said to be \emph{nested inside the double loop of $v_{j_2}$.}
		
		\item If both $\mathrm{Re} \gamma \myp{m}$ and $\mathrm{Re} \gamma \myp{m} - \Theta_{j_2} $ depend on $k_1$ and $k_2$, the time slice $m$ is said to \emph{propagate a crossing with the double loop of $v_{j_2}$}.
	\end{enumerate}
\end{enumerate}

We can now give the following definition which allows for complete classification of the relevant fully paired graphs.
Recall from \cref{sec:leadingmotives} that a graph is called \emph{leading} if it can be obtained by starting from the trivial loop and then iteratively attaching leading motives to the bottom of the graph.
\begin{definition}
	We consider a relevant fully paired graph and iterate through its degree two vertices starting from the bottom of the graph, and consider the double loops associated with the vertices.
	There will be the following cases.
	\begin{enumerate}[(i)]
		\item If all long time slices of the graph are independent of the double loop, we move to the next vertex in the list.
		
		\item If there are long time slices depending on the double loop and if all of them are nested inside the loop, then the double loop is called a \emph{nest} and the graph is called \emph{nested} (see Figure \ref{Fig8}). 
		\begin{figure}
			\centering
			\includegraphics{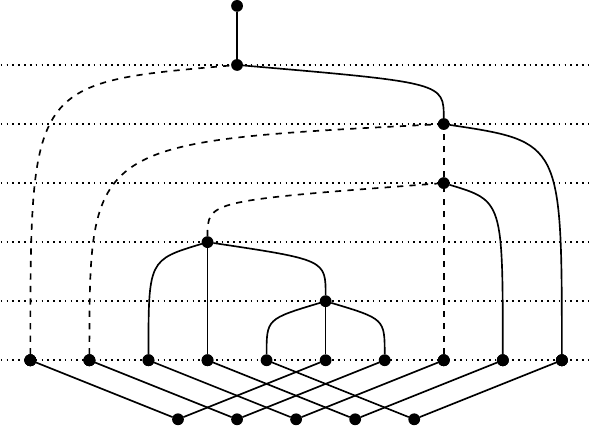}
			\caption{An example of a nested graph.
				This graph is constructed by inserting the leading motive L1 into the edge that connects the two vertices in the motive L5.
				It is easily checked that $\Theta_1 = -\Theta_4$ and $\Theta_2 = -\Theta_3$, and thus the time slice $1$ is nested in the double loop of $v_4$.}
			\label{Fig8}
		\end{figure}
		
		\item If there is a long time slice depending on the double loop but not nested inside it, the topmost of these time slices is called the \emph{last propagated crossing slice} and the graph is called \emph{crossing} (see Figure \ref{Fig7}).
		\begin{figure}
			\centering
			\includegraphics{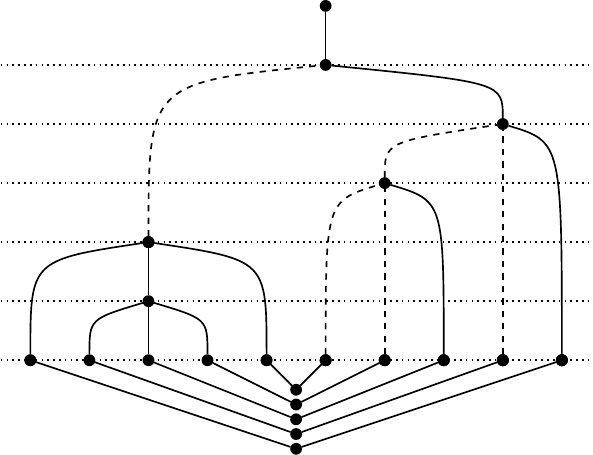}
			\caption{An example of a crossing graph.
				Denoting the free momenta of $v_3$ and $v_4$ by $k_1,k_2$ and $k_1',k_2'$, respectively, it is easily checked that $\gamma \myp{1} = - \Theta_1 = \omega^{\lambda} \myp{k_2'} -\omega^{\lambda} \myp{k_2} + \omega^{\lambda} \myp{k_1+k_2-k_1'} - \omega^{\lambda} \myp{k_1-k_1'-k_2'} \neq \Theta_3 $, meaning that the time slice $1$ propagates a crossing with the double loop of $v_3$.}
			\label{Fig7}
		\end{figure}
	\end{enumerate}
\end{definition}
It turns out (see \cref{lemma:recollisions} below), that the initial sequence of double loops satisfying $\mathrm{(i)}$ above are formed by iteration of leading motives (or immediate recollisions).
Thus, if a graph is not nested or crossing, it will be a leading graph (see \cref{prop:leadinggraphs} below).
A posteriori, another way of phrasing the definition above is thus that if a graph is not leading, we can consider the first degree two vertex from the bottom that does not correspond to an immediate recollision.
Then either there is a long time slice propagating a crossing with the double loop of this vertex (and the graph is crossing), or the double loop is a nest (and the graph is nested).
Note also that these definitions only concern the first degree two vertex in the graph that is not part of an immediate recollision; a nested graph may have crossings with degree two vertices above the initial nest, and a crossing graph may have nests above the initial crossing vertex.

Before providing the proof of this classification of the fully paired graphs, we need to know that the dispersion relation is sufficiently non-degenerate.
The lemma below is taken from \cite{LukkarinenSpohn:WNS:2011}.
The proof holds without change for the modified dispersion relation, since the only needed input is that a function satisfying \cref{ass:DR2} cannot be constant.
However, we recall the proof here since it will be instructive for our treatment of the crossing and nested graphs later.
\begin{lemma}
\label{lemma:Thetacannotindependfreemomenta}
	Suppose that the modified dispersion relation $\omega^{\lambda}$ satisfies \cref{ass:DR2}, and let $v_{j_0}$ be an interaction vertex in a relevant graph.
	Suppose that $f \in \mathscr{F}_e$ for some $e \in \mathscr{E} \myp{v_{j_0}}$, i.e., $k_e$ depends on $k_f$.
	The following hold true.
	\begin{enumerate}[(i)]
		\item $\nabla_{k_f} \Theta_{j_0} \neq 0$.
		
		\item $\Theta_{j_0}$ cannot be independent of all free momenta.
	\end{enumerate}
\end{lemma}
\begin{proof}
	Suppose on the contrary that $\Theta_{j_0}$ is a constant in $k_f$.
	We denote the momenta associated to $v_{j_0}$ by $k_0, k_1, k_2, k_3$ and the edges by $e_0, e_1, e_2, e_3$, respectively.
	For the sake of simplicity, we also denote $\mathscr{F}_i = \mathscr{F}_{e_i}$.
	By \cref{lemma:orientedpathedcomposition} and \cref{corollary:decompositionfreeedgetofreeedge}, each $k_i$ depends on some free momenta and $\mathscr{F}_i \neq \emptyset$. 

	Recall that any given edge can depend on the free edge $f$ only if it is contained in the unique loop defined by $f$ in the spanning tree of the graph.
	This implies that $f$ can appear in zero or exactly two of the four sets $\mathscr{F}_0, \mathscr{F}_1, \mathscr{F}_2, \mathscr{F}_3$.
	Thus, there are unique $j,l$ such that $f \in \mathscr{F}_j \cap \mathscr{F}_l$, and we denote the other two indices by $j',l'$.
	It follows that $k_j = \pm \myp{k_f+u}$ and $k_l = \pm \myp{k_f+u'}$, where $u$, $u'$ are linear combinations of the remaining free momenta, by \cref{lemma:orientedpathedcomposition}.  
 
	Since $\nabla_{k_f} \Theta_{j_0} = 0$, it follows that $\nabla \omega^{\lambda} \myp{k_f} \pm \nabla \omega^{\lambda} \myp{k_f+u'-u} = 0$ for all $k_f \in \T^d$ and any value of $u'-u$.
	If $u'-u$ is not identically zero, then it depends on some other free momentum which is not $k_f$, so by differentiating with respect to this free momentum, it follows that the Hessian of $\omega^{\lambda}$ is uniformly zero.
	Then $\omega^{\lambda}$ has to be a linear map which is also periodic, implying that $\omega^{\lambda}$ is a constant.
	This is however not possible since $\omega^{\lambda}$ satisfies \cref{ass:DR2}, so it follows that $u'=u$.
	Therefore, $k_l=\pm k_j$. 
	Recalling that $k_0=k_1+k_2+k_3$, we have two cases:
	
	\emph{Case 1: $k_l=k_j$.}
	If either $l$ or $j$ is $0$, we  then have $k_{l'}+k_{j'}=0$, where $l',j' \in \Set{1,2,3}$ are the remaining indices, so in this case, the graph would be irrelevant by \cref{rem:irrelevant_graphs}.
	If both $l$ and $j$ are different from $0$, we have for instance $k_0 = k_{j'}+k_j+k_l = k_{j'} \pm 2 \myp{k_f + u}$, where $k_{j'}$ and $u$ are both independent of $k_f$.
	This contradicts the form of the presentation \eqref{eq:decompositionfreeedgetofreeedge}.

	\emph{Case 2: $k_l+k_j=0$.}
	If one of $l,j$ is $0$, say $j=0$ then $k_{l'}+k_{j'}=2 k_j = \pm 2 \myp{k_f + u} $.
	This is not possible since $k_{l'}$ and $k_{j'}$ are both independent of $k_f$.
	Otherwise, if $l$ and $j$ are both non-zero, then $k_l + k_j = 0$ implies by \cref{rem:irrelevant_graphs} that the graph is irrelevant.
	
	Thus our initial assumption cannot hold, and we conclude that $\nabla_{k_f} \Theta_{j_0} \neq 0$.
	The second statement is an obvious corollary of this.
 \end{proof}
 
\begin{lemma}
\label{lemma:recollisions}
	Consider a relevant fully paired graph and suppose that all of its long time slices are independent of double loops of the first $M$ degree two vertices.
	Then all the double loops of these $M$ vertices are immediate recollisions.
	That is, there is a graph from which the full graph can be obtained by iteratively adding $M$ leading motives.
\end{lemma}
\begin{proof}
	We only need to show that the double loop of the first degree two vertex, $v_{i_0}$, is an immediate recollision, and then the statement follows by induction in $M$.
	Denote by $k_1,k_2,k_3$, and $k_0 = k_1+k_2+k_3$ the momenta of the edges connected to $v_{i_0}$, where $k_1$ and $k_2$ are free.
	We know that $v_{i_0}$ cannot be the first vertex of the graph since then we would have $\deg \myp{v_{i_0}}= 0$, so we may consider the time slice $i_0-2$, which is the topmost long time slice below $v_{i_0}$.
	We will show that $v_{i_0}$ and $v_{i_0-1}$ along with their edges constitute a leading motive by showing the following:
	\begin{enumerate}[(i)]
		\item The double loop of $v_{i_0}$ passes through only two interaction vertices, $v_{i_0-1}$ and $v_{i_0}$.
				
		\item The "incoming" momentum is preserved, i.e., the momenta of the edges in $\mathscr{E} \myp{v_{i_0-1}}$ correspond exactly to the momenta of the edges in $\mathscr{E} \myp{v_{i_0}}$, up to a sign.
		
		\item The double loop preserves the phase, $\Theta_{i_0-1} = -\Theta_{i_0}$.
	\end{enumerate}
	Then, since $\mathrm{Re} \gamma \myp{i_0-2}$ is independent of the double loop of $v_{i_0}$, we immediately get from \eqref{eq:varthetaj} that
	\begin{equation}
	\label{eq:recollisions1}
		\Theta_{i_0-1} + \Theta_{i_0}
		= \mathrm{Re} \gamma \myp{i_0-2} - \sum_{j=i_0+1}^N \Theta_j
	\end{equation}
	is also independent of the double loop, since each term $\Theta_j$ is independent for $j > i_0$, by construction of the spanning tree of the graph.
	It follows that $v_{i_0-1}$ must be the $X$-vertex of the double loop of $v_{i_0}$.
	Indeed, assuming for contradiction that this is not the case, then $v_{i_0-1}$ is a $T_n$-vertex, meaning that
	\begin{equation*}
		\Theta_{i_0-1} ={} \pm \omega^{\lambda} \myp{k_n + u} \pm \omega^{\lambda} \myp{k_n+u'} + \alpha',
	\end{equation*}
	for some $u,u',\alpha'$ that do not depend on the free momenta of $v_{i_0}$.
	If for instance $n = 1$, $\Theta_{i_0-1}$ is independent of $k_2$, so we differentiate with respect to $k_2$ to obtain
	\begin{equation*}
		0 ={} \nabla_{k_2} \myp{\Theta_{i_0-1} + \Theta_{i_0}}
		={} \pm \nabla \omega^{\lambda} \myp{k_3} \pm \nabla \omega^{\lambda} \myp{k_2}
		={} \pm \nabla \omega^{\lambda} \myp{k_0-k_1-k_2} \pm \nabla \omega^{\lambda} \myp{k_2},
	\end{equation*}
	which holds for any $k_1,k_2 \in \T^d$, implying that $\omega^{\lambda}$ has to be constant, which is incompatible with \cref{ass:DR2}.
	It follows similarly that $v_{i_0-1}$ cannot be a $T_2$-vertex.
	If $v_{i_0-1}$ is a $T_3$-vertex, we get by differentiating $\Theta_{i_0-1} + \Theta_{i_0}$ with respect to $k_1$,
	\begin{equation*}
		0 ={} \pm \nabla \omega^{\lambda} \myp{k_0-k_1-k_2+u} \pm \nabla \omega^{\lambda} \myp{k_0-k_1-k_2+u'} \pm \nabla \omega^{\lambda} \myp{k_0-k_1-k_2} \pm \nabla \omega^{\lambda} \myp{k_1},
	\end{equation*}
	in which case $\omega^{\lambda}$ would also have to be constant.
	The only remaining possibility is that $v_{i_0-1}$ is the $X$-vertex of the double loop of $v_{i_0}$.
	
	Next we note that since $\mathrm{Re} \gamma \myp{j}$ is independent of $k_1,k_2$ for all $j < i_0-1$ by assumption, we have for all remaining time slices $j < i_0-2$ that
	\begin{equation*}
		\Theta_j 
		= \mathrm{Re} \gamma \myp{j-1} - \mathrm{Re} \gamma \myp{j}
	\end{equation*}
	is independent of the double loop of $v_{i_0}$.
	Supposing for the moment that $v_j$ is a $T$-vertex of the double loop, then exactly two edges of $v_j$ depend on at least one of the free momenta of $v_{i_0}$, say $k_1$.
	It then follows from \cref{lemma:Thetacannotindependfreemomenta} that $\nabla_{k_1} \Theta_j \neq 0$, contradicting the fact that $\Theta_j$ is independent of the double loop.
	Thus we conclude that $v_j$ is not a $T$-vertex of the double loop of $v_{i_0}$ implying the first statement (i) above, since it also cannot be the $X$-vertex.
	
	Let us consider now the loops created by the free edges of $v_{i_0}$ in the spanning tree.
	Each of the loops start at $v_{i_0}$, goes down along a free edge, and passes through $v_{i_0-1}$ before returning to $v_{i_0}$ along the integrated edge in $\mathscr{E}_- \myp{v_{i_0}}$, possibly going through some cluster vertices along the way.
	There are two cases:
	
	\emph{Case 1}:
	If neither of the edges in $\mathscr{E}_- \myp{v_{i_0}}$ is connected directly to $v_{i_0-1}$, then they must each be paired to the edges of $\mathscr{E}_- \myp{v_{i_0-1}}$ through a cluster vertex at the bottom of the graph.
	In this case, the momentum $k_0'$ of the unique edge in $\mathscr{E}_+ \myp{v_{i_0-1}}$ satisfies $k_0' = -k_0 $, so we conclude that each of the momenta of the edges in $\mathscr{E} \myp{v_{i_0-1}}$ is exactly equal to (minus) a momentum of a unique edge in $\mathscr{E} \myp{v_{i_0}}$.
	Since the pairings have to respect the parities of the edges, we find that $\Theta_{i_0-1} = - \Theta_{i_0}$ is the only possible option, and the double loop of $v_{i_0}$ thus corresponds to a gain motive (recall the leading motives in \cref{fig:LeadingMotives}).
	We may cut the motive out of the graph by cutting each of the edges corresponding to $k_0'$ and $k_0$, and afterwards extending these edges to the bottom of the graph and connecting them via a pairing.
	
	\emph{Case 2}: If one of the edges in $\mathscr{E}_- \myp{v_{i_0}}$ is equal to the unique edge in $\mathscr{E}_+ \myp{v_{i_0-1}}$, then the two remaining edges of $\mathscr{E}_- \myp{v_{i_0}}$ must be paired to edges in $\mathscr{E}_- \myp{v_{i_0-1}}$ through clusters.
	Thus we conclude that we can label the edges $ \mathscr{E}_- \myp{v_{i_0-1}} = \Set{e_1',e_2',e_3'}$ and $ \mathscr{E}_- \myp{v_{j_0}} = \Set{e_1,e_2,e_3}$ such that the corresponding momenta satisfy $k_{e_1} = -k_{e_1'}'$, $k_{e_2} = -k_{e_2'}'$,
	and $k_{e_3} = k_{e_0'}'$, where $e_0'$ is the unique edge in $\mathscr{E}_+ \myp{v_{i_0-1}}$.
	It follows that the final momentum satisfies $k_{e_3'}' = k_{e_0'}' - k_{e_1'}' - k_{e_2'}' = k_{e_1} + k_{e_2} + k_{e_3} = k_{e_0} = k_0$, and that the parities of the edges $e_3'$ and $e_0$ have to be equal.
	Consequently, we must have $\Theta_{i_0-1} = - \Theta_{i_0}$, and the double loop of $v_{i_0}$ corresponds to a loss motive.
	We can cut the motive out of the graph by cutting the edges $e_0$ and $e_3'$, and then repairing the graph by gluing the two edges together.
	
	We have thus proved the lemma for $M=1$.
	Assume now that the lemma is true for $M-1$ for all relevant fully paired graphs, and consider a graph where all long time slices are independent of the double loops of the first $M$ degree two vertices.
	By what we have just shown, the double loop of the first degree two vertex is an immediate recollision, so we can cut it out of the graph as described above.
	Since immediate recollisions preserve the phase, the resulting graph then satisfies that all long time slices are independent of the double loops of the first $M-1$ degree two vertices, which concludes the proof of the lemma by induction.
\end{proof}
\begin{proposition}
\label{prop:leadinggraphs}
	Consider a fully paired graph with non-zero associated amplitude and suppose that it is not nested or crossing.
	Then all of its long time slices are trivial, and the graph is leading.
	In general, all relevant fully paired graphs, in which all long time slices are trivial, are leading.
\end{proposition} 
\begin{proof}
	From the assumption, the graph is relevant and all of its long time slices are independent of all of the double loops.
	\cref{lemma:recollisions} implies that the graph is leading, and all long time slices are trivial since immediate recollisions preserve the phase, which is initially $\mathrm{Re} \gamma \myp{0} = 0$.
	
	Conversely, if all long time slices are trivial, they satisfy $\mathrm{Re} \gamma \myp{m} = 0$ independently of all free momenta, and thus by \cref{lemma:recollisions} the graph is leading.
\end{proof}

\subsubsection*{Nested and crossing graphs} 
Before concluding this section, we state and prove the following structural lemmas which will be useful later for analysing the phase factors of nested and crossing graphs.
The assertions of the lemmas are already contained in the proofs of the amplitude bounds for nested and crossing graphs in the setting of the weakly non-linear Schr\"odinger equation \cite{LukkarinenSpohn:WNS:2011}, but for the sake of exposition and later reference, we state the results separately here and flesh out some details of the proofs.
We provide a slightly different exposition than in \cite{LukkarinenSpohn:WNS:2011} that does not rely on the iterative cluster scheme.
\begin{lemma}[Nested graphs]
	\label{lem:NestedGraphs}
	Consider a relevant nested graph and let $i_2$ be the index of the first interaction vertex (from the bottom) which is not an immediate recollision.
	Further, let 
	\begin{equation*}
		N_2 \coloneq \Set{m \mid \text{The time slice $m$ is long and nested inside the double loop of $v_{i_2}$} },
	\end{equation*}
	and denote $j_0 \coloneq \min N_2$.
	Then the following hold true:
	\begin{enumerate}[(i)]
		\item $j_0$ is not the initial time slice, and the vertex $v_{j_0}$ siting at the bottom of the time slice $j_0$ is the $X$-vertex of the double loop of $v_{i_2}$.
		In particular, $j_0-1$ and $j_0$ are both long.
		
		\item The phase factor $\Theta_{j_0} + \Theta_{i_2}$ is independent of the double loops of all degree two vertices up to $v_{i_2}$.
		
		\item For any other degree zero vertex $v_j$ below $v_{i_2}$ which is \emph{not} the degree zero vertex of an immediate recollision, $\Theta_j$ is independent of the double loops of all degree two vertices up to $v_{i_2}$.
		In particular, $v_j$ cannot be a $T$-vertex of the double loop of $v_{i_2}$.
		
		\item The double loop of $v_{i_2}$ corresponds to that of a leading motive, but with a \emph{delayed} recollision.
		In particular, $\Theta_{j_0} = - \Theta_{i_2}$.
	\end{enumerate}
\end{lemma}
\begin{proof}
	Recall first that by definition of nested graphs, the vertex $v_{i_2}$ exists, and there is at least one long time slice that is nested inside its double loop, and any long time slice that depends on the double loop is nested inside it.
	If we denote by $k_1,k_2$ the free momenta of the double loop of $v_{i_2}$, then $\mathrm{Re} \gamma \myp{j_0}$ depends on $k_1,k_2$, but $\mathrm{Re} \myp{j_0} - \Theta_{i_2}$ does not.
	Further, we have $j_0 < i_2-1 $ because $i_2-1$ by definition is not a long time slice, and no time slice above $i_2-1$ can depend on the double loop of $v_{i_2}$.
	
	Note now that the initial time slice does not depend on $k_1,k_2$ since $\mathrm{Re} \gamma \myp{0} = 0$, so $j_0 > 0$ and hence $j_0-1$ is the index of a time slice.
	Assuming for contradiction that $j_0-1$ is a short time slice, then $\deg v_{j_0} = 2$, so $v_{j_0}$ cannot be the bottom interaction vertex of the graph.
	Hence $j_0-1 > 0$, and it must belong to an immediate recollision by definition of $v_{i_2}$, which implies that $j_0-2$ must be a long time slice, and $\mathrm{Re} \gamma \myp{j_0} = \mathrm{Re} \gamma \myp{j_0-2}$.
	It follows that $j_0-2$ is also nested inside the double loop of $v_{i_2}$ which is in contradiction with the choice of $j_0$.
	Thus we conclude that $j_0-1$ is a long time slice that is not nested inside the double loop of $v_{i_2}$, and therefore it must be independent of the double loop (simply because the double loop does not have any crossings).
	
	Note now that since $\mathrm{Re} \gamma \myp{j_0}$ is nested inside the double loop of $v_{i_2}$, $\mathrm{Re} \gamma \myp{j_0} - \mathrm{Re} \gamma \myp{j_0-1} = -\Theta_{j_0}$ must depend on the double loop, and
	\begin{equation*}
		\mathrm{Re} \gamma \myp{j_0} - \Theta_{i_2} - \mathrm{Re} \gamma \myp{j_0-1} 
		= - \myp{\Theta_{j_0} + \Theta_{i_2}}
	\end{equation*}
	must be independent of the double loop.
	By the exact same argument as the one starting from \eqref{eq:recollisions1} in the proof of \cref{lemma:recollisions}, we conclude that $v_{j_0}$ must be the $X$-vertex of the double loop of $v_{i_2}$, which finishes the proof of \emph{(i)}.
	
	Consider the set of degree zero vertices below $v_{i_2}$ that do not belong to an immediate recollision.
	Since any degree two vertex below $v_{i_2}$ by definition belongs to an immediate recollision, the set of indices of these vertices is
	\begin{equation}
	\label{eq:zero_vertex_prime}
		B_0' \coloneq \Set{1\leq j < i_2 \mid \deg v_j = 0, \text{ and } \deg v_{j+1} = 0 \text{ or } j+1=i_2}.
	\end{equation}
	Since both $j_0-1$ and $j_0$ are long time slices, we have $j_0 \in B_0'$.
	Fix now a degree two vertex $v_{j_2}$ with $j_2< i_2$, and a $j \in B_0'$.
	Then either $j=i_2-1$, in which case $\Theta_j$ is independent of the double loop of $v_{j_2}$ (because $v_j$ is above $v_{j_2}$ in the graph), or $\deg v_{j+1} = 0$.
	For the latter case, we note that every long time slice is independent of the double loop of $v_{j_2}$, by definition of $i_2$.
	In particular, $\mathrm{Re} \gamma \myp{j} - \mathrm{Re} \gamma \myp{j-1} = -\Theta_{j}$ is independent of this double loop.
	Taking $j = j_0$ concludes the proof of \emph{(ii)}.
	
	To finish the proof of \emph{(iii)}, we still need to argue that $\Theta_j$ is independent of the double loop of $v_{i_2}$ for $j \in B_0' \setminus \Set{j_0}$.
	By definition of $j_0$, this is automatically true for any $j \in B_0'$ with $j < j_0$.
	Writing $B_0' \cap \Set{i \geq j_0} = \Set{j_0,j_1, \dotsc, j_m}$, we assume for induction for some $1 \leq k \leq m$ that $\Theta_{j_i}$ is independent of the double loop of $v_{i_2}$ for all $1 \leq i \leq k-1 $.
	Since for each such $i$, the vertices between $v_{j_{i-1}}$ and $v_{j_i}$ all belong to immediate recollisions, we can remove the phase factors of these recollisions and obtain
	\begin{equation}
		\label{eq:nested_degree_zero_independence}
		\Theta_{j_k} + \mathrm{Re} \gamma \myp{j_k} - \Theta_{i_2}
		={} \mathrm{Re} \gamma \myp{j_0} - \Theta_{i_2} - \sum_{i=1}^{k-1} \Theta_{j_i},
	\end{equation}
	where the right hand side is independent of the double loop of $v_{i_2}$.
	It follows that the two terms $\Theta_{j_k}$ and $\mathrm{Re} \gamma \myp{j_k} - \Theta_{i_2}$ are either both dependent or both independent of the double loop.
	Assuming for contradiction that they depend on the double loop of $v_{i_2}$, we infer that $\mathrm{Re} \gamma \myp{j_k}$ must be independent, since otherwise the time slice $j_k$ would propagate a crossing, which cannot happen because the double loop of $v_{i_2}$ is a nest.
	From \eqref{eq:nested_degree_zero_independence} we may then conclude that $\Theta_{j_k} - \Theta_{i_2}$ is independent of the double loop of $v_{i_2}$.
	However, since $v_{j_k}$ must then be a $T$-vertex of the double loop, this is not possible by the same argument we used to show that $v_{j_0}$ is the $X$-vertex.
	Hence we conclude that $\Theta_{j_k}$ is independent of the double loop, and further by \cref{lemma:Thetacannotindependfreemomenta} that $v_{j_k}$ cannot be a $T$-vertex, finishing the proof of \emph{(iii)}.
	
	Finally, to see that $\emph{(iv)}$ holds, we note now that by \emph{(iii)}, the double loop of $v_{i_2}$ only goes through $v_{j_0}$ and immediate recollisions.
	Since these preserve the phase and the incoming momentum, the same argumentation as in the proof of \cref{lemma:recollisions} shows that the momenta of $v_{j_0}$ and $v_{i_2}$ are pairwise equal, and that $\Theta_{j_0} = - \Theta_{i_2}$.
	Hence, double loop of $v_{i_2}$ corresponds to that of a leading motive, up to the presence of the immediate recollisions.
	(Alternatively, one may cut all the immediate recollisions below $v_{i_2}$ out of the graph to obtain a second graph, which by \emph{(iii)} satisfies that all its long time slices are independent of the double loop of $v_{i_2}$).
\end{proof}
\begin{lemma}[Crossing graphs]
\label{lem:CrossingGraphs}
	Consider a relevant crossing graph and let $i_2$ be the index of the first interaction vertex (from the bottom) which is not an immediate recollision.
	Further, let $i_0-1$ be the index of the last propagated crossing slice.
	The following properties hold:
	\begin{enumerate}[(i)]
		\item There exist a $p \in \Set{0,1}$ and an $\alpha_1$ which is independent of the double loop of any $v_j$, $j \leq i_2$, such that
		\begin{equation}
		\label{eq:CrossingPhase1}
			\mathrm{Re} \gamma \myp{i_0-1}
			= p \Theta_{i_2} + \Theta_{i_0} + \alpha_1 + \mathrm{Re} \gamma \myp{i_2}.
		\end{equation}
		Further, $v_{i_0}$ is a vertex in the double loop of $v_{i_2}$ (and thus either a $T$- or the $X$-vertex), and $\Theta_{i_0}$ is independent of any free momenta ending before $v_{i_2}$.
		
		\item If $v_{i_0}$ is a $T_j$-vertex, then
		\begin{equation}
		\label{eq:CrossingPhase2}
			\Theta_{i_0} 
			= \pm \omega^{\lambda} \myp{k_j+u} \pm \omega^{\lambda} \myp{k_j+u'} + \alpha'
		\end{equation}
		for some choice of signs, where $\alpha',u',u$ are independent of the double loop of $v_{i_2}$.
		Further, $u'-u$ is not identically zero; it depends on some free momenta, and it cannot depend solely on the free momentum of the top fusion vertex.
		
		\item If $v_{i_0}$ is the $X$-vertex, then 
		\begin{equation}
		\label{eq:CrossingPhase3}
			\Theta_{i_0} 
			= \alpha' + \sum_{j=1}^3 \pm \omega^{\lambda} \myp{k_j+u_j}
		\end{equation}
		for some choice of signs, where $\alpha', u_1,u_2,u_3$ are independent of the double loop of $v_{i_2}$.
		In this case, $p = 1$ in \eqref{eq:CrossingPhase1}, and the $u_j$ cannot all be identically zero; at least one of them depends on some free momenta ending above $v_{i_2}$, and it cannot depend solely on the free momentum of the top fusion vertex.
	\end{enumerate}
\end{lemma}
\begin{proof}
	First, we rewrite
	\begin{equation}
	\label{eq:crossingraph1}
		\mathrm{Re} \gamma \myp{i_0-1} 
		= \Theta_{i_2} + \Theta_{i_0} + a_1 + \mathrm{Re} \gamma \myp{i_2},
	\end{equation}
	where
	\begin{equation}
	\label{eq:crossingraph2}
		a_1 = \sum_{j=i_0+1}^{i_2-1} \Theta_j. 
	\end{equation}
	By construction, $\mathrm{Re} \gamma \myp{i_2}$ does not depend on the edges that are associated to a lower ranked vertex, and hence it does not depend on $k_1,k_2$, nor on any double loop of $v_j$ with $j \leq i_2$.
	Our goal now is to prove that we can choose $p \in \Set{0,1}$ such that
	\begin{equation*}
		\alpha_1 \coloneq a_1 + \myp{1-p} \Theta_{i_2},
	\end{equation*}
	is also independent of any such double loop.
	It is then obvious that \eqref{eq:CrossingPhase1} holds.
	
	Now, by \cref{lemma:recollisions} and definition of $i_2$, for any $j < i_2$ with $\deg \myp{v_j} = 2$, the double loop of $v_j$ is determined by a leading motive, and $\Theta_{j-1} = - \Theta_j$.
	Now, if $j \geq i_0+2$, the terms cancel each other in the sum \eqref{eq:crossingraph2}.
	Assuming for the moment that $ \deg \myp{v_{i_0+1}} = 2$, then $\Theta_{i_0} = -\Theta_{i_0+1}$ and $\mathrm{Re}\gamma \myp{i_0-1} = \mathrm{Re} \gamma \myp{i_0+1}$.
	This implies that $i_0+1$ must be a short time slice, since otherwise also $i_0+1$ propagates a crossing with $v_{i_2}$, which is not allowed by definition of $i_0-1$.
	This means that $\deg \myp{v_{i_0+1}} = \deg \myp{v_{i_0+2}} = 2 $, so $v_{i_0+2}$ cannot be an immediate recollision, and we conclude that necessarily $i_0+2 = i_2$.
	However, in this case $\mathrm{Re} \gamma \myp{i_0-1} = \Theta_{i_2} + \mathrm{Re} \gamma \myp{i_2}$, and thus $i_0-1$ does not propagate a crossing.
	Hence we conclude that we always have $\deg \myp{v_{i_0+1}} = 0$, which implies that the sum \eqref{eq:crossingraph2} cannot contain the phase of just one of the vertices of any given immediate recollision.
	
	It follows that if we consider the set of degree zero vertices between $v_{i_0}$ and $v_{i_2}$ that do not belong to an immediate recollision,
	\begin{equation*}
		I' \coloneq \Set{i_0+1 \leq j \leq i_2-1 \mid \deg \myp{v_j}=0, \text{ and } \deg \myp{v_{j+1}} = 0 \mbox{ or } j +1 = i_2 }, 
	\end{equation*}
	then $a_1 = \sum_{j \in I'} \Theta_j$.
	As a result, $a_1$ is independent of all of the double loops of $v_j$ with $j < i_2$.
	If $I'=\emptyset$, then $a_1 = 0$ and the claim holds for $p=1$.
	If $I \neq \emptyset$, we denote $j' = \min I'$.
	Clearly, $\mathrm{deg} \myp{v_{j'}} = 0$ and the time slice $j'-1 > i_0-1$ is long, with $\mathrm{Re} \gamma \myp{j'-1} = \Theta_{i_2} + a_1 + \mathrm{Re} \gamma \myp{i_2}$.
	In view of the fact that this time slice cannot propagate a crossing, either $a_1$ or $a_1 + \Theta_{i_2}$ is independent of the double loop of $v_{i_2}$.
	For the first case, we take $p=1$ and $\alpha_1 = a_1$.
	For the second case, we take $p=0$ and $\alpha_1 = a_1 + \Theta_{i_2}$.
	This finishes the proof of \eqref{eq:CrossingPhase1}.
	
	To conclude the proof of \emph{(i)}, note now that if $v_{i_0}$ is not a vertex in the double loop of $v_{i_2}$, then $\Theta_{i_0}$ is independent of the double loop.
	By \eqref{eq:CrossingPhase1}, this would imply that either $\mathrm{Re} \gamma \myp{i_0-1}$ or $ \mathrm{Re} \gamma \myp{i_0-1} - \Theta_{i_2}$ is independent of the double loop, contradicting the fact that $i_0-1$ propagates a crossing.
	Finally, recalling that all long time slices are independent of the double loops of any $v_j$ with $j < i_2$ (since these are immediate recollisions), it follows that also $ \Theta_{i_0} = \mathrm{Re} \gamma \myp{i_0-1} - \mathrm{Re} \gamma \myp{i_0} $ is independent of these loops, since we have shown that $i_0$ is a long time slice.
	
	Now, for points \emph{(ii)} and \emph{(iii)}, we note that the equations \eqref{eq:CrossingPhase2} and \eqref{eq:CrossingPhase3} always hold for any $T_j$- and $X$-vertices, respectively.
	We only need to show the statements about the additional momenta appearing in the formulas.
	
	If $v_{i_0}$ is a $T_j$-vertex, the $j$-part of the double loop of $v_{i_2}$ goes through $v_{i_0}$ via two edges $e,e'$.
	If both $e,e'$ are in $\mathscr{E}_- \myp{v_{i_0}}$, then $k_e = \sigma \myp{k_j+u}$ and $k_{e'} = -\sigma \myp{k_j+u'}$ for some choice of $\sigma \in \Set{\pm 1}$, implying that $u'-u = - \sigma \myp{k_e+k_{e'}}$.
	Otherwise, the loop uses $\tilde{e} \in \mathscr{E}_- \myp{v_{i_0}}$ and $\tilde{e}' \in \mathscr{E}_+ \myp{v_{i_0}}$, $k_{\tilde{e}} = \sigma \myp{k_j+u}$ and $k_{\tilde{e}} = \sigma \myp{k_j+u'}$, leading to $u'-u = \sigma \myp{k_{\tilde{e}'} - k_{\tilde{e}}} = \sigma \myp{k_e+k_{e'}}$ with $e,e'$ being the remaining two edges in $\mathscr{E}_- \myp{v_{i_0}}$.
	Now, if $u'-u $ is independent of all free momenta, then \cref{lemma:dependenceoftwoedges} implies that $k_e+k_{e'} = 0$, meaning that the graph is irrelevant.
	Additionally, \cref{lemma:odd_dependence} implies that $k_e+k_{e'}$ and hence $u'-u$ depend on an even number of free momenta, and at least one of these cannot be the free momentum of the top fusion vertex of the graph.
	
	Assume finally that $v_{i_0}$ is the $X$-vertex of the double loop of $v_{i_2}$.
	We show first by induction that all of the  $\Theta_j$ with $j \in I'$ are independent of the double loop of $v_{i_2}$.
	Write $I' = \Set{j_1,\dotsc,j_m}$ and assume for some $1\leq k \leq m$ that $\Theta_{k+1}, \dotsc, \Theta_{j_m}$ are all independent of the double loop, and note that
	\begin{equation*}
		\mathrm{Re} \gamma \myp{j_k-1}
		= \Theta_{i_2} + \Theta_{j_k} + \mathrm{Re} \gamma \myp{i_2} + \sum_{i=k+1}^m \Theta_{j_i}.
	\end{equation*}
	Assuming for contradiction that $\Theta_{j_k}$ depends on the double loop, then $v_{j_k}$ must be a $T$-vertex, and $\mathrm{Re} \gamma \myp{j_k-1}$ must be independent of the double loop of $v_{i_2}$ since the long time slice $j_k-1$ cannot propagate a crossing.
	If $v_{j_k}$ is a $T_1$- or a $T_2$- vertex, then differentiating the above equation with respect to the free momentum of $v_{i_2}$ that $\Theta_{j_k}$ does not depend on, it follows that the gradient of $\omega^{\lambda}$ is constant, which is incompatible with \cref{ass:DR2}.
	On the other hand, if $v_{j_k}$ is a $T_3$-vertex, then $\Theta_{j_k}$ is of the form \eqref{eq:CrossingPhase2}, in which case we have
	\begin{equation*}
		\nabla_{k_1} \Theta_{i_2} 
		= -\nabla_{k_1} \Theta_{j_k} 
		= -\nabla_{k_2} \Theta_{j_k}
		= \nabla_{k_2} \Theta_{i_2}.
	\end{equation*}
	This in turn implies that $\nabla\omega^{\lambda} \myp{k_1} = \pm \nabla \omega^{\lambda} \myp{k_2} $ for all $k_1,k_2 \in \T^d$, so $\nabla\omega^{\lambda}$ is again a constant.
	We conclude that $\Theta_{j}$ cannot depend on $k_1,k_2$ for any $j \in I'$, from which it follows that $p=1$ in \eqref{eq:CrossingPhase1}.
	
	Suppose now for contradiction that $u_1 = u_2 = u_3 = 0$.
	We will show that the double loop of $v_{i_2}$ is then actually a nest, and the graph is nested.
	In this case, the last term in \eqref{eq:CrossingPhase3} must be equal to $ \alpha' = \pm \omega^{\lambda} \myp{k_0}$, and one of the edges $e_0' \in \mathscr{E} \myp{v_{i_0}}$ has corresponding momentum $k_{e_0'} = \pm k_0$.
	Denote by $e_0 \in \mathscr{E}_+ \myp{v_{i_2}}$ the edge with momentum $k_{e_0} = k_0$.
	Since the graph is relevant, it is then apparent from \cref{lemma:dependenceoftwoedges0} and \cref{lemma:removalofedges} that $k_{e_0}$ and $k_{e_0'}$ depend on the same set of free edges, and that removal of $e_0$ and $e_0'$ from the graph splits the graph into two connected components.
	Note that one of these components contains the entire double loop of $v_{i_2}$, and that this component has the structure of a Feynman diagram.
	In particular, the edges $e_0$ and $e_0'$ must have the same parity, since the sum of parities in each time slice remains constant throughout any diagram.
	Thus, we may cut the connected component containing the double loop of $v_{i_2}$ out of the graph, and then attach it to the trivial loop to obtain another momentum graph.
	Since this operation preserves the graph structure (including vertex degrees and the way each edge depends on free momenta),  the resulting graph is fully paired, implying $n_2 = n_0$, and all degree two vertices, except for $v_{i_2}$, are immediate recollisions.
	This implies that the connected component of the double loop of $v_{i_2}$ contains only a single degree zero vertex that does not correspond to an immediate recollision, which must then be $v_{i_0}$.
	Denoting the indices of this set of vertices by $B_0'$, as in \eqref{eq:zero_vertex_prime}, we have shown that $\Theta_j$ is independent of the double loop of $v_{i_2}$ for each $j \in B_0' \setminus \Set{i_0}$.
	Finally, considering $\mathrm{Re} \gamma \myp{i_2}$ and removing all immediate recollisions from this phase, we conclude that
	\begin{equation*}
		\Theta_{i_0} + \Theta_{i_2}
		= - \mathrm{Re} \gamma \myp{i_2} - \sum_{j \in B_0' \setminus \Set{i_0}} \Theta_j
	\end{equation*}
	is independent of the double loop of $v_{i_2}$.
	However, since $p=1$ in \eqref{eq:CrossingPhase1}, it follows that the time slice $i_0-1$ is actually nested in the double loop, and that the graph is nested.
	We conclude that $u_1,u_2,u_3$ cannot all be zero.
	Finally, it follows from \cref{lemma:odd_dependence} that at least one of them must depend on at least two free momenta.
\end{proof}
 
\section{Key parameters and main lemmata}\label{sec:main_lemmata}
In this section, we specify our choices of parameters in the Duhamel expansion, and provide an example of a dispersion relation $\omega$ and interaction $V$ satisfying the assumptions of \cref{sec:assumptions}.
Further, we also collect some main technical results that are needed to treat the terms in the expansion.

\subsection{Key parameters}\label{sec:parameters}
We take all the same parameters as in \cite{LukkarinenSpohn:WNS:2011}, but we recall them here for convenience, and motivate the definitions.
\begin{definition}[Key parameters]
	Let $\delta$ be the constant defined in \cref{ass:DR2} and $\gamma$ be the constant defined in \cref{ass:DR4}.
	We define the parameters
	\begin{equation}
		\label{Def:Para1}
		b=\frac{3}{4},
		\quad \gamma' = \min \myp[\Big]{\frac{1}{4}, 2 \gamma, 2\delta}, 
		\quad a_0 = \frac{\gamma'}{24}, 
		\quad \mbox{ and } b_0 = 16 \myp[\Big]{3+\frac{1}{a_0}}.
	\end{equation}
	Moreover, we also need
	\begin{equation}
		\label{Def:Para2}
		\eps = \lambda^2 
		\mbox{ and } 
		N_0 \myp{\lambda} = \max \myp[\bigg]{1, \floor[\bigg]{ \frac{a_0 \abs{\ln\lambda}}{\ln \expec{\ln\lambda}} }},
	\end{equation}
	where $\floor{x}$ is the integer part of $x\geq 0$. 
	For $N_0=N_0(\lambda)$, we also define
	\begin{equation}
		\label{Def:Para3}
		\kappa' \myp{\lambda} = \lambda^2 N_0^{b_0},
	\end{equation}
	and take the sequence $\kappa = \kappa \myp{\lambda}$ to be
	\begin{equation}
		\label{Def:Para4}
		\kappa_n \myp{\lambda}
		={} \begin{cases}
			0 &	0 \leq n<N_0/2, \\
			\kappa' \myp{\lambda} &	N_0/2 \leq n \leq N_0.
		\end{cases} 
	\end{equation}
\end{definition}
It is immediately apparent that this choice of parameters satisfies $N_0 \to \infty$, $\max_n \kappa_n \to 0$, and
\begin{equation}
	\label{eq:N0limits}
	\frac{N_0 \ln \expec{\ln \lambda}}{\abs{\ln \lambda}} \to a_0
	\qquad \text{and} \qquad
	\frac{N_0 \ln N_0}{\abs{\ln \lambda}} \to a_0,
\end{equation}
as $\lambda \to 0^+$.

The reason for making the specific choice \eqref{Def:Para2} of $N_0 \myp{\lambda}$ is that we will need to control expressions of the type
\begin{equation}
\label{eq:lambda_expression}
	c^{N_0} \lambda^{p_1} N_0^{p_2 N_0 + c'} \myp{\myp{n_1N_0}!}^{n_2} \expec{\ln \lambda}^{n_3 N_0 + c'}
	\to 0
	\qquad \text{when} \qquad
	\lambda \to 0^+,
\end{equation}
for specific values of $c,c' \geq 0$, $n_1,n_2,n_3 \in \N$, and $p_1,p_2 \in \R$, compare e.g. with the results in \cref{sec:graph_estimates} below.
Indeed, \eqref{eq:lambda_expression} holds true whenever $p_1 - a_0 \myp{p_2 + n_1n_2 + n_3} > 0$, with the decay going like a power of $\lambda$.
To see this, simply note for $\lambda < 1$ that we may bound $\expec{\ln \lambda}^{n_3 N_0} \lesssim \lambda^{-a_0 n_3}$ and $N_0^{\myp{p_2+n_1n_2}N_0} \lesssim \lambda^{-a_0 \myp{p_2+n_1n_2}}$, by \eqref{eq:N0limits}.
Hence, using $n! \leq n^n$, we conclude that \eqref{eq:lambda_expression} is bounded by
\begin{equation*}
	\myp{c n_1^{n_1n_2}}^{N_0} \myp{N_0\expec{\ln \lambda}}^{c'} \lambda^{p_1 - a_0 \myp{p_2 + n_1n_2 + n_3}}
	\leq{} C_{\alpha} \lambda^{-\alpha} \myp{\expec{\ln \lambda}}^{2c'} \lambda^{p_1 - a_0 \myp{p_2 + n_1n_2 + n_3}},
\end{equation*}
for any $\alpha > 0$, and the assertion follows.
For instance, we obtain decay up to $\lambda^2$ in
\begin{equation*}
	c^{N_0} \lambda^{-2} N_0^{-\frac{b_0}{4} N_0 + 4N_0 + c'} \myp{4N_0}! \expec{\ln \lambda}^{4 N_0 + c'} \to 0,
\end{equation*}
and decay up to $\lambda^{\gamma'/2}$ in
\begin{equation*}
	c^{N_0} \lambda^{\gamma'} N_0^{4N_0 + c'} \myp{4N_0}! \expec{\ln \lambda}^{4 N_0 + c'} \to 0.
\end{equation*}

Finally, the advantage of choosing the sequence $\kappa$ as in \eqref{Def:Para4} is twofold:
The presence of the non-zero entries in the second half of the sequence allows us to obtain a little extra decay in $\lambda$ for amplitudes of graphs that are large (compared to $N_0$).
This is mainly useful for controlling the error terms.
On the other hand, since the rest of the entries of $\kappa$ vanish, the sequence does not even appear in the amplitudes of graphs that are small compared to $N_0$.
This will be useful when we compute the limit of the leading diagrams in \cref{sec:mainproof}.

\subsection{From phases to resolvents, cluster combinatorics, and integrals over free momenta} 
We now recall the following technical lemmas, whose proofs can be found in \cite[Section 7]{LukkarinenSpohn:WNS:2011}.
\begin{lemma}[From phases to resolvents]
	\label{Lemma:PhaseResolve}
	Let $I$ be a non-empty finite index set and $t>0$.
	Suppose that $A$ is a non-empty subset of $I$, $\gamma_j \in D$, $j\in I$, with $D$ being a compact subset of $\C$.
	We choose an additional time index label $J$, i.e. suppose $J \notin I$ with $\gamma_J \in D$, and denote $A^c = I \setminus A$, $A' = A^c \cup \Set{J}$.
	Then for any path $\Gamma_D$ going once anticlockwise around $D$, the following holds true
	\begin{align}
		\MoveEqLeft[4] \int_{\myp{\R_+}^I} \ids s \delta \myp[\bigg]{t-\sum_{j \in I} s_j} \prod_{j \in I} e^{-i \gamma_j s_j} \nn \\
		={}& -\oint_{\Gamma_D} \frac{\ids z}{2\pi} \int_{\myp{\R_+}^{A'}} \ids s \delta \myp[\bigg]{t-\sum_{j \in A'} s_j} \prod_{j \in A'} e^{-i \gamma_j s_j} \Big|_{\gamma_J=z} \prod_{j \in A} \frac{i}{z-\gamma_j}.
		\label{Lemma:PhaseResolve:1}
	\end{align}
\end{lemma}
\begin{lemma}[Cluster combinatorics]
	\label{lemma:Clustercombinatorics}
	Suppose that \cref{ass:L1} holds.
	There is a constant $c$ such that for all $N>0$ and $0 < \lambda < \lambda_0$
	\begin{equation*}
		\sum_{S \in \pi(I_N)} \prod_{A\in S} \sup_{k,\sigma,\Omega} \abs[\big]{C_{\abs{A}} \myp{k,\sigma,\lambda,\Omega}} 
		\leq{} c^N N!
	\end{equation*}
	If the sum is over non-pairing $S$, the bound can be improved by a factor of $\lambda$.
\end{lemma}
The proof of \cref{lemma:Clustercombinatorics} for the classical case in \cite{LukkarinenSpohn:WNS:2011} carries over to the \emph{fermionic} quantum case without change.
For the non-pairing clusters, \cref{ass:L1} is the only needed input, but for the pairing clusters, the bound relies on the fact that $W_L^{\lambda}$ is bounded, cf. \cref{rem:pair_correl_trunc}, which is not the case for bosons.

In \cite{LukkarinenSpohn:WNS:2011}, the following three results are proved for the unmodified dispersion relation $\omega$.
However, since we require the modified dispersion relation $\omega^{\lambda}$ to satisfy \cref{ass:DR2,ass:DR3} uniformly in $\lambda$, the proofs carry over to our case without change.
\begin{proposition}
	\label{prop:osc_bound}
	Suppose \cref{ass:DR2} holds with constants $C,\delta > 0$, and assume $f \in \ell_1 \myp{\myp{\Z^d}^3}$.
	Then for all $s \in \R$, $k_0 \in \T^d$, and $\sigma,\sigma' \in \Set{\pm 1}$,
	\begin{equation}
		\abs[\bigg]{ \int_{\myp{\T^d}^2} \ids k \ids k' \, e^{is \myp{\omega^{\lambda} \myp{k} + \sigma' \omega^{\lambda} \myp{k'} + \sigma \omega^{\lambda} \myp{k_0 -k -k'}}} \hat{f} \myp{k,k',k_0-k-k'} }
		\leq C \normt{f}{1} \expec{s}^{-1-\delta}.
	\end{equation}
\end{proposition}
As discussed in \cite[Section 7.4]{LukkarinenSpohn:WNS:2011}, \cref{prop:osc_bound} can be used to give rigorous meaning to the energy conservation delta functions in the collision operator \eqref{eq:BN-C} and \eqref{eq:mu-infty}.
\begin{lemma}[Integrals over free momenta - degree one vertex]
	\label{lemma:degree1vertex}
	If \cref{ass:DR3} is satisfied, then for any $k_0 \in \R^d$, $\alpha \in \R$, $\abs{\beta} > 0$, and $\sigma,\sigma' \in \Set{\pm 1}$,
	\begin{equation}
		\label{eq:degree1vertex}
		\int_{\T^d} \ids k \frac{F_1 \myp{\sigma' k}}{\abs{\omega^{\lambda} \myp{k} + \sigma \omega^{\lambda} \myp{k_0-k} - \alpha + i\beta}}
		\lesssim{} \lambda^{-b} \expec{ \ln \abs{\beta} }^2 .
	\end{equation}
\end{lemma}
Here, the function $F_1$ is the one appearing in the construction of the momentum cut-off function $\Phi_1$ in \cref{Propo:Cutoff}.
\begin{lemma}[Integrals over free momenta - degree two vertex]
	\label{lemma:degree2vertex}
	Suppose that \cref{ass:DR2,ass:DR3} are satisfied.
	Then for any $k_0 \in \R^d$, $\alpha \in \R$, $\abs{\beta} > 0$, and $\sigma,\sigma' \in \Set{\pm 1}$,
	\begin{equation}
		\label{eq:degree2vertex}
		\int_{\T^{2d}} \ids k \ids k' \frac{1}{\abs{\omega^{\lambda} \myp{k} + \sigma' \omega^{\lambda} \myp{k'} +\sigma \omega^{\lambda} \myp{k_0-k-k'} - \alpha + i\beta}} 
		\lesssim{} \expec{ \ln \abs{\beta} }.
	\end{equation}
\end{lemma}
The proofs of \cref{lemma:degree1vertex,lemma:degree2vertex} are based on the following formula (see e.g. \cite[Lemma 4.21]{LukkarinenSpohn:KLW:2007}), expressing a resolvent as an oscillating integral,
\begin{equation}
	\label{eq:resolvent_integral}
	\frac{1}{\abs{r + i \beta}}
	={} \int_{\R} \ids s \, e^{isr} F \myp{s;\beta},
\end{equation}
for $r \in \R$ and $0 < \beta \leq 1$, where the function $F \myp{s;\beta}$ satisfies $0 \leq F \myp{s;\beta} \leq \expec{\ln \beta} e^{-\beta \abs{s}} + \1 \myp{\abs{s} \leq 1} \ln \frac{1}{\abs{s}}$.

\cref{Lemma:PhaseResolve,lemma:degree1vertex,lemma:degree2vertex} will be especially useful to bound the amplitudes for graphs that are not fully paired in \cref{sec:graph_estimates} below.

\subsection{Extended Bessel functions and nearest neighbor interactions}
\label{sec:verify_assumptions}
We consider here the modified dispersion relation \eqref{ModifiedDispersion} and take a closer look at \cref{ass:DR1,ass:DR2,ass:DR3,ass:DR4} from \cref{sec:assumptions}.
The goal of this section is to provide an example of $\omega$ and $V$ such that $\omega^{\lambda}$ satisfies these assumptions.
Note first that regardless of $\omega$ and $V$, the quantity $\int_{\Omega^{\ast}} \mathrm{d}k_2 W_L^{\lambda} \myp{k_2} \widehat{V} \myp{0}$ is a constant and thus will only appear as a constant phase factor in the integrals in \cref{ass:DR2,ass:DR3,ass:DR4}.
As a result, we only need to consider the convolution term
\begin{equation}
	\label{QuantityR1}
	\widetilde{R}_{\lambda} \myp{k_1}
	\coloneq{} -\int_{\Omega^{\ast}} \mathrm{d}k_2 W_L^{\lambda} \myp{k_2} \widehat{V} \myp{k_1-k_2},
\end{equation}
and modify the dispersion relation as
\begin{equation}
	\label{ModifiedDispersion1}
	\widetilde{\omega}^{\lambda} 
	\coloneq{} \omega + \lambda \widetilde{R}_{\lambda}
\end{equation}
when checking that the assumptions hold true.
Similar as in \cite{LukkarinenSpohn:WNS:2011}, we consider the nearest neighbour dispersion relation
\begin{equation}
\label{Nearest}
	\omega \myp{k} 
	\coloneq c - \sum_{j=1}^d \cos \myp{2\pi k_j},
\end{equation}
where the constant $c$ is arbitrary, and $ k_j $ denotes the $ j $'th coordinate of $ k \in \T^d $.
We will also assume that the interaction potential is of the same form as $\omega$, i.e.,
\begin{equation}
\label{eq:Vassump}
	\widehat{V} \myp{k} 
	={} \tilde{c} - \sum\limits_{j=1}^d \cos \myp{2 \pi k_j-\alpha_j},
\end{equation}
for some constants $\tilde{c}, \alpha_1,\dotsc, \alpha_d \in \R$ which will not play a role in the following.
The constants $c,\tilde{c}$ have no influence on the absolute values of the integrals in \cref{ass:DR2,ass:DR3,ass:DR4}, so we can assume them to be zero.
Furthermore, we will also take $\alpha_1,\dotsc,\alpha_d$ to be zero in the following, for simplicity.
\begin{lemma}
\label{lemma:Kcrossing}
	There are constants $C, \lambda_0 > 0$ such that for any $t_0,t_1,t_2 \in \R$, $u_1,u_2 \in \myb{0,2 \pi}^d$, and $0\leq \lambda < \lambda_0$, the function $\mathcal{K}$ from \eqref{eq:defp2tx} satisfies
	\begin{equation}
	\label{eq:Kcrossing}
		\normt[\Big]{\mathcal{K} \myp[\Big]{\nonarg ; t_0,t_1,t_2,\frac{u_1}{2 \pi},\frac{u_2}{2 \pi}}}{\ell^3}
		\leq C \prod\limits_{j=1}^d \frac{1}{\expec{r_j}^{\frac{1}{7}}},
	\end{equation}
	where $r_j = \abs{t_0 + t_1 e^{i u_{1,j}} + t_2 e^{i u_{2,j}}}$ for $j = 1, \dotsc, d$.
	Furthermore, it follows that
	\begin{equation}
	\label{eq:l3_dispersivity2}
		\normt{p_t^{\lambda}}{\ell^3}^3
		\leq{} C \expec{t}^{-\frac{3d}{7}},
	\end{equation}
	so that \cref{ass:DR2} is satisfied for $d \geq 3$.
\end{lemma}
\begin{proof}
	We will prove \eqref{eq:Kcrossing} by providing bounds on $\mathcal{K}$ in $\ell^2 \myp{\Z^d}$ and $\ell^4 \myp{\Z^d}$, and then interpolating.
	Note first that we can write
	\begin{align}
		\MoveEqLeft[3] \mathcal{K} \myp[\Big]{x;t_0,t_1,t_2,\frac{u_1}{2 \pi}, \frac{u_2}{2 \pi}} \nn \\
		={}& \frac{1}{\myp{2 \pi}^d} \int_{\myb{0,2\pi}^d} \ids p \, e^{i p \cdot x} \prod_{j=1}^d \myb[\Big]{ e^{i \myp{t_0 \cos \myp{p_j} + t_1 \cos \myp{p_j+u_{1,j}} + t_2 \cos \myp{p_j + u_{2,j}} }} \nn \\
		&\times e^{-i \lambda \myp[\big]{ \int_{\Omega^{\ast}} \ids k W_L^{\lambda} \myp{k} \myp{ t_0 \cos \myp{p_j - 2\pi k_j} + t_1 \cos \myp{p_j+u_{1,j} - 2\pi k_j} + t_2 \cos \myp{p_j + u_{2,j} - 2\pi k_j} } }} } \nn \\
		={}& \frac{1}{\myp{2 \pi}^d} \int_{\myb{0,2\pi}^d} \ids p \, e^{i p \cdot x} \prod_{j=1}^d e^{i r_j \myp[\big]{\cos \myp{p_j+ \phi_j} - \lambda \int_{\Omega^{\ast}} \ids k W_L^{\lambda} \myp{k} \cos \myp{p_j+\phi_j - 2\pi k_j} }},
	\label{eq:Kcrossing1}
	\end{align}
	where $r_j = \abs{t_0 + t_1 e^{i u_{1,j}} + t_2 e^{i u_{2,j}}}$, and the vector $\phi$ is independent of $p$.
	The numbers $\phi_j$ are chosen so that independently of $p_j \in \myb{0,2 \pi}$,
	\begin{align*}
		t_0 \cos \myp{p_j} + t_1 \cos \myp{p_j+u_{1,j}} + t_2 \cos \myp{p_j + u_{2,j}}
		={}& \mathrm{Re}\myb{e^{ip_j} \myp{t_0 + t_1 e^{i u_{1,j}} + t_2 e^{i u_{2,j}} }} \\
		={}& r_j \cos\myp{p_j + \phi_j}.
	\end{align*}
	Applying Parseval's identity to \eqref{eq:Kcrossing1} immediately implies that $\normt{\mathcal{K}}{\ell^2} = \myp{2\pi}^{-d}$.
	On the other hand, continuing from \eqref{eq:Kcrossing1}, we also have
	\begin{align*}
		\MoveEqLeft[4] \abs[\Big]{\mathcal{K} \myp[\Big]{x;t_0,t_1,t_2,\frac{u_1}{2 \pi}, \frac{u_2}{2 \pi}}}^2 \nn \\
		={}& \frac{1}{\myp{2 \pi}^{2d}} \int_{\myb{0,2\pi}^d} \ids p \, e^{i p\cdot x} \int_{\myb{0,2\pi}^d} \ids q \prod_{j=1}^d e^{i r_j \myp{\cos \myp{p_j + q_j + \phi_j} - \cos \myp{q_j + \phi_j}}} \nn \\
		&\times e^{-i \lambda r_j \int_{\Omega^{\ast}} \ids k W_L^{\lambda} \myp{k} \myp{\cos \myp{p_j+q_j+\phi_j - 2\pi k_j} - \cos \myp{q_j+\phi_j - 2\pi k_j}} },
	\end{align*}
	so denoting 
	\begin{equation*}
		\psi_j \myp{p_j,q_j}
		\coloneq{}  \sin \myp[\big]{\frac{p_j}{2} + q_j + \phi_j} - \lambda \int_{\Omega^{\ast}} \ids k W_L^{\lambda} \myp{k} \sin \myp[\big]{\frac{p_j}{2}+q_j+\phi_j - 2\pi k_j},
	\end{equation*}
	it follows, again by Parseval's identity, that
	\begin{align}
		\MoveEqLeft[2] \normt[\Big]{\mathcal{K} \myp[\Big]{\nonarg;t_0,t_1,t_2,\frac{u_1}{2 \pi}, \frac{u_2}{2 \pi}}}{\ell^4}^4 
		={} \frac{1}{\myp{2 \pi}^{4d}} \prod\limits_{j=1}^d \int_{0}^{2\pi} \ids p_j \abs[\bigg]{ \int_{0}^{2\pi} \ids q_j e^{-2 i r_j \sin\myp{\frac{p_j}{2}} \psi_j \myp{p_j,q_j}} }^2.
	\label{eq:Kcrossing3}
	\end{align}
	
	We now ignore the subscript $j$ and fix the $p$-variable, and note that at any critical point where $\partial_q \psi \myp{p,q} = 0$, we must have $\abs{\cos \myp{\frac{p}{2} + q + \phi}} \leq \lambda$, since $0 \leq W_L^{\lambda} \leq 1$.
	It follows for $\lambda$ small enough, independently of $L$, that we can estimate the double derivative at critical points by $\abs{\partial_q^2 \psi\myp{p,q}} \geq \abs{\sin\myp{\frac{p}{2} + q+ \phi}} - \lambda \geq \frac{1}{2}$, so by the stationary phase approximation (e.g. \cite[Proposition 2.6.7, Corollary 2.6.8]{Grafakos-2014}), we conclude for $r \sin \myp{p/2}$ sufficiently large that
	\begin{equation*}
		\abs[\bigg]{ \int_{0}^{2\pi} \ids q \, e^{-2 i r \sin\myp{\frac{p}{2}} \psi \myp{p,q}} }
		\leq{} \frac{C}{\myp{\sin\myp{\frac{p}{2}} r}^{\frac{1}{2}}}.
	\end{equation*}
	Since the $q$-integral is also bounded by $2 \pi$ uniformly in $p$, we obtain for any $r \geq 0$,
	\begin{equation*}
		\int_0^{2 \pi} \ids p \abs[\bigg]{ \int_{0}^{2\pi} \ids q \, e^{-2 i r \sin\myp{\frac{p}{2}} \psi \myp{p,q}} }^2
		\leq{} \int_0^{2 \pi} \ids p \frac{C}{1+\sin\myp{\frac{p}{2}} r}
		\leq{} \frac{C \log \myp{1+r}}{r}.
	\end{equation*}
	Using that the $p$-integral is also uniformly bounded for any $r$, we finally have
	\begin{equation*}
		\int_0^{2 \pi} \ids p \abs[\bigg]{ \int_{0}^{2\pi} \ids q \, e^{-2 i r \sin\myp{\frac{p}{2}} \psi \myp{p,q}} }^2
		\leq{} \frac{C}{\expec{r}^{\frac{6}{7}}}.
	\end{equation*}
	
	Returning to \eqref{eq:Kcrossing3}, we have proven that
	\begin{align*}
		\normt[\Big]{\mathcal{K} \myp[\Big]{\nonarg;t_0,t_1,t_2,\frac{u_1}{2 \pi}, \frac{u_2}{2 \pi}}}{\ell^4}^4
		\leq{}& C \prod\limits_{j=1}^d \frac{1}{\expec{r_j}^{\frac{6}{7}}}.
	\end{align*}
	Noting that $\frac{1}{3} = \frac{\theta}{4} + \frac{1-\theta}{2}$ for $\theta = \frac{2}{3}$, we get by standard interpolation in $\ell^p$,
	\begin{equation*}
		\normt{\mathcal{K}}{\ell^3}
		\leq{} \normt{\mathcal{K}}{\ell^4}^{\theta} \normt{\mathcal{K}}{\ell^2}^{1-\theta}
		={} \normt{\mathcal{K}}{\ell^4}^{\frac{2}{3}} \normt{\mathcal{K}}{\ell^2}^{\frac{1}{3}}
		\leq{} C \prod\limits_{j=1}^d \frac{1}{\myp{\expec{r_j}^{\frac{6}{7}}}^{\frac{2}{12}} }
		={} C \prod\limits_{j=1}^d \frac{1}{\expec{r_j}^{\frac{1}{7}} },
	\end{equation*}
	finishing the proof of \eqref{eq:Kcrossing}.
	The bound \eqref{eq:l3_dispersivity2} follows immediately, noting that $ p_t^{\lambda} \myp{x} = \mathcal{K} \myp{x; t,0,0,0,0} $.
\end{proof}
\begin{lemma}[Constructive interference]
	Suppose that $d \geq 2$.
	Let $M^{\mathrm{sing}} \subseteq \T^d$ be the subset consisting of the points $k \in \T^d$ with all but one component equal to either $0$ or $\frac{1}{2}$.
	There exist $C,\lambda_0>0$, such that for any $k_0 \in \T^d$, $\sigma \in \Set{\pm 1}$, and $t \in \R$, we have for all $0 \leq \lambda < \lambda_0$,
	\begin{equation}
		\label{eq:constr_interf2}
		\abs[\bigg]{\int_{\T^d} \ids k \, e^{-it \myp{\omega^{\lambda} \myp{k} + \sigma \omega^{\lambda} \myp{k-k_0}}}}
		\leq{} \frac{C \expec{t}^{-1}}{d \myp{k_0, M^{\mathrm{sing}}}}.
	\end{equation}
\end{lemma}
\begin{proof}
	From \eqref{eq:Kcrossing1} we have
	\begin{align*}
		\MoveEqLeft[4] \abs[\bigg]{\int_{\T^d} \ids k \, e^{-it \myp{\omega^{\lambda} \myp{k} + \sigma \omega^{\lambda} \myp{k-k_0}}}}
		={} \abs{\mathcal{K} \myp{0;t,\sigma t,0,-k_0,0} } \\
		={}& \frac{1}{\myp{2 \pi}^d} \prod_{j=1}^d \abs[\bigg]{ \int_0^{2\pi} \ids p_j \, e^{i r_j \myp[\big]{\cos \myp{p_j+ \phi_j} - \lambda \int_{\Omega^{\ast}} \ids k W_L^{\lambda} \myp{k} \cos \myp{p_j+\phi_j - 2\pi k_j} }} },
	\end{align*}
	with $r_j = \abs{t} \abs{1+\sigma e^{-2 \pi i k_{0,j}}}$.
	By the argumentation following \eqref{eq:Kcrossing3}, we can choose $\lambda$ small enough, independently of $L$, so that by the stationary phase approximation, we have for any value of $r_j \geq 0$,
	\begin{equation}
	\label{eq:constr_interf3}
		\abs[\bigg]{\int_{\T^d} \ids k \, e^{-it \myp{\omega^{\lambda} \myp{k} + \sigma \omega^{\lambda} \myp{k-k_0}}}}
		\leq C \prod_{j=1}^d \frac{1}{\expec{r_j}^{\frac{1}{2}}}.
	\end{equation}
	Now, using $ \abs{1+\sigma e^{-2 \pi i k_{0,j}}} \geq \abs{\sin \myp{ 2 \pi k_{0,j}}} $ and picking the index $j$ corresponding to the \emph{second} largest of the numbers $\abs{\sin \myp{ 2 \pi k_{0,j}}}$, we find
	\begin{equation*}
		\abs[\bigg]{\int_{\T^d} \id k \, e^{-it \myp{\omega^{\lambda} \myp{k} + \sigma \omega^{\lambda} \myp{k-k_0}}}}
		\leq{} \frac{C}{\expec{t} \abs{\sin \myp{2 \pi k_{0,j}}}},
	\end{equation*}
	so choosing $M^{\mathrm{sing}}$ as described in the statement finishes the proof.
\end{proof}

With \eqref{eq:Kcrossing} in hand, the crossing bounds in \cref{ass:DR4} can be proven in the exact same way as in \cite{LukkarinenSpohn:WNS:2011}.
However, the proof in \cite{LukkarinenSpohn:WNS:2011} can easily be extended to dimensions $d \geq 3$ instead of only $d \geq 4$, so we write out the details here for completeness.
\begin{lemma}[Crossing bounds]
	Suppose that $d \geq 3$.
	Let $u_m \in \T^d$ and $\sigma_m \in \Set{\pm 1}$ for $m = 1,2,3$, and $0 < \zeta \leq 1$.
	Then we have
	\begin{equation}
	\label{eq:crossing1}
		\iint_{\R^2} \ids s \ids t \, e^{- \zeta \abs{s}} \prod\limits_{m=1}^3 \normt{\mathcal{K} \myp{\nonarg; t, \sigma_m s,0,u_m,0}}{\ell^3}
		\leq{} C \zeta^{\frac{1}{7}-1} \prod\limits_{j=1}^d \frac{1}{\abs{\sin \myp{2 \pi u_{n,j}}}^{\frac{1}{7d}}}
	\end{equation}
	for any $n = 1,2,3$, and
	\begin{align}
		\MoveEqLeft[8] \iint_{\R^2} \ids s \ids t \, e^{- \zeta \abs{s}} \normt{p_t^{\lambda}}{\ell^3}^2 \normt{\mathcal{K} \myp{\nonarg; t, \sigma_1 s, \sigma_2 s,u_1,u_2}}{\ell^3} \nn \\
		\leq{}& C \zeta^{\frac{1}{7}-1} \prod\limits_{j=1}^d \frac{1}{\abs{\sin \myp{2 \pi \myp{ u_{2,j}-u_{1,j}}} }^{\frac{1}{7d}}}.
	\label{eq:crossing2}
	\end{align}
	It follows that \cref{ass:DR4} is satisfied with
	\begin{equation}
	\label{eq:FcrDef}
		F^{cr} \myp{u,\zeta} \coloneq C \prod_{j=1}^d \frac{1}{\abs{\sin \myp{2 \pi u_{j}}}^{\frac{1}{7}}}.
	\end{equation}
\end{lemma}
\begin{proof}
	We start out by proving \eqref{eq:crossing1}.
	By \cref{lemma:Kcrossing}, we can bound 
	\begin{equation*}
		\normt{\mathcal{K} \myp{\nonarg; t, \sigma_m s,0,u_m,0}}{\ell^3} \lesssim \prod_{j=1}^d \expec{r_j \myp{u_m}}^{-1/7}, 
	\end{equation*}
	where $r_j \myp{u_m} = \abs{t \pm s e^{i 2 \pi u_{m,j}}} \geq \max \Set{ \abs{\abs{t} - \abs{s}}, \abs{s} \abs{\sin \myp{2 \pi u_{m,j}}} } $.
	We use the latter bound on $r_j$ in a single coordinate $u_{n,i}$ of $u_n$ satisfying that $\abs{\sin \myp{2 \pi u_{n,i}}}$ is as large as possible, and we apply the former bound in all other coordinates of the $u_m$, to get sufficient decay in $t$.
	This means that we can bound
	\begin{align*}
		\MoveEqLeft[6] \iint_{\R^2} \ids s \ids t \, e^{- \zeta \abs{s}} \prod\limits_{m=1}^3 \normt{\mathcal{K} \myp{\nonarg; t, \sigma_m s,0,u_m,0}}{\ell^3} \\
		\leq{}& C \iint_{\R^2} \ids s \ids t \, e^{- \zeta \abs{s}} \prod\limits_{m=1}^3 \prod_{j=1}^d \frac{1}{\expec{r_j \myp{u_m}}^{\frac{1}{7}}}  \\
		\leq{}& C \iint_{\R^2} \ids s \ids t \, e^{- \zeta \abs{s}} \expec{\abs{t} - \abs{s}}^{-\frac{3d-1}{7}} \frac{1}{ \myp{ \abs{s} \abs{\sin \myp{2 \pi u_{n,i}}} }^{\frac{1}{7}} },
	\end{align*}
	where, for $d \geq 3$, the $t$-integral is bounded independently of $s$,
	\begin{equation*}
		\int_{\R} \ids t \expec{\abs{t} - \abs{s}}^{-\frac{3d-1}{7}} 
		\leq{} 2 \int_{\R} \ids t \expec{t}^{-\frac{3d-1}{7}}
		< \infty.
	\end{equation*}
	Finally, noting that $\abs{\sin \myp{2 \pi u_{n,i}}} \geq \prod_{j=1}^d \abs{\sin \myp{2 \pi u_{n,j}}}^{1/d}$ by choice of the coordinate $u_{n,i}$, it follows that \eqref{eq:crossing1} holds.
	
	To prove \eqref{eq:crossing2}, we have similarly, with the use of \eqref{eq:l3_dispersivity2},
	\begin{align*}
		\MoveEqLeft[8] \iint_{\R^2} \ids s \ids t \, e^{- \zeta \abs{s}} \normt{p_t^{\lambda}}{\ell^3}^2 \normt{\mathcal{K} \myp{\nonarg; t, \sigma_1 s, \sigma_2 s,u_1,u_2}}{\ell^3} \\
		\leq{}& C \iint_{\R^2} \ids s \ids t \, e^{- \zeta \abs{s}} \expec{t}^{-\frac{2d}{7}} \prod\limits_{j=1}^d \expec{r_j}^{-\frac{1}{7}},
	\end{align*}
	where in this case we can bound
	\begin{align*}
		r_j 
		={}& \abs{t \pm s \myp{e^{i 2 \pi u_{1,j}} \pm e^{i 2 \pi u_{2,j}} }} \\
		\geq{}& \max \Set[\big]{ \abs[\big]{ \abs{t}-\abs{s} \abs{1 \pm e^{i 2 \pi \myp{u_{2,j} - u_{1,j}} }} }, \abs{s} \abs{ \sin \myp{2 \pi \myp{u_{2,j} - u_{1,j}}} }}.
	\end{align*}
	As before, using the latter bound in the coordinate where $ \abs{ \sin \myp{2 \pi \myp{u_{2,i} - u_{1,i}}} } $ is the largest possible, and the former bound in all other coordinates, leads to
	\begin{align*}
		\MoveEqLeft[2] \iint_{\R^2} \ids s \ids t \, e^{- \zeta \abs{s}} \expec{t}^{-\frac{2d}{7}} \prod\limits_{j=1}^d \expec{r_j}^{-\frac{1}{7}} \\
		\leq{}& \iint_{\R^2} \ids s \ids t \frac{e^{- \zeta \abs{s}} }{\expec{t}^{\frac{2d}{7}}} \frac{1}{ \myp{ \abs{s} \abs{ \sin \myp{2 \pi \myp{u_{2,i} - u_{1,i}}} } }^{\frac{1}{7}}} \prod\limits_{j \neq i} \frac{1}{\expec{ \abs{t}-\abs{s} \abs{1 \pm e^{i 2 \pi \myp{u_{2,j} - u_{1,j}} }} }^{\frac{1}{7}}} \\
		\leq{}& C \int_{\R} \ids s \frac{e^{-\zeta \abs{s}}}{\abs{s}^{\frac{1}{7}}} \prod\limits_{j=1}^d \frac{1}{\abs{ \sin \myp{2 \pi \myp{u_{2,j} - u_{1,j}}} }^{\frac{1}{7d}}},
	\end{align*}
	and \eqref{eq:crossing2} follows.
	
	We still need to check that \eqref{eq:crossingest2} holds with our choice \eqref{eq:FcrDef} of $F^{cr}$.
	For this, consider first $n=1$ and rewrite the resolvent using \eqref{eq:resolvent_integral}.
	We thus have to estimate
	\begin{align*}
		\MoveEqLeft[4] \iint_{\myp{\T^d}^2} \ids k_1 \ids k_2 F^{cr} \myp{k_1 + u} \frac{1}{\abs{\alpha - \Theta \myp{k,\sigma} + i \zeta}} \\
		\leq{}& \int_{\R} \ids s F \myp{s;\zeta} \int_{\T^d} \ids k_1 F^{cr} \myp{k_1 + u} \abs[\bigg]{ \int_{\T^d} \ids k_2 \, e^{-i s \myp{\omega^{\lambda} \myp{k_0-k_1-k_2} + \sigma \omega^{\lambda} \myp{k_2}}} } \\
		\lesssim{}& \int_{\R} \ids s F \myp{s;\zeta} \int_{\T^d} \ids k_1 F^{cr} \myp{k_1 + u} \prod_{j=1}^d \frac{1}{\expec{r_j}^{\frac{1}{2}}},
	\end{align*}
	where we used the estimate \eqref{eq:constr_interf3} with $r_j = \abs{s} \abs{1+\sigma e^{-2 \pi i \myp{k_{1,j}- k_{0,j}} }} $.
	Plugging in the expression for $F^{cr}$ and then applying H\"older's inequality in the $k$-integral (for instance with the conjugate pair $\myp{3,\frac{3}{2}}$), we obtain that this is bounded from above by
	\begin{align*}
		\MoveEqLeft[6] \int_{\R} \ids s \frac{F \myp{s;\zeta}}{\expec{s}^{\frac{d}{2}}} \prod_{j=1}^d  \int_{0}^{1} \ids k_{1,j} \frac{1}{\abs{\sin 2 \pi \myp{k_{1,j} +  u_j}}^{\frac{1}{7}} \abs{\sin 2 \pi \myp{k_{1,j} - k_{0,j}}}^{\frac{1}{2}} } \\
		\lesssim{}& \int_{\R} \ids s \frac{1}{\expec{s}^{\frac{d}{2}}} \myp[\Big]{\expec{\ln \zeta} e^{-\zeta \abs{s}} + \1 \myp{\abs{s} \leq 1} \ln \frac{1}{\abs{s}}}
		\lesssim{} 1,
	\end{align*}
	uniformly for $0 < \zeta \leq 1$ since $d \geq 3$, showing that \eqref{eq:crossingest2} holds with $c_2 = 0$.
	The case $n = 2$ is handled in the same way, and for $n = 3$, we simply make the change of variables $\tilde{k}_1 = k_0 - k_1 - k_2$, which leads back to the case $n=1$, but with $\Theta \myp{k,\sigma}$ replaced by $-\Theta \myp{k,-\sigma}$ in the resolvent.
\end{proof}

\section{Graph estimates}
\label{sec:graph_estimates}
\subsection{Basic estimates}
In this section, we prove estimates on the amplitudes corresponding to graphs that are either higher order (that is, containing clusters that are not pairing), or partially paired (meaning pairing graphs that contain a vertex of degree one).
\begin{lemma}[Basic graph estimate 1]
\label{lemma:BasicGEstimate1}
	There is a constant $C>0$ such that for $1 \leq n \leq N_0$, $s>0$, $S \in \pi \myp{I_{4n+2}}$, and any $k \in \T^d$, we have
	\begin{align}
		\limsup_{L\to\infty} &\lambda^{2n} \sum_{\bar{\sigma},\bar{\sigma}' \in \Set{\pm 1}^{\mathcal{I}_n'}} \1 \myp{\sigma_{n,1} = 1} \1 \myp{\sigma_{n,1}'=-1} \frac{1}{L^d} \int_{\myp{\Omega^{\ast}}^{\mathcal{I}_n'}} \ids \bar{k} \int_{\myp{\Omega^{\ast}}^{\mathcal{I}_n'}} \ids \bar{k}' \nn \\
		&\times \delta_L \myp{k_{n,1} - \myb{k}} \delta_L \myp{k_{n,1}' + \myb{k}} \Delta_{n,\ell} \myp{\bar{k},\bar{\sigma}} \Delta_{n,\ell'} \myp{\bar{k}',\bar{\sigma}'} \prod_{A \in S} \delta_L \myp[\Big]{ \sum_{j\in A} K_j } \nn \\
		&\times\prod_{j=1}^n \abs[\big]{ \Psi_1 \myp{k_{j-1:\ell_j}, \sigma_{j,\ell_j}} \Psi_1 \myp{-k'_{j-1:\ell_j'}, \sigma_{j,\ell_j'}'} } \nn \\
		&\times \abs[\bigg]{ \int_{\myp{\R_+}^{I_{2,2n}}} \mathrm{d}r \delta \myp[\bigg]{s - \sum_{j=2}^{2n} r_j} \prod_{j=2}^{2n} e^{-ir_j \gamma \myp{j;J}} } \nn \\
		\leq{}& C^{1+\tilde{n}_1+\tilde{n}_2} \normt{\widehat{V}}{\infty}^{2n} e^{s \lambda^2} \frac{\myp{s\lambda^2}^{\tilde{n}_0-{n}'_0}}{\myp{\tilde{n}_0-{n}'_0}!} \lambda^{2 + \tilde{n}_2 + \myp{1-b} \tilde{n}_1 - \tilde{n}_0} N_0^{-b_0 n'_0} \expec{\ln n} \expec{\ln\lambda}^{1+\tilde{n}_2+2\tilde{n}_1},
	\label{eq:BasicGEstimate1}
	\end{align}
	where $\tilde{n}_j$ is the number of non-amputated interaction vertices of degree $j$ and $n'_0$ is the number of degree $0$ interaction vertices $v_l$ with $2 < l \leq 2n-N_0+1$.
	We have used the notations $ \mathcal{I}_n' = \mathcal{I}_n \cup \Set{\myp{n,1}} $ and $K = \myp{k_{0,\nonarg}',k_{0,\nonarg}}$ from \cref{Propo:QERR3}, as well as the cell step function $\myb{\nonarg}: \T^d \to \Omega^{\ast}$ from \eqref{eq:discrete-lebesque}.
\end{lemma}
\begin{proof}
	Note first that the constraints on the parities in the expression above imply that the sums over $\sigma$, $\sigma'$ contain only one non-zero term, so they can be ignored.
	We next argue that the $\limsup$ can be taken inside the integrals.
	First, we resolve the momentum constraints of the $\delta_L$-functions using the spanning tree constructed in \cref{Sec:PropMometum}.
	This removes all the vertx $\delta_L$-functions and the integrals over the "integrated edges", and only the integrals over the "free edges" remain.
	It is important to note here that when resolving the momentum constraints, we can rewrite the two additional $\delta_L$-functions connected to the top edges of the graph,
	\begin{equation*}
		\delta_L \myp{k_{n,1} - \myb{k}} \delta_L \myp{k_{n,1}' + \myb{k}}
		={} \delta_L \myp{k_{n,1}' + \myb{k}}^2
		={} L^d \delta_L \myp{k_{n,1}' + \myb{k}}
	\end{equation*}
	which absorbs the prefactor $L^{-d}$ from \eqref{eq:BasicGEstimate1} above.
	Finally, resolving the $k_{n,1}'$-integral removes the last $\delta_L$-function and replaces all dependence on $k_{n,1'}$ with $-\myb{k}$.
	Then, we convert the remaining free momentum integrals into Lebesgue integrals of step functions on $\T^d$, using \eqref{eq:discrete-lebesque}, and note that the resulting integrand is uniformly bounded in $L$.
	Indeed, since we have removed all the $\delta_L$-functions, all that remains of the integrand on the left hand side of \eqref{eq:BasicGEstimate1} are the interaction terms $\Psi_1$ and the time integral, which are bounded by $\normt{\widehat{V}}{\infty}$ and $s^{2n-2}$, respectively.
	Consequently, by dominated convergence we can move the $L \to \infty$ limit inside the integrals.
	By reinserting the momentum constraints, it is possible to write the limit as an expression identical to the left hand side of \eqref{eq:BasicGEstimate1}, but with every $\Omega^{\ast}$ replaced by $\T^d$, and every $\delta_L$ replaced by a continuum $\delta$-function $\delta = \delta_{\T^d}$.
	However, we will refrain from doing this and instead continue working with the resolved momentum constraints in the following.

	There are in total $2n$ interaction vertices.
	We denote $A_j$, $j=0,1,2$, to be the set of time slice indices $2 \leq l < 2n$ such that $\mathrm{deg} \myp{v_{l+1}} = j$.
	It could happen that $A_j=\emptyset$ for some $j$ but we always have $$A_0\cup A_1\cup A_2=\Set{2,3,\dotsc,2n-1}.$$ 
	We set $B = \Set{l\in A_0 \mid l \leq 2n-N_0}$, which can be an empty set.
	For all $l \in B$, it is clear that $\floor{l/2} \leq n-N_0/2$.
	We set $\gamma_j = \gamma \myp{j;J}$ and find by \cref{Lemma:PhaseResolve},
	\begin{align}
	\label{eq:Aestimate1}
		\MoveEqLeft[4] \abs[\bigg]{ \int_{\myp{\R_+}^{I_{2,2n}}} \ids r \delta \myp[\bigg]{s-\sum_{j=2}^{2n}r_j } \prod_{j=2}^{2n} e^{-ir_j \gamma_j} } \nn \\
		\leq{}& \oint_{\Gamma_{n}} \frac{\abs{\ids z}}{2\pi} \int_{\myp{\R_+}^{A'}} \ids r \delta \myp[\bigg]{s-\sum_{j\in A'} r_j } \abs[\big]{e^{-ir_*z}} \prod_{j\in A_0} \abs[\big]{e^{-i r_j \gamma_j}} \prod_{j\in A} \frac{1}{\abs{z-\gamma_j}},
	\end{align}
	where $\Gamma_{n}$ is the contour depicted in Figure \ref{Fig9}, the set $A = \Set{2n}\cup A_1 \cup A_2$ and $A' = \Set{*} \cup A_0$. 
	\begin{figure}
	\centering
		\includegraphics[width=0.85\linewidth]{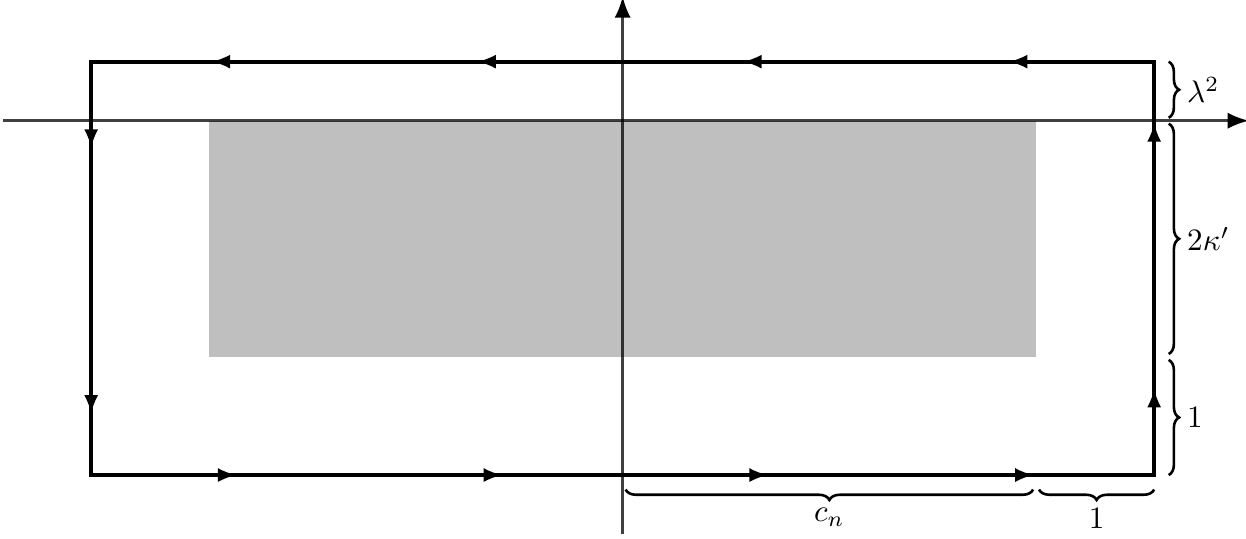}
		\caption{The integration path $\Gamma_n$ of \eqref{eq:Aestimate1}.
		Here, $c_n = 2 (2n+1) (\lVert \omega \rVert_{\infty} + 2 \lVert \widehat{V} \rVert_{\infty} )$, and for any momentum graph with $2n$ vertices, the phase factors $\gamma \myp{j;J}$ always lie in the shaded region.}
	\label{Fig9}
	\end{figure}
	
	In the case that $j\in B$, we have $\mathrm{Im} \myp{-\gamma_j} = \kappa_{n-m}+\kappa_{n-m'}$, in which $m+m'=2+J_+(j-2;J)+J_-(j-2;J)=j$.
	As a consequence, $\min \myp{m,m'} \leq \floor{l/2} \leq n-N_0/2$ and then $\mathrm{Im} \myp{- \gamma_j} \geq \kappa_{n-\min \myp{m,m'}} = \kappa' = \lambda^2 N_0^{b_0} \geq 0$.
	We then find
	\begin{align}
	\label{eq:Aestimate2}
		\MoveEqLeft[4] \int_{\myp{\R_+}^{A'}} \ids r \delta \myp[\bigg]{s-\sum_{j\in A'}r_j } \abs[\big]{e^{-i r_*z}} \prod_{j \in A_0} \abs[\big]{e^{-i r_j \gamma_j}} \nn \\
		\leq{}& e^{s \myp{\mathrm{Im} z}_+} \int_{\myp{\R_+}^B} \ids r \prod_{j \in B} e^{-\kappa' r_j} \int_{\myp{\R_+}^{A'\setminus B}} \ids r \delta \myp[\bigg]{ s-\sum_{j\in B} r_j - \sum_{j \in A'\setminus B} r_j } \nn \\
		\leq{}& e^{s \myp{\mathrm{Im} z}_+} \myp{\kappa'}^{-\abs{B}} \frac{s^{\tilde{n}}}{\tilde{n}!},
	\end{align}
	where $\myp{\mathrm{Im} z}_+ = \max\Set{\mathrm{Im} z,0}$, and
	\begin{equation*}
		\tilde{n} ={} \abs{A' \setminus B}-1 = \abs{A_0 \setminus B} = \abs{A_0} - \abs{B} = \tilde{n}_0 - n'_0.
	\end{equation*}
	Since $\gamma_{2n}=0$, we now deduce
	\begin{align}
	\label{eq:Aestimate3}
		\MoveEqLeft[4] \abs[\bigg]{ \int_{\myp{\R_+}^{I_{2,2n}}} \ids r \delta \myp[\bigg]{ s-\sum_{j=2}^{2n} r_j } \prod_{j=2}^{2n} e^{-ir_j \gamma_j} } \nn \\
		\leq{}& \myp{\kappa'}^{-n'_0} \frac{s^{\tilde{n}_0-n'_0}}{\myp{\tilde{n}_0-n'_0}!} \oint_{\Gamma_{n}} \frac{\abs{\ids z}}{2\pi} \frac{e^{s \myp{\mathrm{Im} z}_+}}{\abs{z}} \prod_{j \in A_1 \cup A_2} \frac{1}{\abs{z-\gamma_j}}. 
	\end{align}

	Now, let us estimate $\abs{\Psi_1 \myp{k_{0;\ell_1}, \sigma_{1,\ell_1}} \Psi_1 \myp{-k_{0;\ell_1'}', \sigma_{1,\ell_1'}'}} \leq \normt{\widehat{V}}{\infty}^2$ to remove the dependence on the two ``amputated'' interaction vertices at the bottom of the trees.
	Supposing that there are any free momenta associated with these two vertices, they will be integrated out next, resulting in a factor of $1$.
	
	Now, let us consider the resolvents $\frac{1}{\abs{z-\gamma_j}}$, $j \in A_1\cup A_2$.
	These vertices depend only on the free momenta associated to edges ending on time slices $j' \geq j$.
	Consider a degree one vertex in the plus tree.
	By \cref{lemma:degreeofafusionvertex}, there is a permutation $\pi$ of $\Set{0,1,2}$ such that $k_{j-1,\ell_j+\pi\myp{1}}$ is free.
	Moreover, both $k_{j-1,\ell_j + \pi \myp{3}}$ and $k_0 \coloneq k_{j-1,\ell_j + \pi \myp{2}} + k_{j-1,\ell_j + \pi \myp{1}}$ are independent of $k_{j-1,\ell_j + \pi \myp{1}}$.
	Then, recalling \cref{Propo:Cutoff} and estimating $\abs{\Psi_1 \myp{k_{j-1;\ell_j},\sigma_{j,\ell_j}}} \leq \normt{\widehat{V}}{\infty} F_1 \myp{k_0}$, we may apply \cref{lemma:degree1vertex} for $k_0 = k_{j-1,\ell_j + \pi \myp{2}} + k_{j-1,\ell_j + \pi \myp{1}}$, $\sigma = 1$, and $\mathrm{d}k = \mathrm{d}k_{j-1,\ell_j + \pi \myp{1}}$.
	For a degree one vertex in the minus tree, the same strategy can be applied with $k_0' = k_{j-1,\ell_j' + \pi' \myp{2}}' + k_{j-1,\ell_j' + \pi' \myp{1}}'$, $\sigma = -1$, and $\mathrm{d}k = \mathrm{d} k_{j-1,\ell_j' + \pi' \myp{1}}'$, by estimating $\abs{\Psi_1 \myp{-k_{j-1;\ell_j'}',\sigma_{j,\ell_j'}'}} \leq \normt{\widehat{V}}{\infty} F_1 \myp{-k_0'}$.
	Consider a degree two vertex, by \cref{lemma:degreeofafusionvertex}, there is a permutation $\pi$ of $\Set{0,1,2}$ such that $k_{j-1,\ell_j + \pi \myp{1}}$, $k_{j-1,\ell_j + \pi \myp{2}}$ are free and   $k_{j-1,\ell_j + \pi \myp{3}} - k_{j-1,\ell_j + \pi \myp{1}} - k_{j-1,\ell_j + \pi \myp{2}}$ is independent of both $k_{j-1,\ell_j + \pi \myp{1}}$, $k_{j-1,\ell_j + \pi \myp{2}}$.
	We can then apply \cref{lemma:degree2vertex} for $k_0 = k_{j-1,\ell_j + \pi \myp{3}}$, $\sigma=1$, and $\mathrm{d}k = \mathrm{d} k_{j-1,\ell_j + \pi \myp{1}}$, $\mathrm{d}k' = \mathrm{d} k_{j-1,\ell_j + \pi \myp{2}}$.
	If this degree two vertex is in the minus tree, $\sigma=-1$.
	For degree zero vertices, we simply use the estimate $\abs{\Psi_1} \leq \normt{\widehat{V}}{\infty}$. 
	
	By \cref{lemma:degree1vertex,lemma:degree2vertex}, we can iteratively evaluate the free momentum integrals in the time direction, from the bottom to the top of the graph.
	At each iteration, there is at most one resolvent factor depending on the corresponding free momenta.
	The remaining free momenta only affect the value of the constant $\alpha$ in \cref{lemma:degree1vertex,lemma:degree2vertex}.
	Using the contour $\Gamma_{n}$, the value for $\beta$ is $\lambda^2$ for the topmost edge of $\Gamma_{n}$.
	The remaining $z$ makes $\abs{z-\gamma_j} \geq 1$ and the upper bound remains valid after changing the constant.
	Thus, from the momentum integrals we gain a contribution $C\lambda^{-b} \expec{\ln \beta}^2$ from each degree one vertex, and $C \expec{\ln \beta}$ from each degree two vertex.

	After iterating through the interaction vertices, nothing remains of the integrand that might depend on any momenta.
	In particular, the resulting bound is uniform in $k \in \T^d$.
	Since we have already evaluated the integral of the free edge of the top fusion vertex, the only integral left is 
	\begin{equation*}
		\oint_{\Gamma_{n}} \frac{\abs{\ids z}}{2\pi} \frac{e^{s \myp{\mathrm{Im}z}_+}}{\abs{z}},
	\end{equation*}
	which is bounded by $C e^{s\lambda^2} \expec{\ln \myp{2n}} \expec{\ln\beta}$.
	Adding up the exponents, the result of the lemma follows.
\end{proof}

\begin{lemma}[Basic graph estimate 2]
\label{lemma:BasicGEstimate2}
	There is a constant $C>0$ such that for $n,n'\ge 0$ and $s>0$, and any $k,k' \in \T^d$, we have
	\begin{align}
	\label{eq:BasicGEstimate2}
		\MoveEqLeft[3] \limsup_{L\to\infty} \lambda^{n+n'} \sum_{\bar{\sigma} \in \Set{\pm 1}^{\mathcal{I}_n'}} \sum_{\bar{\sigma}'\in \Set{\pm 1}^{\mathcal{I}_{n'}'}} \1 \myp{\sigma_{n,1}=1} \1 \myp{\sigma_{n',1}'=-1} \frac{1}{L^d} \int_{\myp{\Omega^{\ast}}^{\mathcal{I}_n'}} \ids \bar{k} \int_{\myp{\Omega^{\ast}}^{\mathcal{I}_{n'}'}} \ids \bar{k}' \nn \\
		&\times \Delta_{n,\ell} \myp{\bar{k},\bar{\sigma}} \Delta_{n',\ell'} \myp{\bar{k}',\bar{\sigma}'} \delta_L \myp{k_{n,1} - \myb{k}} \delta_L \myp{k_{n',1}' - \myb{k'}} \prod_{A \in S} \delta_L \myp[\bigg]{ \sum_{j\in A} K_j } \nn \\
		&\times \prod_{j=1}^n \abs[\big]{ \Psi_1 \myp{k_{j-1:\ell_j}, \sigma_{j,\ell_j}} } \prod_{j=1}^{n'} \abs[\big]{ \Psi_1 \myp{-k'_{j-1:\ell_j'}, \sigma_{j,\ell_j'}'} } \nn \\
		&\times \abs[\bigg]{ \int_{\myp{\R_+}^{I_{0,n+n'}}} \ids r \delta \myp[\bigg]{ s-\sum_{j=0}^{n+n'} r_j } \prod_{j=0}^{n+n'} e^{-ir_j \gamma \myp{j;J}} } \nn \\
		\leq{}& \1 \myp{k+k'=0} C^{1 + n_1 + n_2} \normt{\widehat{V}}{\infty}^{n+n'} e^{s\lambda^2} \frac{\myp{s\lambda^2}^{n_0}}{n_0!} \lambda^{n_2 + \myp{1-b} n_1 - n_0} \nn\\
		&\times \expec{\ln \myp{n+n'+1}} \expec{\ln\lambda}^{1 + n_2 + 2n_1},
	\end{align}
	where $n_j$ is the number of interaction vertices of degree $j$.
\end{lemma}
\begin{proof}
	The proof follows exactly the same lines of arguments with that of \cref{lemma:BasicGEstimate1}, except there are no amputated vertices, and we instead set $B = \emptyset$.
	
	Additionally, when evaluating the $k_{n',1}'$-integral to remove the extra $\delta_L$-functions, we get as in \cref{rem:additional_deltas},
	\begin{equation*}
		\delta_L \myp{k_{n,1} - \myb{k}} \delta_L \myp{k_{n',1}' - \myb{k'}}
		={} \delta_L \myp{\myb{k}+\myb{k'}} \delta_L \myp{k_{n',1}' - \myb{k'}},
	\end{equation*}
	where we can use the first factor to absorb the prefactor $L^{-d}$ in \eqref{eq:BasicGEstimate2} and get in the $L \to \infty$ limit,
	\begin{equation*}
		\lim_{L \to \infty} L^{-d} \delta_L \myp{\myb{k}+\myb{k'}}
		={} \1 \myp{k+k' = 0}.
	\end{equation*}
\end{proof}
\begin{lemma}[Basic graph estimate 3]
\label{lemma:BasicGEstimate3}
	If we change the term 
	\begin{equation*}
		\Psi_1 \myp{k_{0:\ell_1}, \sigma_{1,\ell_1}} \Psi_1 \myp{-k'_{0:\ell_1'}, \sigma_{1,\ell_1'}}
	\end{equation*}
	in \eqref{eq:BasicGEstimate1} to 
	\begin{equation*}
		\Psi_0 \myp{k_{0:\ell_1}, \sigma_{1,\ell_1}} \Psi_0 \myp{-k'_{0:\ell_1'}, \sigma_{1,\ell_1'}},
	\end{equation*}
	the estimate can be improved by a factor of $C \expec{\ln\lambda} \lambda^{bd-2}$ or the corresponding $\mathcal{F}_n^2$ is $0$.
	If $\tilde{n}_1=0$, then the factor can be improved by $C \expec{\ln\lambda} \lambda^{b\myp{d-2}}$.
	If any of the amputated vertices has degree $2$, then the factor can be improved by $C \expec{\ln\lambda} \lambda^{b\myp{d-1}}$. 
\end{lemma}
\begin{proof}
	Let us consider the case in which one of the amputated vertices has degree $2$.
	If this is a minus vertex, then by \cref{Propo:Cutoff} and the fact that $\widehat{V}$ is bounded,
	\begin{equation}
	\label{lemma:BasicGEstimate3:E1}
		\Psi_0 \myp{-k'_{0:\ell_1'}} 
		\lesssim \sum_{e_1,e_2 \in \mathscr{E}_- \myp{v_1}, e_1<e_2} \1 \myp[\big]{ d \myp{- \myp{k_{e_1}+k_{e_2}}, M^{\mathrm{sing}} } \leq \lambda^b },
	\end{equation}
	where $v_1$ is the vertex that all the momenta $k'_{0:\ell_1'}$ are attached to.
	It follows from \cref{lemma:freeedgesorder2} that $-k_{e_1}-k_{e_2}$ depends on the two free momenta of $v_1,$ denoted by $\bar{k},\bar{k}'$, in the form of $\bar{k},\bar{k}'$, or $-\bar{k}-\bar{k}'$.
	When we integrate first over $\bar{k}$, then over $\bar{k}'$, \cref{Propo:Cutoff} can be applied, leading to 
	\begin{equation}
	\label{lemma:BasicGEstimate3:E2}
		\sup_{k_0 \in \T^d} \int_{\T^d} \ids k \1 \myp[\big]{ d \myp{\pm k +k_0,M^{{\mathrm{sing}}}} < \lambda^b }
		\lesssim \lambda^{b(d-1)}. 
	\end{equation}
	Now, we can estimate $\Psi_0 \myp{k_{0:\ell_1}} \lesssim 1$ and the rest of the proof follows the same lines of arguments used in the proof of \cref{lemma:BasicGEstimate1}. 
	
	If the vertex does not have degree two, we can just estimate $\Psi_0 \myp{-k'_{0:\ell_1'}} \lesssim 1$, with the notice that $\widehat{V}$ is bounded, and integrate out on any free momenta attached to it.
	The extra factor in the amputated plus vertex, denoted by $v_2$, can be estimated as
	\begin{equation}
	\label{lemma:BasicGEstimate3:E3}
		\Psi_0 \myp{k_{0:\ell_1}} 
		\lesssim \sum_{e_1,e_2 \in \mathscr{E}_- \myp{v_2}, e_1 < e_2} \1 \myp[\big]{ d \myp{k_{e_1}+k_{e_2}, M^{\mathrm{sing}}} \leq \lambda^b }.
	\end{equation}
	If $v_2$ has degree two, we gain a factor $\lambda^{b \myp{d-1}}$ and the proof follows the same lines of arguments used in the proof of \cref{lemma:BasicGEstimate1}. 
	
	To prove the other claims in the Lemma, we only need to consider the term associated to some fixed pair $e_1 < e_2$ with $e_1,e_2 \in \mathscr{E}_- \myp{v_2}$.
	If $k_{e_1}+k_{e_2}$ is independent of all free momenta, then by \cref{lemma:dependenceoftwoedges}, we conclude that the two initial time vertices belonging to $e_1 \cup e_2$ are isolated from the rest of the initial time vertices, so they must be paired.
	It then follows by \cref{Propo:QERR2} that the graph is irrelevant, so we may therefore assume that $k_{e_1}+k_{e_2}$ depends on some free momenta.
	Among all of those free edges, we denote by $e_3$ the one that is added last in the creation algorithm of the momentum graph (i.e. it is the maximum in the edge ordering).
	Hence, $e_3$ is the first of these free edges we will encounter when iterating through the graph from the bottom to the top.
	Now, we carry on the same same lines of arguments used in the proof of \cref{lemma:BasicGEstimate1}, with the following exception: If the fusion vertex that $e_3$ is attached to is of degree two, we will need the corresponding $\Psi_1$ factor, which is kept unchanged.
	The interaction vertices will be integrated until this fusion vertex is reached.
	There are a few possibilities.
	We note first that by \cref{lemma:odd_dependence}, this fusion vertex cannot be the top fusion vertex of the graph, as $k_{e_1}+k_{e_2}$ depends on an even number of free momenta.
	If the fusion vertex of $e_3$ is $v_2$ itself, we use \eqref{lemma:BasicGEstimate3:E2} to gain a factor $\lambda^{b \myp{d-1}}$.
	If this vertex is a degree one non-amputated interaction vertex, we can remove the resolvent factor using a trivial $L^\infty$ estimate and apply \eqref{lemma:BasicGEstimate3:E2}.
	Since this degree one vertex was previously estimated by a power $\lambda^{-b} \expec{\ln \beta}^2 $ in the proof of \cref{lemma:BasicGEstimate1}, we thus gain a factor $\expec{\ln \beta}^2 \lambda^{b-2+b \myp{d-1}} = \expec{\ln \beta}^2 \lambda^{bd-2}$ compared to that estimate.
	In the last case, $e_3$ is attached to a non-amputated interaction vertex of degree two.
	We denote the two free momenta by $k_1$ and $k_2$, the third integrated one by $k_3$ and $k_0=k_1+k_2+k_3$.
	Then $k_{e_1}+k_{e_2} = \pm k_j+k'_0$ for some $j \in \Set{1,2,3}$ and $k'_0$ is independent of $k_1,k_2$.
	We then need to estimate
	\begin{align}
	\label{lemma:BasicGEstimate3:E4}
		\MoveEqLeft[4] \iint_{\myp{\T^d}^2} \ids k_1 \ids k_2 \1 \myp[\big]{ d \myp{\pm k_j +k_0', M^{\mathrm{sing}}} < \lambda^b } \nn \\
		&\times \frac{\abs{\Psi_1 \myp{\pm \myp{ k_1, k_2, k_0-k_1-k_2}}}}{\abs{ \omega^{\lambda} \myp{k_1} + \sigma'\omega^{\lambda} \myp{k_2} + \sigma \omega^{\lambda} \myp{k_0-k_1-k_2}-\alpha-i\beta }}.
	\end{align}
	If $j=1$, then $\abs{\Psi_1} \lesssim \normt{\widehat{V}}{\infty} F_1 \myp{\pm \myp{k_0-k_1}}$.
	We then change the integration variable $k_1$ to $k = \pm k_1+k_0'$.
	\cref{lemma:degree1vertex} can be applied to the integral of $k_2$ and estimate \eqref{lemma:BasicGEstimate3:E2} can be applied to the $k$ integral.
	This process gives a factor $\expec{ \ln \beta}^2 \lambda^{-b+b\myp{d-1}}$.
	In case $j=2$ or $3$, similar changes of variables can be applied to obtain the same contribution in $\lambda$.
	Therefore, we get an improvement of $ \expec{\ln \beta}^2 \lambda^{b \myp{d-2}}$ in this case.
	After this, we can finish estimating the amplitude by following the same lines of arguments used in the proof of \cref{lemma:BasicGEstimate1}.
	
	If $\tilde{n}_1=0$, then there is no non-amputated degree one vertex, so the worst contribution $\lambda^{bd - 2}$ does not occur.
	This finally gives $ \expec{\ln \beta}^2 \lambda^{b\myp{d-2}}$, and the claims of the Lemma are proved.
\end{proof}

\subsection{Estimating \texorpdfstring{$Q_3^{err}$}{Q3err} and \texorpdfstring{$Q_2^{err}$}{Q2err}}
The graph estimates provided in the previous section suffice to prove that the constructive interference term $Q_2^{err}$ and the amputated term $Q_3^{err}$ from \eqref{Def:OperatorQLambda2} do not contribute in the weak coupling limit.
\begin{proposition}
\label{Lemma:Q3err}
	Suppose that $t>0$ and $0 < \lambda < \lambda_0'$ and define $N_0$ and $\kappa$ as before.
	There are constants $C,c>0$ depending on $\omega$ and $V$ such that for any $k,k' \in \T^d$,
	\begin{align}
		\MoveEqLeft[6] \limsup_{L\to\infty} \frac{1}{L^{2d}} \abs[\big]{Q_3^{err} \myp{\myb{k'},\myb{k},t}}^2 \nn \\
		\leq{}& C t^2 e^t \expec{ct}^{N_0} N_0^{2N_0+2} \myp{4N_0}! \expec{\ln \lambda}^{4N_0+1} \myp[\big]{\lambda + \lambda^{-2} N_0^{-b_0 N_0/4}}
	\label{Lemma:Q3err:1}
	\end{align}
	when $N_0 \myp{\lambda} \geq 44$.
	In particular, in view of \eqref{eq:lambda_expression}, it follows that
	\begin{equation*}
		\lim_{\lambda \to 0} \sup_{k,k'\in \T^d} \limsup_{L \to \infty} \frac{1}{L^d} \abs{Q_3^{err}}
		={} 0.
	\end{equation*}
\end{proposition}
\begin{proof}
	We first  observe  that
	\begin{align}
		\label{Lemma:Q3err:E1}
		\MoveEqLeft[2] \limsup_{L\to\infty} \frac{1}{L^{2d}} \abs[\big]{Q_3^{err} \myp{\myb{k'},\myb{k},t}}^2 
		\leq{} \limsup_{L\to\infty} \bar{Q}_3^{err} \myp{\myb{k},t}^2 \nn \\
		\leq{}& C t^2 \lambda^{-4} \nn \\
		&\times \sup_{0 \leq s \leq t\lambda^{-2}} \sum_{\substack{J \text{ interlaces } \\ \myp{N_0-1,N_0-1}}} \sum_{\ell,\ell' \in G_{N_0}} \sum_{S \in \pi \myp{I_{4N_0+2}}} \limsup_{L\to\infty} \frac{1}{L^d} \abs[\big]{\mathcal{F}^3_{N_0} \myp{S,J,\ell,\ell',\myb{k},s,\kappa }},
	\end{align}
	with $\mathcal{F}^3_n$ as in \eqref{eq:DefF3}.
	By \cref{lemma:BasicGEstimate1}, the $\limsup$ can be bounded by
	\begin{align}
	\label{Lemma:Q3err:E2}
		\MoveEqLeft[10] \sum_{\substack{J \text{ interlaces } \\ \myp{N_0-1,N_0-1}}} \sum_{\ell, \ell' \in G_{N_0}} \sum_{S \in \pi \myp{I_{4N_0+2}}} \prod_{A \in S} \sup_{k_A,\sigma_A,\Omega} \abs{{C}_{\abs{A}} \myp{k_A,\sigma_A,\lambda,\Omega)}} \nn \\
		&\times e^t \expec{ct}^{N_0} \expec{\ln \lambda}^{1+4N_0} \lambda^{2+\tilde{n}_2 - \tilde{n}_0 + \myp{1-b} \tilde{n}_1} N_0^{-b_0 n'_0},
	\end{align}
	where the notations of $\tilde{n}_0, \tilde{n}_1, \tilde{n}_2, n_0'$ are the same as in \cref{lemma:BasicGEstimate1}.
	We set $r = 2N_0 + 1 - \abs{S}$ and note that $\tilde{n}_j = n_j-n_j \myp{2} \geq 0$, $n_0' = n_0 \myp{N_0+1} - n_0 \myp{2}$.
	It then follows from \cref{lemma:numberofnonpairclusters} that $\tilde{n}_2 - \tilde{n}_0 = r + n_0 \myp{2} - n_2 \myp{2}$.
	Using that $n_0 \myp{2} + n_1 \myp{2} + n_2 \myp{2} = 2$, we obtain in the case where $r + n_1 \geq 20$ that $\tilde{n}_2 - \tilde{n}_0 + \myp{1-b} \tilde{n}_1 \geq r + \frac{1}{4} n_1 - 2 \geq 3$, and hence $\lambda^{2+\tilde{n}_2 - \tilde{n}_0 + \myp{1-b} \tilde{n}_1} N_0^{-b_0n'_0} \leq \lambda^5$.
	
	On the other hand, if $r+n_1 \leq 19,$ then by \cref{lemma:n2jn0j}, $n'_0 \geq \frac{N_0+1 - \myp{n_1 + r}}{2} -2 \geq \frac{N_0}{2}-11$, which is at least $\frac{N_0}{4}$ when $N_0 \geq 44$.
	It follows in this case that $\lambda^{2+\tilde{n}_2 - \tilde{n}_0 + \myp{1-b} \tilde{n}_1} N_0^{-b_0n'_0} \leq \lambda^2 N_0^{-b_0N_0/4}$.
	By being a bit more careful with estimating the exponents, it is possible to push the bound on $N_0$ further down, but this is not needed.
	
	Now, we know that the number of terms in the sum over $J$ is smaller than $2^{2N_0-2}$ and the number of terms in the sum over $\ell$ is smaller than $\myp{2N_0}^{N_0}$.
	Using \cref{lemma:Clustercombinatorics}, and collecting all the terms, we get the conclusion of the lemma after adjusting the constants $c$ and $C$. 
\end{proof}
\begin{proposition}\label{Lemma:Q2err}
	Suppose that $t>0$  and $0 < \lambda < \lambda_0'$ and define $N_0$ and $\kappa$ as before.
	There are constants $C,c>0$ depending on $\omega$ and $V$ such that for any $k,k'\in\T^d$,
	\begin{equation}
		\limsup_{L\to\infty} \frac{1}{L^{2d}} \abs[\big]{Q_2^{err} \myp{\myb{k'},\myb{k},t}}^2
		\leq{} Ct^2 e^t \expec{ct}^{N_0} N_0^{2N_0+4} \myp{4N_0}! \expec{\ln \lambda}^{4N_0+2} \lambda^{\frac{1}{4}}.
	\end{equation}
	In particular, in view of \eqref{eq:lambda_expression}, it follows that
	\begin{equation*}
		\lim_{\lambda \to 0} \sup_{k,k'\in \T^d} \limsup_{L \to \infty} \frac{1}{L^d} \abs{Q_2^{err}}
		={} 0.
	\end{equation*}
\end{proposition}
\begin{proof}
	The same argument used to prove \cref{Lemma:Q3err} can be reused to obtain the following estimate for $Q_2^{err}$.
	\begin{align}
		\MoveEqLeft[3] \limsup_{L\to\infty} \frac{1}{L^{2d}} \abs[\big]{Q_2^{err} \myp{\myb{k'},\myb{k},t}}^2 
		\leq{} \limsup_{N\to \infty} \bar{Q}_2^{err} \myp{\myb{k},t}^2 \nn \\
		\leq{}& C N_0 t^2 \lambda^{-4} \nn \\
		&\times \sup_{0\leq s\leq t\lambda^{-2}} \sum_{n=1}^{N_0} \sum_{\substack{J \text{ interlaces } \\ \myp{n-1,n-1}}} \sum_{\ell,\ell' \in G_n} \sum_{S \in \pi \myp{I_{4n+2}}} \limsup_{L\to\infty} \frac{1}{L^d} \abs[\big]{\mathcal{F}^2_n \myp{S,J,\ell,\ell',\myb{k},s,\kappa}},
	\label{Lemma:Q2err:E1}
	\end{align}
	with $\mathcal{F}^2_n$ as in \eqref{eq:DefF2}, and then the $\sup$ term can be bounded, using \cref{lemma:BasicGEstimate3}, as
	\begin{align}
	\label{Lemma:Q2err:E2}
		\MoveEqLeft[6] \sum_{n=1}^{N_0} \sum_{\substack{J \text{ interlaces } \\ \myp{n-1,n-1}}} \sum_{\ell,\ell' \in G_n} \sum_{S \in \pi \myp{I_{4n+2}}} \prod_{A \in S} \sup_{k_A,\sigma_A,\Omega} \abs{C_{\abs{A}} \myp{k_A,\sigma_A,\lambda,\Omega}} \nn \\
		&\times e^t \expec{ct}^n \expec{\ln \lambda}^{2+4n} \lambda^{2+\tilde{n}_2 - \tilde{n}_0 + \myp{1-b} \tilde{n}_1 + db} \1 \myp[\big]{\mathcal{F}^2_n \neq 0} \mathfrak{G},
	\end{align}
	where $\mathfrak{G}=\lambda^{-b}$ when $n_2 \myp{2}>0$, $\mathfrak{G}=\lambda^{-2b}$ when $n_2 \myp{2}=0, \tilde{n}_1=0$, and $\mathfrak{G}=\lambda^{-2}$ otherwise. 
	
	By \cref{lemma:n2jn0j}, we always have $\tilde{n}_2 - \tilde{n}_0 = r -n_2 \myp{2} + n_0 \myp{2} \geq 0$.
	Hence, if $n_2 \myp{2} > 0$, the power of $\lambda$ in the expression above is bounded by $\lambda^{2+b\myp{d-1}} \leq \lambda^{4+\frac{1}{4}}$ in dimension $d \geq 4$.
	
	On the other hand, if $n_2 \myp{2} = 0$, then either $n_0 \myp{2} \geq 1$ (in which case $\tilde{n}_2 - \tilde{n}_0 \geq 1$) or $n_1 \myp{2}=2$.
	In the second case, there is exactly one free momentum attached to the bottom interaction vertex, the amputated vertex in the minus tree.
	It follows by the iterative cluster scheme that there is a cluster $A \in S$ that is connected to exactly two edges of this vertex.
	If $A$ is a pairing, then the graph is irrelevant by \cref{Propo:QERR2}, so we may assume that $\abs{A} \geq 4$, which in turn implies by \cref{lemma:numberofnonpairclusters} that $r \geq 1$.
	We conclude in any case that $\tilde{n}_2 - \tilde{n}_0 \geq 1$.
	
	Thus, if $n_2 \myp{2} = 0$ and $\tilde{n}_1 = 0$, then the power of $\lambda$ in \eqref{Lemma:Q2err:E2} is bounded by $\lambda^{2+1+bd-2b} \leq \lambda^{4+\frac{1}{2}} \leq \lambda^{4+\frac{1}{4}}$.
	In the final case where $n_2 \myp{2} = 0$ and $\tilde{n}_1 \geq 1$, we obtain $\lambda^{2+1+1-b+db-2} = \lambda^{2+b \myp{d-1}} \leq \lambda^{4+\frac{1}{4}}$.
	The conclusion of the proposition now follows by estimating the sums in \eqref{Lemma:Q2err:E2} as in the proof of \cref{Lemma:Q3err}.
\end{proof}

\subsection{Estimating \texorpdfstring{$Q_1^{err}$}{Q1err} and \texorpdfstring{$Q^{main}$}{Qmain}}
The following two results show that for the remaining terms in the Duhamel expansion, we can bound the amplitudes for all the graphs that are not fully paired.
\begin{proposition}\label{Lemma:Q1err}
	Suppose that $t>0$ and $0<\lambda<\lambda_0'$ and define $N_0$ and $\kappa$ as before.
	There are constants $C,c>0$ depending on $\omega$, $V$, and $\lambda_0'$ such that for any $k,k'\in T^d$,
	\begin{align}
	\label{Lemma:Q1err:1}
		\MoveEqLeft[3] \limsup_{L\to\infty} \frac{1}{L^{2d}} \abs{Q_1^{err} \myp{\myb{k'},\myb{k},t}}^2 \nn \\
		\leq{}& Ct^2 e^t \expec{ct}^{N_0} N_0^{2N_0+5+2b_0} \myp{4N_0}! \expec{\ln\lambda}^{4N_0+2} \lambda^{\frac{1}{4}} + Ct^2 N_0^{2+2b_0} \nn \\
		&\times \sup_{\substack{0\leq s\leq t\lambda^{-2} \\ N_0/2 \leq n<N_0} } \sum_{\substack{J \text{ interlaces } \\ \myp{n,n}}} \sum_{\ell,\ell' \in G_n} \sum_{S \in \pi \myp{I_{4n+2}}} \abs{\mathcal{F}^{err,pair}_n \myp{S,J, \ell,\ell',\myb{k},s,\kappa}},
	\end{align}
	where $\mathcal{F}^{err,pair}_n \myp{S,J, \ell,\ell',\myb{k},s,\kappa} = 0$ if the graph created by $S,J, \ell,\ell'$ is not fully paired, otherwise it is defined as
	\begin{align}
	\label{Lemma:Q1err:2}
		\MoveEqLeft[4] \mathcal{F}^{err,pair}_n \myp{S,J, \ell,\ell',\myb{k},s,\kappa} \nn \\
		={}& \myp{-\lambda^2}^n \epsilon \myp{S} \sum_{\bar{\sigma}, \bar{\sigma}' \in \Set{\pm 1}^{\mathcal{I}_n'}} \int_{\myp{\T^d}^{\mathcal{I}_n'}} \int_{\myp{\T^d}^{\mathcal{I}_n'}} \ids \bar{k} \ids \bar{k}' \Delta_{n,\ell} \myp{\bar{k},\bar{\sigma}} \Delta_{n,\ell'} \myp{\bar{k}',\bar{\sigma}'} \nn \\
		&\times \delta \myp{k_{n,1}-k} \delta \myp{k_{n,1}'+k} \1 \myp{\sigma_{n,1}=1} \1 \myp{\sigma_{n,1}'=-1}  \nn \\
		&\times \prod_{j=1}^n \myb[\big]{ \sigma_{j,\ell_j} \Psi_1 \myp{k_{j-1:\ell_j}, \sigma_{j,\ell_j}} \sigma_{j,\ell_j'}' \Psi_1 \myp{-k_{j-1:\ell_j'}', -\sigma_{j,\ell_j'}'} } \nn \\
		&\times \prod_{A=\Set{j,l} \in S} \myb[\big]{ \delta \myp{K_j + K_l} \1 \myp{o_j = - o_l} W^0 \myp{K_j, o_j} } \nn \\
		&\times \int_{\myp{\R_+}^{I_{0,2n}}} \ids r \delta \myp[\bigg]{s-\sum_{j=0}^{2n} r_j } \prod_{j=0}^{2n} e^{-ir_j\gamma \myp{j;J}},
	\end{align}
	where we have reused the notation $K = \myp{k_{0,\nonarg}',k_{0,\nonarg}}$ and $o = \myp{\sigma_{0,\nonarg}', \sigma_{0,\nonarg}}$ from \cref{Propo:QERR3}, as well as the notation from \eqref{eq:Wnotations}.
	$W^0$ is given by \eqref{eq:wignerfunction0}.
	For the pairings $A = \Set{j,l} \in S$, we use the convention $j < l$, so that $o_j$ denotes the parity of the left leg of the pairing.
\end{proposition}
\begin{proof}
	The same argument used to prove \cref{Lemma:Q3err} can be reused to obtain the following estimate for $\bar{Q}_1^{err}$,
	\begin{align}
		\label{Lemma:Q1err:E1}
		\MoveEqLeft[3] \limsup_{L\to\infty} \frac{1}{L^{2d}} \abs{Q_1^{err} \myp{\myb{k'},\myb{k},t}}^2
		\leq{} \limsup_{L\to\infty} \bar{Q}_1^{err} \myp{\myb{k},t}^2 \nn \\
		\leq{}& CN_0^2 \kappa' \myp{\lambda}^2 t^2 \lambda^{-4} \nn \\
		&\times \sup_{\substack{0 \leq s \leq t\lambda^{-2} \\ N_0/2 \leq n < N_0}} \sum_{\substack{J \text{ interlaces } \\ \myp{n,n}}} \sum_{\ell,\ell' \in G_n} \sum_{S \in \pi \myp{I_{4n+2}}} \limsup_{L\to\infty} \frac{1}{L^d} \abs{\mathcal{F}^1_n \myp{S,J,\ell,\ell',\myb{k},s,\kappa}},
	\end{align}
	with $\mathcal{F}^1_n$ as in \eqref{eq:DefF1}.
	Recall here that by definition \eqref{Def:Para3}, $\kappa' \myp{\lambda} = \lambda^2 N_0^{b_0}$ cancels the negative power of $\lambda$ above.
	For any graph that is not fully paired, we take the absolute value up to the time-integration and apply \cref{lemma:BasicGEstimate2}.
	For these terms, we can thus bound the $\limsup$ in \eqref{Lemma:Q1err:E1} by
	\begin{equation}
	\label{Lemma:Q1err:E2}
		e^t \frac{t^{n_0}}{n_0!} \expec{\ln \lambda}^{2+4n} \lambda^{r + \myp{1-b} n_1} C^{2n+1} \prod_{A \in S} \sup_{k_A,\sigma_A,\Omega} \abs{C_{\abs{A}} \myp{k_A,\sigma_A,\lambda,\Omega}}.
	\end{equation}
	If the graph is pairing but not fully paired, $n_1 \geq 1$.
	If the graph is higher order, $r \geq 1$.
	For both types, \eqref{Lemma:Q1err:E2} is bounded by
	\begin{equation}
	\label{Lemma:Q1err:E3}
		e^t \expec{ct}^n \expec{\ln \lambda}^{2+4n} \lambda^{\frac{1}{4}} \prod_{A \in S} \sup_{k_A,\sigma_A,\Omega} \abs{C_{\abs{A}} \myp{k_A,\sigma_A,\lambda,\Omega}}.
	\end{equation}
	
	Let us now consider a fully paired graph.
	In this case, the clusters are pairings with $C_2 \myp{\myp{k',k}, \myp{\sigma',\sigma}, \lambda, \Omega} = W_L^{\lambda} \myp{k,\sigma'} \1 \myp{\sigma' = -\sigma}$ by \cref{rem:pair_correl_trunc}.
	Since by \cref{lem:wigner_estimate},
	\begin{equation}
	\label{Lemma:Q1err:E4}
		\limsup_{L\to\infty} \sup_{k \in \T^d} \abs{W_L^{\lambda} \myp{\myb{k}} - W^0 \myp{k}} 
		\leq{} 2c_0^2 \lambda,
		\mbox{ and } W^0 \myp{k} = W^0 \myp{-k},
	\end{equation}
	for any finite index set $I$ it holds true by induction that
	\begin{equation}
	\label{Lemma:Q1err:E5}
		\limsup_{L\to\infty} \abs[\bigg]{\prod_{j \in I} W_L^{\lambda} \myp{\myb{\pm k_j}, \sigma_j} - \prod_{j \in I} W^0 \myp{k_j,\sigma_j}}
		\leq{} \abs{I} C^{\abs{I}-1} 2c_0^2 \lambda,
	\end{equation}
	with $C = 2c_0^2 \lambda_0' + \normt{W^0}{\infty} < \infty$. 
	For pairing $S$, we have $\abs{S} = 2n+1$ by \cref{lemma:BasicGEstimate2}
	Hence, we can replace in the definition of $\mathcal{F}^1_n$ all $C_2$ terms by $W^0 \myp{K_j,o_j} \1 \myp[\big]{ o_j = -o_l}$ with an error bounded by
	\begin{equation}
	\label{Lemma:Q1err:E6}
		Cn e^t \expec{ct}^n \expec{\ln\lambda}^{2+4n} \lambda.
	\end{equation}
	After making this replacement, we resolve all the $\delta_L$-functions in the resulting expression and rewrite the integrals over the free momenta as Lebesgue integrals as in \eqref{eq:discrete-lebesque}.
	As explained in the beginning of the proof of \cref{lemma:BasicGEstimate1}, we can then apply the dominated convergence theorem to take the $L \to \infty$ limit and obtain an expression of the form \eqref{Lemma:Q1err:2}.
	
	Finally, collecting all the terms and estimating the sums over $J,\ell,\ell'$, and $S$ as before, the conclusion of the proposition holds true.
\end{proof}
By exactly the same argument, we can also prove.
\begin{proposition}\label{Lemma:Qmain}
	Suppose that $t>0$  and $0<\lambda<\lambda_0'$ and define $N_0$ and $\kappa$ as  before.
	There  are constants $C,c>0$ depending on $\omega$, $V$, and $\lambda_0'$ such that for any $k,k'\in \T^d$,
	\begin{align}
	\label{Lemma:Qmain:1}
		\MoveEqLeft[8] \limsup_{L\to\infty} \abs[\Big]{\frac{1}{L^d} Q^{main} \myp{\myb{k'},\myb{k},t} - Q^{main,pair} \myp{k',k,t}} \nn \\
		\leq{}& Ce^t \expec{ct}^{N_0} N_0^{N_0+4} \myp{2N_0}! \expec{\ln\lambda}^{2N_0+2} \lambda^{\frac{1}{4}},
	\end{align}
	where 
	\begin{equation}
	\label{Lemma:Qmain:2}
		Q^{main,pair} \myp{k',k,t} 
		={} \sum_{n=0}^{N_0-1} \sum_{\ell \in G_n} \sum_{S \in \pi \myp{I_{0,2n+1}}} \mathcal{F}^{main,pair}_n \myp{S,\ell,k',k,t/\eps,\kappa}
	\end{equation}
	in which $\mathcal{F}^{main,pair}_n \myp{S,\ell,k',k,t/\eps,\kappa}=0$ if the graph created by $S,\ell$ is not fully paired, otherwise, it is defined as
	\begin{align}
	\label{Lemma:Qmain:3}
		\MoveEqLeft[4] \mathcal{F}^{main,pair}_n \myp{S,\ell,k',k,t/\eps,\kappa} \nn \\
		={}& \myp{-i\lambda}^n \epsilon \myp{S} \sum_{\bar{\sigma} \in \Set{\pm1}^{\mathcal{I}_n''}} \int_{\myp{\T^d}^{\mathcal{I}_n''}} \ids \bar{k} \Delta_{n,\ell} \myp{\bar{k}, \bar{\sigma}} \delta \myp{k_{n,1}-k} \delta \myp{k_{0,0}-k'} \nn \\
		&\times \1 \myp{\sigma_{n,1}=1} \1 \myp{\sigma_{0,0}=-1} \prod_{j=1}^n \myb[\big]{ \sigma_{j,\ell_j} \Psi_1 \myp{k_{j-1:\ell_j}, \sigma_{j,\ell_j}} } \nn \\
		&\times \prod_{A = \Set{j,l} \in S} \myb[\big]{ \delta \myp{k_{0,l}+k_{0,j}} \1 \myp{\sigma_{0,j} = -\sigma_{0,l}} W^0 \myp{k_{0,j},\sigma_{0,j}} } \nn \\
		&\times \int_{\myp{\R_+}^{I_{0,n}}} \mathrm{d}r \delta \myp[\bigg]{ \frac{t}{\eps} - \sum_{j=0}^{n} r_j } \prod_{j=0}^{n} e^{-ir_j \gamma_j},
	\end{align}
	with $\mathcal{I}_n'' = \mathcal{I}_n \cup \Set{\myp{n,1}} \cup \Set{\myp{0,0}}$ as in \cref{Proposition:Ampl}, and $\gamma_j$ defined by $\eqref{eq:gamma_j}$.
	We also use the notation from \eqref{eq:Wnotations} with $W^0$ as in \eqref{eq:wignerfunction0}.
	For the pairings $A = \Set{j,l} \in S$, we use the convention $j < l$, so that $\sigma_{0,j}$ denotes the parity of the left leg of the pairing.
\end{proposition}

\subsection{Crossing graphs}
All that remains now is to handle the contributions of fully paired graphs to $Q_1^{err}$ and $Q^{main}$.
According to the results in \cref{sec:FullyPairedGraphs}, these consist exactly of the crossing, nested, and leading graphs.
We treat first the crossing graphs here.
\begin{proposition}[Crossing graphs]
	\label{Propo:CrossingGraphs} There is a constant $c_0$ depending on $\omega$ and $V$, and a constant $C$ depending on $\omega$, such that the following estimates hold true for all $k,k' \in \T^d$ and all crossing graphs,
	\begin{align}
		\abs[\big]{\mathcal{F}^{err,pair}_n \myp{S,J,\ell,\ell',k,s,\kappa}} 
		\leq{}& C\lambda^{2\gamma} e^{s\lambda^2} \expec{c_0 s \lambda^2}^{n-1} \expec{\ln\lambda}^{3+c_2+2n}, \\
		\abs[\big]{\mathcal{F}^{main,pair}_n \myp{S,\ell,k',k,t\lambda^{-2},\kappa}}
		\leq{}& C \1 \myp{k+k'=0} \lambda^{2\gamma} e^t \expec{c_0t}^{n/2-1} \expec{\ln\lambda}^{3+c_2+n},
	\end{align}
	where the constants $\gamma, c_2$ appear in \cref{ass:DR4}.
\end{proposition}
\begin{proof}
	We consider a general relevant crossing graph with $N = n' + n$ interaction vertices.
	As in \cref{lem:CrossingGraphs}, let $v_{i_2}$ be the first degree two vertex that does not correspond to an immediate recollision, and let $m' = i_0-1$ be the index of the last propagated crossing slice.
	We denote the momenta of the edges in $\mathscr{E} \myp{v_{i_2}}$ by $k_1,k_2,k_3$, and $k_0 = k_1+k_2+k_3$, with $k_1$, $k_2$ being the free momenta.
	
	We follow the general strategy of \cref{lemma:BasicGEstimate1} with a few exceptions, by first rewriting the time simplex integral. 
	We have now $A_1 = \emptyset$, and we take $A = \Set{m',N} \cup A_2$, that is, we bundle the long time slice $m' =i_0-1$ together with the short time slices in order to get an extra resolvent factor where we can then apply the results of \cref{lem:CrossingGraphs}.
	The point is that instead of getting the usual long time slice contribution of $\lambda^{-2}$ for the last propagated crossing slice, we can use the crossing estimates of \cref{ass:DR4} to get a contribution of $\lambda^{2\gamma - 2}$, which gives an improvement of $\lambda^{2 \gamma}$ over the basic estimates.
	Observing that $\abs{A_2} = \abs{A_0} = N/2$, we find
	\begin{align}
		\label{eq:crossingraph4}
		\MoveEqLeft[6] \abs[\bigg]{\int_{\myp{\R_+}^{I_{0,N}}} \ids r \delta \myp[\bigg]{s-\sum_{j=0}^N r_j} \prod_{j=0}^N e^{-ir_j\gamma \myp{j}} } \nn \\
		\leq{}& \frac{s^{N/2}-1}{\myp{N/2-1}!} \oint_{\Gamma_{N/2}} \frac{\abs{\ids z}}{2\pi} \frac{e^{s \myp{\mathrm{Im} z}_+}}{\abs{z}} \prod_{j \in \Set{m'} \cup A_2} \frac{1}{\abs{z-\gamma \myp{j}}}.
	\end{align}
	By \cref{lem:CrossingGraphs}, $\mathrm{Re} \gamma \myp{m'} = p \Theta_{i_2} + \Theta_{i_0} + \alpha_1 + \mathrm{Re} \gamma \myp{i_2}$ for a $p \in \Set{0,1}$, where $\alpha_1$, $\Theta_{i_0}$, and $\mathrm{Re} \gamma \myp{i_2}$ are independent of all free momenta ending before or at $v_{i_2}$.
	We can thus follow the basic iteration process until reaching the time slice $i_2-1$ that ends at the crossing vertex $v_{i_2}$.
	At this stage, the difficulty is the dependence of the factor $\frac{1}{z-\gamma \myp{m'}}$ on the free momenta.
	The trivial estimate $\frac{1}{\abs{z-\gamma \myp{m'}}} \leq 1$ can be applied if $z$ does not belong to the top of the integration path, and the usual procedure of \cref{lemma:BasicGEstimate1} and \cref{lemma:BasicGEstimate2} can then be reused.
	We obtain in this case an additional factor of $\lambda^2$ in the upper bound of the amplitude, compared to the basic estimates.
	
	We proceed now with the more complicated scenario: $z$ belongs to the top of the integration path.
	In this case, $z = \alpha + i\beta$ for some $\abs{\alpha} \leq 1 + c_{N/2} = 1+2 \myp{N+1} \myp{ \normt{\omega}{\infty} + 2 \normt{\widehat{V}}{\infty} }$ and $\beta = \lambda^2$.
	Recalling that $\mathrm{Im} \gamma \myp{m} \leq 0$ for all $m$, we can estimate all the remaining resolvent factors by
	\begin{equation}
		\label{eq:crossingraph5}
		\frac{1}{\abs{\alpha - \gamma \myp{m} + i\beta}} 
		\leq{} \frac{1}{\abs{\alpha - \mathrm{Re} \gamma \myp{m} + i\beta}}.
	\end{equation}
	At this point, we need to estimate the double loop integral of $v_{i_2}$, which now also contains the resolvent factor of $\gamma \myp{m'}$.
	Using \eqref{eq:resolvent_integral} to express the resolvents as oscillating integrals and then applying Fubini shows that
	\begin{align}
		\MoveEqLeft[3] \int_{\myp{\T^d}^2} \ids k_1 \ids k_2 \frac{1}{\abs{\alpha - \mathrm{Re} \gamma \myp{i_2-1} + i\beta}} \frac{1}{\abs{\alpha - \mathrm{Re} \gamma \myp{m'} + i\beta}} \nn \\
		={}& \int_{\myp{\T^d}^2} \ids k_1 \ids k_2 \frac{1}{\abs{\alpha - \mathrm{Re} \gamma \myp{i_2} - \Theta_{i_2} + i\beta}} \frac{1}{\abs{\alpha - \mathrm{Re} \gamma \myp{i_2} - \alpha_1 - p\Theta_{i_2} - \Theta_{i_0} + i\beta}} \nn \\
		\leq{}& \int_{\R^2} \ids r \ids s F \myp{r ; \beta} F \myp{s;\beta} \abs[\bigg]{\int_{\myp{\T^d}^2} \ids k_1 \ids k_2 \, e^{-i \myp{r+ps} \Theta_{i_2} - is\Theta_{i_0}} } \nn \\
		\leq{}& 4 \expec{\ln\beta}^2 \myp[\bigg]{1 + \int_{\R^2} \ids r \ids s \, e^{-\beta \abs{s}} \abs[\bigg]{\int_{\myp{\T^d}^2} \ids k_1 \ids k_2 \, e^{-i \myp{r+ps} \Theta_{i_2} - is\Theta_{i_0}} }}.
	\label{eq:crossingraph6}
	\end{align}
	Here, we can do the change of variables $r \to t = r+ps$ and then estimate the integral using the crossing bounds \cref{ass:DR4}, depending on whether $v_{i_0}$ is a $T$- or the $X$-vertex of the double loop of $v_{i_2}$.
	
	\emph{Case 1:} $v_{i_0}$ is a $T_j$-vertex.
	In this case, we have by \cref{lem:CrossingGraphs} that $\Theta_{i_0} = \pm \omega^{\lambda} \myp{k_j+u} \pm \omega^{\lambda} \myp{k_j+u'} + \alpha'$ for some choice of the signs, where $\alpha',u',u$ are independent of $k_1,k_2$, and $u'-u$ depends on some free momenta ending above $v_{i_2}$.
	The momentum integral becomes
	\begin{align}
		\label{eq:crossingraph7}
		\MoveEqLeft[4] \abs[\bigg]{ \int_{\myp{\T^d}^2} \ids k_1 \ids k_2 \, e^{-it \Theta_{i_2} - is \Theta_{i_0}}} \nn \\
		={}& \abs[\bigg]{ \int_{\myp{\T^d}^2} \ids k_1 \ids k_2 \, e^{-it \myp{\pm\omega_1^{\lambda} \pm\omega_2^{\lambda} \pm \omega_3^{\lambda}} - is \myp{\pm \omega^{\lambda} \myp{k_j+u} \pm\omega^{\lambda} \myp{k_j+u'}}} } \nn \\
		\leq{}& \normt{p_{\pm t}^{\lambda}}{3}^2 \normt{\mathcal{K} \myp{\nonarg; \pm t, \pm s, \pm s, u,u'}}{3}
		={} \normt{p_t^{\lambda}}{3}^2 \normt{\mathcal{K} \myp{\nonarg; t, \pm s, \pm s, u,u'}}{3},
	\end{align}
	where at the end we have used the time-reversal symmetries of $p_t^{\lambda}$ and $\mathcal{K}$, and we will show the validity of the inequality in detail at the end of the proof.
	Plugging \eqref{eq:crossingraph7} back into \eqref{eq:crossingraph6}, we may now use the bound \eqref{eq:DR4_1} from \cref{ass:DR4} to estimate \eqref{eq:crossingraph6} by
	\begin{equation}
		\label{eq:crossingraph8}
		4 \expec{\ln\beta}^2 \beta^{\gamma - 1} \myp{1+F^{\mathrm{cr}} \myp{u'-u,\beta}},
	\end{equation}
	where $u'-u$ depends on some free momenta, and we have used $1 \leq \beta^{\gamma-1}$ since $0 < \gamma \leq 1$.
	We continue iterating the basic estimates of the resolvents for the remaining degree two vertices in \eqref{eq:crossingraph4} until the first of these momenta appears, and note by \cref{lem:CrossingGraphs} that this happens before reaching the top fusion vertex of the graph.
	Thus, this momentum is attached to a degree two vertex, and we recall from \cref{lemma:freeedgesorder2} that the dependence has to be of the form "$\pm k_j$" for a $j \in \Set{1,2,3}$.
	We can then use \eqref{eq:crossingest2} to obtain a factor $\expec{\ln \beta}^{1+c_2}$.
	By comparing with the basic estimates, we have gained an additional factor of 
	\begin{equation}
		\label{eq:crossingraph9}
		C \expec{\ln\beta}^{c_2+1} \beta^{\gamma} \frac{1}{\expec{s\beta}},
	\end{equation}
	which implies the bounds stated in the Proposition.
		
	\emph{Case 2:} $v_{i_0}$ is the $X$-vertex.
	Here, we have by \cref{lem:CrossingGraphs} that $\Theta_{i_0} = \alpha' + \sum_{j=1}^3 \pm\omega^{\lambda} \myp{k_j+u_j} $ for some choice of the signs, where $\alpha', u_1,u_2,u_3$ are independent of $k_1,k_2$, and at least one of the $u_j$ depends on some free momenta ending above $v_{i_2}$.
	We find, as for \eqref{eq:crossingraph7}, that the momentum integral in \eqref{eq:crossingraph6} is bounded by
	\begin{align}
		\label{eq:crossingraph10}
		\abs[\bigg]{ \int_{\myp{\T^d}^2} \ids k_1 \ids k_2 \, e^{-it \Theta_{i_2} - is \Theta_{i_0}}}
		={}& \abs[\bigg]{ \int_{\myp{\T^d}^2} \ids k_1 \ids k_2 \prod_{j=1}^3 e^{-i \myp{\pm t \omega^{\lambda} \myp{k_j} \pm s \omega^{\lambda} \myp{{k_j}+u_j}} }} \nn \\
		\leq{}& \prod_{j=1}^3 \normt{\mathcal{K} \myp{\nonarg; t, \pm s, 0, u_j,0}}{3}.
	\end{align}
	This time, we can apply the estimate \eqref{eq:DR4_2} from \cref{ass:DR4}, and find that \eqref{eq:crossingraph6} is bounded by
	\begin{equation}
	\label{eq:crossingraph11}
		4 \expec{\ln\beta}^2 \beta^{\gamma-1} \myp{1+F^{\mathrm{cr}} \myp{u_{j'},\beta}},
	\end{equation}
	for some index $j'$ satisfying that $u_{j'}$ depends on some free momentum above $v_{i_2}$ which is not the free momentum of the top fusion vertex, again by \cref{lem:CrossingGraphs}.
	From here, we can then follow the same steps used for the $T$-vertex case to get the same improvement \eqref{eq:crossingraph9} for the amplitude.
	
	Finally, we just need to argue that the bounds in \eqref{eq:crossingraph7} and \eqref{eq:crossingraph10} hold.
	For this, we note that $p_t^{\lambda}$ and $\mathcal{K}$ are defined as the inverse Fourier transforms of the functions $e^{-i t \omega^{\lambda}}$ and $e^{-i \myp{ t_0 \omega^{\lambda} + t_1 \omega^{\lambda} \myp{\nonarg+u_1} + t_2 \omega^{\lambda} \myp{\nonarg + u_2}}}$, respectively.
	Thus, we recognize that the integrals in \eqref{eq:crossingraph7} and \eqref{eq:crossingraph10} are both of the form
	\begin{equation*}
		\int_{\myp{\T^d}^2} \ids k_1 \ids k_2 \hat{f}_1 \myp{k_1} \hat{f}_2 \myp{k_2} \hat{f}_3 \myp{k_0-k_1-k_2},
	\end{equation*}
	where the $f_i \in \ell^1 \myp{\Z^d}$ can be either $p_t^{\lambda}$ or $\mathcal{K}$.
	Simply expanding the Fourier transforms and then applying Fubini and H\"older, we find
	\begin{align*}
		\MoveEqLeft[2] \abs[\bigg]{\int_{\myp{\T^d}^2} \ids k_1 \ids k_2 \hat{f}_1 \myp{k_1} \hat{f}_2 \myp{k_2} \hat{f}_3 \myp{k_0-k_1-k_2}} \\
		={}& \abs[\bigg]{\int_{\myp{\T^d}^2} \ids k_1 \ids k_2 \sum_{x_1,x_2,x_3 \in \Z^d} e^{-2 \pi i \myp{ k_1 \cdot \myp{x_1-x_3} + k_2 \cdot \myp{x_2-x_3} + k_0 \cdot x_3 } } f_1 \myp{x_1} f_2 \myp{x_2} f_3 \myp{x_3} } \\
		={}& \abs[\bigg]{\sum_{x_3 \in \Z^d} e^{-2 \pi i k_0 \cdot x_3} f_1 \myp{x_3} f_2 \myp{x_3} f_3 \myp{x_3} }
		\leq{} \normt{f_1}{3} \normt{f_2}{3} \normt{f_3}{3},
	\end{align*}
	which is exactly the stated bound in both \eqref{eq:crossingraph7} and \eqref{eq:crossingraph10}.
\end{proof}

\subsection{Contributions from attaching leading motives to a graph}
\label{sec:AddLeadingMotive}
Before providing bounds on the amplitudes of the leading and nested graphs, we investigate more closely how the amplitude of a graph is affected when attaching a leading motive to the graph.
\begin{lemma}[Adding a leading motive to a graph]
\label{lemma:AddMotive}
	Consider any momentum graph with a pairing cluster at the bottom, and let $\mathcal{M}$ be a leading motive (cf. \cref{fig:LeadingMotives}) that can be attached to this pairing.
	
	Let $\sigma$ denote the parity of the left leg of the pairing where $\mathcal{M}$ is attached, $k_0$ the momentum of the edge in the initial time slice where $\mathcal{M}$ is attached, and $r_2$ the initial time variable.
	Assuming that $\Phi_1 = 1$, adding the motive $\mathcal{M}$ to the graph changes a factor $e^{-\kappa_0 r_2} W^0 \myp{k_0,\sigma}$ to an expression of the form
	\begin{equation}
	\label{eq:AddMotive}
		-\lambda^2 \int_{\myp{\T^d}^3} \id k \delta \myp{k_0-k_1-k_2-k_3} \int_{\myp{\R_+}^2} \id r \delta \myp{r_2-r_1-r_0} \prod_{m=0}^2 e^{-r_m \kappa_{2-m}} F_{\mathcal{M}} \myp{r_1,k,\sigma},
	\end{equation}
	where $k = \myp{k_1,k_2,k_3}$, and the function $F_{\mathcal{M}}$ is explicitly calculable.
	For convenience, we have gathered all the possible contributions in \cref{table:motives}.
	\begin{table}[h]
		\begin{tabularx}{0.7\textwidth}{Xl}
			\toprule
			$\mathcal{M}$		&	\multicolumn{1}{c}{$F_{\mathcal{M}} \myp{r_1,k,\sigma}$} \\ \midrule
			$L1^{\tau}$			&	$\widehat{V} \myp{k_1+k_2}^2 e^{-i r_1 \tau \Theta \myp{k,1}} W_0^{\sigma} W_2 W_3$	\\ \midrule
			$L2^{\tau}$			&	$\widehat{V} \myp{k_1+k_2} \widehat{V} \myp{k_1+k_3} e^{-i r_1 \tau \Theta \myp{k,1}} W_0^{\sigma} W_1 W_3$	\\ \midrule
			$L3^{\tau}$			&	$\widehat{V} \myp{k_1+k_2}^2 e^{-i r_1 \tau \Theta \myp{k,1}} W_0^{\sigma} W_1 \widetilde{W}_2$	\\ \midrule
			$L4^{\tau}$			&	$-\widehat{V} \myp{k_1+k_2} \widehat{V} \myp{k_1+k_3} e^{-i r_1 \tau \Theta \myp{k,1}} W_0^{\sigma} W_2 W_3$	\\ \midrule
			$L5^{\tau}$			&	$-\widehat{V} \myp{k_1+k_2}^2 e^{-i r_1 \tau \Theta \myp{k,1}} W_0^{\sigma} W_1 W_3$	\\ \midrule
			$L6^{\tau}$			&	$-\widehat{V} \myp{k_1+k_2} \widehat{V} \myp{k_1+k_3} e^{-i r_1 \tau \Theta \myp{k,1}} W_0^{\sigma} W_1 \widetilde{W}_2$	\\ \midrule
			$G1^{\sigma}$		&	$\widehat{V} \myp{k_1+k_2} \widehat{V} \myp{k_1+k_3} e^{-i r_1 \sigma \Theta \myp{k,1}} W_1^{-\sigma} W_2^{\sigma} W_3^{\sigma}$	\\ \midrule
			$G2^{\sigma}$		&	$\widehat{V} \myp{k_1+k_2} \widehat{V} \myp{k_1+k_3} e^{i r_1 \sigma \Theta \myp{k,1}} W_1^{-\sigma} W_2^{\sigma} W_3^{\sigma}$	\\ \midrule
			$G3^{\sigma}$		&	$-\widehat{V} \myp{k_1+k_2}^2 e^{-i r_1 \sigma \Theta \myp{k,1}} W_1^{-\sigma} W_2^{\sigma} W_3^{\sigma}$	\\ \midrule
			$G4^{\sigma}$		&	$-\widehat{V} \myp{k_1+k_2}^2 e^{i r_1 \sigma \Theta \myp{k,1}} W_1^{-\sigma} W_2^{\sigma} W_3^{\sigma}$	\\ \bottomrule
		\end{tabularx}
		\vspace{5pt}
		\caption{Contributions to the amplitude from adding a leading motive to a pairing whose left leg has parity $\sigma$.
		For the loss motives, we have introduced the parameter $\tau = \pm1$ to write the contribution from a motive and the conjugate motive on the same line.
		We use the shorthand notation $W_j^{\sigma} \coloneq W^0 \myp{k_j,\sigma} $.}
	\label{table:motives}
	\end{table}
\end{lemma}
\begin{proof}
	Consider first a pairing whose left leg has parity $\sigma = -1$, and suppose that we attach the loss motive L1 to this pairing (recall that the leading motives are gathered in \cref{fig:LeadingMotives}).
	Attaching L1 is exactly the same as attaching the diagram in \cref{fig:L1} to the original graph at the given pairing.
	Recalling by \cref{lem:cluster_sign_order} that the sign $\epsilon \myp{S}$ of the clusters behaves multiplicatively when attaching graphs to each other, the resulting contribution has already been calculated in \eqref{eq:L1} in \cref{ex:L1}, except in the present case we have already taken $L \to \infty$.
	Comparing with the zeroth order amplitude (which corresponds to the trivial graph consisting of just a single pairing),
	\begin{align*}
		\mathcal{F}_0^{main,pair} \myp{k_0',k_0,r_2,\kappa}
		={} \1 \myp{k_0+k_0' = 0} e^{-\kappa_0 r_2} W^0 \myp{k_0},
	\end{align*}
	it is apparent that adding L1 to the given pairing gives a contribution equal to \eqref{eq:AddMotive} with
	\begin{equation*}
		F_{L1} \myp{r_1,k,-1}
		={} \widehat{V} \myp{k_1+k_2}^2 e^{-i r_1 \Theta \myp{k,1}} W_0 W_2 W_3.
	\end{equation*}
	The contributions from adding any of the motives $L2, \dotsc, L6,G1^-,G3^-,G4^-$ is calculated in exactly the same way, providing all the entries in \cref{table:motives} for $\sigma = -1$ and $\tau = 1$, except for $G2^-$.
	
	The contributions from attaching the remaining loss motives $L1^-,\dotsc, L6^-$ are easily obtained by noting that the resulting diagrams are "conjugated" to the ones used to calculate the contributions from $L1,\dotsc,L6$, that is, the diagrams are obtained by reversing the parities of all edges and inverting the order of all edges in each time slice.
	As explained in the proof of \cref{Propo:QERR3}, this affects the corresponding amplitude by simply adding a complex conjugation. (The entire reason for introducing the minus trees in the first place is to handle the conjugated factors in \eqref{Def:OperatorQErrors:Bis}).
	Similarly, we note that the diagram of G$2^-$ is also obtained by conjugating the diagram of G$1^-$, so their amplitudes only differ by a complex conjugation.
	This finishes all the entries in \cref{table:motives} for $\sigma = -1$.
	
	If we instead attach a loss motive to a pairing whose left leg has parity $\sigma = 1$, it only changes the parity of the cluster corresponding to the initial momentum $k_0$.
	In other words, instead of seeing $W^0 \myp{k_0}$, we see now $\widetilde{W}^0 \myp{k_0}$.
	Everything else in the calculations of \cref{ex:L1} would be the same.
	
	For the remaining gain motives, we calculate the contribution from G1 by hand and note that also G2 is obtained from G1 by conjugating the diagram.
	finally. we note that the diagrams of G3 and G4 are obtained by conjugating the diagrams of G$4^-$ and G$3^-$, respectively, except all the parities are reversed.
	This means that the amplitude of G3 can be obtained from that of G$4^-$ by adding a complex conjugation and switching any $W^0 \myp{k_j,\pm 1}$ to $W^0 \myp{k_j,\mp 1}$, and the same holds for G4 and G$3^-$.
\end{proof}
We need to control the time-integrability of factors obtained by iteratively attaching leading motives, so in view of the assertions of \cref{lemma:AddMotive}, we consider the function $G_{s,\tau}^{\lambda}: L^2 \myp{\T^d}^4 \rightarrow L^2 \myp{\T^d}$, defined for $s \in \R$ and $\tau \in \Set{\pm 1}^4$ by
\begin{align}
	\MoveEqLeft[2] G_{s,\tau}^{\lambda} \myb{f_0,f_1,f_2,f_3} \myp{k_0} \nn \\
	\coloneq{}& \int_{\myp{\T^d}^3} \ids k_1 \ids k_2 \ids k_3 \delta\myp{k_0-k_1-k_2-k_3} \widetilde{U} \myp{k_1,k_2,k_3} \prod_{i=0}^3 e^{-i s \tau_i \omega^{\lambda} \myp{k_i}} f_i \myp{k_i},
\label{eq:Gdef}
\end{align}
where $\widetilde{U}$ can be either $\pm \widehat{V} \myp{k_1+k_2}^2$, $\pm \widehat{V} \myp{k_1+k_3}^2$, or $\pm \widehat{V} \myp{k_1+k_2} \widehat{V} \myp{k_1+k_3}$, depending on the leading motive in question.
We will typically abbreviate $G_{s,\tau}^{\lambda} = G_{s,\tau}^{\lambda} \myb{f_0,f_1,f_2,f_3} $.
Note that by Young's inequality,
\begin{align}
	\abs{G_{s,\tau}^{\lambda} \myp{k_0}}
	\leq{}& \abs{f_0 \myp{k_0}} \normt{\widehat{V}}{\infty}^2 \int_{\myp{\T^d}^2} \ids k_1 \ids k_2 \abs{f_1 \myp{k_1}} \abs{f_2 \myp{k_2}} \abs{f_3 \myp{k_0-k_1-k_2}} \nn \\
	={}& \abs{f_0 \myp{k_0}} \normt{\widehat{V}}{\infty}^2 \myp[\big]{\abs{f_1} \ast \abs{f_2} \ast \abs{f_3}} \myp{k_0}
	\leq{} \abs{f_0 \myp{k_0}} \normt{\widehat{V}}{\infty}^2 \prod_{i=1}^3 \normt{f_i}{\frac{3}{2}},
\label{eq:Gbound1}
\end{align}
so $G_{s,\tau}$ indeed maps into $L^2 \myp{\T^d}$, since $\normt{f}{3/2} \leq \normt{f}{2}$.
An analogous definition of $G_{s,\tau}^{\lambda}$ is treated in \cite{LukkarinenSpohn:WNS:2011}, and the proofs of the following bounds on $G_{s,\tau}^{\lambda}$ are essentially the same as in \cite{LukkarinenSpohn:WNS:2011}.
The only differences in our case is the presence of the interaction term $\widetilde{U}$, and the dependence of $\lambda$ in the modified dispersion relation in \eqref{eq:Gdef}.

Consider also the free evolution semigroup $U_t^{\lambda}$ on $\ell^2 \myp{\Z^d}$, that is,
\begin{equation}
	\myp{U_t^{\lambda} g} \myp{x}
	\coloneq{} p_t^{\lambda} \ast g \myp{x}
	={} \sum_{y \in \Z^d} p_t^{\lambda} \myp{x-y} g \myp{y},
	\quad g \in \ell^2 \myp{\Z^d},
\end{equation}
where $p_t^{\lambda} \myp{x} = \int_{\T^d} \ids k e^{2 \pi i k \cdot x - i t \omega^{\lambda} \myp{k}} $ is the free propagator as defined in \eqref{eq:propagator}.
\begin{lemma}
\label{lemma:Gbound}
	We have the pointwise bound
	\begin{equation}
	\label{eq:Gbound2}
		\abs{G_{s,\tau}^{\lambda} \myp{k_0}}
		\leq{} \abs{f_0 \myp{k_0}} \normt{V}{\ell^1 \myp{\Z^d}}^2 \prod_{i=1}^3 \normt{U_{\tau_i s}^{\lambda} \check{f}_i}{\ell^3 \myp{\Z^d}},
	\end{equation}
	and for any $t \in \R$,
	\begin{equation}
	\label{eq:Gbound3}
		\normt{U_t^{\lambda} \check{G}_{s,\tau}^{\lambda}}{\ell^3 \myp{\Z^d}}
		\leq{} \normt{V}{\ell^1 \myp{\Z^d}}^2 \normt{\check{f_0}}{\ell^1 \myp{\Z^d}} \normt{p_{t+\tau_0 s}^{\lambda}}{\ell^3 \myp{\Z^d}} \prod_{i=1}^3 \normt{U_{\tau_i s}^{\lambda} \check{f_i}}{\ell^3 \myp{\Z^d}}.
	\end{equation}
\end{lemma}
\begin{proof}
	Since $p_t^{\lambda}$ is an inverse Fourier transform, $ p_t^{\lambda} \myp{x} = \myp{e^{-it \omega^{\lambda}}}\check{\mkern6mu} \myp{x}$, we have for any $f \in L^2 \myp{\T^d}$ and $t \in \R$,
	\begin{equation}
	\label{eq:free_evol1}
		U_t^{\lambda} \check{f}
		={} p_t^{\lambda} \ast \check{f}
		={} \myp{e^{-it \omega^{\lambda}} f}\check{\mkern6mu},
	\end{equation}
	so that
	\begin{equation}
	\label{eq:free_evol2}
		e^{-it \omega^{\lambda} \myp{k}} f \myp{k}
		={} \myp{U_t^{\lambda} \check{f}}\hat{\mkern6mu} \myp{k}
		={} \sum_{x \in \Z^d} e^{-2 \pi i k \cdot x} \myp{U_t^{\lambda} \check{f}} \myp{x}.
	\end{equation}
	The idea is to use the convolution structure, as in \eqref{eq:Gbound1}, but we cannot immediately take the absolute value inside the integral, so we have to be a bit more careful with the interaction factor.
	We write
	\begin{equation}
	\label{eq:Gexpr}
		G_{s,\tau}^{\lambda} \myp{k_0}
		={} f_0 \myp{k_0} e^{-i \tau_0 s \omega^{\lambda} \myp{k_0}} F \myp{k_0},
	\end{equation}
	with
	\begin{equation}
	\label{eq:Fdef}
		F \myp{k_0} = \int_{\myp{\T^d}^3} \ids k_1 \ids k_2 \ids k_3 \delta \myp{k_0-k_1-k_2-k_3} \widetilde{U} \myp{k_1,k_2,k_3} \prod_{i=1}^3 e^{-i s \tau_i \omega^{\lambda} \myp{k_i}} f_i \myp{k_i},
	\end{equation}
	Note that because of the $\delta$-function, we have $\widehat{V} \myp{k_1+k_2} = \widehat{V} \myp{k_0-k_3}$ and $\widehat{V} \myp{k_1+k_3} = \widehat{V} \myp{k_0-k_2}$.
	Writing $V_{k_0} \myp{x} \coloneq e^{2\pi i k_0 \cdot x} V \myp{-x}$, so that $\widehat{V} \myp{k_0-k_3} = \widehat{V}_{k_0} \myp{k_3}$, and using \eqref{eq:free_evol2}, we can express $F$ as the Fourier transform of a single function, evaluated at $k_0$.
	For instance, if $\widetilde{U} \myp{k_1,k_2,k_3} = \widehat{V} \myp{k_0-k_3}^2 $, we have
	\begin{align*}
		F \myp{k_0}
		={}&	\begin{aligned}[t]	
					\int_{\myp{\T^d}^3} \ids k_1 \ids k_2 \ids k_3 \delta\myp{k_0-k_1-k_2-k_3} &\myp{U_{\tau_1 s}^{\lambda} \check{f}_1}\hat{\mkern6mu} \myp{k_1} \myp{U_{\tau_2 s}^{\lambda} \check{f}_2}\hat{\mkern6mu} \myp{k_2} \\
					\times& \myp{ V_{k_0} \ast V_{k_0} \ast U_{\tau_3 s}^{\lambda} \check{f}_3}\hat{\mkern6mu} \myp{k_3}
				\end{aligned} \\
		={}& \myp{U_{\tau_1 s}^{\lambda} \check{f}_1}\hat{\mkern6mu} \ast \myp{U_{\tau_2 s}^{\lambda} \check{f}_2}\hat{\mkern6mu} \ast \myp{ V_{k_0} \ast V_{k_0} \ast U_{\tau_3 s}^{\lambda} \check{f}_3}\hat{\mkern6mu} \myp{k_0} \\
		={}& \myp[\big]{ \myp{U_{\tau_1 s}^{\lambda} \check{f}_1} \myp{U_{\tau_2 s}^{\lambda} \check{f}_2} \myp{ V_{k_0} \ast V_{k_0} \ast U_{\tau_3 s}^{\lambda} \check{f}_3} }\hat{\mkern6mu} \myp{k_0},
	\end{align*}
	and similarly for the other possible cases of $\widetilde{U}$.
	We conclude that $F$ can be written
	\begin{equation*}
		\label{eq:Fmdef}
		F \myp{k_0}
		={} \sum_{x \in \Z^d} e^{-2 \pi i k_0 \cdot x} \prod_{i=1}^3 h_i \myp{x},
	\end{equation*}
	where the $h_i$ can be either $\pm U_{\tau_i s}^{\lambda} \check{f}_i$, $\pm V_{k_0} \ast \myp{U_{\tau_i s}^{\lambda} \check{f}_i}$, or $\pm V_{k_0} \ast V_{k_0} \ast \myp{U_{\tau_i s}^{\lambda} \check{f}_i}$, depending on the form of $\widetilde{U}$, but with $V_{k_0}$ appearing twice in total in the expression of $F$.
	Now, applying the H\"older and Young inequalities,
	\begin{equation}
	\label{eq:Fmbound}
		\normt{F}{L^{\infty} \myp{\T^d}} 
		\leq{} \normt{\check{F}}{\ell^1 \myp{\Z^d}} 
		\leq{} \prod_{i=1}^3 \normt{h_i}{\ell^3 \myp{\Z}^d} 
		\leq{} \normt{V}{\ell^1 \myp{\Z^d}}^2 \prod_{i=1}^3 \normt{U_{\tau_i s}^{\lambda} \check{f}_i}{\ell^3 \myp{\Z^d}},
	\end{equation}
	since $\normt{V_{k_0}}{\ell^1} = \normt{V}{\ell^1}$.
	This finishes the proof of \eqref{eq:Gbound2}.
	
	To see that \eqref{eq:Gbound3} holds, note first that
	\begin{equation*}
		\normt{\myp{f_0F}\check{\mkern6mu}}{\ell^1 \myp{\Z^d}}
		={} \normt{\check{f}_0 \ast \check{F}}{\ell^1 \myp{\Z^d}}
		\leq{} \normt{\check{f}_0}{\ell^1 \myp{\Z^d}} \normt{V}{\ell^1 \myp{\Z^d}}^2 \prod_{i=1}^3 \normt{U_{\tau_i s}^{\lambda} \check{f}_i}{\ell^3 \myp{\Z^d}}.
	\end{equation*}
	Returning to \eqref{eq:Gexpr} and using \eqref{eq:free_evol2}, we have for any $t \in \R$ that
	\begin{equation*}
		\myp{U_t^{\lambda} \check{G}_{s,\tau}^{\lambda}} \hat{\mkern6mu}
		={} e^{-i t \omega^{\lambda}} G_{s,\tau}^{\lambda}
		={} e^{-i \myp{t + \tau_0 s} \omega^{\lambda}} f_0 F
		={} \myp{U_{t + \tau_0 s}^{\lambda} \myp{f_0F}\check{\mkern6mu}} \hat{\mkern6mu},
	\end{equation*}
	so \eqref{eq:Gbound3} now follows by \eqref{eq:free_evol1} and Young's inequality.
\end{proof}
We will need to control the time-integrability of iterations of $G_{s,\tau}$, i.e., factors of the form \eqref{eq:Gdef} where the input functions $f_i$ are themselves of the form \eqref{eq:Gdef}.
\begin{proposition}
\label{prop:Gbounds}
	Let $G^{\lambda}$ be a factor of the form \eqref{eq:Gdef} obtained by iteration of $M$ leading motive integrals, and let $s_m \in \R$, $\tau_m \in \Set{\pm 1}^4$ be the corresponding time variables and parities.
	Then, for each $m = 1,\dotsc,M$ there are (possibly empty) sets $B_{m,i} \subseteq \Set{1,\dotsc,m-1}$, $i=1,2,3$, such that for any $k \in \T^d$,
	\begin{equation}
	\label{eq:Gbound4}
		\abs{G^{\lambda} \myp{k}}
		\leq{} \abs{h_0 \myp{k}} c_1^{3M} \normt{V}{\ell^1}^{2M} \prod_{m=1}^M \prod_{i=1}^3 \normt[\big]{p_{\tau_{m,i}s_m + \sum_{j\in B_{m,i}} \tau_{j,i} s_j}^{\lambda}}{\ell^3},
	\end{equation}
	where $h_0 \in \Set{1,W^0,1-W^0}$ and $c_1 = \max\Set{1,\normt{\check{W}^0}{\ell^1}, \normt{\myp{1-W^0}\check{\mkern6mu}}{\ell^1}}$.
	Furthermore, there is a subset $B_M \subseteq \Set{1,\dotsc,M}$ such that for any $t \in \R$,
	\begin{equation}
	\label{eq:Gbound5}
		\normt{U_t^{\lambda} \check{G}^{\lambda} }{\ell^3}
		\leq{} c_1^{3M+1} \normt{V}{\ell^1}^{2M} \normt[\big]{p_{t + \sum_{j\in B_M} \tau_{j,i} s_j}^{\lambda}}{\ell^3} \prod_{m=1}^M \prod_{i=1}^3 \normt[\big]{p_{\tau_{m,i}s_m + \sum_{j\in B_{m,i}} \tau_{j,i} s_j}^{\lambda}}{\ell^3}.
	\end{equation}
	
	In particular, if the modified dispersion relation $\omega^{\lambda}$ satisfies \cref{ass:DR2}, then \eqref{eq:Gbound4} provides an upper bound $\abs{G^{\lambda} \myp{k}} \leq \widetilde{G} \myp{s}$, uniformly in $k$ and $\lambda$, such that $\widetilde{G}$ is integrable on $\R^M$.
\end{proposition}
\begin{proof}
	In order to prove \eqref{eq:Gbound4} and \eqref{eq:Gbound5}, we will need an extension of \eqref{eq:Gbound3}.
	First, we need to expand $G^{\lambda}$ in the first argument $f_0$ until we see either a factor $1$, $W^0$, or $\widetilde{W}^0$.
	More precisely, we may write $G^{\lambda} = G_{r_1,\sigma_1}^{\lambda} \myb{f_{1,0},f_{1,1},f_{1,2},f_{1,3}} $, where each $f_{1,i}$ is either of the form \eqref{eq:Gdef} or one of $1,W^0,\widetilde{W}^0$, and $r_1=s_M$, $\sigma_1=\tau_M$.
	If $f_{1,0}$ is of the form \eqref{eq:Gdef}, we keep expanding in the first argument until it is no longer possible, at which point we have for some $ 1 \leq N \leq M$, and some $r_n$, $\sigma_{n,0}$ which are a subset of the $s_m$, $\tau_m$,
	\begin{equation}
	\label{eq:Gexpr2}
		G^{\lambda} \myp{k}
		={} h_0 \myp{k} \prod_{n=1}^N e^{-i r_n \sigma_{n,0} \omega^{\lambda} \myp{k}} \prod_{n=1}^N F_n \myp{k},
	\end{equation}
	where $h_0 \in \Set{1,W^0,\widetilde{W}^0}$, and $F_n$ is of the form \eqref{eq:Fdef} with some $f_{n,i}$ obtained from other iterations, that is, $f_{n,i}$ is either of the form \eqref{eq:Gdef}, or $f_{n,i} \in \Set{1,W^0,\widetilde{W}^0}$.
	In particular, as in \eqref{eq:Fmbound},
	\begin{equation}
	\label{eq:Fmbound2}
		\normt{F_n}{L^{\infty}} 
		\leq{} \normt{\check{F}_n}{\ell^1} 
		\leq{} \normt{V}{\ell^1}^2 \prod_{i=1}^3 \normt{U_{\sigma_{n,i}r_n}^{\lambda} \check{f}_{n,i}}{\ell^3}.
	\end{equation}
	Denoting $H \coloneq{} h_0 \prod_{n=1}^N F_n$, we then have by \eqref{eq:Fmbound2} that
	\begin{equation*}
		\normt{\check{H}}{\ell^1}
		={} \normt{\check{h}_0 \ast \check{F}_1 \ast \cdots \ast \check{F}_N}{\ell^1}
		\leq{} \normt{\check{h}_0}{\ell^1} \normt{V}{\ell^1}^{2N} \prod_{n=1}^N \prod_{i=1}^3 \normt{U_{\sigma_{n,i}r_n}^{\lambda} \check{f}_{n,i}}{\ell^3}.
	\end{equation*}
	Returning to \eqref{eq:Gexpr2} and using \eqref{eq:free_evol2}, we have for any $t \in \R$ that
	\begin{equation*}
		\myp{U_t^{\lambda} \check{G}^{\lambda}}\hat{\mkern6mu}
		={} e^{-i t \omega^{\lambda}} G^{\lambda}
		={} e^{-i \myp[\big]{t + \sum_{n=1}^N r_n \sigma_{n,0}} \omega^{\lambda}} H
		={} \myp{U_{t + \sum_{n=1}^N r_n \sigma_{n,0}}^{\lambda} \check{H}} \hat{\mkern6mu},
	\end{equation*}
	so by \eqref{eq:free_evol1}, we conclude
	\begin{equation}
	\label{eq:Gbound6}
		\normt{U_t^{\lambda} \check{G}^{\lambda}}{\ell^3}
		\leq{} \normt{\check{h}_0}{\ell^1} \normt{V}{\ell^1}^{2N} \normt[\big]{p_{t + \sum_{n=1}^N r_n \sigma_{n,0}}^{\lambda}}{\ell^3} \prod_{n=1}^N \prod_{i=1}^3 \normt{U_{\sigma_{n,i}r_n}^{\lambda} \check{f}_{n,i}}{\ell^3}.
	\end{equation}
	
	We now return to the bounds \eqref{eq:Gbound4} and \eqref{eq:Gbound5}, which we prove simultaneously by induction.
	Note that the bounds have already been established for $M=1$ by \cref{lemma:Gbound}, and note that both bounds also hold for $M=0$, in the case where $G^{\lambda} = h_0$ is just a given function.
	Thus we only need to do the induction step, so assume that \eqref{eq:Gbound4} and \eqref{eq:Gbound5} hold for any factor obtained by $K<M$ iterations of leading motive integrals.
	Write now
	\begin{equation*}
		G^{\lambda} = G_{s_M,\tau_M}^{\lambda} \myb{G_0,G_1,G_2,G_3},
	\end{equation*}
	where each $G_p$, $p=0,1,2,3$, is an iteration of $K_p$ applications of $G_{s,\tau}^{\lambda}$, with $K_p \in \Set{0,\dotsc,M-1}$ and $\sum_{p=0}^3 K_p = M-1$.
	Note that if a time variable $s_m$ appears in $G_p$, then it does not appear in any of the other factors $G_q$, $q \neq p$.
	In other words, there are (possibly empty) disjoint subsequences $\myp{m_n^p}_{n=1}^{K_p} \subseteq \Set{1,\dotsc,M-1}$ whose union is $\Set{0,\dotsc,M-1}$, such that the time variables in $G_p$ are given by $s_{m_n^p}$.
	Then, by the induction hypothesis,
	\begin{equation*}
		\abs{G_0 \myp{k}}
		\leq{} \abs{h_0 \myp{k}} c_1^{3K_0} \normt{V}{\ell^1}^{2K_0} \prod_{n=1}^{K_0} \prod_{i=1}^3 \normt[\big]{p_{\tau_{m_n^0,i} s_{m_n^0} + \sum_{j\in B_{m_n^0,i}} \tau_{j,i} s_j}^{\lambda}}{\ell^3},
	\end{equation*}
	and for $p=1,2,3$,
	\begin{equation*}
		\normt{U_t^{\lambda} \check{G}_p}{\ell^3}
		\leq{} c_1^{3K_p+1} \normt{V}{\ell^1}^{2K_p} \normt[\big]{p_{t + \sum_{j\in B_{K_p}} \tau_{j,i} s_j}^{\lambda}}{\ell^3} \prod_{n=1}^{K_p} \prod_{i=1}^3 \normt[\big]{p_{\tau_{m_n^p,i} s_{m_n^p} + \sum_{j\in B_{m_n^p,i}} \tau_{j,i} s_j}^{\lambda}}{\ell^3}.
	\end{equation*}
	Now, combining these with \eqref{eq:Gbound2} applied to $G^{\lambda}$, we obtain \eqref{eq:Gbound4} by rearranging the factors and setting $B_{M-1,i} \coloneq B_{K_i}$.
	
	To prove \eqref{eq:Gbound5}, the argument is the same, except one starts from \eqref{eq:Gbound6} and applies the induction hypothesis to each of the factors $\normt{U_{\sigma_{n,i}r_n}^{\lambda} \check{f}_{n,i}}{\ell^3}$.
	
	Finally, to easily see that the right hand side of \eqref{eq:Gbound4} is integrable in $s_1,\dotsc,s_M$ uniformly in $k$ and $\lambda$ when $\omega^{\lambda}$ satisfies \cref{ass:DR2}, we note first that $\abs{h_0 \myp{k}} \leq 1$ for all $k$.
	Then we can apply the $\ell^{3}$-dispersivity to get rid of the $\lambda$ dependence, and iteratively integrate out the time variables, starting with $s_M$.
	At the first step, exactly three factors depend on $s_M$, the integral of which we can estimate by
	\begin{align*}
		\int_{\R} \ids s_M \prod_{i=1}^3 \normt[\big]{p_{\tau_{M,i}s_M + \sum_{j\in B_{M,i}} \tau_{j,i} s_j}^{\lambda}}{\ell^3}
		\lesssim{}& \int_{\R} \ids s_M \prod_{i=1}^3 \expec[\Big]{\tau_{M,i}s_M + \sum_{j\in B_{M,i}} \tau_{j,i} s_j}^{-\frac{1+\delta}{3}} \\
		\leq{}& \int_{\R} \ids s_M \expec{s_M}^{-1-\delta},
	\end{align*}
	by use of H\"older's inequality.
	Now, only three of the remaining factors in \eqref{eq:Gbound4} depend on $s_{M-1}$, so we can repeat this procedure $M-1$ more times to integrate out all the time variables.
\end{proof}

\subsection{Leading and nested graphs}
Here we provide the final bounds that we need on the leading and nested graphs in order to finish the proof of the main theorem in the next section.
The contributions from the presence of leading motives in the graphs need to be estimated carefully, and we need to remove the momentum cut-offs $\Phi_1$ in order to use the structure of the $G_{s,\tau}$-factors from \cref{sec:AddLeadingMotive}.
\begin{lemma}[Removing the cut-off]
\label{lem:RemoveCutoff}
	There are constants $c,C>0$ depending on $\omega$, such that the following estimates hold true for all relevant fully paired graphs and for all $k,k' \in \T^d$,
	\begin{align}
		\MoveEqLeft[10] \abs[\Big]{\mathcal{F}^{err,pair}_n \myp{S,J,\ell,\ell',k,s,\kappa} - \mathcal{F}^{err,pair}_n \myp{S,J,\ell,\ell',k,s,\kappa}|_{\Phi_1 \to 1} } \nn \\
		\leq{}& C \normt{\widehat{V}}{\infty}^{2n} n \lambda^{\gamma'} e^{s\lambda^2} \expec{cs\lambda^2}^n \expec{\ln n} \expec{\ln\lambda}^{1+n}, \label{eq:RemoveCutoff1} \\
		\MoveEqLeft[10] \abs[\Big]{ \mathcal{F}^{main,pair}_n \myp{S,\ell,k',k,t\lambda^{-2},\kappa} - \mathcal{F}^{main,pair}_n \myp{S,\ell, k',k,t\lambda^{-2}, \kappa}|_{\Phi_1 \to 1} } \nn \\
		\leq{}& C \1 \myp{k+k'=0} \normt{\widehat{V}}{\infty}^n n \lambda^{\gamma'} e^{t} \expec{ct}^{n/2} \expec{\ln n} \expec{\ln\lambda}^{1+n/2}, \label{eq:RemoveCutoff2}
	\end{align}
	with $\gamma'$ as in \eqref{Def:Para1}.
\end{lemma}
\begin{proof}
	Consider a relevant fully paired graph with $N$ interaction vertices.
	The idea is to iteratively expand the cut-off functions $\Phi_1 = 1 - \Phi_0$ at each interaction vertex $v_m$, going from the bottom of the graph to the top.
	At each step we obtain two terms, one that corresponds to replacing $\Phi_1$ with $1$ at the vertex $v_m$, which we use to continue the iteration, and one that corresponds to replacing $\Phi_1$ with $\Phi_0$ at $v_m$, which we need to estimate.
	
	In the amplitude corresponding to the term with $\Phi_0$, we take the absolute value inside all the momentum integrals, and then estimate any remaining $\Phi_1$ factors and all $W^0$ factors by one, and all interaction factors by $\normt{\widehat{V}}{\infty}$.
	The resulting amplitude looks almost identical to the one in \eqref{eq:BasicGEstimate2}, except there is now a single $\Phi_0$ factor at $v_m$ instead of all the usual interaction factors $\Phi_1$.
	Then, we estimate the time integral as in the basic scheme and obtain an expression with resolvents $\abs{z-\gamma \myp{j}}^{-1}$ for the degree two vertices, of which there are $N/2$ since the graph is fully paired.
	Now, we estimate the $\Phi_0$ factor using \cref{Propo:Cutoff},
	\begin{equation*}
		\Phi_0 \myp{\pm k} \leq \sum\limits_{e_i < e_j} \1 \myp[\big]{ d \myp{\pm \myp{k_{e_i} + k_{e_j}}, M^{\mathrm{sing}}} < \lambda^b },
	\end{equation*}
	where the sum runs over pairs of edges in $\mathscr{E}_- \myp{v_m}$.
	
	Fix now any pair $e_i<e_j$ and note that since the graph is relevant, we must have $k_{e_i} + k_{e_j} \neq 0$ (because of the initial presence of the $\Phi_1$ factors), so by \cref{lemma:dependenceoftwoedges}, $k_{e_i} + k_{e_j}$ must depend on some free momenta.
	We now iterate through the momentum integrals of the degree two vertices, from the bottom to the top of the graph, until we reach a vertex $v_{m'}$ where $k_{e_i} + k_{e_j}$ depends on the free momenta.
	By \cref{lemma:odd_dependence}, $v_{m'}$ must be a degree two vertex, since $k_{e_i} + k_{e_j}$ cannot depend only on the free momentum of the top fusion vertex.
	Thus, we have to estimate
	\begin{equation*}
		\iint_{\myp{\T^d}^2} \ids k_1 \ids k_2 \frac{\1 \myp[\big]{ d \myp{\pm \myp{k_{e_i} + k_{e_j}}, M^{\mathrm{sing}}} < \lambda^b }}{\abs{z-\gamma \myp{m'}}},
	\end{equation*}
	where $k_1,k_2$ are the free momenta of $v_{m'}$.
	However, trivially estimating the resolvent $\abs{z-\gamma \myp{m'}}^{-1} \leq \lambda^{-2}$ and then integrating the characteristic functions yields a bound $\lambda^{b \myp{d-1}-2} \leq \lambda^{\frac{1}{4}} \leq \lambda^{\gamma'}$ for this vertex.
	After this, the iteration scheme of the basic estimate can be continued, yielding a contribution of $\expec{ \ln \lambda}$ from any other degree two vertex.
	In conclusion, we find that the amplitude with the extra characteristic function is bounded by
	\begin{equation*}
		C \normt{\widehat{V}}{\infty}^N \lambda^{\gamma'} e^{s\lambda^2} \expec{cs\lambda^2}^N \expec{\ln N} \expec{\ln\lambda}^{1+N},
	\end{equation*}
	which has to be multiplied by $N$ to give an upper bound on the difference in \eqref{eq:RemoveCutoff1}, \eqref{eq:RemoveCutoff2}, since we remove $\Phi_1$ in $N$ steps.
\end{proof}
\begin{remark}[Dimension]
	\cref{lem:RemoveCutoff} is one of two places where we use the assumption that $d \geq 4$, the other place being in \cref{lemma:BasicGEstimate3}.
	However, it might be worth noting that for the specific case of the nearest neighbour dispersion relation (with corresponding $V$) as in \cref{sec:verify_assumptions}, it is possible to combine the bound \eqref{eq:constr_interf3} with the representation \eqref{eq:resolvent_integral} to obtain a different estimate of the degree one vertex integrals in \cref{lemma:degree1vertex}, of the form
	\begin{equation*}
		\int_{\T^d} \ids k \frac{F_1 \myp{\sigma' k}}{\abs{\omega^{\lambda} \myp{k} + \sigma \omega^{\lambda} \myp{k_0-k} - \alpha + i\beta}}
		\lesssim{} \expec{ \ln \beta } \prod_{j=1}^d \frac{1}{d_{\T} \myp[\big]{k_{0,j}, \Set{0,\frac{1}{2}}}^{1/2}}.
	\end{equation*}
	Here, the right hand side is not uniform in $k_0$, but it has a better dependence on $\lambda$.
	Then, combining with H\"older's inequality, one may avoid using the trivial estimate $\lambda^{-2}$ on the resolvent factor in the proof above.
	This allows for removal of the cut-off functions also in three dimensions, in this specific case.
\end{remark}
\begin{proposition}[Leading graphs]
	\label{Propo:LeadingGraphs}
	There are constants $c_0, C > 0$ depending on $\omega$,$V$ such that the following estimates hold true for all leading graphs and for all $k,k' \in \T^d$,
	\begin{align}
		&\abs[\Big]{\mathcal{F}^{err,pair}_n \myp{S,J,\ell,\ell',k,s,\kappa} \big|_{\Phi_1 \to 1}}
		\leq{} \frac{C}{n!} \myp{c_0\lambda^2 s}^n, \label{eq:LeadingGraphs1} \\
		&\abs[\Big]{\mathcal{F}^{main,pair}_n \myp{S,\ell,k',k,t\lambda^{-2},\kappa} \big|_{\Phi_1\to 1} }
		\leq{} \1 \myp{k+k'=0} \frac{C}{\myp{n/2}!} \myp{c_0 t}^{n/2}. 
	\label{eq:LeadingGraphs2}
	\end{align}
\end{proposition}
\begin{proof}
	The argument is the same for the two bounds, so we write it in detail only for $\mathcal{F}^{main,pair}_n$.
	Recall that any leading graph with $n=2m$ interaction vertices is built up from exactly $m$ leading motives.
	We start from \eqref{Lemma:Qmain:3} and consider first the phase factors $\gamma_j$ given by \eqref{eq:gamma_j},
	\begin{equation*}
		\gamma_j = \sum_{l=j+1}^n \Theta_{l-1:\ell_l} \myp{\bar{k},\bar{\sigma}} - i \kappa_{n-j}.
	\end{equation*}
	Since the leading motives preserve the phase, meaning that $\Theta_{l-1:\ell_l} = -\Theta_{l:\ell_{l+1}} $ for all odd $l$, as discussed in \cref{rem:leading_motives}, the terms in the sum above cancel pairwise, implying for all $j$ that
	\begin{equation*}
		\mathrm{Re} \, \gamma_j
		={}	\begin{cases}
			0,						&	j \text{ even}, \\
			\Theta_{j-1:\ell_j},	&	j \text{ odd}.
		\end{cases}
	\end{equation*}
	Now, since there are two free edges for every leading motive in the graph, we can resolve the momentum constraints as usual and write, by abbreviating $U_j \coloneq U \myp{k_{j-1:\ell_j},\sigma_{j,\ell_j}}$ from \eqref{Def:Uinteraction},
	\begin{align*}
		\MoveEqLeft[4] \mathcal{F}^{main,pair}_{2m} \myp{S,\ell,k',k,t\lambda^{-2},\kappa} \big|_{\Phi_1\to 1} \nn \\
		={}& \myp{-1}^m \lambda^{2m} \epsilon \myp{S} 
		\int_{\myp{\T^d}^{2m}} \prod_{j=1}^m \ids k_{j,1} \ids {k_{j,2}} \prod_{j=1}^{2m} \myb[\big]{ \sigma_{j,\ell_j} U_j } \prod_{A = \Set{j,l} \in S} W^0 \myp{k_{0,j},\sigma_{0,j}} \nn \\
		&\times \int_{\myp{\R_+}^{I_{0,2m}}} \ids r \delta \myp[\bigg]{ \frac{t}{\lambda^2} - \sum_{j=0}^{2m} r_j } \prod_{j=0}^{2m} e^{-r_j \kappa_{2m-j}} \prod_{j=1}^{m} e^{-ir_{2j-1} \Theta_{2j}}.
	\end{align*}
	In this expression, $k_{j,1}$, $k_{j,2}$ denote the free momenta of the $j$'th degree two vertex, and the initial momenta $k_{0,j}$ as well as the factors $U_j$ and $\Theta_j$ may depend on on the top momenta $k,k'$.
	Then, changing the time variables from $r$ to $\myp{\tilde{s},s}$ with $\tilde{s}_i = \lambda^2 r_{2i}$ for $i = 0,1,\dotsc,m$, and $s_i = r_{2i-1}$ for $i = 1,\dotsc,m$, we can rewrite the time integral,
	\begin{align*}
		\MoveEqLeft[4] \int_{\myp{\R_+}^{I_{0,2m}}} \ids r \delta \myp[\bigg]{ \frac{t}{\lambda^2} - \sum_{j=0}^{2m} r_j } \prod_{j=0}^{2m} e^{-r_j \kappa_{2m-j}} \prod_{j=1}^{m} e^{-ir_{2j-1} \Theta_{2j}} \\
		={}& \lambda^{-2m} \int_{\R_+^m} \ids s \prod_{j=1}^{m} e^{-s_j \kappa_{2m-\myp{2j-1}} } \prod_{j=1}^{m} e^{-is_j \Theta_{2j}} \nn \\
		&\times \int_{\R_+^{m+1}} \ids \tilde{s} \delta \myp[\Big]{t - \lambda^2 \sum_{i=1}^m s_i -\sum_{i=0}^{m} \tilde{s}_i } \prod_{j=0}^{m} e^{-\lambda^{-2} \tilde{s}_j \kappa_{2m-2j} }.
	\end{align*}
	Now, evaluating the momentum integrals for each motive leads to iterations of $G_{s_i,\tau_i}$ as in \eqref{eq:Gdef}, yielding a function $\1 \myp{k+k'=0} \widetilde{G}^{\lambda} \myp{k,s;S,\ell}$, which by \cref{prop:Gbounds} has an $s$-integrable upper bound $G \myp{s;S,\ell}$, independent of $k$ and $\lambda$.
	Noting that the sign $\epsilon \myp{S}$ is absorbed into $\widetilde{G}^{\lambda}$, we can thus rewrite $\mathcal{F}^{main,pair}_{2m}$ as
	\begin{align}
		\MoveEqLeft[4] \mathcal{F}^{main,pair}_{2m} \myp{S,\ell,k',k,t \lambda^{-2},\kappa} \big|_{\Phi_1\to 1} \nn \\
		={}& \myp{-1}^m \1 \myp{k+k'=0} \int_{\R_+^m} \ids s \prod_{j=1}^{m} e^{-s_j \kappa_{2m-\myp{2j-1}} } \widetilde{G}^{\lambda} \myp{k,s;S,\ell} \nn \\
		&\times \int_{\R_+^{m+1}} \ids \tilde{s} \delta \myp[\Big]{t - \lambda^2 \sum_{i=1}^m s_i -\sum_{i=0}^{m} \tilde{s}_i } \prod_{j=0}^{m} e^{-\lambda^{-2} \tilde{s}_j \kappa_{2m-2j} }.
	\label{eq:LeadingGraphs4}
	\end{align}
	Finally, taking absolute values inside the integrals and simply estimating the $\kappa$-dependent factors by $1$, we obtain \eqref{eq:LeadingGraphs2} after noting that
	\begin{equation}
		\int_{\R_+^{m+1}} \ids \tilde{s} \delta \myp[\Big]{t - \lambda^2 \sum_{i=1}^m s_i -\sum_{i=0}^{m} \tilde{s}_i }
		={} \frac{1}{m!} \myp[\Big]{ t - \lambda^2 \sum_{i=1}^m s_i }^m \1 \myp[\Big]{0 \leq t - \lambda^2 \sum_{i=1}^m s_i},
	\label{eq:timesimplex}
	\end{equation}
	which is easily shown by induction.
	The constant $c_0$ comes from the integral of the upper bound on $G^{\lambda}$ from \eqref{eq:Gbound4}.
\end{proof}
\begin{proposition}[Nested graphs]
\label{Propo:NestedGraphs}
	There is a constant $c > 0$, depending on $\omega,V$, and a $C>0$ depending on $\omega$, such that the following estimates hold true for all relevant nested graphs, and for all $k, k' \in \T^d$,
	\begin{align}
		\abs[\Big]{ \mathcal{F}^{err,pair}_n \myp{S,J,\ell,\ell',k,s,\kappa}  \big|_{\Phi_1\to 1} }
		\leq{}& C \lambda^{\gamma'} e^{s\lambda^2} \expec{cs\lambda^2}^n \expec{\ln n} \expec{\ln\lambda}^{1+n}, \label{eq:NestedGraphs1} \\
		\abs[\Big]{ \mathcal{F}^{main,pair}_n \myp{S,\ell,k',k,t\lambda^{-2},\kappa}  \big|_{\Phi_1\to 1} }
		\leq{}& C \1 \myp{k+k'=0} \lambda^{\gamma'} e^t \expec{ct}^{n/2} \expec{\ln n} \expec{\ln\lambda}^{1+n/2}. \label{eq:NestedGraphs2}
	\end{align}
\end{proposition}
\begin{proof}
	Let $i_2$ and $j_0$ be as in \cref{lem:NestedGraphs}, that is, $v_{i_2}$ is the first degree two vertex from the bottom that does not belong to an immediate recollision, and $j_0$ is the index of the first time slice that is nested inside the double loop of $v_{i_2}$.
	Consider the contribution of the first $i_2$ time slices to the total phase in the time integral, i.e. $e^{-i \sum_{j=0}^{i_2-1} r_j \gamma \myp{j}}$.
	Recalling that the imaginary part of $\gamma \myp{j}$ is always non-positive just a sum of entries of $\kappa$, we can re-write
	\begin{equation}
	\label{eq:nested_phase}
		\sum_{i=0}^{i_2-1} r_i \gamma \myp{i}
		={} \zeta \myp{r} + \sum_{i=0}^{i_2-1} r_i \sum_{j=i+1}^{i_2} \Theta_j,
	\end{equation}
	where $\mathrm{Im} \zeta \myp{r} \leq 0$ does not depend on any free momenta ending before (or at) $v_{i_2}$.
	Denote by
	\begin{equation*}
		B_j \coloneq \Set{1 \leq i < i_2 \mid \deg v_i = j},
		\qquad j=0,2,
	\end{equation*}
	the indices of the vertices below $v_{i_2}$ with degree $j$.
	Then by choice of $v_{i_2}$, every $v_i$ with $i \in B_2$ ends and immediate recollision, implying that $i-1 \in B_0$ and $\Theta_{i-1} = - \Theta_i$.
	The remaining indices in $B_0$ that do not belong to an immediate recollision are
	\begin{equation*}
		B_0' \coloneq \Set{j \in B_0 \mid j+1 \notin B_2}.
	\end{equation*}
	Since both of the time slices $j_0-1$ and $j_0$ are long, we must have that $j_0 \in B_0'$.
	Continuing to rewrite the phase factor from before, we have
	\begin{align*}
		\sum_{i=0}^{i_2-1} r_i \sum_{j=i+1}^{i_2} \Theta_j
		= \sum_{j \in B_0'\setminus \Set{j_0}} \sum_{i=0}^{j-1} r_i \Theta_j + \myp{\Theta_{i_2}+\Theta_{j_0}} \sum_{i=0}^{j_0-1} r_i + \sum_{j \in B_2} \Theta_j r_{j-1} + \Theta_{i_2} \sum_{i=j_0}^{i_2-1} r_i,
	\end{align*}
	where by \cref{lem:NestedGraphs}, the terms in the first two sums do not depend on any free momenta ending before or at $v_{i_2}$.
	We conclude that the phase factor \eqref{eq:nested_phase} can be written
	\begin{equation}
	\label{eq:nested_phase2}
		\sum_{i=0}^{i_2-1} r_i \gamma \myp{i}
		= \tilde{\zeta} \myp{r} + \sum_{j \in B_2} \Theta_j r_{j-1} + \Theta_{i_2} \sum_{i=j_0}^{i_2-1} r_i,
	\end{equation}
	where $\tilde{\zeta} \myp{r}$ is independent of all free momenta ending before or at $v_{i_2}$ and $\mathrm{Im} \tilde{\zeta} \myp{r} \leq 0$.
	
	We can now follow the basic iterative scheme with a few modifications.
	The idea is to integrate out the free momenta of the vertices below $v_{i_2}$ before taking absolute values inside the remaining integrals.
	Recall from earlier the sets of time slice indices
	\begin{equation*}
		A_j \coloneq \Set{ 0 \leq i < N \mid \deg v_{i+1} = j }
	\end{equation*}
	for $j=0,2$, where now $\abs{A_j} = N/2$ since the graph is fully paired.
	To decompose the time integral, we now take
	\begin{equation*}
		A \coloneq \Set{N} \cup \Set{i_2 \leq i < N \mid \deg v_{i+1} = 2},
	\end{equation*}
	and
	\begin{equation*}
		A' \coloneq \Set{\ast} \cup I_{0,i_2-1} \cup A_0,
	\end{equation*}
	and write, using \cref{Lemma:PhaseResolve},
	\begin{align*}
		\MoveEqLeft[3] \int_{\myp{\R_+}^{I_{0,N}}} \ids r \delta \myp[\Big]{s - \sum_{j=0}^N r_j} \prod_{j=0}^N e^{-i r_j \gamma \myp{j}} \\
		={}& - \oint_{\Gamma_N} \frac{\ids z}{2 \pi} \prod_{j \in A} \frac{i}{z-\gamma \myp{j}} \int_{\myp{\R_+}^{A'}} \ids r \delta \myp[\Big]{s - \sum_{j \in A'} r_j} e^{-i r_{\ast} z} \prod_{j \in I_{0,i_2-1} \cup A_0} e^{-i r_j \gamma \myp{j}} \\
		={}& - \oint_{\Gamma_N} \frac{\ids z}{2 \pi} \frac{i}{z} \prod_{i_2 \leq j \in A_2} \frac{i}{z-\gamma \myp{j}}
		\begin{aligned}[t] 
			&\int_{\myp{\R_+}^{A'}} \ids r \delta \myp[\Big]{s - \sum_{j \in A'} r_j} \\
			&\times e^{-i r_{\ast} z} e^{-i \sum_{j=0}^{i_2-1} r_j \gamma \myp{j}} \prod_{i_2 \leq j \in A_0} e^{-i r_j \gamma \myp{j}}.
		\end{aligned}
	\end{align*}
	Now, inserting \eqref{eq:nested_phase2} into this expression, we may integrate out all the double loop momenta of the immediate recollision vertices below $v_{i_2}$.
	Each of these double loop integrals corresponds to a $G_{r,\tau}^{\lambda}$-factor as in \eqref{eq:Gdef}, which in total yields a factor $G_1^{\lambda} \myp{\bar{k},r}$ that depends only on the recollision time variables $r_j$ with $j \in \Set{i \mid i+1 \in B_2}$, and only on the free momenta of the vertices $v_{i_2}$ and above.
	Note that $G_1^{\lambda}$ is a product of factors of the form as in \cref{prop:Gbounds}, and not all of them have to depend on the free momenta of $v_{i_2}$.
	Taking absolute values inside all remaining integrals except for the double loop integral of $v_{i_2}$, and estimating $\abs{e^{-i r_{\ast} z}} \leq e^{s \lambda^2}$, $\abs{e^{-i \tilde{\zeta} \myp{r}}} \leq 1$, and $\abs{e^{-i r_j \gamma \myp{j}}} \leq 1$ for all remaining $j\in A_0$ above $v_{i_2}$, we arrive for instance at an expression of the form
	\begin{align}
		\abs{\mathcal{F}_n^{err,pairs}}
		\lesssim{}& \lambda^{N} \normt{\widehat{V}}{\infty}^{N-2 \myp{\abs{B_2}+1}} \int_{\myp{\T^d}^{N-2 \myp{\abs{B_2}+1} } } \prod_{i_2 < j \in A_2} \ids k_{j,1} \ids k_{j,2} \nn \\
		&\times \oint_{\Gamma_N} \ids z \frac{e^{\lambda^2s}}{\abs{z}} \prod_{i_2 \leq j \in A_2} \frac{1}{\abs{z-\gamma \myp{j}}} \int_{\myp{\R_+}^{A'}} \ids r \delta \myp[\Big]{s- \sum_{j\in A'} r_j} \nn \\
		&\times \abs[\bigg]{ \int_{\myp{\T^d}^2} \ids k_1 \ids k_2 e^{-i \Theta_{i_2} \sum_{i=j_0}^{i_2-1} r_i} U_{j_0} U_{i_2}  \myb[\Big]{\prod W^0} G_1^{\lambda} \myp{\bar{k},r} },
		\label{eq:nestedbound1}
	\end{align}
	where $k_1$, $k_2$ denote the free momenta of $v_{i_2}$, $\prod W^0$ contains the remaining factors of $W^0$ that might depend on $k_1,k_2$, and $U_{j_0} $, $ U_{i_2}$ are the interaction factors of $v_{j_0}$ and $v_{i_2}$.
	Since the double loop of $v_{i_2}$ corresponds to that of a leading motive by \cref{lem:NestedGraphs}, the interaction factor $U_{j_0} U_{i_2}$ is of the same form as the allowed factors in \eqref{eq:Gdef}, and hence the integral over $k_1,k_2$ also gives rise to a $G_{r,\tau}^{\lambda}$-factor, except it depends on a sum of time variables $\sum_{i=j_0}^{i_2-1} r_i$.
	The resulting factor, which we denote by $G_2^{\lambda} \myp{\bar{k},r}$, depends on the time variables corresponding to the set $B \coloneq I_{j_0,i_2-1} \cup \Set{i \mid i+1 \in B_2}$, but we note that \emph{only} the last application of $G_{r,\tau}^{\lambda}$, corresponding to the double loop of the nest, depends on any of the long time slices.
	
	We will estimate the integral of $G_2^{\lambda}$ using \cref{prop:Gbounds} by first integrating out the time slice $i_2-1$ along with all the long time slices in the nest of $v_{i_2}$.
	Since the time slice $j_0$ is long, we have $ M \coloneq \abs{\Set{i_2-1}\cup \myp{A_0 \cap I_{j_0,i_2-1}}} \geq 2$ time variables to integrate in this step, which allows us to use the dispersive estimate of \cref{ass:DR2} to obtain the extra decay in $\lambda$ that we need for the nested graphs.
	Afterwards, we iterate through the short time slice integrals going from the top of the nest and downwards.
	Applying \eqref{eq:Gbound4}, we thus have to estimate
	\begin{align}
		\MoveEqLeft[2] \int_{\myp{\R_+}^{A'}} \ids r \delta \myp[\Big]{s- \sum_{j\in A'} r_j} \abs[\big]{G_2^{\lambda} \myp{\bar{k},r}} \nn \\
		\leq{}& \int_{\myp{\R_+}^{A' \setminus B}} \ids r \delta \myp[\Big]{s- \sum_{j \in A' \setminus B} r_j} \int_{\myp{\R_+}^B} \ids r \abs[\big]{G_2^{\lambda} \myp{\bar{k},r}} \1 \myp[\Big]{\sum_{j \in B} r_j \leq s} \nn \\
		\leq{}& \frac{s^{\tilde{n}}}{\tilde{n}!} C^{\abs{B_2}+1} \normt{V}{1}^{2 \myp{\abs{B_2}+1}}
			\begin{aligned}[t]
				&\int_{\myp{\R_+}^B} \ids r  \1 \myp[\Big]{\sum_{j \in B} r_j \leq s} \prod_{i=1}^3 \normt[\big]{p_{\sum_{j=j_0}^{i_2-1} r_j + \sum_{\substack{ \scriptscriptstyle{j\in A_2} \\ \scriptscriptstyle{j<i_2-1} } } \tau_{i_2-1,i,j} r_j}^{\lambda}}{3} \\
				&\times \prod_{\substack{m\in A_2 \\ m < i_2-1}} \myp[\bigg]{ \1 \myp[\Big]{\sum_{\substack{ j\in A_2 \\ j<m }} r_j \leq s} \prod_{i=1}^3 \normt[\big]{p_{r_m + \sum_{\substack{ \scriptscriptstyle{j\in A_2} \\ \scriptscriptstyle{j<m} } } \tau_{m,i,j} r_j}^{\lambda}}{3} },
			\end{aligned}
		\label{eq:nestedbound2}
	\end{align}
	with $\tilde{n} = \abs{A'\setminus B} -1 = \abs{A_0 \setminus I_{j_0,i_2-1}}$ and $\tau_{m,i,j} \in \Set{-1,0,1}$.
	Using H\"older's inequality and \cref{ass:DR2}, the integral over $r_{i_2-1}$ and the long time slices is bounded by an integral of the type
	\begin{align*}
		\int_{\R_+^M} \ids r \1 \myp[\Big]{\sum_{i=1}^M r_i \leq s} \expec[\Big]{\sum_{i=1}^M r_i + \alpha}^{-1-\delta}
		={}& \int_0^s \ids t \frac{\abs[\big]{\Set{\sum_{i=1}^M r_i = t}\cap \R_+^M}}{\expec{t+\alpha}^{1+\delta}} \\
		={} \frac{\sqrt{M}}{\myp{M-1}!} \int_0^s \ids t \frac{t^{M-1}}{\expec{t+\alpha}^{1+\delta}}
		\leq{}& \frac{\sqrt{M} s^{M-2}}{\myp{M-1}!} \int_0^s \ids t \frac{t}{\expec{t+\alpha}^{1+\delta}},
	\end{align*}
	where
	\begin{equation*}
		-s \leq -\sum_{\substack{j \in A_2 \\ j < i_2-1}} r_j
		\leq \alpha
		\leq \sum_{\substack{j \in A_2 \\ j_0< j < i_2-1}} r_j + \sum_{\substack{j \in A_2 \\ j < i_2-1}} r_j
		\leq 2s.
	\end{equation*}
	Using $2 \delta \geq \gamma'$ and $\gamma' < 1$, we have
	\begin{equation}
	\label{eq:nestedbound3}
		\int_0^s \ids t \frac{t}{\expec{t+\alpha}^{1+\delta}}
		\leq{} \int_0^s \ids t \frac{1}{\expec{t+\alpha}^{\delta}} + \alpha_- \int_0^s \ids t \frac{1}{\expec{t+\alpha}^{1+\delta}}
		\leq{} C \myp{s^{1-\gamma'/2} + \alpha_-},
	\end{equation}
	which we now plug back into \eqref{eq:nestedbound2}.
	For the term involving $s^{1-\gamma'/2}$, we simply use that the remaining time integrals in \eqref{eq:nestedbound2} are bounded by \cref{prop:Gbounds}.
	For the term involving $\alpha_- \leq \sum_{j \in A_2 ; j < i_2-1} r_j$, we have to iterate through the time integrals of each of the remaining degree two vertices, from the top of the nest and down.
	If $m$ is the short time slice corresponding to the next vertex, there are now exactly three factors left in \eqref{eq:nestedbound2} that depend on $r_m$, as well as a contribution from $\alpha_-$.
	We save all other time variables appearing in $\alpha_-$ for later iterations.
	We estimate the $r_m$-integral by H\"older,
	\begin{align*}
		\int_0^s \ids r_m \, r_m \prod_{i=1}^3 \normt[\big]{p_{r_m + \sum_{\substack{ \scriptscriptstyle{j\in A_2} \\ \scriptscriptstyle{j<m} } } \tau_{m,i,j} r_j}^{\lambda}}{3}
		\leq{} \prod_{i=1}^3 \myp[\bigg]{ \int_0^s \ids r_m \, r_m \expec[\Big]{r_m + \sum_{\substack{ j\in A_2 \\ j<m } } \tau_{m,i,j} r_j }^{-1-\delta} }^{\frac{1}{3}},
	\end{align*}
	where each of the factors is of the same form as in \eqref{eq:nestedbound3} for some $\alpha_i$ with $\abs{\alpha_i} \leq \sum_{j\in A_2 ; j<m } r_j \leq s $, allowing us to continue the iteration in the next step.
	Noting that each iteration may add additional terms that depend on the time slices below (through the $\alpha_i$), and that there are at most $\abs{B_2 + 1} \leq \abs{A_2} = \abs{A_0}$ steps to be carried out, we find that \eqref{eq:nestedbound2} is bounded by
	\begin{align*}
		\int_{\myp{\R_+}^{A'}} \ids r \delta \myp[\Big]{s- \sum_{j\in A'} r_j} \abs[\big]{G_2^{\lambda} \myp{k,r}}
		\lesssim{}& C^{\abs{B_2}+1} \normt{V}{1}^{2 \myp{\abs{B_2}+1}} \frac{\sqrt{M} \abs{A_0}^2}{\tilde{n}! \myp{M-1}!} s^{\tilde{n}+ M-1 -\gamma'/2} \\
		\leq{}& 2 C^{\abs{B_2}+1} \normt{V}{1}^{2 \myp{\abs{B_2}+1}} s^{N/2 -\gamma'/2},
	\end{align*}
	since $\tilde{n} + M-1 = \abs{A_0 \setminus I_{j_0,i_2-1}} + \abs{A_0 \cap I_{j_0,i_2-1}} = \abs{A_0} = N/2$.
	
	Finally, going back to \eqref{eq:nestedbound1}, we may now iterate through the remaining degree two vertices above $v_{i_2}$ as in the scheme of the basic estimates.
	We find that the amplitude is bounded by
	\begin{align*}
		\abs{\mathcal{F}_n^{err,pairs}}
		\leq{}& \widetilde{C} C^{N/2} \normt{V}{1}^N \lambda^{\gamma'/2} e^{s\lambda^2} \myp{s\lambda^2}^{N/2-\gamma'/2} \expec{\ln N} \expec{\ln \lambda}^{1+N} \\
		\leq{}& \widetilde{C} \lambda^{\gamma'/2} e^{s\lambda^2} \expec{c s\lambda^2}^{N/2} \expec{\ln N} \expec{\ln \lambda}^{1+N},
	\end{align*}
	concluding the proof.
\end{proof}

\section{Completion of the proof of the main theorem}\label{sec:mainproof}
We are now able to finish the proof of \cref{thm:main}.
We need the following lemma providing a bound on the number of leading graphs, but we refer to \cite[Lemma 10.1]{LukkarinenSpohn:WNS:2011} for the proof.
\begin{lemma}
\label{Lemma:MainTheorem1}
	Consider a momentum graph with $n'$ interaction vertices in the minus tree and $n$ interaction vertices in the plus tree.
	If $n+n'$ is odd, the graph is not leading.
	If $n+n'$ is even, there are at most 
	\begin{equation}
		4^{n+n'} \myp{n+n'-1}!!
		\leq{} 8^{n+n'} \myp[\Big]{\frac{n+n'}{2}}!
	\end{equation}
	leading graphs, where we define $\myp{-1}!! = 1$. 
\end{lemma}
Note also that the total number of pairing graphs with $n$ interaction vertices is bounded by $5n\myp{2n}!$.
This follows easily from the facts that there are $\myp{2n}!! = 2^n n!$ ways of distributing the interaction vertices across the plus and minus trees, and there are $\frac{\myp{2n+2}!}{2^{n+1} \myp{n+1}!}$ ways of dividing the $2n+2$ initial vertices into sets of $n+1$ pairs.

Letting $c_0$ denote the constant from \cref{Propo:LeadingGraphs}, we choose $t_0 = \myp{2^6 c_0}^{-1}$.
Then, for any $0 < t < t_0$ and $0 < s \leq t \lambda^{-2}$, the series $\sum_{m=0}^{\infty} \myp{2^6 c_0 \lambda^2 s}^m$ is convergent.
\begin{lemma}
\label{Lemma:MainTheorem2}
	Let $t_0$ be the constant specified above, and assume $0 < t < t_0$.
	Then we have
	\begin{equation}
		\label{eq:MainTheorem2_1}
		\lim_{\lambda\to 0} \sup_{k,k' \in \T^d} \limsup_{L\to\infty} \abs[\Big]{\frac{1}{L^d} Q^{\lambda} \myp{\myb{k'},\myb{k},t} - Q^{main,pair} \myp{k',k,t} }
		={} 0,
	\end{equation}
	where $Q^{main,pair} \myp{k',k,t}$ is given by \eqref{Lemma:Qmain:2}.
	Moreover, we also have
	\begin{equation}
	\label{eq:MainTheorem2_2}
		\lim_{\lambda\to 0} \sup_{k,k' \in \T^d} \abs[\bigg]{ Q^{main,pair} \myp{k',k,t} - \sum_{\substack{n=0, \\ n \text{ even}}}^{N_0-1} \sum_{\substack{\text{leading} \\ \text{graphs}}} \mathcal{F}^{main,pair}_n \myp{S,\ell,k',k,t\lambda^{-2},\kappa} \big|_{\Phi_1\to 1} }
		={} 0. 
	\end{equation}
\end{lemma}
\begin{proof}
	Returning to \eqref{Def:OperatorQLambda2}, we have
	\begin{align*}
		\abs[\Big]{\frac{1}{L^d} Q^{\lambda} - Q^{main,pair}}
		\leq{} \abs[\Big]{\frac{1}{L^d} Q^{main} - Q^{main,pair}} + \frac{1}{L^d} \myp[\big]{\abs{Q^{err}_1} + \abs{Q^{err}_2} + \abs{Q^{err}_3}},
	\end{align*}
	where the first, third, and fourth terms vanish in the limit, by \cref{Lemma:Qmain}, \cref{Lemma:Q2err}, and \cref{Lemma:Q3err}, respectively.
	It also follows from \cref{Lemma:Q1err} that we only need to consider fully paired graphs in order to bound $\abs{Q^{err}_1}$.
	Now, considering for instance the nested graphs (of which there are at most $5N_0 \myp{2N_0}!$), we get by combining \cref{lem:RemoveCutoff} and \cref{Propo:NestedGraphs} that the total contribution from these is bounded uniformly for $k,k' \in \T^d$ by
	\begin{equation*}
		Ct^2 \expec{ct}^{N_0} \lambda^{\gamma'} N_0^{4+2b_0+1} \myp{2N_0}! \expec{\ln \lambda}^{2+N_0},
	\end{equation*}
	which vanishes as $\lambda \to 0^+$, in view of \eqref{eq:lambda_expression}.
	Similarly, applying \cref{Propo:CrossingGraphs} shows that the contribution from the crossing graphs also vanishes in the limit.
	For the remaining graphs, which are leading, we use \cref{Propo:LeadingGraphs} combined with \cref{lem:RemoveCutoff} (to remove the cut-off $\Phi_1$), as well as the fact that there are at most $2^{6n} n!$ such graphs, by \cref{Lemma:MainTheorem1}.
	The contribution coming from removing the cut-off is handled in the same way as the contributions from the nested and crossing graphs.
	The final term that we need to bound is thus 
	\begin{equation*}
		N_0^{2+2b_0} \sup_{\substack{0\leq s\leq t\lambda^{-2} \\ N_0/2 \leq n<N_0} } 2^{6n} \myp{c_0\lambda^2 s}^n
		\leq{} N_0^{2+2b_0} \sup_{N_0/2 \leq n<N_0} \myp[\Big]{\frac{t}{t_0}}^n
		\leq{} N_0^{2+2b_0} \myp[\Big]{\frac{t}{t_0}}^{\frac{N_0}{2}},
	\end{equation*}
	which tends to zero as $\lambda \to 0^+$.
	This shows that \eqref{eq:MainTheorem2_1} holds.
	
	The statement \eqref{eq:MainTheorem2_2} about $Q^{main,pair}$ is proved in the same way, starting from \cref{Lemma:Qmain}, and then using \cref{Propo:CrossingGraphs} and \cref{Propo:NestedGraphs} to bound the contributions from the crossing and nested graphs.
	Finally, applying \cref{lem:RemoveCutoff} allows us to remove the momentum cut-off $\Phi_1$ for the remaining leading graphs.
\end{proof}
Finally, the proof of \cref{thm:main} will be completed after computing the sum of the contributions from the leading diagrams in the weak-coupling limit.
Note that combining \cref{Lemma:MainTheorem1} with \cref{Propo:LeadingGraphs} shows that for any $k,k'\in \T^d$ and each even $n$,
\begin{equation*}
	\sum_{\text{leading graphs}} \abs[\Big]{ \mathcal{F}^{main,pair}_n \myp{S,\ell,k',k,t\lambda^{-2},\kappa} \big|_{\Phi_1\to 1} }
	\leq c \1 \myp{k+k'=0} \myp[\Big]{\frac{t}{t_0}}^{\frac{n}{2}},
\end{equation*}
implying for $t < t_0$ that we can move the $\lambda \to 0$ limit inside the sum over $n$ in \eqref{eq:MainTheorem2_2}.
After combining with the lemma below, this finishes the proof of \cref{thm:main}.
\begin{lemma}[Convergence of leading terms]
\label{Lemma:MainTheorem3}
	Take any even integer $n= 2m$, and let $\nu$ be as in \eqref{eq:nu-infty}.
	Then for any $k,k' \in \T^d$, we have
	\begin{align}
		\MoveEqLeft[10] \lim_{\lambda\to 0} \sum_{\text{leading graphs}} \mathcal{F}^{main,pair}_n \myp{S,\ell,k',k,t\lambda^{-2},\kappa} \big|_{\Phi_1\to 1} \nn \\
		={}& \1 \myp{k+k'=0} W^0 \myp{k} \nu \myp{k}^{n/2} \frac{\myp{-t}^{n/2}}{\myp{n/2}!}.
	\label{eq:MainTheorem3}
	\end{align}
\end{lemma}
\begin{proof}
	Take first $\lambda$ small enough so that $N_0 \myp{\lambda} > 4m$.
	This ensures that $\kappa_j=0$ for all $j$ appearing in the expression of $\mathcal{F}^{main,pair}_{2m}$, by the choice of parameters \eqref{Def:Para4}.
	Proceeding as in the proof of \eqref{eq:LeadingGraphs2}, we combine \eqref{eq:LeadingGraphs4} and \eqref{eq:timesimplex} to immediately obtain by dominated convergence,
	\begin{align}
		\MoveEqLeft[2] \mathcal{F}^{main,pair}_{2m} \myp{S,\ell,k',k,t \lambda^{-2},\kappa}\big|_{\Phi_1\to 1} \nn \\
		&={} \1 \myp{k+k'=0} \frac{\myp{-1}^m}{m!} \int_{\R_+^m} \ids s \widetilde{G}^{\lambda} \myp{k,s;S,\ell} \myp[\Big]{ t - \lambda^2 \sum_{i=1}^m s_i }^m \1 \myp[\Big]{\lambda^2 \sum_{i=1}^m s_i \leq t} \nn \\
		&\to{} \1 \myp{k+k'=0} \frac{\myp{-t}^m}{m!} \int_{\R_+^m} \ids s \widetilde{G}^0 \myp{k,s;S,\ell},
		\label{eq:graph_limit}
	\end{align}
	as $\lambda \to 0$.
	Recall here that $\widetilde{G}^{\lambda} \myp{k,s;S,\ell}$ is an $m$-fold iteration of $G_{s_i,\tau_i}$-factors as in \eqref{eq:Gdef}, and that $\widetilde{G}^{\lambda}$ has an $s$-integrable upper bound which is independent of $k$ and $\lambda$, by \cref{prop:Gbounds}.

	We will argue by induction that \eqref{eq:MainTheorem3} holds, or, more specifically, that
	\begin{equation}
	\label{eq:main_ind}
		\sum_{\substack{m-\text{motive} \\\text{leading graphs}}} \int_{\R_+^m} \ids s \widetilde{G}^0 \myp{k,s;S,\ell}
		={} W^0 \myp{k} \nu \myp{k}^m.
	\end{equation}
	For $m=0$, we have the amplitude for the trivial graph consisting of a single loop,
	\begin{equation*}
		\mathcal{F}_0^{main,pair} \myp{k',k,t\lambda^{-2}, \kappa}
		={} \1 \myp{k+k'=0} W^0 \myp{k} e^{-\frac{t}{\lambda^2} \kappa_0}
		={} \1 \myp{k+k'=0} W^0 \myp{k},
	\end{equation*}
	since in this case $\gamma_0 = -i \kappa_0 = 0$ by \eqref{eq:gamma_j} and our choice of $\kappa$ \eqref{Def:Para4}.
	
	Assume now that \eqref{eq:MainTheorem3} holds for $m-1$, and note that the set of leading graphs with $m$ motives is obtained by adding a single leading motive to a graph with $m-1$ motives in such a way that the resulting graph has no interaction vertices in the minus tree, and each $m$-motive leading graph can be obtained like this in exactly one way.
	This means that we can split the sum in \eqref{eq:MainTheorem3} into a double sum, where we first fix a graph with $m-1$ motives and sum over all the ways of adding a leading motive to this fixed graph, and then afterwards sum over all graphs with $m-1$ motives.
	
	Thus we fix a graph with $m-1$ motives and consider all the ways of adding a leading motive to its pairing clusters.
	Since the resulting graph should contain no interaction vertices in the minus tree, when adding a motive to the special pairing that connects the minus tree to the plus tree (whose left leg has parity $-1$), we can only use six of the loss motives, L1,\dots,L6 (recall the leading motives in \cref{fig:LeadingMotives}).
	
	For the remaining pairing clusters there are no restrictions, and we may use any of the $16$ leading motives suitable for any given cluster (recall that only half of the gain motives may be attached to any given pairing).
	We consider first the special pairing connecting the plus and minus trees and show that summing over the ways of adding a motive to this pairing results in the right hand side of \eqref{eq:MainTheorem3}.
	Finally, we will argue that the contributions from adding motives to all the other pairings exactly cancel each other.
	
	Thus, denoting by $k_0 \coloneq -k' = k$ the momentum corresponding to the right leg of the pairing connecting to the minus tree, and $\tilde{k} \coloneq \myp{k_1,k_2,k_3}$, we compute the contribution from adding the admissible loss motives using \cref{lemma:AddMotive}, which changes the factor $W^0 \myp{k_0}$ to
	\begin{align}
		\MoveEqLeft[4] -\int_{\R_+} \ids s_1 \int_{\myp{\T^d}^3} \ids \tilde{k} \delta \myp{k_0-k_1-k_2-k_3} \sum_{j=1}^6 F_{Lj} \myp{s_1,\tilde{k},-1} \nn \\
		={}& -\int_{\R_+} \ids s_1 \int_{\myp{\T^d}^3} \ids \tilde{k} \delta \myp{k_0-k_1-k_2-k_3} e^{-i s_1 \Theta \myp{\tilde{k},1}} W_0 \nn \\ 
		&\qquad\begin{aligned}[t]
			\times \myp[\Big]{& \widehat{V} \myp{k_1+k_2}^2 W_2 W_3 + \widehat{V} \myp{k_1+k_2} \widehat{V} \myp{k_1+k_3} W_1 W_3 \\
			&+\widehat{V} \myp{k_1+k_2}^2 W_1 \widetilde{W}_2 - \widehat{V} \myp{k_1+k_2} \widehat{V} \myp{k_1+k_3} W_2 W_3 \\
			&-\widehat{V} \myp{k_1+k_2}^2 W_1 W_3 - \widehat{V} \myp{k_1+k_2} \widehat{V} \myp{k_1+k_3} W_1 \widetilde{W}_2 }
		\end{aligned} \nn \\
		={}& -W_0 \int_{\R_+} \ids s_1 \int_{\myp{\T^d}^3} \ids \tilde{k} \delta \myp{k_0-k_1-k_2-k_3} e^{-i s_1 \Theta \myp{\tilde{k},1}} \nn \\
		&\qquad\times \widehat{V} \myp{k_1+k_2} \myp[\big]{ \widehat{V} \myp{k_1+k_2} - \widehat{V} \myp{k_1+k_3} } \myp[\big]{W_2 W_3 - W_1 W_3 + W_1 \widetilde{W}_2} 
		\label{eq:loss_sum} \\
		={}& -W^0 \myp{k_0} \nu \myp{k_0} \nn
	\end{align}
	where we have used that the change of variables $k_2 \leftrightarrow k_3$ leaves the phase $\Theta \myp{\tilde{k},1}$ invariant.
	Now summing over all graphs with $m-1$ leading motives, we obtain from \eqref{eq:graph_limit} and the induction hypothesis the total contribution
	\begin{align*}
		\MoveEqLeft[6] \lim_{\lambda\to 0} \sum_{\mathcal{G}_m} \mathcal{F}^{main,pair}_{2m} \myp{S,\ell,k',k,t\lambda^{-2},\kappa} \big|_{\Phi_1\to 1} \nn \\
		={}& \1 \myp{k+k'=0} \frac{\myp{-t}^m}{m!} \nu \myp{k} \sum_{\substack{\myp{m-1}-\text{motive} \\\text{leading graphs}}} \int_{\myp{\R_+}^{I_{m-1}}} \ids s \widetilde{G} \myp{k,s;S,\ell}  \\
		={}& \1 \myp{k+k'=0} W^0 \myp{k} \nu \myp{k}^m \frac{\myp{-t}^m}{m!},
	\end{align*}
	where $\mathcal{G}_m$ is the subset of the $m$-motive leading graphs obtained by attaching one of the six admissible loss motives to an $m-1$-motive graph at the pairing connecting to the minus tree.
	
	As mentioned above, we still need to show that the contributions from all other leading graphs cancel out, so fix a graph consisting of $m-1$ leading motives, as well as any pairing at the bottom of the graph which does not connect to the minus tree.
	Denote by $\sigma$ the parity of the left leg of the pairing, and by $k_0$ the momentum corresponding to this leg.
	As above, we want to sum up all the possible contributions from attaching a leading motive to this pairing, but this time there are $16$ possibilities.
	Assuming for simplicity that $\sigma=-1$, we collect first the terms from \cref{lemma:AddMotive} with the same phase and continue from \eqref{eq:loss_sum},
	\begin{align*}
		\MoveEqLeft[4] -\int_{\R_+} \ids r \int_{\myp{\T^d}^3} \ids \tilde{k} \delta \myp{k_0-k_1-k_2-k_3} \\
		&\times \myp[\bigg]{\sum_{j=1}^6 F_{Lj} \myp{r,\tilde{k},-1} + F_{G2^-} \myp{r,\tilde{k},-1} +F_{G4^-} \myp{r,\tilde{k},-1}} \\
		={}& -\int_{\R_+} \ids r \int_{\myp{\T^d}^3} \ids \tilde{k} \delta \myp{k_0-k_1-k_2-k_3} e^{-i r \Theta \myp{\tilde{k},1}} \\
		&\qquad\qquad \begin{aligned}[t]
					&\times \widehat{V} \myp{k_1+k_2} \myp[\big]{ \widehat{V} \myp{k_1+k_2} - \widehat{V} \myp{k_1+k_3} } \\
					&\times \myp[\big]{ \underbrace{ W_0 \myp[\big]{W_2 W_3 - W_1 W_3 + W_1 \widetilde{W}_2} -\widetilde{W}_1 W_2 W_3 }_{\mathrlap{\! = \widetilde{W}_0 \widetilde{W}_1 W_2 W_3 - W_0 W_1 \widetilde{W}_2 \widetilde{W}_3}} }
				\end{aligned} \\
		={}& \int_{\R_+} \ids r \int_{\myp{\T^d}^3} \ids \tilde{k} \delta \myp{k_0-k_1-k_2-k_3} e^{-i r \Theta \myp{\tilde{k},1}} \\
		&\qquad \times \frac{1}{2} \myp[\big]{ \widehat{V} \myp{k_1+k_2} - \widehat{V} \myp{k_1+k_3} }^2 \myp[\big]{\widetilde{W}_0 \widetilde{W}_1 W_2 W_3 - W_0 W_1 \widetilde{W}_2 \widetilde{W}_3 },
	\end{align*}
	where we in the last equality do a change of variables $k_2 \leftrightarrow k_3$ to complete the square.
	Similarly, we get the same expression for the sum of the remaining terms, but with a flipped sign in the phase of the exponential factor, giving the total contribution
	\begin{align}
		\MoveEqLeft[4] \frac{1}{2} \int_{\R_+} \ids r \int_{\myp{\T^d}^3} \ids \tilde{k} \delta \myp{k_0-k_1-k_2-k_3} \myp[\big]{ e^{-i r \Theta \myp{\tilde{k},1}} + e^{i r \Theta \myp{\tilde{k},1}} } \nn \\
		&\times \myp[\big]{ \widehat{V} \myp{k_1+k_2} - \widehat{V} \myp{k_1+k_3} }^2 \myp[\big]{\widetilde{W}_0 \widetilde{W}_1 W_2 W_3 - W_0 W_1 \widetilde{W}_2 \widetilde{W}_3 }.
	\label{eq:main_sum}
	\end{align}
	We recognise, at least formally, the Boltzmann-Nordheim collision operator \eqref{eq:BN-C}, so since $W^0$ is a stationary solution to the Boltzmann-Nordheim equation, this should vanish.
	Now, rewriting and using the explicit expression \eqref{eq:wignerfunction0} of $W^0$ with $\beta = 1/T$ denoting the inverse temperature,
	\begin{align*}
		\widetilde{W}_0 \widetilde{W}_1 W_2 W_3 - W_0 W_1 \widetilde{W}_2 \widetilde{W}_3
		={}& W_0 W_1 W_2 W_3 \\
		&\times \myp[\Big]{ \myp[\Big]{\frac{1}{W_0}-1} \myp[\Big]{\frac{1}{W_1}-1} - \myp[\Big]{\frac{1}{W_2}-1} \myp[\Big]{\frac{1}{W_3}-1} } \\
		={}& W_0 W_1 W_2 W_3 \myp[\big]{ e^{\beta \myp{\omega_0+\omega_1}} - e^{\beta \myp{\omega_2 + \omega_3}} } \\
		={}& W_0 W_1 W_2 W_3 e^{\beta \myp{\omega_0 + \omega_1}} \myp{ 1 - e^{\beta \Theta \myp{\tilde{k},1}} },
	\end{align*}
	so doing the time integral in \eqref{eq:main_sum} from $0$ to some $M>0$ gives
	\begin{align*}
		\MoveEqLeft[4] \lim_{M\to \infty} \frac{i \beta}{2} \int_{\myp{\T^d}^3} \ids \tilde{k} \delta \myp{k_0-k_1-k_2-k_3} \myp[\big]{ e^{-i M \Theta \myp{\tilde{k},1}} - e^{i M \Theta \myp{\tilde{k},1}} } \nn \\
		&\times \myp[\big]{ \widehat{V} \myp{k_1+k_2} - \widehat{V} \myp{k_1+k_3} }^2 W_0 W_1 W_2 W_3 e^{\beta \myp{\omega_0 + \omega_1}} \frac{ 1 - e^{\beta \Theta \myp{\tilde{k},1}} }{\beta \Theta \myp{\tilde{k},1}}.
	\end{align*}
	With our assumptions on $V$ and $\omega$, the integrand (the expression on the second line above) is smooth and bounded on $\myp{\T^d}^3$, so it is the Fourier transform of a function in $\ell^1 \myp{\myp{\Z^d}^3}$.
	We conclude by \cref{prop:osc_bound} that the limit above is zero, so the contributions coming from attaching a leading motive to our fixed pairing (assuming that $\sigma = -1$) cancel each other.
	The argument when $\sigma = 1$ is exactly the same, so we conclude that \eqref{eq:main_ind} holds, and combining with \eqref{eq:graph_limit} finishes the proof of \eqref{eq:MainTheorem3}, and hence also of \cref{thm:main}.
\end{proof}

\bibliographystyle{plain}

\bibliography{QuantumBoltzmann}

\def\cprime{$'$}
\begin{thebibliography}{10}

\bibitem{balescu1975equilibrium}
R.~Balescu.
\newblock Equilibrium and nonequilibrium statistical mechanics.
\newblock {\em NASA STI Recon Technical Report A}, 76, 1975.

\bibitem{Benedetto-2005}
D.~Bendetto, F.~Castella, R.~Esposito, and M.~Pulvirenti.
\newblock On the weak-coupling limit for bosons and fermions.
\newblock {\em Mathematical Models and Methods in Applied Sciences},
  15(12):1811--1843, 2005.

\bibitem{Benedetto-2004}
D.~Benedetto, F.~Castella, R.~Esposito, and M.~Pulvirenti.
\newblock Some considerations on the derivation of the nonlinear quantum
  {B}oltzmann equation.
\newblock {\em Journal of Statistical Physics}, 116:381--410, 2004.

\bibitem{Benedetto-2006}
D.~Benedetto, F.~Castella, R.~Esposito, and M.~Pulvirenti.
\newblock Some considerations on the derivation of the nonlinear quantum
  {B}oltzmann equation {II}: The low density regime.
\newblock {\em Journal of Statistical Physics}, 124:951--996, 2006.

\bibitem{Benedetto-2007}
D.~Benedetto, F.~Castella, R.~Esposito, and M.~Pulvirenti.
\newblock A short review on the derivation of the nonlinear quantum {B}oltzmann
  equations.
\newblock {\em Communications in Mathematical Sciences}, 5(S1):55--71, 2007.

\bibitem{CarChe-26}
E.~C{\'a}rdenas and T.~Chen.
\newblock Quantum boltzmann dynamics and bosonized particle-hole interactions
  in fermion gases.
\newblock {\em Journal of Statistical Physics}, 193(1):1, 2026.

\bibitem{Deng-26-ICM}
Y.~Deng.
\newblock Hilbert's sixth problem: Particles and waves.
\newblock In {\em Proceedings of the International Congress of Mathematicians
  2026 - Volume 5: Invited Lectures (Sections 9-11)}, pages 264--284, 2026.

\bibitem{DenHan-23pre2}
Y.~Deng and Z.~Hani.
\newblock Derivation of the wave kinetic equation: full range of scaling laws.
\newblock To appear in \emph{Memoirs of the American Math. Society}. arXiv
  e-print, \href{https://arxiv.org/abs/2301.07063}{arXiv:2301.07063}, 2023.

\bibitem{DenHan-23}
Y.~{Deng} and Z.~{Hani}.
\newblock {Full derivation of the wave kinetic equation}.
\newblock {\em Inventiones Mathematicae}, 233(2):543--724, 2023.

\bibitem{DenHan-23pre}
Y.~Deng and Z.~Hani.
\newblock Long time justification of wave turbulence theory.
\newblock arXiv e-print,
  \href{https://arxiv.org/abs/2311.10082}{arXiv:2311.10082}, 2023.

\bibitem{DenHanMa-25pre}
Y.~Deng, Z.~Hani, and X.~Ma.
\newblock {H}ilbert's sixth problem: derivation of fluid equations via
  {B}oltzmann's kinetic theory.
\newblock arXiv e-print:
  \href{https://arxiv.org/abs/2503.01800}{arXiv:2503.01800}, 2025.

\bibitem{DenHanMa-24pre}
Y.~Deng, Z.~Hani, and X.~Ma.
\newblock Long time derivation of the {B}oltzmann equation from hard sphere
  dynamics.
\newblock To appear in \emph{Annals of Math}. arXiv e-print:
  \href{https://arxiv.org/abs/2408.07818}{arXiv:2408.07818}, 2025.

\bibitem{ErdSalYau-04}
L.~Erd\"{o}s, M.~Salmhofer, and H.-T. Yau.
\newblock On the quantum boltzmann equation.
\newblock {\em Journal of Statistical Physics}, 116:367--380, 08 2004.

\bibitem{erdHos2008quantum}
L.~Erd{\H{o}}s, M.~Salmhofer, and H.-T. Yau.
\newblock Quantum diffusion of the random schr{\"o}dinger evolution in the
  scaling limit.
\newblock {\em Acta mathematica}, 200(2):211--277, 2008.

\bibitem{Grafakos-2014}
L.~Grafakos.
\newblock {\em Classical Fourier analysis}.
\newblock Graduate Texts in Mathematics. Springer New York, NY, 2014.

\bibitem{hugenholtz1983derivation}
N.~M. Hugenholtz.
\newblock Derivation of the boltzmann equation for a fermi gas.
\newblock {\em Journal of statistical physics}, 32(2):231--254, 1983.

\bibitem{lanford1975time}
O.~E. Lanford.
\newblock Time evolution of large classical systems.
\newblock In {\em Dynamical systems, theory and applications}, pages 1--111.
  Springer, 1975.

\bibitem{LukMeiSpo-15}
J.~Lukkarinen, P.~Mei, and H.~Spohn.
\newblock Global well-posedness of the spatially homogeneous hubbard-boltzmann
  equation.
\newblock {\em Communications on Pure and Applied Mathematics}, 68(5):758--807,
  2015.

\bibitem{LukPirVuo-26pre}
J.~Lukkarinen, S.~Pirnes, and A.~Vuoksenmaa.
\newblock Finite lattice kinetic equations for bosons, fermions, and discrete
  {NLS}.
\newblock arXiv e-print,
  \href{https://arxiv.org/abs/2601.10486}{arXiv:2601.10486}, 2026.

\bibitem{LukkarinenSpohn:KLW:2007}
J.~Lukkarinen and H.~Spohn.
\newblock Kinetic limit for wave propagation in a random medium.
\newblock {\em Arch. Ration. Mech. Anal.}, 183(1):93--162, 2007.

\bibitem{lukkarinen2009not}
J.~Lukkarinen and H.~Spohn.
\newblock Not to normal order-notes on the kinetic limit for weakly interacting
  quantum fluids.
\newblock {\em Journal of Statistical Physics}, 134(5-6):1133--1172, 2009.

\bibitem{LukkarinenSpohn:WNS:2011}
J.~Lukkarinen and H.~Spohn.
\newblock Weakly nonlinear {S}chr\"odinger equation with random initial data.
\newblock {\em Invent. Math.}, 183(1):79--188, 2011.

\bibitem{Nordheim}
L.W. Nordheim.
\newblock Transport phenomena in {E}instein-{B}ose and {F}ermi-{D}irac gases.
\newblock {\em Proc. Roy. Soc. London A}, 119:689, 1928.

\bibitem{prigogine1962prigogine}
I.~Prigogine.
\newblock Non-equilibrium statistical mechanics.
\newblock {\em Interscience Publ., New York}, 1962.

\bibitem{Salmhofer-09}
M.~Salmhofer.
\newblock Clustering of fermionic truncated expectation values via functional
  integration.
\newblock {\em Journal of Statistical Physics}, 134:941--952, 2009.

\bibitem{Solovej-2014notes}
J.~P. Solovej.
\newblock Many body quantum mechanics.
\newblock Lecture notes at the Erwin Schroedinger Institute 2014, available
  online at
  \url{https://web.math.ku.dk/solovej/MANYBODY/mbnotes-ptn-5-3-14.pdf}.

\bibitem{UehlingUhlenbeck:TPI:1933}
E.~A. Uehling and G.~E. Uhlenbeck.
\newblock Transport phenomena in {E}instein-{B}ose and {F}ermi-{D}irac gases.
\newblock {\em Phys. Rev.}, 43:552--561, 1933.

\bibitem{van1954quantum}
L.~Van~Hove.
\newblock Quantum-mechanical perturbations giving rise to a statistical
  transport equation.
\newblock {\em Physica}, 21(1-5):517--540, 1954.

\end{thebibliography}

\end{document}